\documentclass[letterpaper,12pt,oneside]{book}

\usepackage{latexsym}
\usepackage{amsmath,amssymb,amsfonts}
\usepackage{amsthm}
\usepackage[T1]{fontenc} \usepackage{newtxtext}\usepackage[varg]{newpxmath} 

\usepackage[top=25truemm,bottom=18truemm,left=12truemm,right=62truemm]{geometry}
\renewcommand{\baselinestretch}{0.95}

\usepackage{eucal} %for mathcal
\usepackage{bxpapersize}
\usepackage{enumitem} %for list 
\usepackage[pdftex]{graphicx} %for submission
\graphicspath{{./figure/}}
\usepackage{subcaption} %caption for minipage
\usepackage{simplewick} %for wick's theorem
\usepackage{pb-diagram} 
\usepackage{arydshln} %for dashed line in Table
\usepackage{here} % to put figure here
\usepackage{makeidx} %for index
\usepackage{fancybox}
\usepackage{dashbox}
\usepackage{mathrsfs} %for mathscr
\usepackage{accents} % for widebar

\usepackage{wrapfig} % figure in text

\usepackage{xtab} % for a long table over multiple pages
\usepackage{marginfix} %comment on the right
\usepackage{mycc} % item numbering unification

\usepackage[sf,bf,topmarks,calcwidth,pagestyles]{titlesec}
\usepackage{titletoc}

\usepackage[
bookmarks=false,
colorlinks=true,
linkcolor=blue,
citecolor=red,
]{hyperref}

\makeindex
\usepackage{fancyhdr} 
\usepackage[absolute,overlay]{textpos}
\textblockorigin{20mm}{30mm} %origin

\newtheorem{theorem}{Theorem}[chapter]
\newtheorem{definition}[theorem]{Definition}
\newtheorem{corollary}[theorem]{Corollary}
\newtheorem{lemma}[theorem]{Lemma}
\newtheorem{remark}[theorem]{Remark}

\newtheorem{example}[theorem]{Example}

\renewcommand{\theequation}{\arabic{chapter}.\arabic{equation}}
\renewcommand{\thefigure}{\arabic{chapter}.\arabic{figure}}
\renewcommand\proofname{Proof}

\def\qed{\hfill $\Box$}

\newcommand{\A}[4]{
  \left[
    \begin{array}{cc}     %\A{a11}{a12}
        #1 & #2 \\        %     {a21}{a22}
        #3 & #4
    \end{array}
  \right]}

\newcommand{\AH}[2]{
  \left[
    \begin{array}{cc}    %\AH{a11}{a12}
        #1 & #2          
    \end{array}
  \right]}

\newcommand{\AV}[2]{
  \left[
    \begin{array}{c}     %\AV{a11}
        #1 \\            %      {a21}
        #2
    \end{array}
  \right]}

\newcommand{\AVl}[2]{
  \left[
    \begin{array}{l}     %\AV{a11}
        #1 \\            %      {a21}
        #2
    \end{array}
  \right]}

\newcommand{\dtf}[4]{
  \left[
\renewcommand{\arraystretch}{1.2}
  \begin{array}{c|c}
      #1 & #2 \\ \hline
      #3 & #4
  \end{array}
\renewcommand{\arraystretch}{1}
  \right]}

\newcommand{\nA}[4]{
    \begin{array}{cc}     %\macA{a11}{a12}
        #1 & #2 \\        %     {a21}{a22}
        #3 & #4
    \end{array}}

\newcommand{\nAV}[2]{
    \begin{array}{c}     %\macAV{a11}
        #1 \\            %      {a21}
        #2
    \end{array}}

\newcommand{\nAH}[2]{
    \begin{array}{cc}    %\macAH{a11}{a12}
        #1 & #2          
    \end{array}}

\newcommand{\B}[9]{
  \left[
    \begin{array}{ccc}    %\B{a11}{a12}{a13}
        #1 & #2 & #3 \\   %     {a21}{a22}{a23}
        #4 & #5 & #6 \\   %     {a31}{a32}{a33}
        #7 & #8 & #9
    \end{array}
  \right]}

\newcommand{\BV}[3]{
  \left[
    \begin{array}{c}      %\BV{a11}
        #1 \\             %   {a21}
        #2 \\             %   {a31}
        #3 
    \end{array}
  \right]}

\newcommand{\BH}[3]{
  \left[
    \begin{array}{ccc}    %\BH{a11}{a12}{a13}
        #1 & #2 & #3 
    \end{array}
  \right]}

\newcommand{\kakkon}[2]{
  \left\{
    \begin{array}{l}    
        #1 \\   
        #2
    \end{array}
  \right.}

\newcommand{\maths}[1]{\mathcal{#1}}
\newcommand{\im}{{\rm i}}
\newcommand{\ex}{\, {\rm e}}
\newcommand{\uu}[1]{^{\, \scalebox{0.8}{$#1$}}}
\newcommand{\dd}[1]{_{\scalebox{0.8}{$#1$}}}
\newcommand{\urm}[1]{^{\scalebox{0.8}{\textrm{#1}}}}
\newcommand{\drm}[1]{_{\scalebox{0.8}{\textrm{#1}}}}
\newcommand{\mm}[1]{{\ooalign{\(#1\)\cr\hidewidth\(\diagdown\)\hidewidth\cr}}}
\newcommand{\trans}{^{\scalebox{0.62}{\textrm{T}}}}
\newcommand{\transinv}{^{-\scalebox{0.62}{\textrm{T}}}}
\newcommand{\inv}{^{\scalebox{0.65}{$-1$}}}
\newcommand{\ga}{\maths{A}} %gauge field
\newcommand{\gel}{A}  %coefficient of gauge field
\newcommand{\win}{\theta} %weight function
\newcommand{\step}{\bm{1}} %step function
\newcommand{\lie}{G} %lie algebra
\newcommand{\qu}{q} %quadrature vector
\newcommand{\toqu}{R} %quadrature vector
\newcommand{\pole}{\maths{P}}
\newcommand{\zero}{\maths{Z}}
\newcommand{\tf}{P}  %transfer function
\newcommand{\pg}{Y} %propagator
\newcommand{\ipg}{\im \pg} % i propagator
\newcommand{\sel}{K}  %self-energy
\newcommand{\isel}{\im K}
\newcommand{\mas}{M}  %system variable
\newcommand{\ves}{v}  %vacuum expectation
\newcommand{\lag}{\maths{L}} %lagrangian
\newcommand{\ham}{\maths{H}} %hamiltonian
\newcommand{\hil}{\mathscr{H}} %hilbert space
\newcommand{\uni}{\mathrm{U}} %hilbert space

\newcommand{\bm}[1]{\mbox{\boldmath$#1$}}
\newcommand{\ffrac}[2]{{\displaystyle \frac{#1}{#2}}}
\newcommand{\intt}{{\displaystyle \int}}

\newcommand{\mzero}[1]{\langle 0|{#1}|0\rangle}

\newcommand{\bra}[1]{\langle{#1}|}
\newcommand{\ket}[1]{|{#1}\rangle}
\newcommand{\dgg}{^{\dag}}
\newcommand{\ddgg}{^{\ddagger}}
\newcommand{\simm}{^{\scalebox{0.8}{$\sim$}}}
\newcommand{\nn}{\nonumber}
\newcommand{\Tr}{{\rm Tr} \ }
\newcommand{\diag}{\textrm{diag} }
\newcommand{\pTr}[1]{{\rm Tr}\dd{#1} \ }
\newcommand{\rea}{{\rm Re}}
\newcommand{\ima}{{\rm Im}}

\newcommand{\T}{{\rm T}}

\newcommand{\hc}{{\rm h.c.}}
\newcommand{\chain}[1]{\maths{C}(#1)}
\newcommand{\dchain}[1]{\widetilde{\maths{C}}(#1)}
\newcommand{\homo}{\maths{F}}
\newcommand{\dhomo}{\widetilde{\maths{F}}}

\newcommand{\be}[1]{\begin{enumerate}[#1]}
\newcommand{\bee}{\begin{enumerate}}
\newcommand{\bl}{\begin{list}{}{}}
\newcommand{\el}{\end{list}}
\newcommand{\ee}{\end{enumerate}}
\newcommand{\fsh}[1]{{\ooalign{\(#1\)\cr\hidewidth\(/\)\hidewidth\cr}}}
\newcommand{\wick}[2]{ \bm{\bigl[} #1  \, \bm{\bigl|} \,  #2 \bm{\bigr]} } 
\newcommand{\en}[1]{\small$ #1 $\normalsize}

\makeatletter
  \def\widebar{\accentset{{\cc@style\underline{\mskip10mu}}}}
  \def\Widebar{\accentset{{\cc@style\underline{\mskip8mu}}}}
\makeatother

\begin{document}
\begin{titlepage}
\vspace*{18mm}

%\vspace*{-10mm}
\hspace{7mm}
{\Large \textsf{Yanagisawa M.}} %\hspace{3mm} \textsf{yanagi(at)cds.caltech.edu}

\vspace{5mm}
\begin{picture}(500,10)(60,0)
%(yoko tate)
\put(77,5){\line(1,0){445}}
\end{picture}

\vspace{25mm}

\hspace{7mm}
\scalebox{1}[1.1]{
{\LARGE{\textbf{\textsf{Systems Theory of Classical and Quantum Fields}}}}
}

\vspace{3mm}
\hspace{79mm}
\scalebox{1}[1.1]{
{\LARGE{\textbf{\textsf{and}}}}
}

\vspace{3mm}
\hspace{8mm}
\scalebox{1}[1.1]{
{\LARGE{\textbf{\textsf{Applications to Quantum Computing and Control}}}}
}

\vspace{45mm}
\begin{figure}[H]
\hspace{-5mm}
\includegraphics[keepaspectratio,width=200mm]{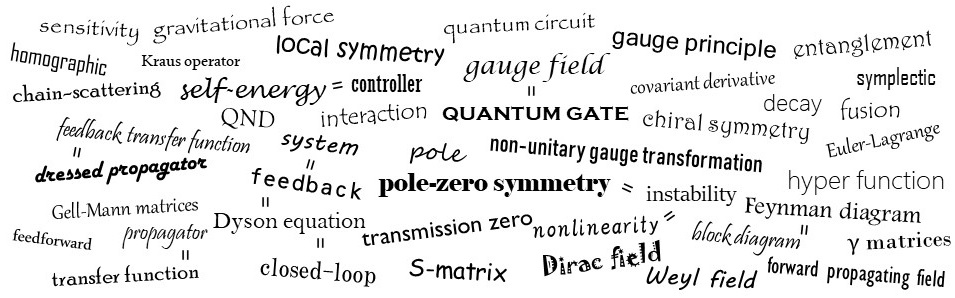} 
\end{figure}

\newpage

\thispagestyle{empty}
\vspace*{18mm}

\hspace{7mm}
{\Large \textsf{Yanagisawa M.}}\\

\hspace{7mm}
{\large \textsf{yanagi(at)cds.caltech.edu}}

%%%%%%%%%%%%%%%%%%%%%%%

%\begin{textblock}{width}[baseposition](x,y)
\begin{textblock}{6.5}[0,0](9,17)
\noindent
Special thanks to Sachiko Yanagisawa and Akihito Osawa.
This book wouldn't have been possible without their support
during this difficult time of the COVID-19 pandemic.
\hspace{7mm} Aug. 2020

\end{textblock}

\end{titlepage}

\clearpage
%\baselineskip 4ex
%%%%%%%%%%%%%%%%%%%%%%%%%%%%%%%%%%%%%%%%%%%%%%%%%%%%%%%%%%%%%%%%%%%%%%%
%\pagestyle{empty}

\pagenumbering{roman}

%\addcontentsline{toc}{chapter}{Abstract}
%\input{abs.tex}
%\input{ack.tex}
%\addcontentsline{toc}{chapter}{Acknowledgements}

\tableofcontents
\clearpage
\setcounter{page}{0}
\pagenumbering{arabic}

\pagestyle{fancy}
\renewcommand{\chaptermark}[1]{\markboth{\MakeUppercase{\chaptername}\ 
\thechapter:\ #1}{}}
\renewcommand{\sectionmark}[1]{\markright{\thesection\ \hspace{.8mm} #1}}

\chapter{Outline: Basic ideas and notation}
\label{chap:1}
\thispagestyle{fancy}

\section{Basic ideas}

We explore a field theoretical approach to 
quantum computing and control in this book.
Readers are not required to have a background of 
systems theory, field theory and related applications 
such as quantum computing beyond the elementary level.
Some basic preliminaries are reviewed in 
Chapters \ref{ctqm} - \ref{chap:gauge}.
In this chapter, 
we briefly explain the ideas of this book,
especially 
a relationship between quantum gates and gauge transformations.
Let us start with the following question:

\subsection{What is System?}
\label{sec:what}

The word \textit{system} is often used without a clear definition.
This sometimes causes confusion and misunderstanding,
especially when we talk about feedback.
Let us see a couple of examples
to think about this problem.
\index{system}

\begin{wrapfigure}[0]{r}[53mm]{49mm}
\centering
\vspace{-3mm}
\includegraphics[keepaspectratio,width=35mm]{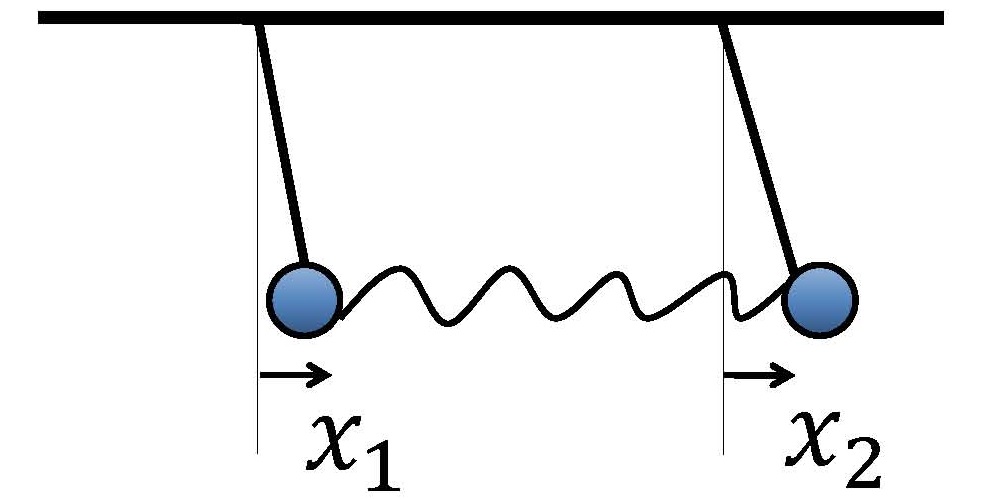} 
\caption{
\small
Two pendulums connected by a spring.
\normalsize
}
\label{fig-pendulums}
\end{wrapfigure}

Consider two identical pendulums connected by a spring
as in Figure \ref{fig-pendulums}.
What do we call a system in this case?
The definition of a system depends on
what we focus on for equations of motion.
We are usually interested in the behavior of the two weights,
so it is reasonable to regard each weight as a system
and write equations of motion for \en{x_1, x_2} as
\small\begin{align}
 \frac{d}{dt}  \AV{x_1}{\dot{x}_1}
&=
 \cdots, %f_1(x_1,\dot{x}_1,x_2,\dot{x}_2),
\hspace{9mm}
 \frac{d}{dt}  \AV{x_2}{\dot{x}_2}
=
 \cdots. %f_2(x_1,\dot{x}_1,x_2,\dot{x}_2).
\label{pen-el}
\end{align}\normalsize

There is another description of the coupled pendulums.
The two weights are efficiently described by new variables
\small\begin{align}
 x
&\equiv
 x_2+x_1,
\hspace{13mm}
 X 
\equiv
 x_2-x_1.
\end{align}\normalsize
In this case, 
(\ref{pen-el}) is decoupled as
\small\begin{align}
 \frac{d}{dt}\AV{x}{\dot{x}}
&=
 f(x,\dot{x}),
\hspace{9mm}
 \frac{d}{dt}\AV{X}{\dot{X}}
=
 F(X,\dot{X}).
\label{de-pen}
\end{align}\normalsize
It is well known that the eigenfrequency of \en{ x } 
is the same as a single pendulum, 
hence \en{ x } describes oscillations without coupling.
The effect of the spring appears in \en{ X }.
Obviously, 
we are more interested in the second equation of (\ref{de-pen}),
i.e.,
the interaction between the two weights.
In this case, 
the interaction is regarded as a system.

\newpage

What we call a system
varies with
what we focus on for equations of motion.
Basically the first and second descriptions are equivalent.
However, the second one is more useful
when we consider applications to quantum computing and feedback.
Next, we explore this perspective in more detail.

\subsection{What is System? (from a field theoretical perspective)}

Let us consider electron-electron scattering:
\vspace{0mm}
\begin{figure}[H]
\centering 
\includegraphics[keepaspectratio,width=45mm]{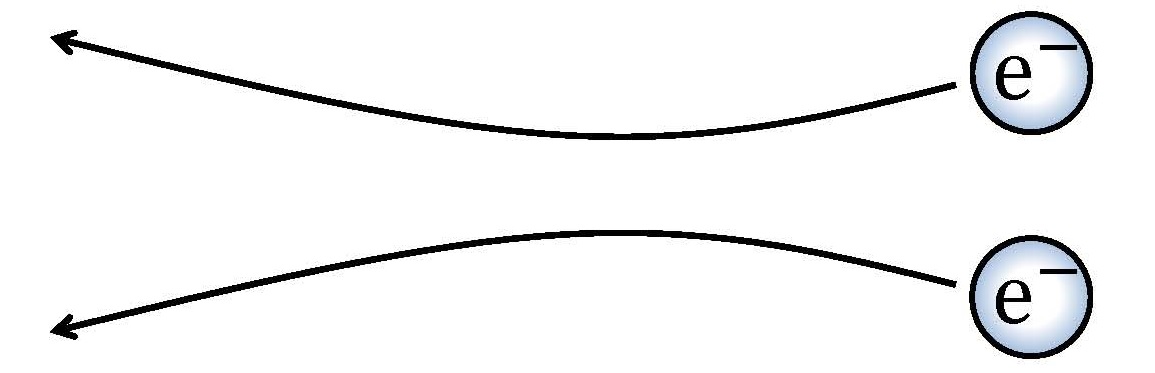} 
\end{figure}
\vspace{-1mm}

\noindent
What do we call a system in this picture?
Here two electrons propagating in free space come together, 
interact with each other 
and propagate away into free space again.
For weakly interacting electrons,
this scattering process is similar to the coupled pendulums.
If we follow the first interpretation of the coupled pendulums,
each electron is regarded as a system.

The second interpretation is different.
This can be easily seen from 
the corresponding Feynman diagram 
(the \textit{t}-channel process):\index{Feynman diagram}

\vspace{0mm}
\begin{figure}[H]
\centering
\includegraphics[keepaspectratio,width=55mm]{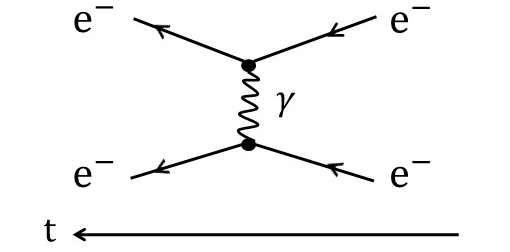}
\end{figure}
\vspace{-1mm}

\noindent
In this diagram, 
the interaction is mediated by the photon represented by the wiggly line,
which corresponds to the spring of the coupled pendulums.
A main purpose of drawing Feynman diagrams 
is to calculate a propagator from the incoming electrons 
to the outgoing ones.
Once we obtain the propagator,
the Feynman diagram can be expressed as follows:

\vspace{-2mm}
\begin{figure}[H]
\centering
\includegraphics[keepaspectratio,width=52mm]{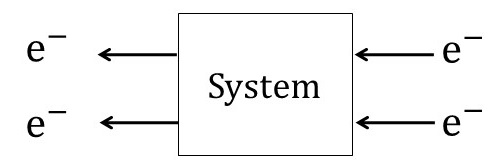}
\end{figure}
\vspace{0mm}

\noindent
In systems theory, 
this is called a block diagram\index{block diagram}
in which
the box represents a \textit{system}
and the two incoming and outgoing arrows are
the input and the output \textit{signals} of the system, 
respectively.

Now it is not difficult to tell 
what is a system in the scattering process
by comparing the Feynman and block diagrams.
A system is the interaction (the photon),
whereas the two electrons are (input/output) signals.\index{signal}
The difference between the system and signals is 
\textit{locality}.\index{locality}
The electrons are described by fields in free space.
On the other hand,
the interaction (the photon) results from 
a local symmetry of the electrons 
according to a gauge theory.
In this interpretation,
a system is defined by a local gauge transformation.
\index{gauge transformation (local)}

\subsection{Gauge theory and quantum gates}

The consideration above 
leads us to another application of the gauge theory.
For instance, 
quantum computing is implemented by successive applications
of \textit{local} operations at different points in spacetime.
This can be described by local gauge transformations.
As an example,
let us consider the concatenation of two QND (quantum non-demolition) 
SUM gates:\index{QND gate}\index{non-demolition}

\vspace{-1mm}
\begin{figure}[H]
\centering
\includegraphics[keepaspectratio,width=53mm]{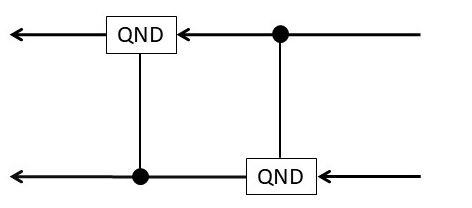}
\end{figure}
\vspace{-2mm}

\noindent
The two qubits/qumodes (the upper and lower lines)
correspond to the electrons 
in the example of electron-electron scattering.
According to the first interpretation,
each qubit is regarded as a system as follows:

\vspace{-1mm}
\begin{figure}[H]
\centering
\includegraphics[keepaspectratio,width=52mm]{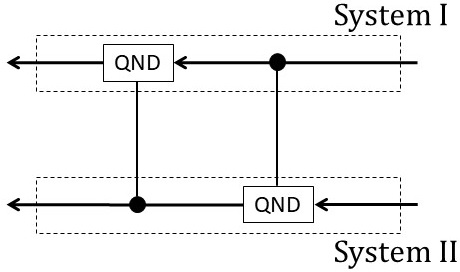}
\end{figure}
\vspace{-1mm}

\noindent
However, our approach is different.
The qubits/qumodes are signals and 
each QND gate is regarded as a system:

\vspace{0mm}
\begin{figure}[H]
\centering
\includegraphics[keepaspectratio,width=53mm]{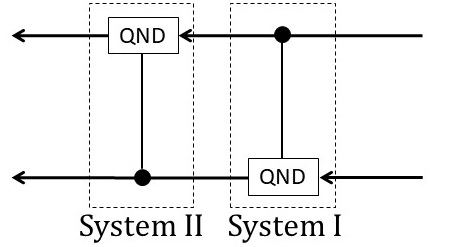}
\end{figure}
\vspace{-0mm}

This idea provides a new way of looking at quantum computing.
In a conventional approach,
a unitary operator (or a Hamiltonian)
is a starting point to describe a computational process.
In our approach, 
an inverse problem is considered,
i.e.,
given a computing process,
we regard it as a local gauge transformation
and 
derive its Lagrangian using a gauge theory.
The advantage of this approach is in 
the formulation of feedback and analysis of nonlinear dynamics,
which will be briefly explained in the rest of this chapter.

There are some issues to resolve
before proceeding on to feedback, 
though.
We always assume that gauge transformations are unitary
in the gauge theory.
This is true for fermions.
However, some bosonic quantum gates are not 
described by unitary gauge transformations,
even though underlying time evolution is unitary.
We need to extend a gauge theory
to non-unitary transformations
to calculate the Lagrangians of such quantum gates.
Next, 
we consider this problem shortly.

\subsection{Non-unitary gauge transformations}
\label{sec:1gq}

Non-unitary gauge transformations can be found 
in simple examples.\index{non-unitary gauge transformation}
Here we consider the QND gate again.\index{QND gate}
For bosons,
the input and output signals are described 
by conjugate observables \en{\{ \xi,\eta\} } 
satisfying canonical quantization
\small\begin{align}
\hspace{5mm}
\Xi(\qu) 
\equiv 
 [\xi,\eta] 
= 
 \im,
\hspace{10mm}
\left(
 \qu
\equiv
 \AV{\xi}{\eta}.
\right)
\label{inquantiz}
\end{align}\normalsize
The input-output relation of the QND gate is expressed as

\vspace{-2mm}
\begin{figure}[H]
\begin{minipage}{0.3\hsize}
\vspace{2mm}
\centering
\includegraphics[keepaspectratio,width=47mm]{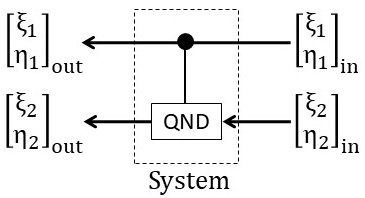}
\end{minipage}
\begin{minipage}{0.7\hsize}
%\vspace{-5mm}
\small\begin{align}
 \left[
\begin{array}{c}
 \xi_1  \\
 \eta_1 \\
 \xi_2  \\
 \eta_2
\end{array}
\right]\drm{out}
&=
\underbrace{
\left[
\begin{array}{c:c}
 I & \nA{0}{}{}{g} \\ \hdashline
 \nA{-g}{}{}{0} & I
\end{array}
\right]
}_{{{\displaystyle \equiv \tf}}}
 \left[
\begin{array}{c}
 \xi_1  \\
 \eta_1 \\
 \xi_2  \\
 \eta_2
\end{array}
\right]\drm{in},
\label{insq}
\end{align}\normalsize
\end{minipage}
\end{figure}
\vspace{-2mm}

\noindent
where \en{ g } is a real constant.
We regard \en{ \tf } as a local gauge transformation.
Obviously, 
\en{ \tf } is not unitary.
However, 
it should be described by unitary time evolution
as far as the QND gate operates on quantum fields.
Using the gauge theory,
we can find the interaction Lagrangian of this gate as
\small\begin{align}
 \lag
= 
 g \xi_1  \eta_2.
\label{sum-e0}
\end{align}\normalsize
This is Hermitian
and the resulting evolution operator 
\en{ \sim \ex\uu{\im\lag} } is unitary.
A rigorous formulation of this idea 
is found in Chapters \ref{chap:gauge} and \ref{chap:boundary}.

It is important to note that 
not all matrices are acceptable 
as non-unitary gauge transformations.
There must be some constraints on the matrix \en{ \tf }.
For instance,
canonical quantization (\ref{inquantiz})
has to be satisfied by both inputs and outputs
because they are traveling fields in free space.
Let us express (\ref{insq}) as
\en{ \qu\drm{out}=\tf \qu\drm{in} } by abuse of notation.
Then \en{ \tf } has to satisfy
\small\begin{align}
 \Xi(\qu\drm{out})
=
 \Xi(\tf \qu\drm{in})
=
 \Xi(\qu\drm{in}),
\label{psymplec}
\end{align}\normalsize
which means that the QND gate is 
not only a gauge transformation, but it is also 
a symplectic transformation.\index{symplectic}
If we define a matrix \en{ \lie }
via the Cayley transform as\index{Cayley transform}
\small\begin{align}
 \tf
&=
 \frac{1+\lie}{1-\lie},
\end{align}\normalsize
then \en{ \lie } satisfies the following two conditions:
\small\begin{align}
\hspace{40mm}
\kakkon{\Sigma\dd{z} \lie + \lie\dgg\Sigma\dd{z} = 0,}
       {\Sigma\dd{y} \lie + \lie\trans \Sigma\dd{y} = 0,}
\hspace{11mm}
 \Sigma\dd{i} \equiv \A{\sigma\dd{i}}{}{}{\sigma\dd{i}},
\end{align}\normalsize
where \en{ \sigma\dd{i} } are the Pauli matrices.
These two conditions guarantee that 
\en{ \lag } is Hermitian and \en{ \tf } is symplectic.
Conversely, 
if these conditions are satisfied,
any gate can be physically realizable.
Detailed analyses are found
in Chapters \ref{chap:gauge} and \ref{chap:sym}.

The formulation of the example above
depends on what differential equations 
are used to describe the signals \en{ \qu } 
because 
the Lagrangian (\ref{sum-e0}) is obtained through
a covariant derivative in the gauge theory.
This motivates us to introduce
forward and backward traveling fields.

\newpage

\subsection{Forward and backward traveling fields and stability}
\label{sec:1fb}

\begin{wrapfigure}[0]{r}[53mm]{49mm}
\centering
\vspace{-3mm}
\includegraphics[keepaspectratio,width=35mm]{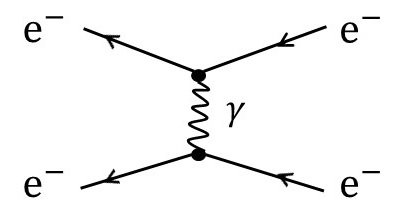} 
\caption{
\small
Feynman diagram of electron-electron scattering.
\normalsize
}
\label{fig-comparison1} 
\end{wrapfigure}

Let us get back to the electron-electron scattering.
In the Feynman diagram,
the electrons (the Dirac equation) are represented by arrowed lines,
whereas
the photon (Maxwell's equations)
is a wiggly line with no arrows,
as in Figure \ref{fig-comparison1}.
The arrows express the flow of particles in spacetime.
That corresponds to the directions of signals
in the case of quantum computing as in Figure \ref{fig-comparison2}.
Now we have a question:
If the computational signals are implemented by optical lasers,
how can we describe them?

\begin{wrapfigure}[0]{r}[53mm]{49mm}
\centering
\vspace{15mm}
\includegraphics[keepaspectratio,width=35mm]{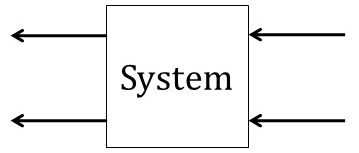} 
\caption{
\small
Block diagram.
\normalsize
}
\label{fig-comparison2} 
\end{wrapfigure}

The optical signals should be represented by arrowed lines
due to unidirectionality.
However, Maxwell's equations are not convenient 
to describe the input-output relation of such signals.
We need a simple and approximate 
model to describe unidirectionally propagating bosons.

To this end,
let us look at a relationship 
between the Dirac and wave equations:

\vspace{3mm}
\begin{picture}(300,60)(65,0)
%(yoko tate)
\put(85,40){\dashbox{KG equation (second order)}}
\put(228,45){\vector(1,0){67}}
\put(300,40){\fbox{Dirac equation (first order)}}
\put(230,34){factorization}
\put(230,7){factorization}
\put(150,21){massless}
\put(140,32){\vector(0,-1){20}}
\put(330,21){massless}
\put(380,32){\vector(0,-1){20}}
\put(78,0){\fbox{Wave equation (second order)}} 
\put(228,3){\vector(1,0){67}}
\put(300,0){\doublebox{Weyl equation (first order)}}
\end{picture}
\vspace{6mm}

\noindent
Dirac derived his first order differential equation\index{Dirac equation}
by factoring the Klein-Gordon equation.\index{Klein-Gordon equation}
By analogy,
the optical signals might be described by 
the Weyl equation instead of (Maxwell's) wave equations.\index{Weyl equation}
However, the Weyl equation is for fermions.
If we use it for bosons,
we will confront the negative energy problem.

The Weyl equation is still a likely candidate as an approximation.
The negative energy problem occurs when we consider
contributions from all frequencies.
It can be avoided 
by limiting to effective frequencies.
This trick is sometimes used in 
a stochastic approach to quantum optics (Appendix \ref{app:qsde}).

As a simple example,
let us consider a one-dimensional wave equation
\small\begin{align}
 \Box \phi = 
 (\partial\dd{t}^2 - \partial\dd{z}^2)\phi =0,
\label{pre-kg}
\end{align}\normalsize
where the four-position is 
\en{ x=(t,0,0,z) }.
If we regard this as a Klein-Gordon equation,
the corresponding Weyl equation is given as
\small\begin{align}
 \A{\partial\dd{t} + \partial\dd{z}}{}
   {}{\partial\dd{t} - \partial\dd{z}}
 \AV{\phi\dd{+}}{\phi\dd{-}}
&=
 0.
\label{pre-dirac}
\end{align}\normalsize
It is reasonable to assume that 
this first order differential equation has a solution of the form
\en{ \phi\dd{\pm} \sim \ex\uu{\im \omega (t \mp z)} } 
in free space.
Hence 
\en{ \phi\dd{+} } and \en{ \phi\dd{-} } can be thought of as 
waves traveling forward and backward,
respectively.\index{forward traveling field}\index{backward traveling field}

It is important to note that 
there also exist solutions \en{ \phi\sim \ex\uu{a(t\pm z)} }
with real \en{ a }.
At a fixed point, say \en{ z=0 }, 
\en{ \phi } increases/decreases in time 
if \en{ a } is positive/negative.
This type of solutions are not practical to assume in free space, 
but they are relevant 
if \en{ \phi } is defined \textit{locally}.

This is one of reasons that we have considered 
the definition of a system in detail in Section \ref{sec:what}.
When \en{ \phi } is a signal in free space,
it describes traveling waves.
However, 
once \en{ \phi } is defined as a system,
it 
may exhibit decaying/amplifying
behavior.
This is actually related to feedback.
We discuss this issue next.

\subsection{Feedforward and feedback}
\label{sec:1ff}

The meaning of feedback is sensitive to 
how we define a system.\index{feedback}
To see this, 
let us consider the concatenation of two QND gates again:

\vspace{-1mm}
\begin{figure}[H]
\centering
\includegraphics[keepaspectratio,width=50mm]{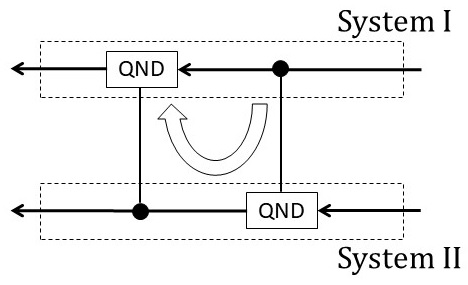}
\end{figure}
\vspace{-2mm}

\noindent
If we regard each traveling field %(upper and lower lines) 
as a system as in the diagram above,
the flow of information is expressed by the bold arrow:
The information of System I is transferred to System II 
at the first gate;
The information is fed back from System II to System I 
at the second gate.
This information flow can be thought of as feedback
for System I.

In our approach,
a system is defined in a different way.
The upper and lower solid lines are regarded as signals 
and 
the two QND gates are systems,
as in the diagram below.
In this case,
the two outputs of System I are fed into System II, but no
information is fed back from System II to System I.
This configuration is called feedforward.\index{feedforward}

\vspace{-2mm}
\begin{figure}[H]
\centering
\includegraphics[keepaspectratio,width=49mm]{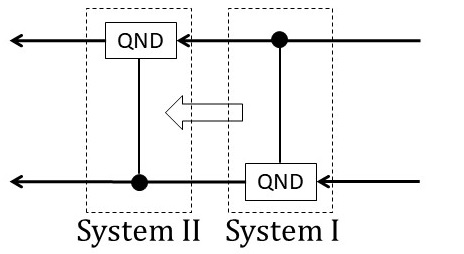}
\end{figure}
\vspace{-1mm}

The same setup is regarded either as feedback or as feedforward.
This sometimes causes confusion.
The definition of feedback depends critically on 
how we define a system.
Then how can we make feedback
in the setup of the second diagram?

Feedback is a process that violates causality 
in some sense.
(In reality,
no feedback violates causality because of a time delay.)
In the current case,
the signals propagates from right to left,
so a possible feedback process is made by 
sending a signal back against the flow of the signals,
as shown in the following diagram:

\vspace{0mm}
\begin{figure}[H]
\begin{minipage}{0.5\hsize}
\centering
\includegraphics[keepaspectratio,width=57mm]{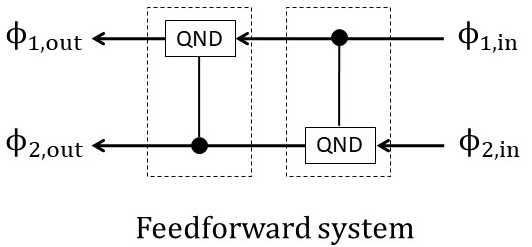}
\end{minipage}
\begin{minipage}{0.5\hsize}
\centering
\includegraphics[keepaspectratio,width=57mm]{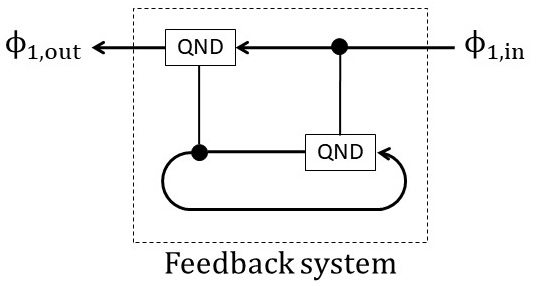}
\end{minipage}
\end{figure}
\vspace{-1mm}

\noindent
In the right diagram,
the information of one system is fed back to itself through 
a closed loop.
This loop is a \textit{locally} defined field 
that can be regarded as a new single system,
called a feedback system.
The transfer function of the feedback system is 
different from the feedforward system 
because it interacts with itself.
This feature is similar to the Dyson equation.
Next, 
we discuss a relationship 
between feedback and the Dyson equation.

\newpage

\subsection{Feedback and the Dyson equation}
\label{sec:1fd}

\begin{wrapfigure}[0]{r}[53mm]{49mm} 
\centering
\vspace{-7mm}
\includegraphics[keepaspectratio,width=33mm]{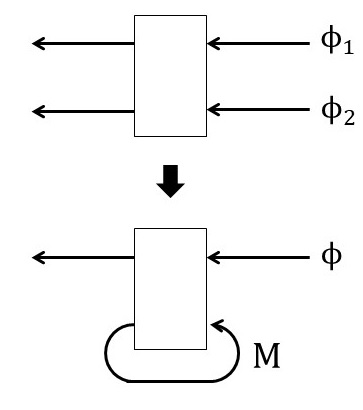}
\caption{
\small
(Upper) two-input and two-output quantum gate.
(Lower) feedback.
\normalsize
}
\label{fig-feedbackdyson-1}
\end{wrapfigure}

Feedback adds new degrees of freedom to a system.
To see this,
consider a two-input and two-output quantum gate.
If we form a closed loop \en{ \mas }
as in Figure \ref{fig-feedbackdyson-1},
how does it behave?
This feedback process can be examined 
by re-interpreting as in Figure \ref{fig-feedbackdyson-2}
where 
we prepare a free field \en{ \phi } 
and closed-loop field \en{ \mas } independently,
and then,
let them interact with each other through the gate.

Let \en{ \ipg } be the transfer function of \en{ \mas }
before it interacts with the free field \en{ \phi }.
(The imaginary unit \en{ \im } is a convention in physics.)
It is shown in Chapter \ref{chap:additional} that 
the feedback process is equivalent to 
a block diagram in Figure \ref{fig-nf}
where \en{ \sel } is self-energy determined 
by the gate.\index{self-energy}
In systems theory,
this is known as negative feedback,\index{negative feedback} 
and \en{ \sel } is called a controller.
The transfer function of the negative feedback is given by
\begin{subequations}
\label{nftfe}
\small\begin{align}
%   \hspace{-40mm}
 z&= \ipg e,
\\ %\hspace{-40mm}
 e&= u-\isel z,
\\ %\hspace{-40mm}
\therefore 
 z&= \ffrac{\ipg}{1+\ipg \isel}u,
\label{nftf}
%\kakkon{y= \ffrac{\im Y}{1+\im Y\im K}u,}
%       {e= \ffrac{1}{1+\im Y\im K}u.}
\end{align}\normalsize
\end{subequations}

\begin{wrapfigure}[0]{r}[53mm]{49mm} 
\centering
\vspace{-9mm}
\includegraphics[keepaspectratio,width=33mm]{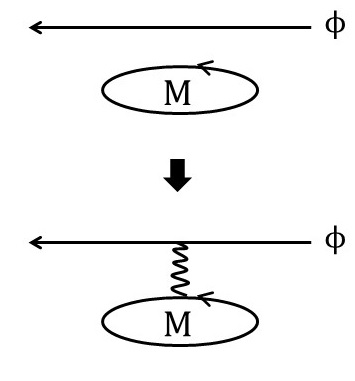}
\caption{
\small
Reinterpretation of the feedback of 
Figure \ref{fig-feedbackdyson-1}.
\normalsize}
\label{fig-feedbackdyson-2}
\end{wrapfigure}

\noindent
which indicates that after the interaction,
\en{ \ipg } transforms into
\small\begin{align}
 \ipg
\to
 \ipg^L 
&\equiv
 \frac{\ipg}{1+\ipg \isel}.
\label{dresstf}
\end{align}\normalsize
The RHS can be rewritten as 
\small\begin{align}
 \pg^L
&= 
 \pg + \pg \sel \pg^L.
\label{intdys}
\end{align}\normalsize
This is known as the Dyson equation,\index{Dyson equation}
and
\en{ \pg } and \en{ \pg^L } are called 
bare and dressed propagators,
respectively.\index{bare propagator}\index{dressed propagator}

\begin{wrapfigure}[0]{r}[53mm]{49mm} 
\centering
\vspace{15mm}
\includegraphics[keepaspectratio,width=45mm]{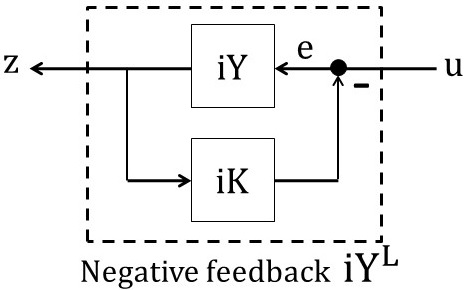}
\caption{
\small
Block diagram of negative feedback.
\normalsize
}
\label{fig-nf}
\end{wrapfigure}

Let us consider a simple case where
\en{ \mas } is a single-mode (\en{ \omega_0 }) closed loop 
of the forward traveling field.
The bare propagator is given via the Laplace transform 
as (Chapter \ref{chap:fp})
\small\begin{align}
 \ipg
&=
 \frac{1}{s-\im \omega_0}.
\label{freeg}
\end{align}\normalsize
For linear interactions,
the self-energy \en{ \isel } is constant.
If the gate is unitary,
the corresponding self-energy is 
a positive constant \en{ \isel = a>0 }
(Chapter \ref{sec:tfs}).
Then
the dressed propagator turns out to be
\small\begin{align}
 \ipg^L
&=
 \frac{1}{s-\im\omega_0 + a},
\label{gllinear}
\end{align}\normalsize
which is stable (dissipative),
i.e., 
\en{ \ipg^L \sim \ex\uu{-(a-\im\omega_0)t} } 
in the time domain.
The self-energy \en{ \isel } plays an important role
for the stability of \en{ \ipg^L }
because \en{ \isel=a } determines 
the singular point in (\ref{gllinear}).

For nonlinear interactions,
\en{ \isel } is no longer constant
and stability analysis is not straightforward,
though 
it is still possible to show 
some dynamical characteristics of \en{ \isel }
using system theoretical techniques.
The key is canonical quantization.
\en{ \isel } is influenced by 
the symplectic structure of quantum systems.
Next, we consider this problem.

\newpage

\subsection{Nonlinearity and instability}
\label{sec:1ni}

To see how the self-energy \en{ \isel } is 
influenced by canonical quantization,
we follow the same argument as Section \ref{sec:1gq},
in which we have discussed a relationship between 
acceptable non-unitary gauge transformations
and canonical quantization.

We have seen in Section \ref{sec:1fd} that 
\en{ \isel } can be regarded as a feedback controller.
In general,
a feedback system is expressed by the following block diagram:

\vspace{1mm}
\begin{figure}[H]
\centering
\includegraphics[keepaspectratio,width=60mm]{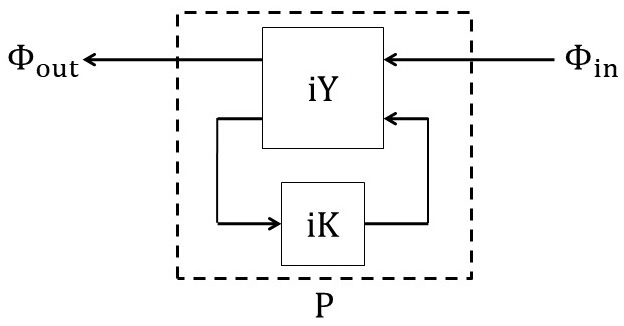}
\end{figure}
\vspace{-5mm}

\noindent
where 
\small\begin{align}
 \Phi\drm{out}
\equiv
 \AV{\phi\drm{out}}{\phi\drm{out}\dgg},
\qquad
 \Phi\drm{in}
\equiv
 \AV{\phi\drm{in}}{\phi\drm{in}\dgg}.
\end{align}\normalsize
The input \en{ \Phi\drm{in} } and 
the output \en{ \Phi\drm{out} } 
are signals in free space,
which means that 
they satisfy the same canonical quantization.
Let us define 
\small\begin{align}
 \Xi(\Phi) 
\equiv
 \phi\phi\dgg - \phi\dgg \phi. 
\end{align}\normalsize
Then the feedback system \en{ \tf } satisfies
\small\begin{align}
 \Xi(\Phi\drm{out})
= 
 \Xi(\tf\Phi\drm{in})=\Xi(\Phi\drm{in}).
\label{psymplec2}
\end{align}\normalsize
This condition is the same as (\ref{psymplec})
in which \en{ \tf } was a static matrix.
Now \en{ \tf } is a dynamical system.
In addition, 
it has a feedback structure.
Let us express \en{ \ipg } as
\small\begin{align}
\ipg
=
 \A{\ipg_{11}}{\ipg_{12}}
   {\ipg_{21}}{\ipg_{22}}.
\end{align}\normalsize
Then \en{ \tf } is given by
\small\begin{align}
 \tf
&=
 \ipg\dd{11}+\ipg\dd{12} \isel (1-\ipg\dd{22} \isel)\inv  \ipg\dd{21}.
\label{fbdig}
\end{align}\normalsize

Now we apply (\ref{psymplec2}) to this system.
For nonlinear interactions,
\en{ \isel } is a function of the complex variable \en{ s } 
in the frequency domain.
It is shown in Chapter \ref{chap:sym} that 
\en{ \isel } involves proper rational functions 
\en{ k\dd{\alpha}(s) \ (\alpha=1,2,3) }
such that 
\begin{subequations}
\small\begin{align}
 k_1(s) &= -k_2(-s),
\\
 k_3(s) &= k_3(-s). 
\end{align}\normalsize
\end{subequations}
Both indicate instability.
For example, 
if \en{ k_1(s) } is stable, 
\en{ k_2(s) } is inevitably unstable,
and \textit{vice versa}.
Also, 
if \en{ k_3\sim 1/(s^2-a^2) }, 
its time response is 
\en{ k_3(t)\sim \sinh (at) }.
Then the signals are exponentially amplified 
by the self-energy in the feedback loop.
We examine this type of instability 
using an example of third order nonlinearities
in Chapter \ref{chap:nonlinear}.

\newpage

\section{Structure of the book}

This book consists of three parts:
basics, and classical and quantum approaches.

\vspace{2mm}
\noindent
\textbf{[Part I] Basics}:
Chapter \ref{ctqm} - Chapter \ref{chap:gauge}

\begin{wrapfigure}[0]{r}[53mm]{49mm} 
\centering
\vspace{-10mm}
\includegraphics[keepaspectratio,width=48mm]{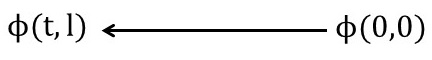}
\caption{
\small
Free field.
\normalsize
}
\label{fig-structure-1}
\end{wrapfigure}

\vspace{2mm}
\noindent
The basics of systems theory and field theory are reviewed.
In Chapter \ref{ctqm},
we introduce a transfer function (propagator) 
in a systems theoretical way.
A propagator is usually defined 
from one point to another in spacetime
as in Figure \ref{fig-structure-1}.

\begin{wrapfigure}[0]{r}[53mm]{49mm} 
\centering
\vspace{5mm}
\includegraphics[keepaspectratio,width=27mm]{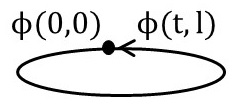}
\caption{
\small
Closed-loop field.
\normalsize
}
\label{fig-structure-2}
\end{wrapfigure}

A situation is different when we consider feedback.
A propagator is defined 
from one point to the same point in space
after traveling a closed loop,
as in Figure \ref{fig-structure-2}.
It is a function of only time
and the Laplace transform is convenient
to describe it in the frequency domain.
Then we can clearly see similarities between 
systems theory and field theory.

\vspace{2mm}
\noindent
\textbf{[Part II] Classical approach}:
Chapter \ref{chap:boundary} - Chapter \ref{chap:inter}

\vspace{2mm}
\noindent
This part is dedicated to the development of mathematical tools 
to show how dissipation/amplification appear in feedback systems.
It is summarized as a chart below.

%\noindent
As a first step,
we define a quantum gate with a transfer function,
i.e.,
an input-output relation.
For instance,
the transfer function is given by SU(2)
if the quantum gate operates as rotations 
in a two dimensional space such as beamsplitters.

\begin{wrapfigure}[17]{r}[42mm]{80mm}
\vspace{3mm}
\begin{picture}(260,170)(68,0)
\put(135,157){\fbox{
\begin{minipage}{40mm}
 Transfer function of a gate (a static matrix)
\end{minipage}
}}
      \put(87,132){Lagrange's method}
      \put(95,118){or a gauge theory}
       \put(202,132){Euler-Lagrange equation}
       \put(202,118){or \textit{S}-matrices}
   \put(185,143){\vector(0,-1){30}}
   \put(195,113){\vector(0,1){30}}
%%%%%%%%%%%%%
\put(145,100){\fbox{Lagrangian of the gate}}
       \put(200,75){Feedback}
   \put(190,93){\vector(0,-1){30}}
%%%%%%%%%%%%%
\put(144,50){\fbox{Lagrangian of a system}}
       \put(200,32){Euler-Lagrange equation}
       \put(200,18){or \textit{S}-matrices}
   \put(190,42){\vector(0,-1){30}}
%\put(70,0){\fbox{Transfer function of the system (a proper rational matrix)}}
\put(115,-7){\fbox{
\begin{minipage}{50mm}
Transfer function of the system (a proper rational matrix)
\end{minipage}
}}
\end{picture}
\end{wrapfigure}

In the second step, 
we derive the interaction Lagrangian of the gate
from the transfer function.
There are two methods for it:
Lagrange's method of undetermined multipliers and a gauge theory.
In Lagrange's method,
the transfer function is considered as boundary conditions 
between the input and the output.
In a gauge theoretical approach,
it is regarded as a local gauge transformation.

In the third step,
we define a system 
by connecting inputs and outputs across the quantum gate.
This is regarded as feedback, 
as discussed in Section \ref{sec:1ff}.
The Lagrangian of this feedback system is obtained by 
equating the input to the output in the Lagrangian of the gate.

In the last step,
the transfer function of the feedback system is derived
from the interaction Lagrangian obtained in the third step.
Classically, 
this is done by the Euler-Lagrange equation.
The same result can be obtained through \textit{S}-matrices, 
which is shown in Part III.

\vspace{2mm}
\noindent
\textbf{[Part III] Quantum approach}:
Chapter \ref{chap:fqf} - Chapter \ref{epilogue}
\vspace{2mm}

\noindent
In this part,
we derive the transfer functions of quantum gates and systems
using \textit{S}-matrices.
The relationship to the Dyson equation is also discussed.
The \textit{S}-matrix approach is useful to analyze 
the input-output relations of nonlinear systems.

We also consider pole-zero symmetry.
This is a property resulting from 
the symplectic structure of of quantum systems.
In systems theory, 
this is also known as allpass.\index{allpass}
It turns out from this symmetry that 
instability (amplification) inevitably appears
in nonlinear quantum systems.

Lastly,
our formulation is applied to different types of systems
such as spins, 
chiral symmetry breaking, 
modeling of thermal noise,
gravitational wave detection,
fusion and decay processes.

\section{Diagrams}

Due to the cross-disciplinary aspects of our approach,
some notation and diagrams may be unconventional.
First, 
we introduce \en{ \phi }, \en{ \psi } and \en{ \mas }
to represent fields:
\begin{subequations}
\label{defphim}
\small\begin{align}
  \phi(x)&: \ \mbox{\normalsize bosonic free fields,\small}
 \\
  \psi(x)&: \ \mbox{\normalsize fermionic free fields,\small}
\\
  \mas(x)&: \ \mbox{\normalsize closed-loop fields.\small}
\end{align}\normalsize
\end{subequations}

An input (output) is defined by a field 
at the initial (final) point in spacetime.
Denoted by \en { \tf } is a system in systems theory,
whereas \en{ \ipg } in field theory.
In most cases,
we consider two-input and two-output systems:
\small\begin{align}
 \phi\drm{out} 
&=
 \tf \phi\drm{in},
\quad \mbox{or} \quad
 \phi\drm{out} 
=
 \ipg \phi\drm{in},
\label{sysphyinout}
\end{align}\normalsize
where
\small\begin{align}
 \phi\drm{out}
=
 \AV{\phi_1}{\phi_2}\drm{out},
\qquad
 \phi\drm{in}
=
 \AV{\phi_1}{\phi_2}\drm{in}.
\end{align}\normalsize 
For non-unitary gates,
we need to consider both \en{ \phi } and \en{ \phi\dgg }.
In this case, 
we use
\small\begin{align}
 \Phi\drm{out}
=
\left[
\begin{array}{c}
 \phi_1 \\
 \phi_1\dgg \\ \hdashline
 \phi_2 \\
 \phi_2\dgg
 \end{array}
\right]\drm{out},
\qquad
 \Phi\drm{in}
=
\left[
\begin{array}{c}
 \phi_1 \\
 \phi_1\dgg \\ \hdashline
 \phi_2 \\
 \phi_2\dgg
 \end{array}
\right]\drm{in}.
\end{align}\normalsize
The following mean field of the input and the output often 
appears in our formulation:
\small\begin{align}
 \mm{\phi}
&=
 \frac{\phi\drm{in}+\phi\drm{out}}{2},
\qquad
 \mm{\Phi}
=
 \frac{\Phi\drm{in}+\Phi\drm{out}}{2}.
\end{align}\normalsize

We use three types of diagrams from systems theory, quantum
computing and field theory.
They are basically interchangeable, 
but we choose an appropriate one depending on the situation.

\vspace{1mm}

\begin{wrapfigure}[5]{r}[15mm]{55mm}
\centering
\vspace{-3mm}
\includegraphics[keepaspectratio,width=40mm]{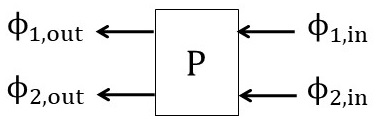}
\end{wrapfigure}

\noindent
[1. \textbf{Block diagram}]
In systems theory, a system is represented by a
block diagram.\index{block diagram}
For example, a two-input and two-output system is depicted as
in the figure on the right.
Note that the time axis is defined from right to left to be consistent with
the input-output relation (\ref{sysphyinout}).

\vspace{1mm}
\begin{wrapfigure}[4]{r}[15mm]{55mm}
\vspace{-7mm}
\centering
\hspace{-2mm}
\includegraphics[keepaspectratio,width=40mm]{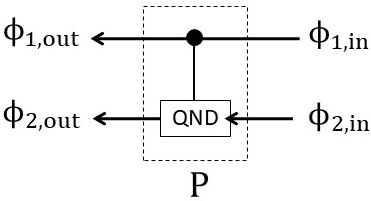}
\end{wrapfigure}

\noindent
[2. \textbf{Quantum circuit}]
In quantum computing, 
a gate is regarded as a system 
and represented by a transfer function \en{ \tf }
(the dashed box) as in the figure on the right.
Note that the time axis is defined from right to left
for consistency with the block diagram.

\vspace{1mm}

\begin{wrapfigure}[6]{r}[15mm]{55mm}
\centering
\vspace{-6mm}
\includegraphics[keepaspectratio,width=42mm]{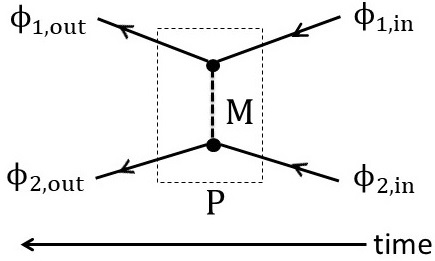}
\end{wrapfigure}

\noindent
[3. \textbf{Feynman diagram}]
In field theory, the Feynman diagram\index{Feynman diagram}
is commonly used to describe scattering processes.
The mediating field \en{ \mas } is regarded as a system
(the dashed box) as in the figure on the right.
Here again the time axis is defined from right to left for consistency.

\newpage

\section{List of symbols}

Here is the list of notation that we use throughout this book:

%\vspace{3mm}

%\begin{flushleft}

\begin{textblock}{10}[0,0](-2.1,2.7)

\small\begin{tabular}{rcl}

$x\uu{\mu}$ & & $=(t,\bm{x})$: Four-position. \\
      & &  In one dimension, $x=(t,0,0,z)$. \\

$\partial\dd{\mu}$ & & $=\ffrac{\partial}{\partial x\uu{\mu}} 
                      = \left(\partial\dd{t}, \nabla\right)$.\\

$\eta\uu{\mu\nu}$ & & $=\textrm{diag}(1,-1,-1,-1)$: \\
         & &  The Minkowski metric. \\

$p\uu{\mu}$ & & $=(E,\textbf{p}) = \im \left(\partial\dd{t},-\nabla\right)$: Four-momentum.\\

$p\dd{\mu}$ & & $=(E,-\textbf{p}) = \im \left(\partial\dd{t},\nabla\right) = \im \partial\dd{\mu}$. \\

$p\cdot x$ & & $=p^0x^0-\textbf{p}\cdot\bm{x}$. \\

%%%%%%%%%%%%%%%%%%%%%

\\

$s$ & & $=-\im p^0$: Complex frequency \\
 & &   of the Laplace transform.\\

$\dtf{A}{B}{C}{D}$ & & $=C(sI-A)\inv B+D$: \\
   & & Transfer function defined in (\ref{2.10}) \\

$\lambda(A)$ & & Eigenvalues of a matrix $A$. \\ 

$\pole$ & & Poles (Section \ref{sec:pz}). \\ 

$\zero$ & & Transmission zeros (Definition \ref{def:zero}), \\

$P\trans$ & & Transpose of a matrix $P$. \\ 

$P\simm(s)$ & & $=P\trans(-s)$ defined in (\ref{simm}). \\ 

%$\maths{C}(P)$ & & Chain-scattering representation of $P$ defined in
%         (\ref{sec4.1-2}).\\ 
%
%$\tilde{\maths{C}}(P)$ & & Dual chain-scattering representation of
%         $P$ defined in (\ref{dualcs}).\\
%
%$\maths{F}(\cdot \ ; \ \cdot)$ & & 
% Homographic transformation defined in (\ref{homog}).\\ 
%
%$\tilde{\maths{F}}(\cdot \ ; \ \cdot)$ & & 
% Dual homographic transformation defined in (\ref{homo}).\\ 

%$\Tr[\cdots]$& & Trace of $\cdots$. \\

%${\rm rank}[\cdots]$ & & Rank of $\cdots$.\\ 

%$\cdots\dgg$ & & Adjoint of an operator. \\

%%%%%%%%%%%%%%%%%%%%%

\\

$\psi$ & & Dirac field or spin field. \\

$\mas$ & & Closed-loop field, \\
    & & see Sections \ref{sec:p=0}, \ref{sec:weyltf}, and \ref{sec:ltf}. \\

$\pi$ & & Canonical momenta.\\

$\sigma\uu{i}$ & & Pauli matrices (Section \ref{gmatrix}).\\

$\sigma\dd{\pm}$ & & Raising and lowering matrices. \\ 
%     & & see Section \ref{gmatrix}.\\

$\gamma\uu{\mu}$ & & $=(\gamma^0,\bm{\gamma})$:
             $\gamma$ matrices (Section \ref{gmatrix}).\\

$\beta$ & & $=\gamma^0$ defined in (\ref{gamma}). \\

$\bm{\alpha}$ & & $=\beta\bm{\gamma}$ defined in (\ref{gamma}). \\

$\gamma^5$ & & $=\im \gamma^0 \gamma^1\gamma^2\gamma^3$: Chirality,\\
          & &  defined in (\ref{lrchiral}). \\

$\widebar{\psi}$ & & $=\psi\dgg\gamma^0$: Dirac adjoint. \\

$\fsh{\partial}$ & & $=\gamma\uu{\mu}\partial\dd{\mu}$: Feynman slash. \\

$\partial\dd{\pm}$ & & $=\partial\dd{t} \pm \partial\dd{z}$. \\

$\phi_F$ & & Forward traveling solution to  \\
         & &    the Weyl equation, see (\ref{wpm}). \\

$\phi_B$ & & Backward traveling solution to \\
         & &   the Weyl equation, see (\ref{wpm}). \\

$\Lambda\dd{\pm}$ & & Projection onto $\pm$energy subspaces, \\
        & &         defined in (\ref{projections}). \\

\end{tabular}\normalsize

\end{textblock}

%\end{flushleft}

%\vspace{5mm}
\begin{picture}(10,350)(-230,180)
%(yoko tate)
\put(0,480){\line(0,-1){550}}
\end{picture}

%\begin{flushleft}

\begin{textblock}{10}[0,0](7.5,2.7)

\small\begin{tabular}{rcl}

$\phi$ & & Forward traveling field. \\

$\Phi$ & & =$\AVl{\phi}{\phi\dgg}$. \\

$\xi$ & & $=\ffrac{\phi+\phi\dgg}{\sqrt{2}}$: $x$-quadrature 
            defined in (\ref{quadrature}). \\

$\eta$ & & $=\ffrac{\phi-\phi\dgg}{\im\sqrt{2}}$: $y$-quadrature
             defined in (\ref{quadrature}). \\

$\qu$ & & $=\AV{\xi}{\eta}$.\\

$\mas\dd{\xi}$ & & $=\ffrac{M+M\dgg}{\sqrt{2}}$. \\

$\mas\dd{\eta}$ & & $=\ffrac{M-M\dgg}{\im\sqrt{2}}$. \\

$\phi\drm{in}$ & & Input field. \\
$\phi\drm{out}$ & & Output field. \\

$\mm{\phi}$ & & $=\ffrac{\phi\drm{in}+\phi\drm{out}}{2}$. \\

%$\ex\uu{\win\uu{\alpha}\lie\dd{\alpha}}$ & & Gauge transformation in which
%\\ & &
%     $\win$ is a real function and $\lie$ is a matrix. \\

$\ga$ & & Gauge field.\\

$\lie$ & & Reactance matrix.\\

\\

%%%%%%%%%%%%%%%%
%%%%%%%%%%%%%%%%

$\lag$ & & Lagrangian density. \\

$L$ & & $=\int dz \lag$: Lagrangian. \\ 

$S^L$ & & \textit{S}-matrix corresponding to $L$. \\

${\rm T}$ & & Time-ordering operator. \\

$\step(x)$ & & Heaviside step function. \\

$\wick{A}{B}$ & & $= \mzero{\T AB}=\contraction{}{A}{\hspace{1mm}}{\hspace{0mm}} AB$: \\
 & & 
 Contraction defined in (\ref{def:contractiond}) and (\ref{def:contraction}).\\

$\wick{A}{B}^{L}$ & & $=\ffrac{\mzero{\T AB S^L }}{\mzero{S^L}}$: \\
     & &          Contraction under a Lagrangian $L$. \\

$A\ddgg$ & & Double dagger defined in (\ref{def:shadirac}) for \\
          & &    fermions and (\ref{def:sha}) for bosons.\\ 

%$\im S_{A|C}$  & &= $\wick{A}{C\ddgg}$: Transfer function from $C$ to $A$
%              for Dirac fields, see Definition \ref{def:diractf}. \\
%$\im T_{A|C}$ & &= $\wick{A}{C\ddgg}$: Transfer function from $C$ to $A$
%              for spin fields, see
%              Definition \ref{def:spintf}. \\
%$\ipg\dd{A|C}$  & &= $\wick{A}{C\ddgg}$: Transfer function from $C$ to $A$
%              for forward traveling fields, see
%              Definition \ref {def:tf}. \\
%$\ipg\dd{A|C}^L$  & &= $\wick{A}{C\ddgg}^L$: Transfer function from $A$ to $C$
%              for forward traveling fields
%              under a Lagrangian $L$.\\

$\ipg_{A|C}$  & &= $\wick{A}{C\ddgg}$: Transfer function from $C$ to $A$ \\
             & &    defined in (\ref{ediractf}) and (\ref{ptf-f}). \\

$\ipg_{A|C}^{L}$  & &= $\wick{A}{C\ddgg}^L$: 
                     Transfer function from $C$ to $A$  \\
             & &    under $L$ defined in (\ref{ptf}). \\

$\sel$  & & Self-energy, see Section \ref{sec:dyson}. \\

$g\ast \phi$ & & $=\int d\tau \ g(t-\tau)\phi(\tau)$ defined in (\ref{conv}). \\

$\phi\dgg\circ g$ & & $=\int d\tau \ \phi\dgg (t+\tau) g(\tau)$ 
                      defined in (\ref{conv}). \\

\end{tabular}\normalsize

\end{textblock}

%\end{flushleft}

\chapter*{}

\begin{center}
\begin{Large}
\begin{bfseries}
\textsf{Part I}
\\
\vspace{10mm}
\textsf{Basics of systems theory and field theory}
\end{bfseries}
\end{Large}
\end{center}

\chapter{Systems theory}
\label{ctqm}
\thispagestyle{fancy}

The basics of systems theory are introduced in this chapter.
Linear time-invariant systems can be described in two different ways:
state equations (in the time domain) and
transfer functions (in the frequency domain).
These representations are equivalent 
under suitable conditions,
but the transfer function representation is more convenient 
because the dynamical properties of a system
can be examined with poles and transmission zeros.
This helps us to investigate nonlinear interactions of quantum fields
through perturbation theory,
as will be seen in Chapter \ref{chap:sym}.

\section{State equation}
\label{sdtf}

\marginpar{\vspace{23mm}
\small
In terms of the Feynman diagram, 
$ \mas $ and $ \{\phi\drm{in},\phi\drm{out}\} $ 
correspond to internal and external lines, 
respectively.
In other words,
$ \mas $ is \textit{off shell},
whereas $ \{\phi\drm{in},\phi\drm{out}\} $ are \textit{on shell}.
\normalsize
 }

Linear time-invariant systems are described by
a \textit{state equation}\index{state equation (systems theory)}
\begin{subequations}
\label{2.1}
\small\begin{align}
 \dot{\mas}(t)
&= 
 A \mas(t) + B\phi\drm{in}(t), 
\label{2.1a} \\
 \phi\drm{out}(t) 
&= 
 C \mas(t) + D\phi\drm{in}(t),
\label{2.1b}
\end{align}\normalsize
\end{subequations}
where \en{ A,B,C } and \en{ D } are coefficient matrices, and
\small\begin{align}
\hspace{12mm}
 \mas:& \ \mbox{\normalsize state vector, \small}
\nn\\
 \phi\drm{in}:& \ \mbox{\normalsize input vector, \small}
\nn\\
 \phi\drm{out}:& \ \mbox{\normalsize output vector. \small}
\nn
\end{align}\normalsize

The input \en{ \phi\drm{in} } and 
the output \en{ \phi\drm{out} }
are accessible signals, 
whereas \en{ \mas } is a local variable 
that is not accessible directly.
Assume that \en{ \mas(0)=0 }.
Then \en{ \mas } is written as
\small\begin{align}
 \mas(t)
&=
 \int_0\uu{t} d\sigma \ \ex\uu{A(t-\sigma)} B\phi\drm{in}(\sigma).
\label{sol}
\end{align}\normalsize
A system is called 
\textit{controllable}\index{controllable (systems theory)}
if for any given vector \en{ \mas' }  and time \en{t'},
there exits \en{ \phi\drm{in} } 
that can achieve \en{ \mas(t')=\mas' }.
Such an input can be constructed 
if a Gramian matrix\index{Gramian matrix}
\small\begin{align}
 W\dd{c}(t)
&=
 \int_0\uu{t} \left(\ex\uu{-A\sigma }B\right) 
 \left(\ex\uu{-A\sigma}B\right)\trans d\sigma
\end{align}\normalsize
is nonsingular. 
In fact,
consider an input of the form
\small\begin{align}
 \phi\drm{in}(t)
&=
 \left(\ex\uu{- At}B\right)\trans W\dd{c}\inv (t') \ex\uu{- At'} \mas'.
\end{align}\normalsize
Then it follows from (\ref{sol}) that 
\small\begin{align}
 \mas(t')
&=
 \int_0\uu{t'} \ex\uu{A(t'-\sigma)} B 
 \left[ 
 \left( \ex\uu{- A\sigma}B \right)\trans 
 W\dd{c}\inv (t') 
 \ex\uu{- At'} \mas' \right] d\sigma
=
 \mas'.
\end{align}\normalsize
The nonsingularity of the Gramian matrix can be written 
as a rank condition.

\begin{lemma}
Assume that \en{ \mas } is an n-component vector.
A system (\ref{2.1}) is controllable if and only if
\small\begin{align}
 \mathrm{rank } \left[B \quad AB \quad \cdots \quad A^{n-1}B \right]
=
 n.
\end{align}\normalsize
This is also expressed as
\small\begin{align}
 \mathrm{rank } \left[\begin{array}{cc}\lambda I\!-\!A &
      B\end{array}\right]
=
 n, \quad \forall \lambda.
\end{align}\normalsize
\end{lemma}

A dual notion of controllability is 
\textit{observability}.\index{observable (systems theory)}
If the initial condition \en{ \mas(0) } of a system
\begin{subequations}
 \label{2.21}
\small\begin{align}
 \dot{\mas}
&=
 A\mas,
\\ 
 \phi\drm{out}
&=
 C\mas,
\end{align}\normalsize
\end{subequations}
is uniquely determined from the output segment 
\en{ \{\phi\drm{out}(t) \mid 0\leq t<t'\} }
of arbitrary length,
the system is said to be 
\textit{observable}.
The observability of the system is associated with
a Gramian matrix\index{Gramian matrix}
\small\begin{align}
 W\dd{o}(t)
&=
 \int_0\uu{t} 
 \left( C \ex\uu{A\sigma} \right)\trans 
 \left(C \ex\uu{A\sigma} \right) d\sigma.
\end{align}\normalsize
If this is nonsingular,
the initial condition \en{ \mas(0) } 
is determined from the output as
\small\begin{align}
 \mas(0)
&=
 W\dd{o}\inv (t)\int_0\uu{t} 
 \left( C \ex\uu{A\sigma}\right)\trans 
 \phi\drm{out}(\sigma) d\sigma.
\end{align}\normalsize
Observability is also expressed as 
a rank condition on the matrices \en{A} and \en{C}:

\begin{lemma}
\label{lem:obs}
Assume that \en{ \mas } is an n-component vector.
A system (\ref{2.1}) is observable if and only if
\small\begin{align}
 \mathrm{rank }\left[ 
 \begin{array}{l}
 C \\ CA \\ \vdots \\ CA^{n-1}
 \end{array}
 \right]
=
 n.
\end{align}\normalsize
This is equivalent to a condition
\small\begin{align}
 \mathrm{rank }\left[
 \begin{array}{c}
 \lambda I-A \\ 
 C
 \end{array}
 \right]
=
 n, 
\quad \forall \lambda.
\end{align}\normalsize
\end{lemma}

If a system is controllable and observable,
it is said to be 
\textit{minimal}\index{minimal (systems theory)} 
or \textit{irreducible}.\index{irreducible (systems theory)}
Minimality is essential to describe the system
in the frequency domain because it guarantees 
the equivalence of time and frequency domain expressions,
which is explained in the next section.

\section{Transfer function}
\label{sec:tf-sys}

The Laplace transform\index{Laplace transform} of
the state equation (\ref{2.1}) subject to the initial condition 
\en{ \mas(0)=0 } yields
\begin{subequations}
\small\begin{align}
 s\mas(s)
&= 
 A\mas(s)+B\phi\drm{in}(s),
\\ 
 \phi\drm{out}(s) 
&= 
 C\mas(s)+D\phi\drm{in}(s),
\end{align}\normalsize
\end{subequations}
where \en{ s } is a complex frequency.
Eliminating \en{ \mas(s) }, 
we get the following 
input-output relation:
\small\begin{align}
  \phi\drm{out}(s) 
= 
 \tf(s)\phi\drm{in}(s),
  \label{2.8}
\end{align}\normalsize
where
\small\begin{align}
  \tf(s) 
&=
 C(sI-A)\inv B+D
\equiv 
 \dtf{A}{B}{C}{D}.
\label{2.10}
\end{align}\normalsize
The matrix \en{ \tf(s) } is called 
a \textit{transfer function}\index{transfer function (systems theory)}.
It represents a transformation 
from the input \en{ \phi\drm{in}(s) }
to the output \en{ \phi\drm{out}(s) } 
in the frequency domain.
The input-output relation is depicted by 
the following block diagram\index{block diagram}:

\vspace{1mm}

\begin{figure}[H]
\centering
\includegraphics[keepaspectratio,width=60mm]{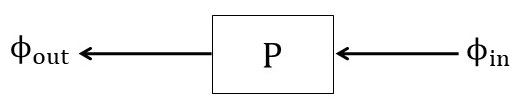}
\end{figure}

\vspace{1mm}

In (\ref{2.8}),
the state vector \en{ \mas } is completely invisible,
which means that 
the transfer function does not depend on 
how we choose \en{ \mas }.
For example,
let \en{ T } be a nonsingular matrix
and define a state vector \en{ \mu } as
\small\begin{align}
 \mu 
&= 
 T\mas.
\end{align}\normalsize
Then the state equation (\ref{2.1}) can be written as
\begin{subequations}
\small\begin{align}
 \dot{\mu}
&= 
 TAT\inv  \mu + TB\phi\drm{in},
\\
 \phi\drm{out}
&= 
 C T\inv \mu +D \phi\drm{in}.
\end{align}\normalsize
\end{subequations}
In the frequency domain,
this system is represented as
\small\begin{align}
 \phi\drm{out}
&=
 \dtf{TAT\inv }{TB}{CT\inv }{D} \phi\drm{in}.
\end{align}\normalsize
Compared to (\ref{2.10}), we have
\small\begin{align}
 \dtf{TAT\inv }{TB}{CT\inv }{D}
&=
 \dtf{A}{B}{C}{D},
\label{simtrans}
\end{align}\normalsize
which means that 
the transfer function is invariant under 
any similarity transformation.\index{similarity transformation}
This property is often used to simplify transfer functions.

\newpage

\begin{example}
\label{ex:similarity}
\rm
A transfer function 
\small\begin{align}
\hspace{10mm}
 \tf(s)
&=
 \frac{s-2|g|^2}{s+2|g|^2}
\hspace{10mm}
(g\in\mathbb{C})
\label{simexy}
\end{align}\normalsize
can be rewritten as
\small\begin{align}
 \tf(s)
&=
 2g\frac{1}{s+2|g|^2}(-2g\uu{*})+1
=
 \dtf{-2|g|^2}{-2g\uu{*}}{2g}{1}.
\end{align}\normalsize
The RHS is equivalent to a state equation 
\begin{subequations}
\label{exsys}
\small\begin{align}
 \dot{\mas}
&=
 -2|g|^2\mas-2g\uu{*} \phi\drm{in},
\\
 \phi\drm{out}
&= 
 2g\mas+\phi\drm{in}. 
\end{align}\normalsize
\end{subequations}
This is called a 
\textit{realization}\index{realization (systems theory)}
of the transfer function \en{\tf}.
Note that a realization is not unique.
For example, 
by taking \en{ T=2g } in (\ref{simtrans}),
\en{ \tf } is also represented as
\small\begin{align}
 \tf(s)
&=
 \dtf{-2|g|^2}{-4|g|^2}{1}{1},
\end{align}\normalsize
to which a state equation is written as
\begin{subequations}
\small\begin{align}
 \dot{\mas}&=-2|g|^2\mas-4|g|^2\phi\drm{in},
\\
 \phi\drm{out}&= \mas+\phi\drm{in},
\end{align}\normalsize
\end{subequations}
where the same symbol \en{ \mas } is used 
for the state vector, 
but it represents a different physical variable 
from (\ref{exsys}).
It is also beneficial to see that 
the following transfer functions are equivalent:
\small\begin{align}
 \dtf{\nA{-2|g|^2}{0}{0}{\alpha}}{\nAV{-2g\uu{*}}{\beta}}
          {\hspace{1mm} 2g \hspace{6mm} 0}{1}
&=
 \dtf{-2|g|^2}{-2g\uu{*}}{2g}{1},
\label{unobmin}
\end{align}\normalsize
for any constants \en{ \alpha,\beta }.
In the LHS,
the system has two modes \en{ \{-2|g|^2, \alpha\} }.
Note that \en{ \alpha } is an unobservable mode.
On the other hand, 
the RHS is controllable and observable, i.e., minimal.
This means that 
unobservable/uncontrollable modes are 
redundant degrees of freedom
and the system becomes minimal 
when we eliminate the redundancy.
\qed
\end{example}

\begin{lemma}
\label{lem:twomin}
Suppose that two minimal realizations 
have an identical transfer function
\small\begin{align}
 \dtf{A_1}{B_1}{C_1}{D_1}
&=
 \dtf{A_2}{B_2}{C_2}{D_2}.
\label{2.23}
\end{align}\normalsize
Then \en{ D_1=D_2 }, and 
there exists a nonsingular matrix \en{ T } such that 
\small\begin{align}
 \dtf{A_2}{B_2}{C_2}{D_2}
&= 
 \dtf{T\inv A_1T}{T\inv B_1}{C_1T}{D_1}.
\end{align}\normalsize
\end{lemma}

\newpage

\section{Poles and transmission zeros}
\label{sec:pz}

Consider a system \en{ \tf } with a minimal realization
\small\begin{align}
 \tf(s)
&=
 \dtf{A}{B}{C}{D}
=
 C(s-A)\inv B +D.
\label{pzy}
\end{align}\normalsize

\begin{wrapfigure}[0]{r}[53mm]{49mm}
\centering
\vspace{-5mm}
\includegraphics[keepaspectratio,width=47mm]{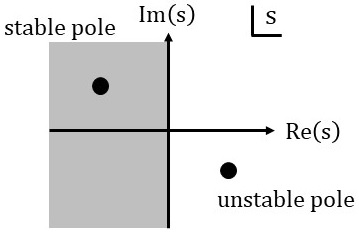}
\caption{
\small
Stable and unstable poles in the complex ($s$) plane.
\normalsize
}
\label{fig-polezero-1}
\end{wrapfigure}

\begin{definition}
\label{def:pole}
For a system \en{ \tf },
poles \en{ \pole(\tf) =\{p\in \mathbb{C}\} }\index{pole (systems theory)}
are defined as
\small\begin{align}
 \pole(\tf) 
&=
 \lambda(A),
\end{align}\normalsize
where \en{ \lambda(A) }
are the eigenvalues of the \en{ A }-matrix.
\end{definition}

\en{ p \in \lambda(A) }
determines the stability of the system.
Roughly speaking,
the system is expressed as
\small\begin{align}
 \tf(s)
\sim
 \frac{1}{s-A}
\quad
 \leftrightarrow
\quad
 \tf(t)\sim \ex\uu{At}.
\label{sysphycores}
\end{align}\normalsize
If the real part of \en{ \lambda(A) }
is negative \en{ \rea\lambda(A)<0},
the output is bounded for any bounded input.
Such a system (or \en{ A }-matrix) 
is said to be \textit{stable}.\index{stable (systems theory)}
A stable pole is a singular point 
in the left half of the complex (\en{ s }) plane,
as in Figure \ref{fig-polezero-1}.
A system is stable
when all poles are in the left half plane.
If there is a pole in the right half plane,
the system is unstable.

\begin{wrapfigure}[0]{r}[53mm]{49mm}
\centering
\vspace{-50mm}
\includegraphics[keepaspectratio,width=47mm]{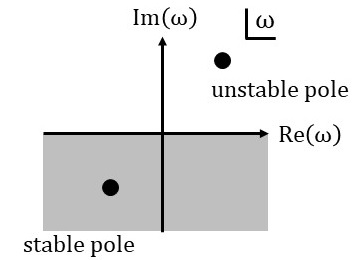}
\caption{
\small
Stable and unstable poles in the complex ($\omega$) plane.
\normalsize
}
\label{fig-polezero-2}
\end{wrapfigure}

\begin{remark}
In physics,
a transfer function is called a propagator,
and its frequency response is conventionally written in
the Fourier transform with a frequency \en{ \omega=\im s }.
For example, 
a propagator \en{ \pg } corresponding to (\ref{sysphycores}) 
is written as
\small\begin{align}
 \pg &= \frac{1}{\omega-\im A}.
\end{align}\normalsize
In this case, 
a stable pole \en{ p\in \lambda(A) \ \ (\rea(p)<0) } 
is in the lower half of the complex (\en{ \omega }) plane,
as in Figure \ref{fig-polezero-2}.
The transfer function \en{ \tf } and 
the propagator \en{ \pg } 
are related to each other as
\small\begin{align}
 \tf = \im \pg.
\end{align}\normalsize
\label{sec:rem-sys}
\end{remark}

Now let us introduce something opposite to the pole.
Suppose a system 
\small\begin{align}
 \tf(s)
=
 \frac{s-z}{s-p}.
\end{align}\normalsize
This diverges at the pole \en{ s=p }.
On the other hand,
it becomes zero at \en{ s=z }.
No information is transferred from the input to the output
at this point.
There are different ways to generalize this idea to matrices.
Here we introduce a \textit{transmission zero}.

\begin{definition}
\label{def:zero}
For a system of the form (\ref{pzy}),
\en{ \zero(\tf) \equiv \{z\in \mathbb{C}\} }
is said to be a transmission zero
\index{transmission zero (systems theory)}
if there exist vectors \en{ \xi } and \en{ \eta } such that
\small\begin{align}
 \AH{\xi\trans}{\eta\trans}\A{z-A}{-B}{-C}{-D}
&=
 0.
\label{zero}
\end{align}\normalsize 
\end{definition}

The meaning of a transmission zero is clear.
It follows from (\ref{zero}) that
\small\begin{align}
 \eta\trans \left[C(z-A)\inv B+D\right]&=0,
\quad \mbox{or} \quad
 \eta\trans \tf(z) =0.
\end{align}\normalsize
Then the output 
\en{ \phi\drm{out} = \tf\phi\drm{in} } satisfies
\small\begin{align}
 \eta\trans \phi\drm{out}(z)
&= 
 \eta\trans \tf(z)\phi\drm{in}(z) 
=
 0,
\end{align}\normalsize
which implies that 
at a transmission zero \en{ s=z },
no information is transferred from the input to the output 
along the vector \en{ \eta }.

\section{Circuits}
\label{sec:circuits}

Let us introduce algebraic operations to transfer functions.
Suppose that 
a transfer function \en{ \tf\dd{i} } has a realization
\small\begin{align}
 \tf\dd{i}(s)
=
 \dtf{A\dd{i}}{B\dd{i}}{C\dd{i}}{D\dd{i}}.
\end{align}\normalsize

\begin{wrapfigure}[0]{r}[53mm]{49mm}
\centering
\vspace{-5mm}
\includegraphics[keepaspectratio,width=35mm]{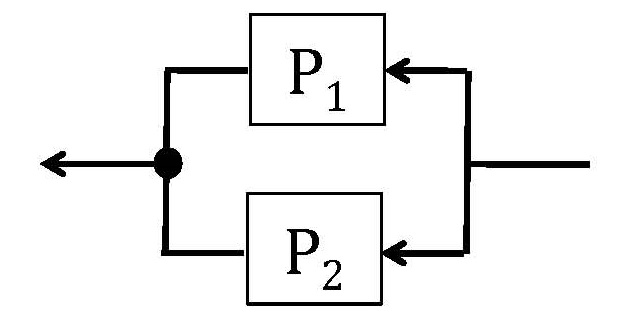}
\caption{
\small
Parallel connection.
\normalsize
}
\label{fig-connection1}
\end{wrapfigure}

The sum of two transfer functions, 
a parallel connection\index{parallel connection}
(Figure \ref{fig-connection1}), 
is given as
\small\begin{align}
\hspace{-5mm}
 \tf_1(s)+\tf_2(s)
= \left[
        \begin{array}{cc|c}
          A_1&0&B_1\\ 0&A_2&B_2\\ \hline C_1&C_2&D_1+D_2
        \end{array}
        \label{2.12}
      \right].
\end{align}\normalsize

\begin{wrapfigure}[0]{r}[53mm]{49mm}
\centering
\vspace{10mm}
\includegraphics[keepaspectratio,width=40mm]{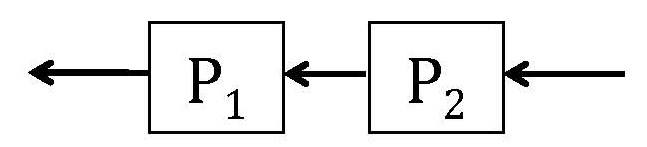}
\caption{
\small
Cascade connection.
\normalsize
}
\label{fig-connection2}
\end{wrapfigure}

\noindent
The product of two transfer functions,
a cascade connection\index{cascade connection}
(Figure \ref{fig-connection2}),
is given as
\small\begin{align}
  \tf_1(s)\tf_2(s) 
&= 
 \left[
        \begin{array}{cc|c}
         A_1&B_1C_2&B_1D_2\\ 0&A_2&B_2\\ \hline C_1&D_1C_2&D_1D_2
        \end{array}
      \right] 
 = 
 \left[
        \begin{array}{cc|c}
          A_2&0&B_2\\ B_1C_2&A_1&B_1D_2\\ \hline D_1C_2&C_1&D_1D_2
        \end{array}
      \right].
        \label{2.13}
\end{align}\normalsize

\noindent
If \en{ D } is invertible,
the inversion \en{ \tf\inv } 
is well defined and given by
\small\begin{align}
  \tf\inv (s)= \left[
        \begin{array}{c|c}
          A-BD\inv C & BD\inv  \\ \hline -D\inv C & D\inv 
        \end{array}
      \right].
        \label{2.15}
\end{align}\normalsize

\noindent
We also introduce a special type of operation 
that is often used in systems theory:
\small\begin{align}
 \tf\simm (s) 
&\equiv
 \tf\trans(-s)
=
 \dtf{-A\trans}{C\trans}{-B\trans}{D\trans}.
\label{simm}
\end{align}\normalsize

Basically, 
any circuit can be calculated from these operations.
For example,
\small\begin{align}
  [I-\tf_1(s)\tf_2(s)]\inv \tf_1(s) & 
= \left[
    \begin{array}{cc|c}
     A_1+B_1D_2VC_1 & B_1WC_2 & B_1W \\ B_2VC_{1} &
     A_2+B_2VD_1C_2 & B_2VD_1 \\ \hline VC_1 & VD_1C_2 &
     VD_1
    \end{array} 
  \right],
\label{2.14}
\end{align}\normalsize
where 
\en{ V \equiv (I- D_1D_2)\inv, \ W \equiv (I- D_2D_1)\inv }.
We have seen this form in the dressed propagator (\ref{dresstf}).

It is worth noting that these operations are algebraic,
and do not necessarily represent experimental implementations.
For example, 
quantum signals cannot be split 
as in Figure \ref{fig-connection1},
though the parallel connection is often used 
to calculate the Dyson equation in perturbation theory.

\chapter{Dirac field}
\label{chap:df}
\thispagestyle{fancy}

The Dirac equation is briefly reviewed here.
We first show how the Dirac equation is derived from 
the Klein-Gordon equation.
This helps us to understand a relationship between 
the Weyl equation and wave equations,
which leads to a forward traveling field 
in Chapter \ref{chap:fp}.
We also introduce a closed-loop field 
under periodic boundary conditions.
This is used for the formulation of feedback 
in Chapters \ref{chap:feedback} and \ref{sec:tfs}.

\section{The Dirac equation}
\label{sec:dweq}

Let us introduce basic notation first.
\en{ x\uu{\mu}=(t,\bm{x}) } is a four-position.\index{four-position}
If we consider a one-dimensional case, 
it is understood as \en{ x=(t,0,0,z) }.
A four-momentum is\index{four-momentum}
\small\begin{align}
 p\uu{\mu}=(E,\bm{p}) = (\im\partial\dd{t},-\im \nabla).
\end{align}\normalsize
The Minkowski metric\index{Minkowski metric}
is \en{ \eta\uu{\mu\nu}=(+,-,-,-) }.
The inner product of \en{ x } and \en{ p } is written as
\en{ p\cdot x=p^0x^0-\bm{p}\cdot\bm{x} }.

\subsection{The energy-momentum relation}

In special relativity theory,
the energy-momentum equation\index{energy-momentum equation}
is written as
\small\begin{align}
 p\uu{\mu}p\dd{\mu} - m^2
&=
0,
\label{eme}
\end{align}\normalsize
where \en{ m } is a mass parameter.
The Klein-Gordon equation\index{Klein-Gordon equation}
is defined from this relation:
\small\begin{align}
  (p\uu{\mu}p\dd{\mu} - m^2)\phi&=0.
\end{align}\normalsize

Let us rewrite the energy-momentum equation as
\small\begin{align}
 p\uu{\mu}p\dd{\mu} - m^2
&=
 (\gamma\uu{\mu}p\dd{\mu} - m)(\gamma\uu{\nu}p\dd{\nu} + m),
\label{kgf}
\end{align}\normalsize
where \en{ \gamma\uu{\mu} } are coefficients to be calculated.
Note that 
\small\begin{align}
 \mbox{RHS}
&=
 \frac{1}{2}
 \bigl\{\gamma\uu{\mu},\gamma\uu{\nu}\bigr\}
 p\dd{\mu}p\dd{\nu} - m^2.
\end{align}\normalsize
The factorization (\ref{kgf}) is relevant if 
\small\begin{align}
 \{\gamma\uu{\mu},\gamma\uu{\nu}\}&= 2\eta\uu{\mu\nu}.
\label{gamr}
\end{align}\normalsize
\en{ \gamma\uu{\mu} } need to be \en{ 4\times 4 } matrices 
to satisfy this relation.
They are called 
\textit{\en{ \gamma } matrices}\index{\en{ \gamma } matrices}.

%\newpage

\subsection{The Dirac equation}

The Dirac equation\index{Dirac equation}
is defined from the factorization (\ref{kgf}):
\small\begin{align}
 (\gamma\uu{\mu}p\dd{\mu} - m)\psi
 =0.
\end{align}\normalsize
Using 
\en{p\dd{\mu} = (E,-\bm{p}) = \im (\partial\dd{t},\nabla) = \im \partial\dd{\mu}},
this can be rewritten as
\begin{subequations}
\label{dirac1}
\small\begin{align}
 \bigl( \im\fsh{\partial} - m \bigr)\psi&=0,
\\
%%%%%%%%%%%%%%%%%%%%%%%%
 \widebar{\psi}\bigl( \im\overleftarrow{\fsh{\partial}} + m \bigr) &=0.
\end{align}\normalsize
\end{subequations}
where 
we have introduced\index{Dirac adjoint}\index{Feynman slash}
\begin{subequations}
\small\begin{align}
\hspace{25mm}
 \widebar{\psi}
&\equiv
 \psi\dgg\gamma^0,
\hspace{5mm}(\mbox{\normalsize Dirac adjoint\small})
\\
\hspace{25mm}
 \fsh{\partial}
&\equiv 
 \gamma\uu{\mu}\partial\dd{\mu}.
\hspace{5mm}(\mbox{\normalsize Feynman slash\small})
\end{align}\normalsize
\end{subequations}

For further study,
let us express \en{ \gamma\uu{\mu} } as
\small\begin{align}
 \gamma\uu{\mu}
&=
 (\gamma^0,\bm{\gamma})
\equiv 
 (\beta,\beta\bm{\alpha}).
\label{gamma}
\end{align}\normalsize
Then (\ref{dirac1}) is written as
\small\begin{align}
\hspace{32mm}
 \im\partial\dd{t}\psi
=
 \ham \psi,
\hspace{7mm}
 \ham
\equiv
 \bm{\alpha}\cdot\bm{p}+\beta m.
\label{diracs}
\end{align}\normalsize
As will be seen in (\ref{dich}),
\en{ \ham } is a Hamiltonian.
The Dirac equation is therefore reduced to 
an eigenvalue problem \en{ \ham \psi = E\psi },
which will be considered in Section \ref{sec:freedirac}.

\subsection{Lagrangian and canonical quantization}

The Lagrangian (density)\index{Lagrangian (Dirac)}
of the Dirac field is given by
\small\begin{align}
 \lag
&=
 \widebar{\psi}\bigl(\im\fsh{\partial}-m\bigr)\psi,
\label{lagd}
\end{align}\normalsize
The Dirac equation is obtained 
from the Euler-Lagrange equation as
\begin{subequations}
\small\begin{align}
 \partial\dd{\mu}\frac{\partial \lag}{\partial(\partial\dd{\mu} \widebar{\psi})}
-
 \frac{\partial \lag}{\partial \widebar{\psi}} 
=
 0
\quad & \Rightarrow \quad 
 \bigl(\im\fsh{\partial}-m \bigr)\psi=0
%%%%%%%%%%%%%%%%%%%%%%%%%%%%%%%%%%%%%%%%%%%%%%
\\ 
\partial\dd{\mu}\frac{\partial \lag}{\partial(\partial\dd{\mu} \psi)}
-
 \frac{\partial \lag}{\partial \psi} 
=
 0
\quad & \Rightarrow \quad 
 \widebar{\psi} \bigl( \im\overleftarrow{\fsh{\partial}} + m \bigr) = 0.
\end{align}\normalsize
\end{subequations}
The corresponding canonical momenta\index{canonical momentum (Dirac)} 
are defined as 
\begin{subequations}
\label{diracmo}
\small\begin{align}
 \pi_{\alpha}
&\equiv
 \frac{\partial\lag}{\partial (\partial\dd{t} \psi_{\alpha})}
=\im\psi_{\alpha}\dgg,
\\
 \widebar{\pi}_{\alpha}
&\equiv
 \frac{\partial\lag}{\partial (\partial\dd{t} \widebar{\psi}_{\alpha})}
=0.
\end{align}\normalsize
\end{subequations}
The Hamiltonian (density)\index{Hamiltonian (Dirac)}
of the Dirac field is then written as
\small\begin{align}
 \ham
=
 \pi(\partial\dd{t} \psi) - \lag
=
 \psi\dgg \bigl( -\im\bm{\alpha}\cdot\nabla + \beta m \bigr) \psi.
\label{dich}
\end{align}\normalsize

Canonical quantization\index{canonical quantization (Dirac)} 
is introduced to the Dirac field as
\small\begin{align}
 \{\pi_{\alpha}(t,\bm{x}),\psi_{\beta}(t,\bm{x}')\}
&=
 \im\delta_{\alpha\beta}\delta(\bm{x}-\bm{x}'),
\end{align}\normalsize
where \en{ \{\cdot,\cdot\} } is the anticommutator.
From (\ref{diracmo}), 
this turns out to be
\small\begin{align}
 \{\psi_{\alpha}\dgg(t,\bm{x}),\psi_{\beta}(t,\bm{x}')\}
&=
 \delta_{\alpha\beta}\delta(\bm{x}-\bm{x}').
\label{quant}
\end{align}\normalsize

\subsection{$\gamma$ matrices}
\label{gmatrix}

There are three well-known bases 
for the matrix representation of 
\en{ \gamma\uu{\mu} }.\index{\en{ \gamma } matrices}
As will be seen later, 
the Dirac equation can be simplified by choosing 
appropriate basis.

\marginpar{\vspace{+17mm}
\small
$\alpha\uu{i}\equiv\beta\gamma\uu{i}$,
see (\ref{gamma}).\normalsize
 }

Denoted by \en{ \sigma\uu{i} \ (i=1,2,3) } 
are the Pauli matrices.\index{Pauli matrix}
The Dirac basis\index{Dirac basis}
is useful to find free field solutions to the Dirac equation,
as will be seen in Section \ref{sec:freedirac}.
In this basis, 
\en{ \gamma } matrices are expressed as
\small\begin{align}
\hspace{-3mm}
 \gamma^0 = \beta = \A{1}{0}{0}{-1}\otimes I,
\quad
 \gamma\uu{i} =\A{0}{1}{-1}{0}\otimes \sigma\uu{i},
\quad
 \alpha\uu{i} = \A{0}{1}{1}{0}\otimes \sigma\uu{i}.
\label{diracbasis}
\end{align}\normalsize
The second choice is the Weyl basis\index{Weyl basis}
in which the Weyl equation is simplified,
as will be seen in Section \ref{rem:weyl}:
\small\begin{align}
\hspace{-3mm}
%\scalebox{0.9}{$ 
 \gamma^0 = \beta = \A{0}{1}{1}{0}\otimes I,
\quad
 \gamma\uu{i} =  \A{0}{1}{-1}{0}\otimes \sigma\uu{i},
\quad
 \alpha\uu{i} = \A{-1}{0}{0}{1}\otimes \sigma\uu{i}.
%$}
\label{weylbasis}
\end{align}\normalsize
The last one is the Majorana basis.\index{Majorana basis}
The advantage of this representation is that
the Dirac equation becomes real.
In this basis, \en{ \gamma } matrices are expressed as
\small\begin{align}
%\scalebox{0.9}{$ 
 \gamma^0=\A{0}{\sigma^2}{\sigma^2}{0},
\
 \gamma^1=\A{\im \sigma^3}{0}{0}{\im\sigma^3},
\
 \gamma^2=\A{0}{-\sigma^2}{\sigma^2}{0},
\
 \gamma^3=\A{-\im \sigma^1}{0}{0}{-\im\sigma^1}.
%$}
\end{align}\normalsize

\begin{remark}
\label{rem:pauli}
Following standard notation, 
we express the Pauli matrices\index{Pauli matrix} as
\small\begin{align}
%\scalebox{0.9}{$ 
 \sigma^1=\A{}{1}{1}{},
\qquad
 \sigma^2=\A{}{-\im}{\im}{},
\qquad
 \sigma^3=\A{1}{}{}{-1}.
%$}
\label{pauli}
\end{align}\normalsize
For later use,
we define 
raising and lowering matrices\index{raising matrix}\index{lowering matrix} 
as
\small\begin{align}
%\scalebox{0.9}{$ 
 \sigma\dd{+} =\ffrac{\sigma^1+\im\sigma^2}{2}=\A{}{1}{0}{},
\qquad
 \sigma\dd{-} = \ffrac{\sigma^1-\im\sigma^2}{2}=\A{}{0}{1}{}.
%$}
\label{rlop}
\end{align}\normalsize
\end{remark}

\subsection{The Weyl equation}
\label{rem:weyl}

The Weyl equation,\index{Weyl equation}
the Dirac equation for massless particles,
is written as
\small\begin{align}
 (\im\partial\dd{t} - \bm{\alpha}\cdot \bm{p})\psi
=
 \im(\partial\dd{t} + \bm{\alpha}\cdot \nabla)\psi
=
 0.
\label{weyleq}
\end{align}\normalsize
Let us introduce chirality \en{ \gamma_5 }\index{chirality} as
\small\begin{align}
 \gamma_5
&\equiv
 \im \gamma^0 \gamma^1 \gamma^2 \gamma^3,
\end{align}\normalsize
and projections \en{ L,R } as
\begin{subequations}
\small\begin{align}
 L
&\equiv
 \frac{1}{2}(I-\gamma_5),
\\
 R
&\equiv
 \frac{1}{2}(I+\gamma_5).
\label{lrchiral}
\end{align}\normalsize
\end{subequations}
Left- and right-chiral two-component Weyl spinors
\en{ \psi_L,\psi_R } are defined as 
\begin{subequations}
\small\begin{align}
\psi_L &\equiv L\psi,
\\
\psi_R &\equiv R\psi.
\end{align}\normalsize
\end{subequations}

In the Weyl basis (\ref{weylbasis}),
\small\begin{align}
%\scalebox{0.95}{$ 
 \gamma_5
=
 \A{-1}{}{}{1}\otimes I_{2},
\quad
 L
=
 \A{1}{}{}{0}\otimes I_{2},
\quad
 R
=
 \A{0}{}{}{1}\otimes I_{2},
%$}
\end{align}\normalsize
hence the Weyl field is expressed as
\small\begin{align}
 \psi
=
 \AV{\psi_L}{\psi_R}.
\end{align}\normalsize
Chirality is conserved in the Weyl equation
because (\ref{weyleq}) is written as
\small\begin{align}
 \A{\im(\partial\dd{t}-\bm{\sigma}\cdot\nabla)}{}
   {}{\im(\partial\dd{t}+\bm{\sigma}\cdot\nabla)}
 \AV{\psi_L}{\psi_R}
=
 0.
\label{weylmat}
\end{align}\normalsize
As the Dirac field satisfies the Klein-Gordon equation,
the left- and right-chiral Weyl spinors satisfy wave equations.
In fact, 
using \en{ \{\sigma\dd{i},\sigma\dd{j}\}=2\delta\dd{ij} },
we can express a wave equation as
\small\begin{align}
 -(\partial\dd{t}^2-\nabla^2)\psi_{L,R}
&=
 \Bigl[\im(\partial\dd{t}-\bm{\sigma}\cdot\nabla)\Bigr]
 \Bigl[\im(\partial\dd{t}+\bm{\sigma}\cdot\nabla)\Bigr]
 \psi_{L,R}
=0.
\label{lrwave}
\end{align}\normalsize

Let us consider a one-dimensional case \en{ x=(t,0,0,z) }.
The Weyl equation (\ref{weylmat}) is written as
\small\begin{align}
%\scalebox{0.95}{$ 
  \left[
    \begin{array}{cc:cc} 
        \im(\partial\dd{t}-\partial\dd{z}) & & & \\      
      & \im(\partial\dd{t}+\partial\dd{z}) & & \\ \hdashline
    & & \im(\partial\dd{t}+\partial\dd{z}) & \\
  & & & \im(\partial\dd{t}-\partial\dd{z})
    \end{array}
  \right]
  \left[
    \begin{array}{c} 
        \psi\dd{L(+)} \\   
        \psi\dd{L(-)} \\ \hdashline
        \psi\dd{R(+)} \\      
        \psi\dd{R(-)}
    \end{array}
  \right]
=0,
%$}
\label{chiralweyl}
\end{align}\normalsize
where \en{ (+) } and \en{ (-) } represent 
spin up and down along the \en{ z } axis, 
respectively.
Let us define 
\small\begin{align}
\partial\dd{\pm}
\equiv
 \partial\dd{t} \pm \partial\dd{z},
\quad
 \phi_F
\equiv
 \AV{\psi\dd{L(-)}}{\psi\dd{R(+)}},
\quad
 \phi_B
\equiv
 \AV{\psi\dd{L(+)}}{\psi\dd{R(-)}}.
\end{align}\normalsize
The Weyl equation is then rewritten as
\small\begin{align}
 \A{\im\partial\dd{+}}{}
   {}{\im\partial\dd{-}}
 \AV{\phi_F}{\phi_B}
&=0.
\label{wpm}
\end{align}\normalsize
\en{ \phi_{F,B} } are forward (backward) traveling waves.
%As with $\psi_{L,R}$,
They are actually solutions to the wave equation:
\small\begin{align}
 -(\partial\dd{t}^2-\partial\dd{z}^2)\phi_{F,B}
&=
(\im\partial\dd{-})(\im\partial\dd{+})\phi_{F,B} 
=
 0.
\end{align}\normalsize
In conclusion,
we have the following relationship:

\vspace{5mm}

\begin{picture}(300,70)(70,0)
%(yoko tate)
\put(100,50){\fbox{Klein-Gordon equation}}
\put(225,55){\vector(1,0){60}}
\put(300,50){\fbox{Dirac equation}}
\put(227,45){factorization}
\put(227,11){factorization}
\put(175,26){ $m=0$ }
\put(170,40){\vector(0,-1){20}}
\put(300,26){ $m=0$ }
\put(340,40){\vector(0,-1){20}}
\put(138,0){\fbox{Wave equation}} 
\put(225,5){\vector(1,0){60}}
\put(300,0){\fbox{Weyl equation}}
\end{picture}
\vspace{3mm}

\newpage

\section{Plane wave solutions to the Dirac equation}
\label{sec:freedirac}

\subsection{Eigenvalues and eigenvectors}

Let us consider the Dirac equation (\ref{diracs})
in the Dirac basis (\ref{diracbasis}):
\small\begin{align}
 \im\partial\dd{t}\psi
&=
 (\bm{\alpha}\cdot\bm{p}+\beta m) \psi
=
 \A{m}{\bm{p}\cdot\bm{\sigma}}{\bm{p}\cdot\bm{\sigma}}{-m}
 \psi.
\end{align}\normalsize
Assume a plane wave solution of the form
\small\begin{align}
 \psi(x)
&=
 w \ex\uu{-\im p \cdot x},
\qquad
 \left(p\cdot x = Et-\bm{p}\cdot\bm{x}\right)
\end{align}\normalsize
where \en{ w } is a four-component vector.
The Dirac equation is now reduced to an eigenvalue problem
\small\begin{align}
 \A{m}{\bm{p}\cdot\bm{\sigma}}{\bm{p}\cdot\bm{\sigma}}{-m}
 w
=
 E w.
\end{align}\normalsize
This matrix has positive and negative eigenvalues 
\en{ E=\pm E\dd{\bm{p}} }, \  
\en{ \bigl(E\dd{\bm{p}}=\sqrt{\bm{p}^2+m^2}\ge 0 \bigr) }.
The corresponding eigenvectors are given by
\begin{subequations}
\label{eigenw}
 \small\begin{align}
&
  E = +E\dd{\bm{p}}: \quad
 w\uu{(+)}(\bm{p},s)
=
 \sqrt{\ffrac{E\dd{\bm{p}}+m}{2m}}
 \AV{1}
    {+\ffrac{\bm{p}\cdot\bm{\sigma}}{E\dd{\bm{p}}+m}}
 \otimes \chi\dd{s},
\label{uu}\\  &
%%%%%%%%%%%%%%%%%%
 E = -E\dd{\bm{p}}: \quad
 w\uu{(-)}(\bm{p},s)
=
 \sqrt{\ffrac{E\dd{\bm{p}}+m}{2m}}
 \AV{-\ffrac{\bm{p}\cdot\bm{\sigma}}{E\dd{\bm{p}}+m}}
    {1}
 \otimes \chi\dd{s},
 \end{align}\normalsize
\end{subequations}
where the normalization has been chosen as
\en{ \widebar{w}\uu{(\pm)} w\uu{(\pm)} =\pm 1 },
and 
\en{ \chi\dd{s} } 
is a two-spinor\index{spinor} 
\small\begin{align}
 \ffrac{\bm{p}\cdot\bm{\sigma}}{|\bm{p}|}\chi\dd{s}
=
 s\chi\dd{s}, \quad (s=\pm 1.)
\end{align}\normalsize

Conventionally, the eigenvectors are redefined as
\begin{subequations}
\label{uv}
\small\begin{align}
&
%\scalebox{0.95}{$ 
 u(\bm{p},s)
\equiv
 w\uu{(+)}(+\bm{p},+s) 
=
 \sqrt{\ffrac{E\dd{\bm{p}}+m}{2m}}
 \AV{1}{\ffrac{\bm{p}\cdot\bm{\sigma}}{E\dd{\bm{p}}+m}}\otimes \chi\dd{s},
%$}
\\ &
%\scalebox{0.95}{$ 
v(\bm{p},s)
\equiv
 w\uu{(-)}(-\bm{p},-s)
=
 \sqrt{\ffrac{E\dd{\bm{p}}+m}{2m}}
 \AV{\ffrac{\bm{p}\cdot\bm{\sigma}}{E\dd{\bm{p}}+m}}{1}\otimes \chi\dd{-s}.
%$}
\end{align}\normalsize
\end{subequations}
These vectors satisfies
\begin{subequations}
\label{uvorth1}
\small\begin{align}
 u\dgg(\bm{p},s)u(\bm{p},s') &= \ffrac{E\dd{\bm{p}}}{m} \delta\dd{ss'},
&
 \widebar{u}(\bm{p},s)u(\bm{p},s') &= \delta\dd{ss'},
\\
 v\dgg(\bm{p},s)v(\bm{p},s') & =\ffrac{E\dd{\bm{p}}}{m} \delta\dd{ss'},
&
 \widebar{v}(\bm{p},s)v(\bm{p},s') &= \delta\dd{ss'},
\\
 u\dgg(\bm{p},s)v(-\bm{p},s') &= 0,
&
 \widebar{u}(\bm{p},s)v(\bm{p},s') &= 0.
\end{align}\normalsize
\end{subequations}

\if0
These vectors are \textit{not} normalized:
\small\begin{align}
\hspace{-2mm}
%\scalebox{0.95}{$ 
 u\dgg(\bm{p},s)u(\bm{p},s')=\ffrac{E\dd{\bm{p}}}{m} \delta\dd{ss'},
\quad
 v\dgg(\bm{p},s)v(\bm{p},s')=\ffrac{E\dd{\bm{p}}}{m} \delta\dd{ss'},
\quad
 u\dgg(\bm{p},s)v(-\bm{p},s')=0.
%$}
\label{uvorth1}
\end{align}\normalsize
However, they are orthonormalized in the following sense:
\small\begin{align}
%\scalebox{0.95}{$ 
 \widebar{u}(\bm{p},s)u(\bm{p},s')=\delta\dd{ss'},
\quad
 \widebar{v}(\bm{p},s)v(\bm{p},s')=\delta\dd{ss'},
\quad
 \widebar{u}(\bm{p},s)v(\bm{p},s')=0.
%$}
\label{uvorth2}
\end{align}\normalsize
\fi

Let us define projection operators 
onto positive and negative energy subspaces as
\begin{subequations}
\label{projections}
\small\begin{align}
&
 \Lambda\dd{+}
\equiv 
 +\sum\dd{s} u(\bm{p},s)\widebar{u}(\bm{p},s)
=
 \ffrac{1}{2m}
 \A{E\dd{\bm{p}}+m}{-\bm{p}\cdot\bm{\sigma}}
   {\bm{p}\cdot\bm{\sigma}}{-E\dd{\bm{p}}+m}
=
 \ffrac{\fsh{p}+m}{2m},
\\ &
 \Lambda\dd{-}
\equiv 
 -\sum\dd{s} v(\bm{p},s)\widebar{v}(\bm{p},s)
=
 \ffrac{1}{2m}
 \A{-E\dd{\bm{p}}+m}{\bm{p}\cdot\bm{\sigma}}
   {-\bm{p}\cdot\bm{\sigma}}{E\dd{\bm{p}}+m}
=
 \ffrac{-\fsh{p}+m}{2m}.
\end{align}\normalsize
\end{subequations}
Note that there is a minus sign 
in the definition of \en{ \Lambda\dd{-} }.
It is not difficult to see
\small\begin{align}
 & \Lambda^2\dd{\pm} = \Lambda\dd{\pm},
\qquad
 \Lambda\dd{\pm} \Lambda\dd{\mp} = 0,
\qquad
 \Lambda\dd{+} + \Lambda\dd{-} = I.
\end{align}\normalsize

\subsection{Free field solutions}
\label{sec:ffsdirac}

Eventually, 
positive and negative energy solutions are given as
\begin{subequations}
\small\begin{align}
&
 \psi\uu{(+)}(x)
=
 {\displaystyle \sum\dd{s}}
 \intt\ffrac{d^3p}{(2\pi)^{3/2}} \sqrt{\ffrac{m}{E\dd{\bm{p}}}}
 \, a(\bm{p},s) \, u(\bm{p},s) \,
 \ex\uu{-\im p\cdot x},
\\ &
 \psi\uu{(-)}(x)
=
 {\displaystyle \sum\dd{s}}
 \intt\ffrac{d^3p}{(2\pi)^{3/2}} \sqrt{\ffrac{m}{E\dd{\bm{p}}}}
 \, b\dgg(\bm{p},s) \, v(\bm{p},s) \,
 \ex\uu{\im p\cdot x},
\end{align}\normalsize
\end{subequations}
where 
\en{ \sqrt{m/E\dd{\bm{p}}} } comes from (\ref{uvorth1}).
This factor is not necessary, 
but it simplifies canonical quantization,
as will be seen in (\ref{quantab}).
A general solution is given as
\small\begin{align}
 \psi(x)
=
  \psi\uu{(+)}(x) + \psi\uu{(-)}(x).
\label{dsfree}
\end{align}\normalsize
It is easy to see that 
\en{ \Lambda\dd{\pm} } are projections 
onto positive and negative energy spaces:
\small\begin{align}
 \Lambda\dd{\pm} \psi
=
 \psi\uu{(\pm)}.
\end{align}\normalsize

\subsection{Creation and annihilation operators}

It follows from (\ref{uvorth1}) that 
\begin{subequations}
\label{abext}
\small\begin{align}
 a(\bm{p},s)
&=
 \sqrt{\frac{m}{E\dd{\bm{p}}}}
 \int \frac{d^3x}{(2\pi)^{3/2}}  \, 
 \ex\uu{\im p\cdot x} \, 
 u\dgg(\bm{p},s) \, 
 \psi(x),
\\
 b\dgg(\bm{p},s)
&=
 \sqrt{\frac{m}{E\dd{\bm{p}}}}
 \int \frac{d^3x}{(2\pi)^{3/2}}  \, 
 \ex\uu{-\im p\cdot x} \, 
 v\dgg(\bm{p},s) \, 
 \psi(x).
\end{align}\normalsize
\end{subequations}
Canonical quantization (\ref{quant}) is rewritten as
\small\begin{align}
 \{a(\bm{p},s),a\dgg(\bm{p}',s')\}
=
 \{b\dgg(\bm{p},s),b(\bm{p}',s')\}
=
 \delta\dd{ss'}\delta(\bm{p}-\bm{p}')&,
\label{quantab}
\end{align}\normalsize
and all other anticommutation relations are zero.
\en{ a\dgg, b\dgg } are regarded as the creation operators 
of particles and antiparticles\index{antiparticle}, 
respectively.
The Hamiltonian (\ref{dich}) can be expressed as
\small\begin{align}
 \ham
&=
 \sum\dd{s} \int d^3p \ E\dd{\bm{p}} 
\Bigl[
 a\dgg(\bm{p},s) a(\bm{p},s) + b\dgg(\bm{p},s)b(\bm{p},s)
\Bigr].
\end{align}\normalsize
The first and second terms represents the energy 
of particles and antiparticles, 
respectively.

%\newpage

\section{Transfer function of the Dirac field}
\label{sec:dfde}

Here we introduce transfer functions through two steps:
a contraction \en{ \wick{\cdot}{\cdot} } 
and 
a double dagger \en{ \ddagger }.
These symbols are unconventional, 
but they are convenient when we consider quantum gates.

\begin{definition}
For operators \en{ A } and \en{ B }, 
a contraction\index{contraction (Dirac)} is defined as
\begin{subequations}
\label{def:contractiond}
\small\begin{align}
 \wick{A(x_2)}{B(x_1)}\dd{\alpha\beta}
&\equiv
 \contraction{}{\phi}{\hspace{5mm}}{\hspace{0mm}}
 A\dd{2\alpha}B\dd{1\beta}
\\ &= 
 \mzero{\T \ A\dd{2\alpha} B\dd{1\beta}}
\\ &=
 \mzero{\step(t_2-t_1) \, A\dd{2\alpha} B\dd{1\beta}
       -\step(t_1-t_2) \, B\dd{1\beta}  A\dd{2\alpha}},
\end{align}\normalsize
\end{subequations}
where \en{ \T } is the time-ordering operator\index{time-ordering operator} 
and 
\en{ \step(t) } is the Heaviside step function.
We have also simplified the notation as 
\en{ A\dd{\alpha}(x_2)=A\dd{2\alpha} }.
\end{definition}

By definition, it follows that 
\small\begin{align}
-\wick{A_2}{B_1}\dd{\alpha\beta}
 &= 
 \wick{B_1}{A_2}\dd{\alpha\beta}.
\label{abba}
\end{align}\normalsize
It is also easy to show that 
for the Dirac field \en{ \psi } and an arbitrary operator \en{ B },
\begin{subequations}
\label{dicont}
\small\begin{align}
 (\im\fsh{\partial}-m) \wick{\psi(x)}{B(0)}
&=
 +\im\delta(t)\gamma^0\bra{0}\{\psi(x),B(0)\}\ket{0},
\label{dicont-1} \\
\wick{B(0)}{\widebar{\psi}(x)}  (\im\overleftarrow{\fsh{\partial}}+m)
&=
 -\im\delta(t)\bra{0}\{B(0),\widebar{\psi}(x)\}\ket{0}\gamma^0.
\end{align}\normalsize
 \end{subequations}

\begin{definition}
Given an operator \en{ A }, 
its double dagger \en{ A\ddgg }\index{$\ddagger$ (Dirac)} 
is defined as an operator satisfying
\small\begin{align}
 \{A\dd{\alpha}(x_2),A\ddgg\dd{\beta}(x_1)\}
&=
 \gamma^0\dd{\alpha\beta} \delta(\bm{x}_2-\bm{x}_1).
\label{def:shadirac}
\end{align}\normalsize
\end{definition}

Note that \en{ A\ddgg } does not necessarily exist for all \en{ A }.
Since \en{ \gamma^0 } is self-adjoint,
\small\begin{align}
 \{A\dd{\alpha},A\ddgg\dd{\beta}\}
&= 
 \gamma^0\dd{\beta\alpha}=\{A\ddgg\dd{\alpha},A\dd{\beta}\}.
\end{align}\normalsize
Hence, 
if \en{ A\ddgg } exists, 
it satisfies 
\small\begin{align}
 (A\ddgg)\ddgg&=A.
\end{align}\normalsize
For the Dirac field \en{ \psi },
\small\begin{align}
 \psi\ddgg &=\widebar{\psi},
\qquad
 (\widebar{\psi})\ddgg =\psi.
\end{align}\normalsize

\begin{definition}
\label{def:diractf}
A transfer function\index{transfer function (Dirac)} 
from \en{ B(x_1)=B_1 } to \en{ A(x_2)=A_2 } is defined as
\small\begin{align}
\ipg\dd{A|B}(x_2,x_1)
&\equiv
 \wick{A_2}{B\ddgg_1}.
\label{ediractf}
\end{align}\normalsize
\end{definition}

If \en{ A=B=\psi },
this is a Green's function of the Dirac equation:
\small\begin{align}
\hspace{23mm}
 (\im\fsh{\partial}-m) \pg\dd{\psi|\psi}(x,0) 
= 
 \delta^{(4)}(x).
\hspace{7mm}
\because (\ref{dicont-1})
\label{feygreen}
\end{align}\normalsize
\en{ \pg\dd{\psi|\psi} } is 
is known as the Feynman propagator.\index{Feynman propagator}

\begin{theorem}
Let \en{ q=(q^0,\bm{q}) } be a four-momentum.
The Feynman propagator is expressed in momentum space as
\begin{subequations}
\label{tfdf}
\small\begin{align}
 \ipg\dd{\psi|\psi}(q)
&=
 \bigl(-\im\fsh{q} +\im m + \epsilon \bigr)\inv .
\\
 \ipg\dd{\widebar{\psi}|\widebar{\psi}}(q)
&=
 \bigl(-\im\fsh{q} -\im m + \epsilon \bigr)\inv .
\end{align}\normalsize
\end{subequations}
\end{theorem}

\begin{proof}
Recall that 
the free Dirac field has been given in (\ref{dsfree}).
%\small\begin{align}
% \psi(x)
%&=
% \sum\dd{s}\int\frac{d^3p}{(2\pi)^{3/2}} \sqrt{\frac{m}{E\dd{\bm{p}}}}
%\left[
% a(\bm{p},s)u(\bm{p},s)
% \ex\uu{-\im p\cdot x}
%+
% b\dgg(\bm{p},s)v(\bm{p},s) 
% \ex\uu{\im p\cdot x}
%\right].
%\end{align}\normalsize
Using the projections (\ref{projections})
and 
anticommutation relation (\ref{quantab}),
we have
\begin{subequations}
\label{fprop1}
\small\begin{align}
 \ipg\dd{\psi|\psi}(x)
&=
 \int\frac{d^3p}{(2\pi)^3}\frac{m}{E\dd{\bm{p}}}
\left[
 \step(t)  \Lambda\dd{+} \ex\uu{-\im p\cdot x}
+\step(-t) \Lambda\dd{-} \ex\uu{\im p\cdot x}
\right]
\label{fprop1-1}\\ &=
 \int\frac{d^3p}{(2\pi)^3}
\left[\step(t)  \frac{\fsh{p}+m}{2E\dd{\bm{p}}} \ex\uu{-\im p\cdot x}
     -\step(-t) \frac{\fsh{p}-m}{2E\dd{\bm{p}}} \ex\uu{\im p\cdot x}
\right]
\\ &=
 \int\frac{d^3p dk}{(2\pi)^4\im}
\left[\frac{E\dd{\bm{p}}\gamma^0-\bm{p}\cdot\bm{\gamma}+m}{2E\dd{\bm{p}}}
      \frac{\ex\uu{-\im (E\dd{\bm{p}}-k)t+\im\bm{p}\cdot\bm{x}}}{k-\im \epsilon}
\right.\nonumber \\ &\hspace{25mm} \left.
     -\frac{E\dd{\bm{p}}\gamma^0-\bm{p}\cdot\bm{\gamma}-m}{2E\dd{\bm{p}}}
      \frac{\ex\uu{\im (E\dd{\bm{p}}-k)t-\im\bm{p}\cdot\bm{x}}}{k-\im \epsilon}
\right].
\label{fprop1-2}
\end{align}\normalsize
\end{subequations}
In the last line,
we have used an integral representation of the step function
\small\begin{align}
 \step(t)
=
 \lim\dd{\epsilon\to 0\dd{+}}
 \int \frac{dk}{2\pi\im}\frac{\ex\uu{\im kt}}{k-\im\epsilon}.
\label{deltak}
\end{align}\normalsize
To rewrite (\ref{fprop1-2}),
let us introduce the following change of variables:
\begin{subequations}
\small\begin{align}
 \mbox{\normalsize first term \small}&
 \kakkon{q^0\equiv E\dd{\bm{p}}-k,}{\bm{q}\equiv \bm{p},}
\quad
 \mbox{\normalsize second term \small}
 \kakkon{q^0\equiv -E\dd{\bm{p}}+k,}{\bm{q}\equiv -\bm{p}.}
\end{align}\normalsize
\end{subequations}
Note that \en{ E\dd{\bm{p}}=\sqrt{\bm{p}^2+m^2}=E\dd{\bm{q}} }.
Then we have %(\ref{fprop1-2}) is written as
\begin{subequations}
\label{fprop2}
\small\begin{align}
 \ipg\dd{\psi|\psi}(x)
&=
 \int\frac{d^4q}{(2\pi)^4\im} \frac{1}{2E\dd{\bm{q}}}
\left[
 \frac{E\dd{\bm{q}}\gamma^0-\bm{q}\cdot\bm{\gamma}+m}
      {E\dd{\bm{q}}-q^0-\im\epsilon}
-\frac{E\dd{\bm{q}}\gamma^0 +\bm{q}\cdot\bm{\gamma}-m}
      {E\dd{\bm{q}}+q^0-\im\epsilon}
\right]
 \ex\uu{-\im q\cdot x}
\\ &=
 \int\frac{d^4q}{(2\pi)^4}
\frac{\fsh{q}+m}
     {-\im(q^2 - m^2 +\im\epsilon)}
\ex\uu{-\im q\cdot x}
\\ &=
 \int\frac{d^4q}{(2\pi)^4}
\frac{1}
     {-\im\fsh{q} +\im m +\epsilon}
\ex\uu{-\im q\cdot x},
\end{align}\normalsize
\end{subequations}
which completes the assertion.
Note that in the second line, 
we have replaced \en{ 2E\dd{\bm{q}}\epsilon \to \epsilon }
because 
\en{ \epsilon\to 0\dd{+} }.
In the third line, 
we have used
\small\begin{align}
\hspace{20mm}
 \frac{\fsh{q}+z}{q^2-z^2}
=
 \frac{1}{\fsh{q}-z}, 
\qquad
 (z=m-\im \epsilon)
\end{align}\normalsize
where the RHS is understood as an inverse matrix.
\end{proof}

\newpage

\subsection{\en{ \bm{p}= 0 } and the Laplace transform}
\label{sec:p=0}

\begin{wrapfigure}[0]{r}[53mm]{49mm}
\centering
\vspace{10mm}
\includegraphics[keepaspectratio,width=40mm]{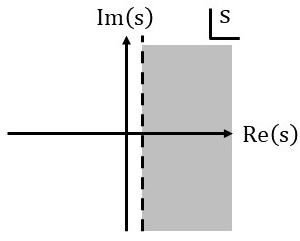}
\caption{
\small
Region of convergence for the (1,1)-element 
(causal function) of (\ref{fprop5}).
\normalsize
}
\label{fig-roc-1}
\end{wrapfigure}

As a special case of the Feynman propagator,
let us consider a Dirac field \en{ \mas } with zero momentum.
Setting \en{ \bm{p}=0 } in (\ref{fprop1-1}),
we have (in the Dirac basis)
\small\begin{align}
 \ipg\dd{\mas|\mas}(t)
&=
 \A{\step(t) \ex\uu{-\im m t}}{}
   {}{\step(-t) \ex\uu{ \im m t}}.
\label{fprop4}
\end{align}\normalsize
The positive and negative energy components are 
expressed by causal and anticausal functions, 
respectively.
Since this propagator is a function of time,
we use the two-sided (bilateral) 
Laplace transform\index{two-sided Laplace transform}
to get a frequency-domain expression:
\small\begin{align}
 \ipg\dd{\mas|\mas}(s)
&=
 \A{\ffrac{1}{s+\im m}}{}
   {}{\ffrac{1}{-s+\im m}},
\label{fprop5}
\end{align}\normalsize
where the region of convergence is \en{ \rea(s)>0 } (\en{ \rea(s)<0 })
for the causal (anticausal) function
(Figure \ref{fig-roc-1}).
This is also obtained 
by setting %\en{ q=(\im s,0,0,0)}
\en{ s\equiv -\im q^0 } and \en{ \bm{q}=0 } 
in (\ref{tfdf}):
\small\begin{align}
\ipg\dd{\mas|\mas}(s) 
&=
 [s\gamma^0 +\im m+\epsilon]\inv 
=
 \A{\ffrac{1}{s+\im m +\epsilon}}{}
 {}{\ffrac{1}{-s+\im m +\epsilon}}.
\label{fourier-1}
\end{align}\normalsize
In this expression, the positive and negative energy components are, 
respectively,
expressed by stable and unstable transfer functions
(Figure \ref{fig-roc-2}).

\begin{wrapfigure}[0]{r}[53mm]{49mm}
\centering
\vspace{-30mm}
\includegraphics[keepaspectratio,width=40mm]{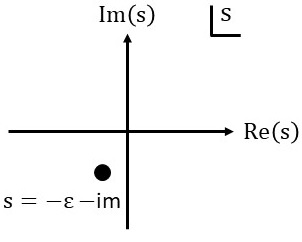}
\caption{
\small
Stable pole of the (1,1)-element 
of (\ref{fourier-1}).
\normalsize
}
\label{fig-roc-2}
\end{wrapfigure}

This result is also understood
from the perspective of Sato's hyperfunctions.\index{hyperfunction}
In the Dirac basis,
(\ref{feygreen}) is written for \en{ x=(t,0,0,0) } as
\small\begin{align}
 \A{\left(\im \partial\dd{t} - m \right) \ipg\uu{(+)}(t)}
 {}{}
 { \left(-\im \partial\dd{t} - m \right) \ipg\uu{(-)}(t)}
&=
 \im \delta(t).
\label{sato0}
\end{align}\normalsize
Let us consider the (1,1)-element. After the Fourier transform,
it is expressed as
\small\begin{align}
 (\omega - m) \ \ipg\uu{(+)}(\omega)
&=
 \im.
\label{sato1}
\end{align}\normalsize
In the hyperfunction approach,
we regard \en{ \omega } as a complex number 
and express an analytic function \en{ f(t) } as
\small\begin{align}
 f(t) 
\mapsto
 [f(\omega)\step\dd{\ima(\omega)>0}] 
=  
 [-f(\omega)\step\dd{\ima(\omega)<0}],
\end{align}\normalsize
where \en{ \step\dd{\ima(\omega)\lessgtr 0} } are 
indicator (step) functions.
The RHS of (\ref{sato1}) is then written as
\small\begin{align}
 [\im\step\dd{\ima(\omega)>0}]
=
 [-\im\step\dd{\ima(\omega)<0}].
\end{align}\normalsize
Accordingly, 
(\ref{sato1}) has two independent solutions 
\begin{subequations}
\label{sato4}
\small\begin{align}
 \ipg\drm{R}\uu{(+)}(\omega)
&=
 [\step\dd{\ima(\omega)>0} \ \frac{\im}{\omega - m}],
\\
 \ipg\drm{A}\uu{(+)}(\omega)
&=
 [\step\dd{\ima(\omega)<0} \ \frac{-\im}{\omega - m}].
\end{align}\normalsize
\end{subequations}
The inverse Fourier transform is defined as
\begin{subequations}
\small\begin{align}
 \ipg\drm{R}\uu{(+)}(t)
&=
 \int_{-\infty +\im\epsilon}^{\infty +\im\epsilon} 
 \frac{dt}{2\pi} \ex\uu{-\im\omega t}
 \frac{\im}{\omega-m}
=
 \step(t)\ex\uu{-\im m t},
\\
 \ipg\drm{A}\uu{(+)}(t)
&=
 \int_{-\infty -\im\epsilon}^{\infty -\im\epsilon} 
 \frac{dt}{2\pi} \ex\uu{-\im\omega t}
 \frac{-\im}{\omega-m}
=
 -\step(-t)\ex\uu{-\im m t},
\end{align}\normalsize
\end{subequations}
where 
\en{ \epsilon>0 }.
This means that 
\en{\ipg\drm{R}\uu{(+)}} and \en{\ipg\drm{A}\uu{(+)}} 
are retarded and advanced Green's functions, 
respectively.
\index{retarded Green's function}
\index{advanced Green's function}
It is also possible to rewrite (\ref{sato4}) as
\begin{subequations}
\small\begin{align}
 \ipg\drm{R}\uu{(+)}(s)
&=
 \frac{1}{s + \im m + \epsilon},
\\
 \ipg\drm{A}\uu{(+)}(s)
&=
 \frac{1}{-s - \im m +\epsilon},
\end{align}\normalsize
\end{subequations}
where \en{ s\equiv -\im\omega }.
Likewise, the (2,2)-element of (\ref{sato0}) has 
two solutions
\begin{subequations}
\small\begin{align}
 \ipg\drm{R}\uu{(-)}(s)
&=
 \frac{1}{s - \im m + \epsilon},
\\
 \ipg\drm{A}\uu{(-)}(s)
&=
 \frac{1}{-s + \im m +\epsilon}.
\end{align}\normalsize
\end{subequations}
Compared to (\ref{fourier-1}), we have
\small\begin{align}
 \ipg\dd{M|M}
=
 \A{\ipg\drm{R}\uu{(+)}}{}
   {}{\ipg\drm{A}\uu{(-)}}.
\end{align}\normalsize
The positive and negative energy components are 
the retarded and advanced Green's functions, 
respectively.

\subsection{Closed-loop Weyl field}
\label{sec:weyltf}

In Chapter \ref{chap:1},
we have distinguished systems and signals due to locality.
From this perspective,
the Dirac field \en{ \psi } with continuous momentum in free space 
is regarded as a signal,
whereas 
the zero-momentum field \en{ \mas } is a (single-mode) system.
Here we consider a discrete multi-mode system
using periodic boundary conditions
in one dimension \en{ x=(t,0,0,z) }.

For simplicity, we consider the Weyl field (\en{ m=0 }).
In the plane wave expansion,
periodic boundary conditions in a finite interval \en{ z\in[0,l] }
are written as
\small\begin{align}
 \exp\left[\im p z\right]
&=
 \exp\left[\im p(z+l)\right]
\quad  \Rightarrow  \quad
 p
=
 \frac{2\pi}{l}n.
\quad
 \left(n\in \mathbb{Z}\right)
\end{align}\normalsize
Then forward and backward traveling solutions to the Weyl equation (\ref{wpm}) 
are, respectively, given by
\begin{subequations}
\label{weylfex}
\small\begin{align}
 \phi_F(x)
&=
 \frac{1}{\sqrt{l}} \sum_{n=-\infty}^\infty
\left[
 A\dd{n}     u_F  \ex\uu{-\im \ffrac{2\pi}{l} n(t-z)}
+
 B\dgg\dd{n} v_F  \ex\uu{\im\ffrac{2\pi}{l}n(t-z)}
\right],
\label{weylfex-1}\\
 \phi_B(x)
&=
 \frac{1}{\sqrt{l}} \sum_{n=-\infty}^\infty
\left[
 C\dd{n}     u_B  \ex\uu{-\im\ffrac{2\pi}{l}n(t+z)}
+
 D\dgg\dd{n} v_B  \ex\uu{\im\ffrac{2\pi}{l}n(t+z)}
\right].
\label{weylfex-2}
\end{align}\normalsize
\end{subequations}
Using the orthonormality of the vectors \en{ u_{F(B)},v_{F(B)} },
one can readily show that 
\small\begin{align}
 \{A\dd{m},A\dd{n}\dgg\} 
&=
 \{B\dgg\dd{m},B\dd{n}\}
=
 \{C\dd{m},C\dd{n}\dgg\} 
=
 \{D\dgg\dd{m},D\dd{n}\}
=
 \delta\dd{mn}.
\label{weylfex-anti}
\end{align}\normalsize

\begin{wrapfigure}[0]{r}[53mm]{49mm}
\centering
\vspace{-0mm}
\includegraphics[keepaspectratio,width=45mm]{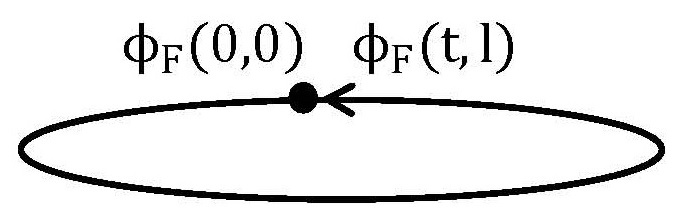}
\caption{
\small
Closed-loop field.
\normalsize
}
\label{fig-periodw}
\end{wrapfigure}

Now let us focus on the forward traveling solution \en{ \phi_F }.
We are interested in a transfer function 
from a fixed point in space, say \en{ z=0 },
to the same point after traveling in a closed loop 
under the periodic boundary conditions\index{closed loop (Weyl)}
as in Figure \ref{fig-periodw}.
Substituting (\ref{weylfex-1}) into 
the definition (\ref{ediractf}) yields
\small\begin{align}
 \ipg\dd{\phi\dd{F}|\phi\dd{F}}(t,l;0,0)
&=
 \sum_n \frac{1}{l}
 \left[ 
 \Lambda\dd{F+} \step(t) \ex\uu{-\im\ffrac{2\pi}{l}nt} 
+
 \Lambda\dd{F-} \step(-t)\ex\uu{ \im\ffrac{2\pi}{l}nt} 
 \right],
\label{wftf}
\end{align}\normalsize
where \en{ \Lambda\dd{F\pm} } are projections
onto positive and negative energy spaces.
The two-sided (bilateral) Laplace transform\index{two-sided Laplace transform}
yields
\begin{subequations}
\label{wftf-c}
\small\begin{align}
 \ipg\dd{\phi\dd{F}|\phi\dd{F}}(s)
&=
 \frac{1}{l}\sum_{n=-\infty}^{\infty}
 \left[
 \frac{\Lambda\dd{F+}}{ s+\im\ffrac{2\pi}{l}n}
+
 \frac{\Lambda\dd{F-}}{-s+\im\ffrac{2\pi}{l}n}
 \right]
\\ &=
 \frac{1}{2\im}
 \left[
 \Lambda\dd{F+} \cot \left(\pi\frac{ sl}{2\pi\im}\right)
+
 \Lambda\dd{F-} \cot \left(\pi\frac{-sl}{2\pi\im}\right)
 \right]
\\ &=
 \frac{\Lambda\dd{F+}-\Lambda\dd{F-}}{2} 
 \frac{1+\ex\uu{-sl}}{1-\ex\uu{-sl}},
\end{align}\normalsize
\end{subequations}
where we have used 
the partial fraction expansion of the cotangent function
in the second line.
Likewise, 
for the backward traveling solution \en{ \phi_B },
\small\begin{align}
 \ipg\dd{\phi\dd{B}|\phi\dd{B}}(s)
&=
 \frac{\Lambda\dd{B+}-\Lambda\dd{B-}}{2} 
 \frac{1+\ex\uu{-sl}}{1-\ex\uu{-sl}}.
\end{align}\normalsize

Note that the forward and backward traveling fields are independent 
because 
they belong to different subspaces as in (\ref{wpm}).
Hence the closed-loop transfer function 
is simply given by the sum of them:
\small\begin{align}
 \ipg\dd{\phi|\phi}(s)
&=
 \ipg\dd{\phi\dd{F}|\phi\dd{F}}(s) 
+  
 \ipg\dd{\phi\dd{B}|\phi\dd{B}}(s)
=
 \frac{\Lambda\dd{+}-\Lambda\dd{-}}{2} 
 \frac{1+\ex\uu{-sl}}{1-\ex\uu{-sl}}.
\label{periodweyltf}
\end{align}\normalsize

So far, 
we have considered an infinite number of modes in the closed loop.
For a finite number of modes,
the transfer function (\ref{wftf-c}) is modified as
\small\begin{align}
 \ipg\dd{\phi|\phi}(s)
&=
 \frac{1}{l}\sum\dd{n:\textrm{finite}}
 \left[
 \frac{\Lambda\dd{+}}{ s+\im\ffrac{2\pi}{l}n}
+
 \frac{\Lambda\dd{-}}{-s+\im\ffrac{2\pi}{l}n}
 \right].
\label{finiteweyl}
\end{align}\normalsize
In particular, 
if a specific mode \en{ n_0 } is chosen in the closed loop,
a single-mode transfer function is given by
\begin{subequations}
\label{singleweyl}
\small\begin{align}
 \ipg\dd{\phi|\phi}(s)
&=
 \frac{1}{l}
 \left[
 \frac{\Lambda\dd{+}}{ s+\im\ffrac{2\pi}{l}n_0}
+
 \frac{\Lambda\dd{-}}{-s+\im\ffrac{2\pi}{l}n_0}
 \right]
\\ &=
 \A{\ffrac{1}{sl+\im 2\pi n_0}}{}{}
   {\ffrac{1}{-sl+\im 2\pi n_0}},
\end{align}\normalsize
\end{subequations}
which is the same form as (\ref{fprop5}).
Furthermore, 
if \en{ n_0=0 },
\small\begin{align}
 \ipg\dd{\phi|\phi}(s)
&=
 (\Lambda\dd{+}-\Lambda\dd{-})
 \frac{1}{sl}
=
 \A{\ffrac{1}{sl}}{}
   {}{\ffrac{1}{-sl}}.
\end{align}\normalsize
This single-mode transfer function is also obtained 
from a limit \en{ l\to 0 } in (\ref{periodweyltf}):
\begin{subequations}
\label{periodweyll}
\small\begin{align}
 \ipg\dd{\phi|\phi}(s)
&=
 \frac{\Lambda\dd{+}-\Lambda\dd{-}}{2} 
 \frac{1+\ex\uu{-sl}}{1-\ex\uu{-sl}},
\\ &  \sim
 (\Lambda\dd{+}-\Lambda\dd{-})
 \frac{1}{sl},
\end{align}\normalsize
\end{subequations}
which means that 
the zero-momentum field \en{ \mas } of Section \ref{sec:p=0} 
is equivalent to a closed-loop field 
in a very small interval \en{ l\ll 1 }.

\chapter{Forward traveling field}
\label{chap:fp}
\thispagestyle{fancy}

In this chapter, 
we introduce (nonrelativistic) 
forward and backward traveling fields 
to describe massless particles propagating unidirectionally 
in one dimension \en{ x=(t,0,0,z) }.
These fields are used for quantum gates and circuits
in Chapters \ref{chap:boundary} and \ref{chap:circuit}.
We also introduce forward traveling closed-loop fields
in the same way as we have done for the Weyl field
in Section \ref{sec:weyltf}.
This is used for the formulation of 
feedback in Chapter \ref{chap:feedback}.

\section{Lagrangian of the forward traveling field}
\label{sec:fpfbasics}

In Section \ref{rem:weyl},
the one-dimensional Weyl equation
has been decomposed into the forward and backward traveling components,
\en{ \phi_F } and \en{ \phi_B }:
\begin{subequations}
 \small\begin{align}
 \im\partial\dd{+} \phi_{F}(x) &= 0,
\label{for-weyl} \\
 \im\partial\dd{-} \phi_{B}(x) &= 0,
 \end{align}\normalsize
\end{subequations}
where the four-position is \en{ x=(t,0,0,z) } and 
\en{ \partial\dd{\pm}\equiv \partial\dd{t} \pm \partial\dd{z} }.
Free field solutions are obtained from 
a continuum limit \en{ l\to \infty } in (\ref{weylfex}):
\begin{subequations}
\small\begin{align}
 \phi_F(x)
&=
 \int \frac{d\omega}{\sqrt{2\pi}}
 \left[
 a(\omega)     u_F \ex\uu{-\im\omega(t-z)}
+
 b\dgg(\omega) v_F \ex\uu{ \im\omega(t-z)}
 \right],
\label{weylf-1}
\\
 \phi_B(x)
&=
 \int \frac{d\omega}{\sqrt{2\pi}}
 \left[
 c(\omega)     u_B \ex\uu{-\im\omega(t+z)}
+
 d\dgg(\omega) v_B \ex\uu{ \im\omega(t+z)}
 \right].
\label{weylb-1}
\end{align}\normalsize
\end{subequations}
The first and second terms describe particles and antiparticles, 
respectively.

\marginpar{\vspace{33mm}
\footnotesize
Since we separate the positive and negative energy components,
forward and backward traveling fields are nonrelativistic.
\normalsize
}

In the Weyl equation, 
the forward traveling component \en{ \phi_F } is a two-component vector.
In an appropriate basis, this is expressed as
\small\begin{align}
 \phi_F(x)
&=
 \phi(x) u_F + \varphi\dgg(x) v_F
=
 \AV{\phi}{\varphi\dgg}.
\label{separate}
\end{align}\normalsize
where \en{ \phi } is a scalar field in the positive energy subspace,
satisfying (\ref{for-weyl}).
This is called 
a \textit{forward traveling field}.\index{forward traveling field}
(Likewise,
a \textit{backward traveling field}\index{backward traveling field} 
is defined by the positive energy component of \en{ \phi_B }.)
The Lagrangian density\index{Lagrangian (forward traveling)}
of \en{ \phi } is given by
\small\begin{align}
% L&=\int dz \ \lag,
%\\
 \lag
&=
 \phi\dgg \im\partial\dd{+} \phi.
\label{forlag}
\end{align}\normalsize
It follows from the Euler-Lagrange equation that
\begin{subequations}
\label{elfp}
\small\begin{align}
 \partial\dd{+}
 \frac{\partial \lag}{\partial(\partial\dd{+} \phi\dgg)}
-\frac{\partial \lag}{\partial \phi\dgg} =0
\quad & \Rightarrow \quad 
 \im\partial\dd{+} \phi=0,
\\ 
\partial\dd{+}
 \frac{\partial \lag}{\partial(\partial\dd{+} \phi)}
-\frac{\partial \lag}{\partial \phi} =0
\quad & \Rightarrow \quad 
 \im \partial\dd{+} \phi\dgg  =0.
\end{align}\normalsize
\end{subequations}
Canonical momenta\index{canonical momentum (forward traveling)} 
corresponding to \en{ \phi } and \en{ \phi\dgg }
are, respectively, given as
\small\begin{align}
 \pi
&\equiv
 \frac{\partial \lag}{\partial(\partial\dd{t}\phi)}
=\im \phi\dgg,
\quad
 \pi\dgg
\equiv
 \frac{\partial \lag}{\partial(\partial\dd{t}\phi\dgg)}
=0.
\label{phimoment}
\end{align}\normalsize
We introduce 
canonical quantization\index{canonical quantization (forward traveling)}
to the forward traveling field as
\small\begin{align}
 [\phi(t,z),\pi(t,z')]
=
 [\phi(t,z),\im\phi\dgg(t,z')]
=
 \im \delta(z-z').
\label{fpcom}
\end{align}\normalsize

Note that there are different choices of Lagrangian
for the forward traveling field.
Sometimes it is convenient to use a symmetric form defined as
\small\begin{align}
 \lag
&=
 \frac{1}{2}\left[ 
 \phi\dgg (\im\partial\dd{+} \phi) 
-
 (\im\partial\dd{+}\phi\dgg) \phi\right]
=
 \frac{1}{2} \Phi\dgg \sigma\dd{z} \im \partial\dd{+} \Phi,
\label{lagfor}
\end{align}\normalsize
where 
\small\begin{align}
 \sigma\dd{z} = \A{1}{}{}{-1}, \quad
 \Phi \equiv \AVl{\phi}{\phi\dgg}.
\end{align}\normalsize
The corresponding canonical momenta are given by
\small\begin{align}
\Pi \equiv 
 \AH{\pi}{\pi\dgg}
=
 \frac{\partial \lag}{\partial(\partial\dd{t}\Phi)}
=
 \frac{\im}{2}\AH{\phi\dgg}{-\phi}.
\end{align}\normalsize

\section{Transfer function of the forward traveling field}
\label{sec:tf}

We introduce a transfer function in the same way as the Dirac field.

\begin{definition}
Given two operators \en{ A } and \en{ B }, 
a contraction\index{contraction (forward traveling)} is defined as
\small\begin{align}
 \wick{A(x_2)}{B(x_1)}&\equiv
 \contraction{}{\phi}{\hspace{3mm}}{\hspace{0mm}}
 A_2B_1
\nn\\ &=
  \mzero{\T \ A_2B_1}
\nn\\ &=
 \mzero{\step(t_2-t_1) \, A_2B_1
       +\step(t_1-t_2) \, B_1A_2}.
\label{def:contraction}
\end{align}\normalsize
\end{definition}

Note that the sign of the second term is different from the Dirac field 
because the forward traveling field is bosonic.
By definition, 
it follows that
\small\begin{align}
 \wick{A_2}{B_1}=\wick{B_1}{A_2}.
\label{prop:ordersym}
\end{align}\normalsize
For a forward traveling field \en{ \phi } 
and an arbitrary operator \en{ B },
\begin{subequations}
\label{forsha}
\small\begin{align}
\hspace{-15mm}
 \im\partial\dd{+}
  \wick{\phi(x)}{B(0)} 
&= 
 +\im \delta(t) \bra{0} \, [\phi(0,z),B(0,0)] \, \ket{0},
\label{forsha1} \\ 
 \im\partial\dd{+} 
 \wick{B(0)}{\phi(x)}
&= 
 -\im \delta(t) \bra{0} \, [B(0,0),\phi(0,z)] \, \ket{0}.
\label{forsha2}
\end{align}\normalsize
\end{subequations}

\begin{definition}
\label{def:dd}
Given an operator \en{ A }, 
its double dagger \en{ A\ddgg }\index{$\ddagger$ (forward traveling)} 
is defined as an operator satisfying
\small\begin{align}
 [A(t,z),A\ddgg(t,z')]&=\delta(z-z').
\label{def:sha}
\end{align}\normalsize
\end{definition}

Note that there does not necessarily exist \en{ A\ddgg }
for all \en{ A }.
If it exists,
we have
\small\begin{align}
 (A\ddgg)\ddgg &= -A,
\qquad
 (\im A)\ddgg = -\im A\ddgg.
\label{prop:sha2}
\end{align}\normalsize
For \en{ \phi },
the double dagger is similar to the Hermitian conjugate
\small\begin{align}
 \phi\ddgg&=\phi\dgg,
\qquad
(\phi\dgg)\ddgg = -\phi.
\label{phisha}
\end{align}\normalsize
However, 
for quadrature operators\index{quadrature operator} defined as
\small\begin{align}
 \xi &\equiv \frac{\phi\dgg+\phi}{\sqrt{2}},
\qquad
 \eta \equiv \im\frac{\phi\dgg-\phi}{\sqrt{2}},
\label{quadrature}
\end{align}\normalsize
the double dagger is different from the Hermitian conjugate:
\small\begin{align}
\xi\ddgg &= -\im \eta,
\qquad
 \eta\ddgg = \im \xi.
\end{align}\normalsize

\marginpar{\vspace{-55mm}
\footnotesize
The quadrature operators are also expressed as \\ \\
$\AVl{\xi}{\eta}= \frac{1}{\sqrt{2}}\A{1}{1}{-\im}{\im} \AVl{\phi}{\phi\dgg}$.
\normalsize
}

\marginpar{\vspace{-20mm}
\footnotesize
In general, $A\ddgg$ corresponds to a canonically conjugate pair of $A$,
rather than the Hermitian conjugate.
\normalsize
}

\begin{definition}
\label{def:tf}
A transfer function\index{transfer function (forward traveling)} 
from \en{ B(x_1)=B_1 } to \en{ A(x_2)=A_2 } is defined by
\small\begin{align}
 \ipg\dd{A|B}(x_2;x_1)
&\equiv  
 \wick{A_2}{B\ddgg_1}.
\label{ptf-f}
\end{align}\normalsize
\end{definition}

\marginpar{\vspace{-0mm}
\footnotesize
Note that 
$\pg\dd{\phi|\phi}$ and $\pg\dd{\phi\dgg|\phi\dgg}$
are retarded and advanced propagators,
respectively.
For the quadratures,\\
\\
$ \pg\dd{\xi|\xi}
=
 \pg\dd{\eta|\eta}
=
 \ffrac{\pg\dd{\phi|\phi} + \pg\dd{\phi\dgg|\phi\dgg}}{2},
$ \\
\\
which correspond to the Feynman propagator.
\normalsize
}

We have defined a transfer function in this way 
because it can simplify input-output relations.
For example, 
\small\begin{align}
 \AV{\xi}{\eta}\drm{out}
&=
 \A{\ipg\dd{\xi\drm{out} | \xi\drm{in} }}
   {\ipg\dd{\xi\drm{out} | \eta\drm{in} }}
   {\ipg\dd{\eta\drm{out} | \xi\drm{in} }}
   {\ipg\dd{\eta\drm{out} | \eta\drm{in} }}
 \AV{\xi}{\eta}\drm{in}.
\label{visinout}
\end{align}\normalsize
where 
the (1,1)-element is a probability amplitude 
from \en{ \xi\drm{in} } to \en{ \xi\drm{out} },
of which a correct form is given by a contraction 
between \en{ \xi\drm{out} } and \en{ -\im \eta\drm{in} } 
(not between \en{ \xi\drm{out} } and \en{ \xi\drm{in} }!):
\begin{subequations}
\small\begin{align}
 \ipg\dd{\xi\drm{out} | \xi\drm{in} }
&=
 \wick{\xi\drm{out} }{(\xi\drm{in})\ddgg}
\\ &=
 \wick{\xi\drm{out} }{-\im\eta\drm{in}}.
\end{align}\normalsize
\end{subequations}

For later use,
we consider transfer functions for two different situations:
a static transfer function and
closed-loop transfer function.
In both cases, the transfer functions are defined 
at the same point in space.

\subsection{Static transfer function}

Let us consider a transfer function defined as
\small\begin{align}
 \ipg\dd{\phi|\phi}(t)
&=
 \delta(t).
\end{align}\normalsize
In the hyperfunction approach,
it has two independent solutions:
\begin{subequations}
\label{fthrou}
\small\begin{align}
\mbox{\normalsize causal: \small} \
 \ipg\dd{\phi|\phi}(\omega)
&=
 [\step\dd{\ima(\omega)>0}],
\label{fthrou-1}\\
\mbox{\normalsize anticausal: \small} \
 \ipg\dd{\phi|\phi}(\omega)
&=
 [-\step\dd{\ima(\omega)<0}].
\label{fthrou-2}
\end{align}\normalsize
\end{subequations}
These correspond to the probability amplitudes 
of \en{ z \to z+0 } and \en{ z \to z-0 }, 
respectively.
The sign of the transfer function changes
if we reverse the direction of propagation of the field.

\newpage

\subsection{Closed-loop forward traveling field}
\label{sec:ltf}

\begin{wrapfigure}[0]{r}[53mm]{49mm}
\vspace{-0mm}
\centering
\includegraphics[keepaspectratio,width=47mm]{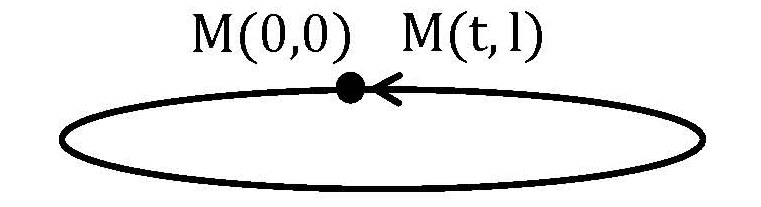}
\caption{
\small
Closed-loop field.
\normalsize
}
\label{fig-periodf}
\end{wrapfigure}

Let us consider a closed-loop field\index{closed loop (forward traveling)} 
\en{ \mas(t,z) } in the same way as Sections \ref{sec:weyltf}.
As shown in Figure \ref{fig-periodf},
this is a forward traveling field
that propagates in a finite interval \en{ z\in[0,l] }
under periodic boundary conditions
\small\begin{align}
 \mas(t,l)&=\mas(t,0).
% \mas(t,l)&=\mas(0,0).
\label{ftpbound}
\end{align}\normalsize
Here we calculate transfer functions
from \en{ z=0 } to the same point after traveling in a closed loop:
\begin{subequations}
\label{ftcltf}
\small\begin{align}
 \ipg\dd{\mas|\mas}(t)
&=
 \wick{\mas(t,l)}{\mas\dgg(0,0)},
\label{ftcltf-1}\\ 
 \ipg\dd{\mas\dgg|\mas\dgg}(t)
&=
 -\wick{\mas\dgg(t,l)}{\mas(0,0)}.
\label{ftcltf-2}
\end{align}\normalsize
\end{subequations}
Since \en{ \ipg\dd{\mas|\mas} } is a retarded Green's function,
it is given by the causal component 
of \en{ \ipg\dd{\phi\dd{F}|\phi\dd{F}} },
i.e., the coefficient of \en{ \Lambda\dd{F+} } in (\ref{wftf-c}):
\small\begin{align}
 \ipg\dd{\mas|\mas}(s)
&=
 \frac{1}{2} 
 \frac{1+\ex\uu{-sl}}{1-\ex\uu{-sl}}.
\end{align}\normalsize
On the other hand,
\en{ \ipg\dd{\mas\dgg|\mas\dgg} }
is given by the anticausal component
of \en{ \ipg\dd{\phi\dd{F}|\phi\dd{F}} }.
Noting the sign difference between bosons and fermions in the contraction,
we have
\small\begin{align}
\ipg\dd{\mas\dgg|\mas\dgg}(s)
&=
 \frac{1}{2} 
 \frac{1+\ex\uu{-sl}}{1-\ex\uu{-sl}}.
\end{align}\normalsize

Let us introduce a vector 
\small\begin{align}
 \bm{\mas}
&=
 \AV{\mas}{\mas\dgg}.
\end{align}\normalsize
A transfer function from \en{ \bm{\mas} } to \en{ \bm{\mas} } is defined as
\small\begin{align}
 \ipg\dd{\bm{\mas}|\bm{\mas}}
&\equiv
 \A{\ipg\dd{\mas|\mas}}{\ipg\dd{\mas|\mas\dgg}}
   {\ipg\dd{\mas\dgg|\mas}}{\ipg\dd{\mas\dgg|\mas\dgg}}.
\end{align}\normalsize
The off-diagonal elements are zero
because 
\en{\ipg\dd{\mas|\mas\dgg}=-\wick{\mas}{\mas} \sim \bra{0}\mas\mas\ket{0}=0}.
As a result, 
we get 
\small\begin{align}
 \ipg\dd{\bm{\mas}|\bm{\mas}}
=
 \A{\ffrac{1}{2}\ffrac{1+\ex\uu{-sl}}{1-\ex\uu{-sl}}}{} 
   {}{\ffrac{1}{2}\ffrac{1+\ex\uu{-sl}}{1-\ex\uu{-sl}}}.
\label{mmatrix}
\end{align}\normalsize
A single-mode transfer function is obtained 
in the same way as Sections \ref{sec:weyltf}:
\small\begin{align}
 \ipg\dd{\bm{\mas}|\bm{\mas}}
&=
 \A{\ffrac{1}{sl+2\pi \im n_0}}{}{}
   {\ffrac{1}{sl+2\pi \im n_0}}.
\label{singlen}
\end{align}\normalsize
Likewise, 
in a limit \en{ l\ll 1 },
\small\begin{align}
 \ipg\dd{\bm{\mas}|\bm{\mas}}
&=
 \A{\ffrac{1}{sl}}{}{}
   {\ffrac{1}{sl}}.
\end{align}\normalsize

\subsection{Asymmetry}
\label{lem:asymmetry}

The forward traveling field has unique asymmetry 
due to its unidirectionality.
%This property is useful to calculate transfer functions
%from Feynman diagrams.
Note that the two matrix elements of (\ref{mmatrix})
are the same function.
Each element is written as
\begin{subequations}
\label{asym1}
\small\begin{align}
\mbox{\normalsize (1,1)-element: \small} \hspace{5mm}
 \ipg\dd{\mas|\mas}(t)
&=
\hspace{2.5mm} \wick{\mas(t)}{\mas\dgg(0)}
\xrightarrow{LT}
\hspace{2.5mm} \wick{\mas}{\mas\dgg}(s),
\label{asym1-1}
\\
\mbox{\normalsize (2,2)-element: \small} \hspace{2mm}
 \ipg\dd{\mas\dgg|\mas\dgg}(t)
&=
 -\wick{\mas\dgg(t)}{\mas(0)}
\xrightarrow{LT}
 -\wick{\mas\dgg}{\mas}(s),
\label{asym1-2}
\end{align}\normalsize
\end{subequations}
where \en{ \xrightarrow{LT} } represents the Laplace transform.
As a result,
\small\begin{align}
 \wick{\mas}{\mas\dgg}
&=
 -\wick{\mas\dgg}{\mas}.
\label{asym}
\end{align}\normalsize

This asymmetry is obtained in a different way.
Note that (\ref{mmatrix}) is an odd function in \en{ s }.
Then (\ref{asym1-2}) is written as
\begin{subequations}
\small\begin{align}
 -\wick{\mas\dgg(t)}{\mas(0)}
=&
 -\wick{\mas(0)}{\mas\dgg(t)}
\qquad \because (\ref{prop:ordersym})
\\ \xrightarrow{LT} &
 -\wick{\mas}{\mas\dgg}(-s)
\\ =& \hspace{4.5mm}
 \wick{\mas}{\mas\dgg}(s),
\hspace{13mm} \because \mbox{odd in \en{ s} }
\end{align}\normalsize 
\end{subequations}
which is (\ref{asym1-1}).
We have used the time-reversal property of the Laplace transform
in the second line.

It is important to note that 
(\ref{asym1-1}) and (\ref{asym1-2}) correspond to
the causal and anticausal solutions,
as in (\ref{fthrou}) of the hyperfunction approach.
Hence 
the asymmetry (\ref{asym}) is %related to time reversal
regarded as 
a conversion rule between the causal and anticausal solutions,
rather than an equality.
In fact,
it is used to reverse the direction of arrowed edges
in Feynman diagrams (Section \ref{sec:feydiag}).
The LHS of (\ref{asym}) is expressed 
by an edge pointing in the forward (causal) direction in time,
whereas the RHS is backward (anticausal).
When we calculate a transfer function using a Feynman diagram,
we need to put all edges in the forward direction 
using (\ref{asym}).
Then we obtain a correct form of the transfer function.

\chapter{Gauge theory}
\label{chap:gauge}
\thispagestyle{fancy}

In this chapter, 
we briefly review a gauge theory.
It is one of the most successful methods in field theory.
In particular, 
Yang-Mills theory is essential for the unification of 
the electromagnetic and weak forces, quantum chromodynamics, 
and the theory of the strong force.
We apply the same idea to quantum computing 
by regarding quantum gates as local gauge transformations.
First we review global symmetries and Noether's theorem,
and then,
introduce a gauge theory through local symmetries.
We also consider an extension of the gauge theory 
to non-unitary transformations.
Gauge transformations are always unitary for fermions,
but some quantum gates are described 
by non-unitary transformations for bosons.

\section{Global symmetries and Noether's theorem}
\label{sec:gaugep}

If observables are invariant under unobservable field transformations,
the corresponding action or Lagrangian should be invariant as well.
For example, 
a global phase shift (gauge transformation)\index{gauge transformation (global)}
of a Dirac field \en{ \psi } is not detectable:
\small\begin{align}
\hspace{25mm}
 \psi \to \psi'= \ex\uu{\im \win}\psi 
\hspace{10mm}
 (\win: \mbox{real constant}.)
\label{gps}
\end{align}\normalsize
Obviously the Lagrangian of the Dirac field is invariant 
under this transformation:
\small\begin{align}
 \lag(\psi', \partial\dd{\mu}\psi')
=
 \widebar{\psi}\bigl(\im\fsh{\partial}-m \bigr) \psi
=
 \lag(\psi, \partial\dd{\mu}\psi).
\label{lagd-1}
\end{align}\normalsize
The invariance of the Lagrangian is called \textit{symmetry}.
\index{symmetry}
In this case, 
it is called a \textit{global symmetry}\index{symmetry (global)}
because \en{ \win } is constant and 
the transformation (\ref{gps}) is defined globally.

Conversely,
if a Lagrangian has such symmetry,
there exists a conserved current.
This is known as \textit{Noether's theorem}.\index{Noether's theorem}
To see this, 
assume that a field \en{ \psi } has an infinitesimal change
\small\begin{align}
 \psi
\to
 \psi'=\psi+\delta\psi,
\label{psichange}
\end{align}\normalsize
and define a \textit{Noether current}\index{Noether current} as
\small\begin{align}
 j\uu{\mu}
\equiv
 \frac{\delta\lag}{\delta(\partial\dd{\mu}\psi)} \delta\psi.
\end{align}\normalsize
Then
\small\begin{align}
 \partial\dd{\mu} j\uu{\mu}
&=
 \partial\dd{\mu}
 \left[
 \frac{\delta\lag}{\delta(\partial\dd{\mu}\psi)} 
\right]
 \delta \psi
+
 \left[
 \frac{\partial\lag}{\delta(\partial\dd{\mu}\psi)}
\right] 
 \partial\dd{\mu} (\delta\psi).
\label{dnc}
\end{align}\normalsize
In response to (\ref{psichange}),
the Lagrangian (density) changes as
\begin{subequations}
\label{noether1}
\small\begin{align}
 \delta \lag(\psi,\partial\dd{\mu}\psi)
&=
 \frac{\delta\lag}{\delta \psi}\delta \psi
+
 \left[
 \frac{\delta\lag}{\delta(\partial\dd{\mu}\psi)} 
 \right]
 \delta(\partial\dd{\mu}\psi)
\\ &=
 \frac{\delta\lag}{\delta \psi}\delta \psi
+
 \left[
 \frac{\delta\lag}{\delta(\partial\dd{\mu}\psi)} 
 \right]
 \partial\dd{\mu} (\delta\psi)
\\ &=
 \frac{\delta\lag}{\delta \psi}\delta \psi
- \partial\dd{\mu}
 \left[
 \frac{\delta\lag}{\delta(\partial\dd{\mu}\psi)} 
\right]
 \delta \psi
+
 \partial\dd{\mu} j\uu{\mu}
\quad 
\because (\ref{dnc})
\\ &=
 \left[
 \frac{\delta\lag}{\delta \psi}
-
 \partial\dd{\mu} \frac{\delta\lag}{\delta(\partial\dd{\mu}\psi)}
\right]
 \delta\psi
+
 \partial\dd{\mu}
 j\uu{\mu}.
\end{align}\normalsize
\end{subequations}
The first term is zero due to the Euler-Lagrange equation.
The global symmetry \en{ \delta\lag=0 } 
indicates that \en{ j\uu{\mu} } is a conserved current:
\small\begin{align}
 \partial\dd{\mu} j\uu{\mu}
&=
 0.
\label{noether3}
\end{align}\normalsize

\begin{example}
\rm
Let us consider an infinitesimal phase shift of the form (\ref{gps}):
\small\begin{align}
 \psi'
&= 
 \ex\uu{-\im \delta\win}\psi
\sim
 \psi + \delta \psi,
\qquad
 (\delta \psi\equiv -\im \delta\win \, \psi.)
\end{align}\normalsize
For the Dirac field,
the Noether current is written as
\small\begin{align}
 j\uu{\mu}
&=
 \widebar{\psi} \gamma\uu{\mu} \psi
=
 (\psi\dgg \psi, \psi\dgg \bm{\alpha}\psi)
\equiv
 (\rho,\bm{j}),
\end{align}\normalsize
where we have dropped \en{ \delta\win }.
In this case, 
(\ref{noether3}) is written as
\small\begin{align}
 \partial\dd{t} \rho + \nabla\cdot \bm{j} 
=
 0,
\end{align}\normalsize
which corresponds to a continuity equation.
\qed
\end{example}

\section{Local symmetries and a gauge principle}
\label{sec:gsym}

Now assume that \en{ \win=\win(x) }.
The Lagrangian \en{ \lag(\psi,\partial\dd{\mu}\psi) } 
is no longer invariant under the phase shift (\ref{gps})
because \en{ \partial\dd{\mu}\theta\not= 0 }. 
This influences some observables such as \en{ p\dd{\mu}=\im \partial\dd{\mu} },
which means that some kind of force is acting on the field \en{ \psi }.
On the other hand, the source of the force is also influenced
by \en{ \psi } as back action so that 
a total Lagrangian remains unchanged.
This is called a \textit{gauge principle}.\index{gauge principle}
In this case,
the invariance of the total Lagrangian 
is called a \textit{local symmetry}.\index{symmetry (local)}

We are interested in an interaction Lagrangian 
creating the force on \en{ \psi }.
To find it,
we consider multiple fields \en{ \{\psi_1,\psi_2,\cdots\} }
\small\begin{align}
 \psi
\equiv
 \BV{\psi_1}{\psi_2}{\vdots},
\hspace{14mm}
 \widebar{\psi}
\equiv
 \BH{\widebar{\psi}_1}{\widebar{\psi}_2}{\cdots},
\label{psivec}
\end{align}\normalsize
and a local gauge transformation\index{gauge transformation (local)}
\small\begin{align}
\hspace{-9mm}
\psi \to \psi'= \uni\psi,
\hspace{9mm}
 \widebar{\psi} \to \widebar{\psi}'= \widebar{\psi}\uni\dgg.
\end{align}\normalsize
\en{ \uni } is a unitary operator written as
\small\begin{align}
 \uni(x)=\exp[ \win\uu{a}(x) \lie\dd{a}],
\label{gauget}
\end{align}\normalsize
where
\en{ \win\uu{a} } is a (real) guage function\index{gauge function}
and \en{ \lie\dd{a} } is  a constant matrix satisfying 
\en{ \lie\dd{a}\dgg=-\lie\dd{a} }.

\vspace{-3mm}

\subsection{Interaction Lagrangian}
\label{sec:gauge-int}

To find the interaction Lagrangian,
we first note that 
the phase factor \en{ \partial\dd{\mu}\win\uu{a}\not= 0 }
is the reason for the change of the Lagrangian \en{ \lag }.
To make \en{ \lag } invariant,
we introduce phase modulation known as 
a \textit{gauge covariant derivative}\index{gauge covariant derivative}
\small\begin{align}
 D\dd{\mu}
&\equiv
 \partial\dd{\mu}  - \ga\dd{\mu},
\label{codev}
\end{align}\normalsize
where 
\en{ \ga\dd{\mu} } is a variable 
to cancel \en{ \partial\dd{\mu}\win\uu{a}\not= 0 } out.
This is called a \textit{gauge field}.\index{gauge field}
Note that 
as \en{ \psi\to \psi' },
the gauge field also transforms as \en{ \ga\dd{\mu}\to\ga\dd{\mu}' }
and hence \en{ D\dd{\mu}\to  D\dd{\mu}' }.
This transformation is determined to satisfy the gauge principle:
\small\begin{align}
\lag(\psi,D\dd{\mu}\psi)
= 
 \lag(\psi',D\dd{\mu}'\psi').
\label{delete}
\end{align}\normalsize

Assume that the gauge field transforms as
\small\begin{align}
 \ga\dd{\mu}'
&=
 \uni \ga\dd{\mu} \uni\dgg  
 + 
 (\partial\dd{\mu} \uni)\uni\dgg.
\label{gaugef}
\end{align}\normalsize
This is equivalent to 
\small\begin{align}
 \uni\dgg D\dd{\mu}' \uni 
&= 
 \uni\dgg (\partial\dd{\mu} - \ga\dd{\mu}') \uni
=
 D\dd{\mu},
\label{gtderiv}
\end{align}\normalsize
from which (\ref{delete}) follows.
As a result,
if there exists a gauge field \en{ \ga\dd{\mu} } satisfying (\ref{gaugef}),
the Lagrangian
\small\begin{align}
 \lag(\psi, D\dd{\mu}\psi)
&=
 \lag(\psi,\partial\dd{\mu}\psi) + \lag\urm{int}
\label{lagD}
\end{align}\normalsize
is locally symmetric.
Here
\small\begin{align}
 \lag\urm{int}
&=
 - \widebar{\psi} \im \gamma\uu{\mu} \ga\dd{\mu} \psi
\label{gaugeint}
\end{align}\normalsize
is regarded as an interaction Lagrangian
creating the force on \en{ \psi }.

\subsection{Gauge fixing and quantum gates}

We have introduced the gauge theory 
because it is useful to describe quantum gates.
For instance, if \en{ \lie } is a scalar, 
\en{ \psi } and \en{ \gel\dd{\mu} } transform as
\begin{subequations}
\small\begin{align}
 \psi' &= \uni \psi,
\hspace{16mm} (\uni = \ex\uu{\win\lie})
\label{exu1tf-1} \\
 \gel\dd{\mu}' &= \gel\dd{\mu} + \partial\dd{\mu} \win.
 \hspace{9mm}
 \because (\ref{gaugef})
\label{exu1tf-2}
\end{align}\normalsize
\end{subequations}
Then 
\en{ \gel\dd{\mu}=(\phi,\bm{\gel}) } is identified with
an electromagnetic four-potential. 
\en{ \gel\dd{\mu} } and  \en{ \gel'\dd{\mu} } 
are equivalent due to the gauge symmetry.
However,
we sometimes break this symmetry 
to simplify Maxwell's equations for \en{ \gel'\dd{\mu} }
by choosing a specific gauge parameter \en{ \win } 
(gauge fixing).\index{gauge fixing}
For example,
if we choose \en{ \win } so that 
\en{ \Delta \win =-\nabla\cdot\bm{\gel} },
then \en{ \gel\dd{\mu}' } satisfies \en{ \nabla\cdot \bm{\gel}'=0 }.
This is known as the Coulomb gauge.

Here we introduce a different interpretation for \en{ \win }.
We regard (\ref{exu1tf-1}) as 
the input-output relation of a quantum gate.
In this case,
\en{ \win } serves as a switch to turn the gate on and off,
i.e.,
the gate is off (\en{ \uni=I })
if \en{ \win=0 }
and 
the gate is on (\en{ \uni\not= I })
if \en{ \win\not= 0 }.
In this interpretation,
\en{ \gel\dd{\mu} } corresponds to the state of the gate
and 
(\ref{exu1tf-2}) describes 
how \en{ \gel\dd{\mu} } changes during the gate operation.
\en{ \partial\dd{\mu}\win } is regarded as 
a cost to produce the output \en{ \psi' }.
In Chapter \ref{chap:boundary},
we formulate quantum gates based on this interpretation.

\newpage

\subsection{Remarks on the gauge field}
\label{sec:rem-gf}

\begin{wrapfigure}[0]{r}[53mm]{49mm}
\centering
\includegraphics[keepaspectratio,width=49mm]{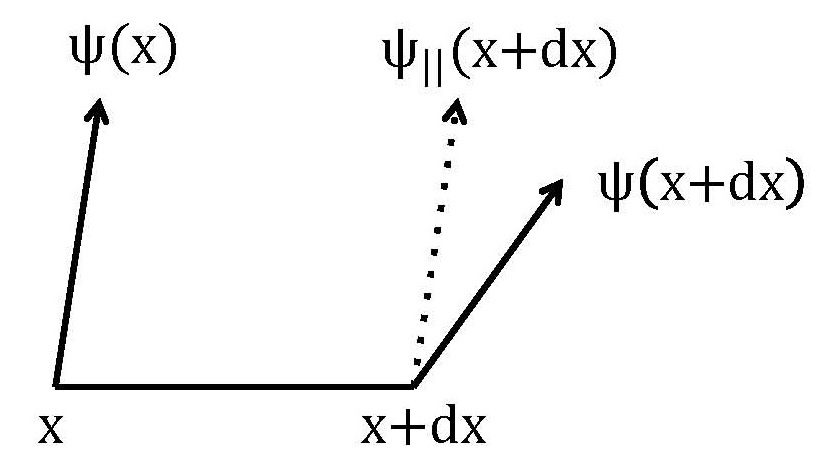}
\caption{
\small
Parallel translation of $\psi(x)$.
\normalsize
}
\label{fig-in-diffg-1}
\end{wrapfigure}

Let us briefly review a mathematical interpretation of the gauge field 
from a differential geometrical perspective.
A parallel translation is 
the key to understanding the gauge field.
We first note that 
\en{ \partial\dd{\mu} } is defined as
\small\begin{align}
 \psi(x+dx)
=
 \psi(x) + \partial\dd{\mu} \psi(x) dx\uu{\mu}.
\label{gd-1}
\end{align}\normalsize
Let \en{ \psi\dd{\parallel}(x+dx) } 
be a parallel translation of \en{ \psi(x) } 
to the point \en{ x+dx } 
as in Figure \ref{fig-in-diffg-1}.
Its differential is proportional to 
an infinitesimal rotation of the field, 
which is called a 
\textit{connection}\index{connection (differential geometry)}
in differential geometry, 
as
\small\begin{align}
 \psi\dd{\parallel}(x+dx)
=
 \psi(x) + \ga\dd{\mu} \psi(x) dx\uu{\mu}.
\label{gd-2}
\end{align}\normalsize
The rotation operator \en{ \ga\dd{\mu} } corresponds to the gauge field.
The gauge covariant derivative\index{gauge covariant derivative}
\en{ D\dd{\mu} = \partial\dd{\mu}-\ga\dd{\mu} }
is obtained from the difference between (\ref{gd-1}) and (\ref{gd-2}):
\small\begin{align}
 \psi(x+dx) -  \psi\dd{\parallel}(x+dx)
&=
 D\dd{\mu} \psi(x) dx\uu{\mu}.
\label{geo-cov}
\end{align}\normalsize

\begin{wrapfigure}[0]{r}[53mm]{49mm}
\centering
\includegraphics[keepaspectratio,width=45mm]{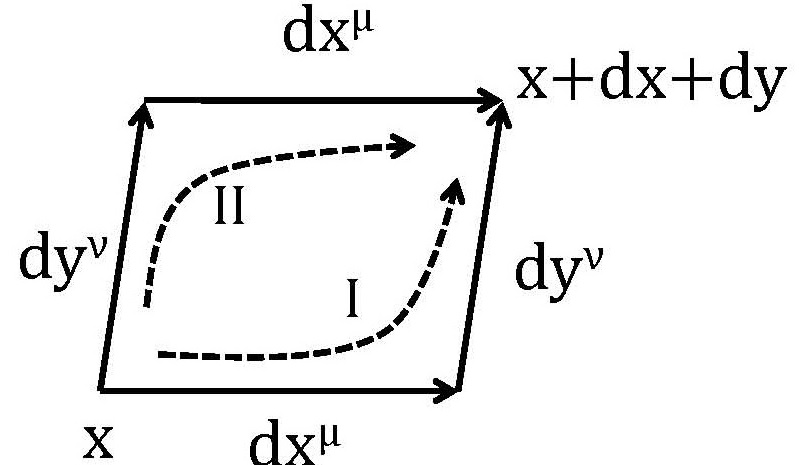}
\caption{
\small
Parallel translation along two paths I and II.
\normalsize
}
\label{fig-in-diffg-2}
\end{wrapfigure}

Next we consider a parallel translation of \en{ \psi } 
along the paths I and II in Figure \ref{fig-in-diffg-2}.
From (\ref{geo-cov}), 
we get
\begin{subequations}
\small\begin{align}
 \psi\urm{I}\dd{\parallel}(x+dx+dy)
&=
 \psi(x+dx+dy)
-D\dd{\mu}\psi(x+dy)dx\uu{\mu}
\\ &\hspace{30mm}
-D\dd{\nu}\psi(x+dx)dy\uu{\nu}
+D\dd{\nu}D\dd{\mu}\psi(x)dx\uu{\mu}dy\uu{\nu},
\nn\\
 \psi\urm{II}\dd{\parallel}(x+dx+dy)
&=
 \psi(x+dx+dy)
-D\dd{\mu}\psi(x+dy)dx\uu{\mu}
\\ &\hspace{30mm}
-D\dd{\nu}\psi(x+dx)dy\uu{\nu}
+D\dd{\mu}D\dd{\nu}\psi(x)dx\uu{\mu}dy\uu{\nu}.
\nn
\end{align}\normalsize
\end{subequations}
When we translate \en{ \psi } around the closed path,
its differential is given as
\small\begin{align}
 \Delta \psi(x)
&=
 \psi\urm{I}\dd{\parallel}(x+dx+dy) - \psi\urm{II}\dd{\parallel}(x+dx+dy)
=
 F\dd{\mu\nu} \psi(x) dx\uu{\mu}dy\uu{\nu},
\end{align}\normalsize
where \en{ F\dd{\mu\nu} } is called 
a \textit{field strength}\index{field strength}
or a \textit{curvature}\index{curvature (differential geometry)}
in differential geometry,
written as
\small\begin{align}
 F\dd{\mu\nu}
&\equiv 
 -[D\dd{\mu},D\dd{\nu}]
=
 \partial\dd{\mu} \ga\dd{\nu}
-
 \partial\dd{\nu} \ga\dd{\mu}
-
 [\ga\dd{\mu}, \ga\dd{\nu}].
\end{align}\normalsize

\marginpar{\vspace{10mm}
\footnotesize
In general, 
one can consider $\lag\urm{\gel}$ 
including higher order of $\partial\dd{\mu}$.
\normalsize
}

The Lagrangian of the gauge field \en{ \ga\dd{\mu} }
is defined from the field strength.
Up to second order of \en{ \partial\dd{\mu} },
it is written as
\small\begin{align}
 \lag\urm{\gel}
=
 -\frac{1}{2} \textrm{Tr} \left[ F\dd{\mu\nu} F\uu{\mu\nu} \right].
\label{lagw}
\end{align}\normalsize
This Lagrangian is also invariant under the gauge transformation.
To see this, 
note that the transformation of \en{ D\dd{\mu} } satisfies 
(\ref{gtderiv}).
The field strength transforms as
\small\begin{align}
 F'\dd{\mu\nu}
=
 \uni F\dd{\mu\nu} \uni\dgg,
\label{fstg}
\end{align}\normalsize
under which \en{ \lag\urm{\gel} } is invariant.
Finally
a symmetric Lagrangian of the entire process 
is given by 
\small\begin{align}
 \lag\urm{entire}
&=
 \lag(\psi,\partial\dd{\mu}\psi) 
+ 
 \lag\urm{int}
+ 
 \lag\urm{\gel}.
\end{align}\normalsize

\newpage

Historically
the significance of the gauge principle was recognized 
after Weyl discovered Maxwell's equations in \en{ \lag\urm{\gel} } 
for the abelian case of U(1).
Then
Yang and Mills 
found \en{ \lag\urm{\gel} } for non-abelian cases. 
Their ideas are used to unify 
the electromagnetic, weak, and strong interactions
as U(1) \en{ \times } SU(2) \en{ \times } SU(3)
in the Standard Model
where 
\en{ \lag\urm{\gel} } describes what kind of particles 
(the photon, weak bosons and gluons)
become force carriers for these three forces.

The situation is different 
when we apply the gauge theory to quantum gates.
The gauge field 
is 
what we tailor to control \en{ \psi }
and always implemented by the electromagnetic interaction.
In addition,
for instantaneous operations of quantum gates,
the gauge field is approximated to be static effectively.
Hence 
we focus on the interaction \en{ \lag\urm{int} } 
for the study of quantum gates
in what follows.
(The dynamics of the gauge field 
will be revisited in Chapter \ref{epilogue}.)

\section{Examples of local symmetries}
\label{sec:exlsym}

\subsection{U(1) symmetry}
\label{sec:u1sym}

The simplest example is U(1) symmetry.
It is defined by \en{ \lie=\im g \ (g\in \mathbb{R}) }
in (\ref{gauget}):
\small\begin{align}
 \uni(x)
&=
 \exp[\win(x) \im g].
\end{align}\normalsize
Let us express the gauge field as 
\en{ \ga\dd{\mu}=\gel\dd{\mu}\lie }.
From (\ref{gaugef}),
\en{ \gel\dd{\mu} } transforms as 
\small\begin{align}
 \gel\dd{\mu}'
&=
 \gel\dd{\mu} + \partial\dd{\mu}\win,
\label{u1w}
\end{align}\normalsize
which is the same transformation as an electromagnetic four-potential.
%$\gel\dd{\mu}=(\phi,\bm{A})$.
The interaction Lagrangian (\ref{gaugeint}) is written as
\small\begin{align}
 \lag\urm{int}
&= 
  g\widebar{\psi} \gamma\uu{\mu} \psi A\dd{\mu}.
\label{u1int}
\end{align}\normalsize
If the spin structure of \en{ \psi } is negligible,
only the scalar potential \en{ A_0 = \phi } is relevant
\small\begin{align}
 \lag\urm{int}
= 
  g\psi\dgg \psi \phi,
\label{u1yukawa}
\end{align}\normalsize
which is Yukawa's interaction.\index{Yukawa's interaction}
If the source of the potential is so heavy 
that \en{ \phi } is not influenced by the interaction,
it can be replaced with a classical potential \en{ V(\bm{x}) } as
\small\begin{align}
 \lag\urm{int}
= 
  g\psi\dgg \psi V.
\label{u1yukawa-1}
\end{align}\normalsize

\subsection{SU(2) symmetry}
\label{subsec:yangmills}

The second most-used example is SU(2).
Let us consider a local gauge transformation
defined by \en{ \lie\dd{a}=\im g T\dd{a} }:
\small\begin{align}
 \uni(x)
=
 \exp\left[\win\uu{a}(x) \im g T\dd{a} \right],
\label{su2ex}
\end{align}\normalsize
where \en{ \win\uu{a} \ (a=x,y,z) } are real functions and
\en{ g } is a real constant.
\en{ T\dd{a} = \sigma\dd{a}/2 }
with the Pauli matrices \en{ \sigma\dd{a} }
are the basis of su(2) and satisfy
\small\begin{align}
 [T\dd{a}, T\dd{b}]
&=
 \im \varepsilon\dd{abc} T\dd{c}.
\label{tcom}
\end{align}\normalsize
For simplicity, 
we assume that \en{ \uni } can be approximated as
\small\begin{align}
 \uni
\sim 
 1 + \win\uu{a} \im g T\dd{a}.
\end{align}\normalsize
Let us express the gauge field as
\small\begin{align}
\hspace{24mm}
  \ga\dd{\mu}
&=
 \gel\dd{\mu}\uu{a} \lie\dd{a}.
\hspace{10mm}
 (\gel\dd{\mu}\uu{a}\in \mathbb{R})
\label{gaugef-su2}
\end{align}\normalsize
Then (\ref{gaugef}) is written as
\small\begin{align}
 \gel\dd{\mu}^{\prime a}
&=
 \gel\dd{\mu}\uu{a}
+
 \partial\dd{\mu}\win\uu{a} 
+
 g\varepsilon\dd{bca}\gel\dd{\mu}\uu{b}\win\uu{c}.
\label{gaugef3}
\end{align}\normalsize
Compared to (\ref{u1w}),
the third term represents a non-abelian aspect of SU(2).
This gauge field is known as 
the \en{ W } and \en{ Z } bosons 
that mediate the weak interaction.
However, 
it can be described by 
the electromagnetic four-potential of U(1) 
for a special case,
as in the next example.

\begin{example}
\label{ex:su2gauge}
\rm
Let us consider a case where two Dirac fields 
\small\begin{align}
 \psi
&=
 \AV{\psi_1}{\psi_2}
\end{align}\normalsize
transform via the SU(2) gauge transformation
\small\begin{align}
 \psi'
=
 \ex\uu{\win \im 4g T\dd{y}} \psi.
\label{su2extheta}
\end{align}\normalsize
In this case, 
\small\begin{align}
\hspace{25mm}
 \lie\dd{a}
=
 \left(0, \im 4g T\dd{y}, 0\right), 
\qquad
 T\dd{y} = \frac{\sigma\dd{y}}{2}.
\label{bunkai-y}
\end{align}\normalsize
Only the \en{ y }-component involves in the transformation.
Let us redefine as \en{ A\dd{\mu} \equiv \gel\dd{\mu}\uu{y} }.
Then (\ref{gaugef3}) is written as
\small\begin{align}
 A\dd{\mu}'
&=
 A\dd{\mu}
+
 \partial\dd{\mu}\win,
\label{su2single}
\end{align}\normalsize
and the interaction Lagrangian (\ref{gaugeint}) is expressed as
\small\begin{align}
 \lag\urm{int} 
&=
 -2\im g
\left[
 \widebar{\psi}_1 \gamma\uu{\mu} \psi_2
-
 \widebar{\psi}_2 \gamma\uu{\mu} \psi_1
\right] A\dd{\mu}.
\label{su2gaugeintl}
\end{align}\normalsize
Under the classical limit \en{ A\dd{\mu}\to V(x) } 
introduced in Section \ref{sec:u1sym},
it is written as
\small\begin{align}
 \lag\urm{int} 
&=
 -2\im g
\left[
 \psi\dgg_1\psi_2
-
 \psi\dgg_2\psi_1
\right]V(x).
\label{su2gaugelag}
\end{align}\normalsize
This form is often seen as 
a rotating wave approximation\index{rotating wave approximation}
in electromagnetism.
\qed
\end{example}

\begin{remark}
\label{rem:u2u1}
The force carriers are determined by the number of the gauge parameters.
In the case of $\mathrm{SU(2)}$,
there are three parameters that are used to 
describe the weak interaction in the Standard Model.
However, 
the transformation (\ref{su2extheta}) is parameterized only by \en{ \win }.
This is essentially the same as $\mathrm{U(1)}$
even though \en{ T\dd{y}\in\mathrm{su(2)} }.
In fact, 
(\ref{su2single}) is the same form as (\ref{u1w}) for $\mathrm{U(1)}$.
This means that 
this $\mathrm{SU(2)}$ gauge transformation 
can be implemented by the electromagnetic interaction.
\end{remark}

%\newpage

\section{Non-unitary symmetry}
\label{subsec:nonunitary}

\subsection{Issues on non-unitary transformations}
\label{subsec:issues}

So far, 
we have considered local symmetries
only for unitary gauge transformations.
There are two reasons for it.
First,
the mass term is invariant only when \en{ \uni(x) } is unitary: 
\small\begin{align}
 m\widebar{\psi}\psi 
&\to
 m\widebar{\psi'}\psi'
=
 m \widebar{\psi}\uni\dgg(x) \uni(x) \psi.
\end{align}\normalsize
Second,
unitarity is required for the uncertainty relation.
For example, if we consider 
a non-unitary transformation\index{non-unitary gauge transformation}
\small\begin{align}
\hspace{10mm}
 \psi\to \psi'
&=
 \ex\uu{r}\psi, 
\hspace{10mm}
  (r\in\mathbb{R})
\end{align}\normalsize
the resulting field obviously 
violates the uncertainty relation.

However, these problems can be avoided
if we consider a transformation
\small\begin{align}
 \AVl{\psi}{\psi\dgg}
&\to
 \AVl{\psi'}{\psi^{\dagger \prime}}
=
 \A{\ex\uu{r}}{}{}{\ex\uu{-r}} \AVl{\psi}{\psi\dgg}.
\label{psisq}
\end{align}\normalsize
Note that 
this expression is irrelevant for spinors
because \en{ \psi } and \en{ \psi\dgg } are column and row vectors, 
respectively.
Let us consider a complex scalar field
such as the forward traveling field \en{ \{\phi, \phi\dgg\} }.
Then 
we can express (\ref{psisq}) with a non-unitary matrix \en{ V } as
\small\begin{align}
 \AVl{\phi}{\phi\dgg}
&\to
 \AVl{\phi'}{\phi^{\dagger \prime}}
=
 V \AVl{\phi}{\phi\dgg}
\label{psisq-2}
\end{align}\normalsize

It is worth noting that 
this is different from conventional gauge transformations.
In general,
gauge transformations are introduced to 
describe interactions between different particles.
However, 
\en{ \{\phi,\phi\dgg\} } are not different particles.
Rather, they are different phases of the same particles.
This situation is similar to chirality.
We regard (\ref{psisq-2}) as a transformation 
describing interactions between particles 
including internal degrees of freedom.

It is also important to note that not all matrices
are available for \en{ V }.
There must be some conditions for \en{ V }
to be physically meaningful transformations.
We consider this problem in this section.

\subsection{Non-unitary gauge condition and symmetry}
\label{subsec:covariant}

Let us consider 
\en{ N }-forward traveling fields \en{ \{\phi_1,\phi_2,\cdots,\phi_N\} }
denoted by
\small\begin{align}
 \Phi\dd{\alpha}
&\equiv
 \AV{\phi\dd{\alpha}}{\phi\dd{\alpha}\dgg},
\quad
 \Phi
\equiv
 \BV{\Phi_1}{\vdots}{\Phi_N}.
\end{align}\normalsize
Since both \en{ \phi\dd{\alpha} } and \en{ \phi\dd{\alpha}\dgg } are involved,
these are efficiently described by 
the symmetric form of the Lagrangian (\ref{lagfor}):
\small\begin{align}
\hspace{20mm}
 \lag(\Phi,\partial\dd{+}\Phi)
&=
 \frac{1}{2} \Phi\dgg \Sigma\dd{z} \im \partial\dd{+} \Phi,
\qquad
  \Sigma\dd{z}
\equiv
 \B{\sigma\dd{z}}{}{}
   {}{\ddots}{}
   {}{}{\sigma\dd{z}}.
\end{align}\normalsize

%\newpage

Let us consider a local transformation of the form
\small\begin{align}
\hspace{15mm}
 \Phi(x)
&\to
 \Phi'(x)
=
 V(x)\Phi(x),
\hspace{10mm}
 V(x)
=
 \exp[\win\uu{a}(x) \lie\dd{a}]
\label{nontrans}
\end{align}\normalsize
where \en{ \win\uu{a} } are real functions
and 
\en{ \lie\dd{a} } are \en{ 2N \times 2N } matrices.
In general,
\en{ \lie\dd{a}\dgg \not= -\lie\dd{a} },
so
\en{ V } is not necessarily unitary.
To find an interaction Lagrangian,
we follow the same argument as the unitary case.
Let us introduce a gauge field
\small\begin{align}
\hspace{20mm}
 \ga
=
 \gel\uu{a}\lie\dd{a}, \qquad (\gel\uu{a} \in \mathbb{R},)
\end{align}\normalsize
and a covariant derivative 
\small\begin{align}
 D\dd{+}
&\equiv
 \partial\dd{+} - \ga.
\label{nonunicov}
\end{align}\normalsize

A transformation \en{ \ga \to \ga' } follows from the gauge principle
\small\begin{align}
 \lag(\Phi,D\dd{+}\Phi)
&=
 \lag(\Phi',D\dd{+}'\Phi').
\label{nonunigp}
\end{align}\normalsize
This is satisfied if
\small\begin{align}
 D\dd{+} 
&=
 \Sigma\dd{z} V\dgg \Sigma\dd{z} D\dd{+}' V.
\label{nonunid}
\end{align}\normalsize
Since \en{ \Sigma\dd{z}^2=I }, we have
\small\begin{align}
 \Sigma\dd{z} V\dgg \Sigma\dd{z} 
=
 \exp\left[
 \win\uu{a} \Sigma\dd{z} \lie\dd{a}\dgg \Sigma\dd{z}
 \right]
\equiv
 V\uu{\circ}.
\end{align}\normalsize
Then (\ref{nonunid}) is written as
\en{ D\dd{+} = V\uu{\circ} D\dd{+}' V }
and hence
\small\begin{align}
 \partial\dd{+} - \ga
=
 V\uu{\circ} V \partial\dd{+}
+
 V\uu{\circ} (\partial\dd{+} V)
-
 V\uu{\circ} \ga' V.
\label{nonunid2}
\end{align}\normalsize
Comparing the coefficients of \en{ \partial\dd{+} }, 
we get \en{ V\uu{\circ} V= I }.
This is satisfied if 
\small\begin{align}
 \Sigma\dd{z} \lie\dd{a} + \lie\dd{a}\dgg \Sigma\dd{z} 
&=
 0,
\label{nonunicond}
\end{align}\normalsize
which is called 
a \textit{non-unitary gauge condition}\index{non-unitary gauge condition}.
From (\ref{nonunid2}), 
\en{ \ga } transforms as
\small\begin{align}
 \ga
\to
 \ga'
=
 V \ga V\inv
+
 (\partial\dd{+} V)V\inv,
\label{gaugefn}
\end{align}\normalsize
which is the same form as the unitary case (\ref{gaugef}).
In a simple case where 
\en{ \lie\dd{a} } is single-component \en{ \lie\dd{a}=(\lie,0,\cdots) },
the gauge field is expressed as \en{ \ga=\gel\lie }
and (\ref{gaugefn}) is reduced to the same form as U(1):
\small\begin{align}
 \gel^{\prime}
&=
 \gel + \partial\dd{+} \win.
\label{nonuniw}
\end{align}\normalsize

Finally, 
a symmetric Lagrangian under the non-unitary transformation is given as
\small\begin{align}
 \lag(\Phi,D\dd{+} \Phi)
&=
 \lag(\Phi, \partial\dd{+} \Phi) +  \lag\urm{int},
\end{align}\normalsize
where \en{  \lag\urm{int} } is an interaction Lagrangian
\small\begin{align}
  \lag\urm{int}
&=
- \frac{1}{2} \Phi\dgg \im \Sigma\dd{z}  \ga \Phi.
\label{nonintl}
\end{align}\normalsize

Let us show that this Lagrangian is Hermitian.
Note that \en{ \gel\uu{a}\in \mathbb{R} }.
The non-unitary gauge condition (\ref{nonunicond}) turns out to be
\small\begin{align}
 \Sigma\dd{z} \ga +\ga\dgg \Sigma\dd{z}
&=
 0,
\end{align}\normalsize
from which it follows that
\small\begin{align}
  (\lag\urm{int})\dgg 
&=
 \frac{1}{2}\Phi\dgg \im \ga\dgg \Sigma\dd{z} \Phi
=
 -\frac{1}{2}\Phi\dgg \im \Sigma\dd{z} \ga \Phi
= 
  \lag\urm{int}.
\end{align}\normalsize
As a result,
the corresponding time evolution operator 
\en{ S \sim \exp(\im  L\urm{int}) }
is unitary.
This contains a solution to the issue raised in Section \ref{subsec:issues},
i.e.,
the uncertainty principle is not violated
under the non-unitary gauge transformations.
Unitarity is not required for the gauge transformations,
yet 
the resulting time evolution ends up being unitary.
This is a consequence of the gauge principle.

\begin{remark}
\label{rem:nonunialag}
The Lagrangian of the gauge field
is given in the same form as the unitary case (\ref{lagw})
because 
the transformation of the gauge field (\ref{gaugefn})
is the same as the unitary case.
\end{remark}

The non-unitary gauge condition (\ref{nonunicond})
is concerned with $\mathrm{su(N,N)}$.
For \en{ N=1 },
it is written as
\small\begin{align}
\hspace{20mm}
 \lie = g\uu{x} \sigma\dd{x} + g\uu{y} \sigma\dd{y} + \im D,
\hspace{5mm}
 (g\uu{x},g\uu{y}\in \mathbb{R})
\label{nonusin}
\end{align}\normalsize
where \en{ D } is an arbitrary real diagonal matrix.
Obviously \en{ \lie\dgg\not= -\lie }
unless \en{ g\uu{x}=g\uu{y} = 0 }.
If \en{ \Tr D=0 }, then \en{ \lie \in } $\mathrm{su(1,1)}$.
This corresponds to 
a degenerate parametric amplifier,\index{degenerate parametric amplifier}
as will be seen in the next example.

\begin{example}
\label{ex:nonuni1}
\rm
As an example of (\ref{nonusin}),
let us consider 
\small\begin{align}
\hspace{10mm}
 \lie
&=
 2\A{}{g}
    {g\uu{*}}{},
\hspace{10mm}
 \bigl(g\in \mathbb{C}.\bigr)
\end{align}\normalsize
Note that \en{ \lie\dgg \not= -\lie }.
The corresponding gauge transformation is given as
\small\begin{align}
 V(x)
&=
 \ex\uu{\win(x)\lie}
=
 \A{\cosh(2|g|\win)}{\ffrac{g}{|g|}\sinh(2|g|\win)}
   {\ffrac{g\uu{*}}{|g|}\sinh(2|g|\win)}{\cosh(2|g|\win)},
\label{exlint}
\end{align}\normalsize
which is not unitary.
(This form is known as the Bogoliubov transformation.)\index{Bogoliubov transformation}
In this case, the interaction Lagrangian (\ref{nonintl}) is written as
\small\begin{align}
 \lag\urm{int}
&=
 \im \gel \bigl( g\uu{*}\phi\phi - g\phi\dgg\phi\dgg \bigr).
\label{nonsu22}
\end{align}\normalsize
This form is known as a degenerate parametric amplifier 
in quantum optics.
As shown in (\ref{nonuniw}),
the gauge field transforms as
\small\begin{align}
 \gel
\to
 \gel^{\prime}
&=
 \gel + \partial\dd{+} \win.
\end{align}\normalsize
Note that (\ref{nonsu22}) is not symmetric under 
a global U(1) transformation
\en{ \phi\to \ex\uu{\im\alpha}\phi },
which means that particle number is not conserved
due to Noether's theorem.
However, 
if we regard \en{ gA } as a pump field
and introduce \en{ gA\to \ex\uu{\im 2\alpha}gA },
the global U(1) symmetry is recovered.
Then 
the first term of (\ref{nonsu22}) describes 
a process where the pump field splits to two forward traveling fields,
and the second term is its reverse process.
%The gauge field $\gel$ mediates these decay and fusion processes.
This example will be revisited in Chapter \ref{chap:boundary}
where the roles of \en{ \gel } and \en{ \win } become clearer.
\qed
\end{example}

%\if0
\marginpar{\vspace{-30mm}
\footnotesize
This U(1) symmetry breaking 
is also the key to understanding the BCS theory
where 
the Bogoliubov transformation 
is used to diagonalize a Hamiltonian.
\normalsize
}
%\fi

\chapter*{}

\begin{center}
\begin{Large}
\begin{bfseries}
\textsf{Part II}
\\
\vspace{10mm}
\textsf{Classical formulation of quantum gates and systems}
\end{bfseries}
\end{Large}
\end{center}

\chapter{Quantum gates}
\label{chap:boundary}
\thispagestyle{fancy}

In this chapter, 
we introduce operations on the forward traveling and Dirac fields,
called quantum gates.\index{quantum gate}
In a conventional setting,
given the Lagrangian of a quantum gate,
we calculate equations of motion to find its input-output relation.
Here we consider an inverse problem.
A quantum gate is defined by an input-output relation 
from which we calculate the corresponding interaction Lagrangian.
We show two methods for it:
Lagrange's method of undetermined multipliers
and the gauge theory.
In Lagrange's method,
the input-output relation is regarded as a boundary condition.
On the other hand, in the gauge theoretical approach,
the input-output relation is regarded as a local gauge transformation.
In either case, we obtain the same interaction Lagrangian.
Various types of examples are also shown for 
unitary and non-unitary quantum gates.

\section{Quantum gates for forward traveling fields}

This section illustrates a general procedure to calculate 
the interaction Lagrangians of unitary and non-unitary quantum gates
on the forward traveling field \en{ \phi }.
Note that unitary gates are described by only \en{ \phi } as
\small\begin{align}
\hspace{22mm}
 \phi 
\to 
 \phi'
=
 \uni(x)\phi,
\qquad
 \uni\uni\dgg=I.
\label{qgdef}
\end{align}\normalsize
On the other hand,
both \en{ \phi } and \en{ \phi\dgg } are necessary
for non-unitary gates.
Let us introduce a vector form
\small\begin{align}
 \Phi
\equiv
 \AVl{\phi}{\phi\dgg}.
\end{align}\normalsize
Then non-unitary gates are expressed as
\small\begin{align}
 \Phi
\to
 \Phi'
=
 V(x) \Phi,
\label{qgdef2}
\end{align}\normalsize
where 
\en{ V } satisfies 
the non-unitary gauge condition (\ref{nonunicond}).
Most of useful and interesting quantum gates are
non-unitary,
as will be seen in Section \ref{sec:exgate}.

There are two methods for the calculation 
of the interaction Lagrangians:
Lagrange's method of undetermined multipliers
and the gauge theory.
The non-unitary case includes the unitary one.
However,
the unitary case is simple and instructive 
to illustrate the procedure.
We consider the unitary case first.

\newpage

\subsection{Unitary gates: Lagrange's method}
\label{subsec:passive}
\index{unitary gate}

\begin{wrapfigure}[0]{r}[53mm]{49mm}
\centering
\includegraphics[keepaspectratio,width=45mm]{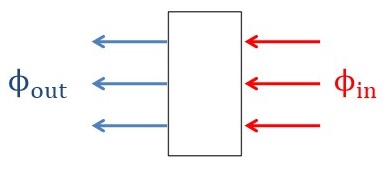}
\caption{
\small
Input and output fields.
\normalsize
}
\label{fig-ugateinout}
\end{wrapfigure}

Let us consider multiple forward traveling fields
\en{ \{\phi_1,\phi_2, \cdots\} }.
These fields are regarded as an input
before entering the gate,
and an output after leaving the gate,
as in Figure \ref{fig-ugateinout}.
In general,
the gate is spatially distributed.
The input and the output are expressed as
\begin{subequations}
\small\begin{align}
 \phi\drm{in}
&\equiv
 \BV{\phi\drm{1,in}}{\phi\drm{2,in}}{\vdots}
\, \ =
 \BV{\phi_1}{\phi_2}{\vdots}\drm{in},
\\
 \phi\drm{out}
&\equiv
 \BV{\phi\drm{1,out}}{\phi\drm{2,out}}{\vdots}
=
 \BV{\phi_1}{\phi_2}{\vdots}\drm{out}.
\end{align}\normalsize
\end{subequations}
For these fields, Lagrangians are written as
\begin{subequations}
\small\begin{align}
 \lag\drm{in}
&=
  \phi\drm{in}\dgg \im\partial\dd{+} \phi\drm{in},
\\
 \lag\drm{out} 
&=
  \phi\drm{out}\dgg \im\partial\dd{+} \phi\drm{out}.
\end{align}\normalsize
\end{subequations}

\begin{wrapfigure}[0]{r}[53mm]{49mm}
\centering
\vspace{-30mm}
\includegraphics[keepaspectratio,width=45mm]{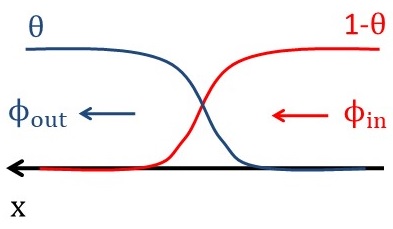}
\caption{
\small
The input $\phi_{\textrm{in}}$ 
and the output $\phi_{\textrm{out}}$ 
are defined 
in the domains of the weight functions $1-\win$ and $\win$,
respectively.
\normalsize
}
\label{fig-weight-0}
\end{wrapfigure}

To combine these Lagrangians,
we introduce a weight function \en{ \win(x)=\win(t,z) }.
\index{weight function (gauge theory)}
For example, 
the input and output can be weighted 
by \en{ \win } as in Figure \ref{fig-weight-0}.
We usually assume that this function satisfies
\small\begin{align}
& 0\le \win(x) \le 1.
\end{align}\normalsize
\en{ 1-\win } and \en{ \win }, respectively,
define the domains of the input and the output so that
\begin{subequations}
\small\begin{align}
 \lim_{z\to \infty}\win(x) &= 1, 
\\
 \lim_{z\to -\infty}\win(x) &= 0. 
\end{align}\normalsize
\end{subequations}
Then both the input and the output are described by
\small\begin{align}
 \lag^{\win} 
&=
 \win  \lag\drm{out} 
+ 
 (1-\win) \lag\drm{in}. 
\label{fpdlag}
\end{align}\normalsize

Note that \en{ \win } and \en{ 1-\win } are partially overlapping.
In this domain,
the input and the output are mixed
and we need to specify 
how they are associated with each other.
This defines a quantum gate.
Here we consider an input-output relation
\small\begin{align}
\hspace{20mm}
 \phi\drm{out}(x)
&=
 \tf \phi\drm{in}(x), 
\qquad
 \tf \tf\dgg = 1,
\label{qgtf0}
\end{align}\normalsize
where \en{ \tf} is a unitary (transfer function) matrix.
This is expressed via 
the Cayley transform\index{Cayley transform}
as
\small\begin{align}
\hspace{20mm}
 \tf
&=
 \frac{1+\lie}{1-\lie},
\qquad \lie\dgg=-\lie.
\label{qgtf1}
\end{align}\normalsize
Then the input-output relation is rewritten as
\small\begin{align}
 (\phi\drm{in} - \phi\drm{out}) 
+ 
 \lie(\phi\drm{in} + \phi\drm{out})
&=0.
\label{gbconp}
\end{align}\normalsize

We regard this as a constraint on \en{ \lag^{\win} }
and apply Lagrange's method of undetermined multipliers
\small\begin{align}
 L
&=
 \int dz \
\Bigl[
 \lag^{\win}
+ 
 \bm{a}\trans
\Bigl\{ 
 (\phi\drm{in}-\phi\drm{out}) + \lie (\phi\drm{in}+\phi\drm{out})
\Bigr\}
\Bigr],
\label{l-multi}
\end{align}\normalsize
where \en{ \bm{a} } is a Lagrange multiplier.
The Euler-Lagrange equation 
yields
\begin{subequations}
\small\begin{align}
 \phi\drm{out}\dgg \ \im \partial\dd{+} \win
+
 \bm{a}\trans(1-\lie)
&=0,
\\
 \phi\drm{in}\dgg \ \im \partial\dd{+} \win
+
 \bm{a}\trans (1+\lie) 
&=0.
\end{align}\normalsize
\end{subequations}
There are two choices for \en{ \bm{a} }.
Note that 
if \en{ \lie=0 }, 
the gate does not operate 
and 
the input-output relation is
\en{ \phi\drm{out} = \phi\drm{in} } in (\ref{l-multi}).
This means that 
we need to choose a multiplier 
that is independent of \en{ \lie }:
\small\begin{align}
 \bm{a}\trans
&=
 - \mm{\phi}\dgg \im \partial\dd{+} \win,
\end{align}\normalsize
where
\small\begin{align}
 \mm{\phi} 
\equiv 
 \frac{\phi\drm{in}+\phi\drm{out}}{2}.
\end{align}\normalsize

The resulting Lagrangian is written as the sum of 
a free field component \en{ \lag\urm{f} } (without \en{ \lie }) and 
an interaction component \en{ \lag\urm{int} } (with \en{ \lie })
as
\small\begin{align}
 L
&=
 \int dz \ \bigl[\lag\urm{f} + \lag\urm{int} \bigr],
\end{align}\normalsize
where 
\begin{subequations}
\label{unilag} 
\small\begin{align}
 \lag\urm{f}
&=
 \lag^{\win}
-
 \mm{\phi}\dgg (\im \partial\dd{+} \win) 
  (\phi\drm{in} - \phi\drm{out}),
\label{unilag-1}\\
%%%%%%%%%%%%%%%%%%%%%%%%%%%%%%%%%%%%%%%%%%%%%%
 \lag\urm{int}
&=
 - 2
 \mm{\phi}\dgg (\im \partial\dd{+} \win) \, \lie \, \mm{\phi}.
\label{unilag-2}
\end{align}\normalsize
\end{subequations}

\begin{remark}
\label{rem:freeindep}
Since 
\en{ \{\phi_1,\phi_2,\cdots\} } do not interact with each other in free space,
the free field Lagrangian (\ref{unilag-1}) can be expressed as
\small\begin{align}
 \lag^{\mathrm{f}}
&=
 \lag_1^{\mathrm{f}} +  \lag_2^{\mathrm{f}} + \cdots.
\end{align}\normalsize
\end{remark}

\begin{wrapfigure}[0]{r}[53mm]{49mm}
\centering
\vspace{-0mm}
\includegraphics[keepaspectratio,width=45mm]{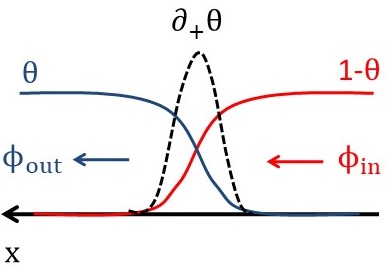}
\caption{
\small
The quantum gate is defined in a domain where $\partial_+\win\not=0$.
\normalsize
}
\label{fig-weight-1}
\end{wrapfigure}

Let us see how the input-output relation is reproduced from this Lagrangian.
The Euler-Lagrange equation yields
\small\begin{align}
 \win \im\partial\dd{+}\phi\drm{out}
+
 (1-\win)\im\partial\dd{+}\phi\drm{in}
-
 (\im\partial\dd{+}\win)  \Bigl[
 (\phi\drm{in}-\phi\drm{out}) 
+ 
 \lie (\phi\drm{in}+\phi\drm{out})
\Bigr]
=0.
\label{xboundary}
\end{align}\normalsize
This shows that the gate is effective
only in a domain where \en{ \partial\dd{+}\win(x)\not= 0 }
as depicted in Figure \ref{fig-weight-1}.
As an example,
consider a lumped gate for which 
\en{ \win(x) } is a step function
\begin{subequations}
\label{unistep}
\small\begin{align}
 \win(x) &= \step(z),
\\
 \partial\dd{+} \win(x) &= \delta(z).
\end{align}\normalsize
\end{subequations}
Then (\ref{xboundary}) reads
\begin{subequations}
\label{lamstep}
\small\begin{align}
\hspace{20mm}
  z < 0: & \quad \im\partial\dd{+}\phi\drm{in}=0,
\\
  z = 0: & \quad  \phi\drm{out} =  \tf \phi\drm{in},
\qquad \tf \equiv \frac{1+\lie}{1-\lie},
\\
  z > 0: & \quad \im\partial\dd{+}\phi\drm{out}=0.
\end{align}\normalsize
\end{subequations}
The local gauge transformation is reconstructed at \en{ z=0 }, 
as expected.

\newpage

\subsection{Unitary gates: Gauge theoretical approach}
\label{subsec:uniLgauge}
\index{unitary gate}

\begin{wrapfigure}[0]{r}[53mm]{49mm}
\centering
\vspace{25mm}
\includegraphics[keepaspectratio,width=45mm]{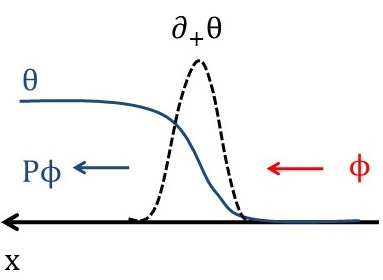}
\caption{
\small
Quantum gate with the input $\phi$ and the output $\tf \phi$.
\normalsize
}
\label{fig-weightG}
\end{wrapfigure}

In the gauge theoretical approach,
a quantum gate is defined by a local gauge transformation.
Let us rewrite the unitary transfer function \en{ \tf } as
\small\begin{align}
 \tf
&=
 \frac{1+\lie}{1-\lie}
\sim
 1+2\lie
\sim
 \exp\bigl[2\lie\bigr],
\label{tfpg}
\end{align}\normalsize
where \en{ \lie } serves as a Lie-algebra associated with \en{ \tf }.
The input-output relation is written as 
\small\begin{align}
 \phi(x) &\to \phi'(x) = \tf\phi(x)  = \exp\bigl[2\lie \bigr]\phi(x).
\label{tfpg-g}
\end{align}\normalsize
We regard this as a gauge transformation
and apply the gauge theory to find an interaction Lagrangian.
However, it is not straightforward 
because (\ref{tfpg-g}) defined \textit{globally}.

To define the gate as a \textit{local} transformation,
we introduce a weight function \en{ \win } 
as in Figure \ref{fig-weightG}.
The input-output relation is redefined as
\small\begin{align}
 \phi(x) 
&\to 
 \phi'(x) 
= 
 \exp\bigl[ \win(x) 2 \lie  \bigr]\phi(x).
\label{gategauge}
\end{align}\normalsize
In this expression,
\en{ \phi'=\phi } in \en{ \{x|\win(x)=0\} }, 
which is the input coming into the gate.
On the other hand, 
\en{ \phi'= \tf \phi } in \en{ \{x|\win(x)=1\} }, 
which is the output leaving the gate.
Hence 
(\ref{gategauge}) is well defined as a local gauge transformation.

Now let us apply the gauge theory.
The free field Lagrangian of the forward traveling field 
is given by
\small\begin{align}
 \lag(\phi,\partial\dd{+}\phi)
&=
 \phi\dgg \im\partial\dd{+}\phi.
\end{align}\normalsize
For the local transformation (\ref{gategauge}),
the covariant derivative is written as 
\small\begin{align}
 D\dd{+}
&=
 \partial\dd{+} - 2 \gel \lie,
\end{align}\normalsize
where \en{ \gel } is a gauge field that transforms as
\small\begin{align}
 \gel \to \gel'
=
 \gel + \partial\dd{+}\win.
\label{gelchange}
\end{align}\normalsize
A locally symmetric Lagrangian is given by
\small\begin{align}
 \lag(\phi,D\dd{+}\phi)
&=
 \lag(\phi,\partial\dd{+}\phi) 
- 
  2 \phi\dgg \im \gel \lie  \phi.
\label{gateLintgauge}
\end{align}\normalsize

This is consistent with 
the result of Lagrange's method (\ref{unilag}).
Note that 
the input and the output are not explicitly labeled 
in the gauge theoretical treatment.
Setting \en{ \phi\drm{in}=\phi\drm{out}=\phi } in (\ref{unilag}),
we get
\small\begin{align}
 \lag
&=
 \lag(\phi,\partial\dd{+}\phi) 
 - 
 2 \phi\dgg (\im \partial\dd{+} \win) \lie \phi.
\end{align}\normalsize
This is the same as (\ref{gateLintgauge}),
except that 
the gauge field \en{ \gel } is replaced to \en{ \partial\dd{+} \win }.
This means that
the quantum gate does nothing if \en{ \gel'=\gel }.
The gate is turned on 
only when 
there is a `gap' \en{ \partial\dd{+}\win=\gel'-\gel\not= 0 }
in the gauge field.

In fact, 
both the input \en{ \phi } and the output \en{ \phi' } are solutions 
to the same differential equation. %in free space.
The difference between \en{ \phi } and \en{ \phi' } is a gauge.
When \en{ \phi } goes through the gate,
its gauge is modulated by the amount of \en{ \partial\dd{+}\win }.
Then \en{ \phi } transforms into \en{ \phi' }.
This is similar to magnetic monopoles and the Berry phase 
(Appendix \ref{app:berry})
for which the gap of the gauge field results in a magnetic flux.

\subsection{Non-unitary gates: Lagrange's method}
\label{subsec:nonuniL}
\index{non-unitary gate}

For a non-unitary gate,
we introduce the following notation:
\small\begin{align}
 \Phi%\drm{1,in}
&\equiv
 \AV{\phi}{\ \phi\dgg},
\hspace{7mm}
 \Phi\drm{in}
\equiv
 \BV{\Phi_1}{\Phi_2}{\vdots}\drm{in},
\hspace{7mm}
 \Phi\drm{out}
\equiv
 \BV{\Phi_1}{\Phi_2}{\vdots}\drm{out}.
\end{align}\normalsize
These fields are described by
the symmetric form of the Lagrangian (\ref{lagfor}):
\begin{subequations}
\small\begin{align}
 \lag\drm{in}
&=
  \frac{1}{2} \Phi\drm{in}\dgg \Sigma\dd{z} \im\partial\dd{+} \Phi\drm{in},
\\
 \lag\drm{out}
&=
  \frac{1}{2} \Phi\drm{out}\dgg \Sigma\dd{z} \im\partial\dd{+} \Phi\drm{out},
\end{align}\normalsize
\end{subequations}
where
\small\begin{align}
 \Sigma\dd{z}
&\equiv
 \B{\sigma\dd{z}}{}{}
   {}{\sigma\dd{z}}{}
   {}{}{\ddots}.
\end{align}\normalsize
The inputs \en{ \Phi\drm{in} } and the outputs \en{ \Phi\drm{out} }
are distinguished by a weight function \en{ \win(x) },
as in the unitary case:
\small\begin{align}
 \lag^{\win}
&=
 \win     \lag\drm{out} 
+ 
 (1-\win) \lag\drm{in} 
\label{fpdlaga}
\end{align}\normalsize

Assume that 
the input-output relation of the non-unitary gate is given by
\small\begin{align}
 (\Phi\drm{in} - \Phi\drm{out}) + \lie (\Phi\drm{in} + \Phi\drm{out})
&=0,
\label{gbcona}
\end{align}\normalsize
where \en{ \lie } is a complex matrix satisfying 
the non-unitary gauge condition (\ref{nonunicond}).
The interaction Lagrangian of the non-unitary gate
is obtained in the same way as the unitary case.
Let us consider a Lagrangian
\small\begin{align}
 L
&=
 \int dz 
\Bigl[
 \lag^{\win}
+ 
 \bm{a}\trans
\Bigl\{ 
 (\Phi\drm{in}-\Phi\drm{out}) + \lie (\Phi\drm{in}+\Phi\drm{out})
\Bigr\}
\Bigr],
\end{align}\normalsize
where \en{ \bm{a} } is a Lagrange multiplier 
determined by the Euler-Lagrange equation as
\small\begin{align}
 \bm{a}\trans
&=
 - \frac{1}{2} \mm{\Phi}\dgg  \Sigma\dd{z} \im \partial\dd{+} \win,
\end{align}\normalsize
where
\small\begin{align}
 \mm{\Phi} \equiv \frac{\Phi\drm{in}+\Phi\drm{out}}{2}.
\end{align}\normalsize
The resulting free field Lagrangian \en{ \lag\urm{f} }
and interaction Lagrangian \en{ \lag\urm{int} } 
are given by
\begin{subequations}
\label{nonunilag}
\small\begin{align}
% L
%&=
% \int dz \Bigl[\lag\urm{f}+ \lag\urm{int} \Bigr],
%\\
 \lag\urm{f}
&=
 \lag^{\win}
-
 \frac{1}{2} \mm{\Phi}\dgg \Sigma\dd{z} 
 (\im\partial\dd{+} \win) (\Phi\drm{in} - \Phi\drm{out}),
\\
 \lag\urm{int}
&=
 - 
 \mm{\Phi}\dgg \Sigma\dd{z} 
 (\im \partial\dd{+} \win) \, 
 \lie \, \mm{\Phi}.
\end{align}\normalsize
\end{subequations}

\subsection{Non-unitary gates: Gauge theoretical approach}
\label{subsec:nonuniG}
\index{non-unitary gate}

In the gauge theoretical approach, 
we start with a Lagrangian
\small\begin{align}
 \lag(\Phi,\partial\dd{+}\Phi)
&=
 \frac{1}{2} \Phi\dgg \Sigma\dd{z} \im\partial\dd{+} \Phi.
\end{align}\normalsize
Let us rewrite the transfer function as
\small\begin{align}
 P
&=
 \frac{1+\lie}{1-\lie}
\sim
 \exp\bigl[2\lie\bigr].
\end{align}\normalsize
A basic procedure is the same as the unitary case
(Section \ref{subsec:uniLgauge}).
We define a local gauge transformation 
by introducing a weight function \en{ \win(x) } as 
\small\begin{align}
 \Phi(x)
&\to
 \Phi'(x)
=
 \exp\bigl[ \win(x)  2 \lie  \bigr] \Phi(x).
\end{align}\normalsize
The corresponding covariant derivative is given by
\small\begin{align}
 D\dd{+}
&=
 \partial\dd{+} - 2 \gel \lie,
\end{align}\normalsize
where \en{ \gel } transforms as
\small\begin{align}
  \gel \to \gel' = \gel + \partial\dd{+} \win.
\end{align}\normalsize
The Lagrangian of the gate is then given by
\small\begin{align}
 \lag(\Phi,D\dd{+}\Phi)
&=
 \lag(\Phi,\partial\dd{+}\Phi)
- 
 \Phi\dgg \Sigma\dd{z}  \im \gel \lie \Phi,
\label{nonintgauge}
\end{align}\normalsize
which is the same form as 
the result of Lagrange's method (\ref{nonunilag}), 
as explained in the unitary case.

\section{Some remarks on $\lie$}
\label{sec:remarks}

\subsection{How to find $\lie$}

The interaction Lagrangian is determined by \en{ \lie },
and it is an important task to find \en{ \lie }
for given quantum gates.
We have considered the two expressions of the transfer function:
\begin{subequations}
\small\begin{align}
\hspace{-20mm}
\mbox{\normalsize Lie algebra: \small} \ 
 \tf &= \ex\uu{2\lie},
\\
\hspace{-20mm}
\mbox{\normalsize Cayley transform: \small} \
 \tf &= \frac{1+\lie}{1-\lie}. 
\end{align}\normalsize
\end{subequations}
These are approximately equivalent.
In general, 
they are different, especially when we consider multiple gates.
For example,
suppose that two gates with \en{ \lie_1 } and \en{ \lie_2 }
are connected in a concatenated way.
In the Lie algebraic expression,
it is written as
\small\begin{align}
 \ex\uu{2\lie}
=
 \ex\uu{2\lie_1}
 \ex\uu{2\lie_2}.
\label{gcase}
\end{align}\normalsize
Up to second order, \en{ \lie } is obtained as
\small\begin{align}
 \lie
\sim
 \lie_1 + \lie_2 + [\lie_1,\lie_2].
\label{ggaugelie}
\end{align}\normalsize

On the other hand,
as a Cayley transform,
the input-output relation is written as
\small\begin{align}
 \Phi\drm{out}
=
 \frac{I+\lie_1}{I-\lie_1}
 \frac{I+\lie_2}{I-\lie_2}
 \Phi\drm{in}.
\end{align}\normalsize
This can be expressed as
\small\begin{align}
  (\Phi\dd{\mathrm{in}} - \Phi\dd{\mathrm{out}}) 
+ 
 (\lie\dd{\mathrm{in}} \Phi\dd{\mathrm{in}} 
+ \lie\dd{\mathrm{out}} \Phi\dd{\mathrm{out}})
&=0.
\label{tfXY}
\end{align}\normalsize
To obtain \en{ \lie },
we need to rewrite it as
\small\begin{align}
  (\Phi\dd{\mathrm{in}} - \Phi\dd{\mathrm{out}}) 
+ 
 \lie (\Phi\dd{\mathrm{in}} + \Phi\dd{\mathrm{out}})
&=0.
\label{gbcona-1}
\end{align}\normalsize
There are different ways to do this.
For example,
by equating two transfer functions
\small\begin{align}
 (1-\lie\dd{\mathrm{out}})\inv (1+\lie\dd{\mathrm{in}})
=
 \frac{1+\lie}{1-\lie},
\label{pcompare}
\end{align}\normalsize
we get
\small\begin{align}
 \lie
&=
 (2+\lie\dd{\mathrm{in}}-\lie\dd{\mathrm{out}})\inv 
 (\lie\dd{\mathrm{in}}+\lie\dd{\mathrm{out}}).
\label{pcompare1}
\end{align}\normalsize
In general,
this is different from (\ref{ggaugelie}).
These expressions will be used 
to find the interaction Lagrangian of quantum circuits 
in Chapter \ref{chap:circuit}.

\subsection{$\lie$ as a Lie algebra, reactance matrix, and self-energy}

The matrix \en{ \lie } has 
different interpretations in different fields.
For example,
in the gauge theory,
\en{ \lie } is considered as a generator or a Lie algebra
because the transfer function is expressed as
\small\begin{align}
 \tf
=
 \ex\uu{2\lie}.
\label{tfcay2}
\end{align}\normalsize

In scattering theory,
a propagator is sometimes expressed as
\small\begin{align}
 \ipg
=
 \frac{1+\im R}{1-\im R},
\end{align}\normalsize
where \en{ R } is called a reactance matrix.\index{reactance matrix}
As explained in Remark \ref{sec:rem-sys},
\en{ \ipg \sim \tf }.
Compared to the Cayley transform\index{Cayley transform} 
\small\begin{align}
 \tf
=
 \frac{I+\lie}{I-\lie},
\label{tfcay}
\end{align}\normalsize
\en{ \lie } corresponds to the reactance matrix.

It is also possible to express the transfer function \en{ \tf } as
\small\begin{align}
 \tf
&=
 \frac{I+\lie}{I-\lie}
\sim
 I +  2\lie ( I+\lie+\lie^2+\cdots).
\end{align}\normalsize
The first term is the direct-through term of the scattering process
and the second term corresponds to the Dyson equation
(Chapter \ref{chap:additional}).
In this case,
\en{ \lie } is related to self-energy.\index{self-energy}

\newpage

\section{Examples of quantum gates}
\label{sec:exgate}

\subsection{SU(2) gate}
\label{subsec:su2}

\begin{wrapfigure}[0]{r}[53mm]{49mm}
\centering
\vspace{-5mm}
\includegraphics[keepaspectratio,width=40mm]{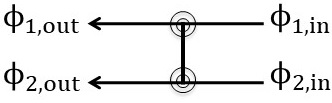}
\caption{
\small
SU(2) gate.
\normalsize
}
\label{fig-su2}
\end{wrapfigure}

Consider a gate with two inputs and two outputs 
as in Figure \ref{fig-su2}:
\begin{subequations}
\small\begin{align}
 \phi\drm{in}
&\equiv
 \AV{\phi_1}{\phi_2}\drm{in},
\\ %\quad
 \phi\drm{out}
&\equiv
 \AV{\phi_1}{\phi_2}\drm{out}.
\end{align}\normalsize
\end{subequations}
Assume that 
the inputs and the outputs are related to each other 
by SU(2).\index{SU(2) gate (forward traveling)}
In general, 
SU(2) is represented by matrices of the form
\small\begin{align}
 \textrm{SU(2)}
=
 \left\{ \
 \A{\alpha}{\beta}{-\beta\uu{*}}{\alpha\uu{*}} 
\ \ \middle|  \ \
 \alpha,\beta\in \mathbb{C}, \ 
 |\alpha|^2+|\beta|^2=1
\right\}.
\end{align}\normalsize
If \en{ \alpha } is real,
\en{ \alpha } and \en{ \beta } can be parameterized as
\small\begin{align}
\hspace{15mm}
 \alpha &=\frac{1-|g|^2}{1+|g|^2}, 
\quad
 \beta = \frac{2g}{1+|g|^2}.
\qquad 
 (g\in \mathbb{C})
\end{align}\normalsize
The input-output relation of this gate is written as
\small\begin{align}
 \AV{\phi_1}{\phi_2}\drm{out}
&=
 \frac{1}{1+|g|^2}
 \A{1-|g|^2}{2g}{-2g\uu{*}}{1-|g|^2}
 \AV{\phi_1}{\phi_2}\drm{in}.
\label{sincos}
\end{align}\normalsize

This is rewritten as
\small\begin{align}
\label{fboundary}
 (\phi\drm{in}-\phi\drm{out}) + \lie (\phi\drm{in}+\phi\drm{out}) &= 0, 
\end{align}\normalsize
where
\small\begin{align}
 \lie
\equiv
 \A{}{g}{-g\uu{*}}{}.
\end{align}\normalsize
In Lagrange's method,
the interaction Lagrangian is given by (\ref{unilag}):
\small\begin{align}
 \lag\urm{SU(2)}
&=
-2 (\im\partial\dd{+}\win) 
  \Bigl[
 g \mm{\phi}_1\dgg  \mm{\phi}_2 
- 
 g\uu{*} \mm{\phi}_2\dgg  \mm{\phi}_1
\Bigr].
\label{su2lag}
\end{align}\normalsize
In the gauge theoretical method,
it is given by (\ref{gateLintgauge}):
\small\begin{align}
 \lag\urm{SU(2)}
&=
 - 2\im \gel
\left[
 g \phi\dgg_1 \phi_2
-
  g\uu{*} \phi\dgg_2   \phi_1
\right].
\label{lagv}
\end{align}\normalsize
This is the same as the rotating wave approximation 
in Example \ref{ex:su2gauge}.

\begin{remark}
This gate is parameterized only by \en{ \win }.
As explained in Remark \ref{rem:u2u1},
it is essentially equivalent to \en{ \mathrm{U(1)} }.
However, 
we still call it an \en{ \mathrm{SU(2)} } gate for convenience.
A general \en{ \mathrm{SU(2)} } transformation 
can be constructed through the cascade connection
of \en{ \mathrm{SU(2)} } gates
(Section \ref{sec:generalsu2}).
\end{remark}

\newpage

\subsection{Displacement gate}
\label{subsec:dis}

\begin{wrapfigure}[0]{r}[53mm]{49mm}
\vspace{5mm}
\centering
\includegraphics[keepaspectratio,width=35mm]{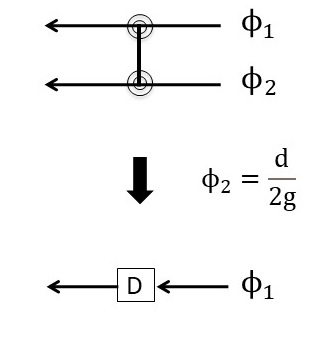}
\caption{
\small
Displacement gate.
\normalsize
}
\label{fig-displace}
\end{wrapfigure}

A displacement gate\index{displacement gate} 
is an operation to displace a field in the phase space.
It is approximately described through the SU(2) gate
by replacing the fields \en{ \phi_2 } to a scalar
\small\begin{align}
\hspace{10mm}
 \phi_2
= 
 \frac{d}{2g},
\qquad
 (d\in \mathbb{C})
\end{align}\normalsize
which is written as
\small\begin{align}
\hspace{33mm}
 \mm{\phi}_2
=
 \frac{d}{2g}.
\qquad
 \left(
\mm{\phi}_2 \equiv \frac{\phi\drm{2,in}+\phi\drm{2,out}}{2}
\right)
\label{dissc}
\end{align}\normalsize
The resulting gate is single-input and single-output 
as in Figure \ref{fig-displace}.
Substituting (\ref{dissc}) into (\ref{su2lag}),
we get
\begin{subequations}
\label{lagdis}
\small\begin{align}
 \lag\urm{f}
&= 
 \lag_1\urm{f},
\\
 \lag\urm{D}
&=
 - (\im\partial\dd{+}\win) 
 \Bigl[ \mm{\phi}_1\dgg d
- 
 d\uu{*}\mm{\phi}_1
 \Bigr],
\end{align}\normalsize
\end{subequations}
where \en{ \lag_1\urm{f} } is the \en{ \phi_1 }-component of 
the free field Lagrangian
(Remark \ref{rem:freeindep})
\small\begin{align}
\hspace{15mm}
 \lag_1\urm{f}
\equiv
  \lag_1^{\win}
-
 \mm{\phi}_1\dgg 
 (\im\partial\dd{+}\win) 
 (\phi\dd{1,\textrm{in}}-\phi\dd{1,\textrm{out}}).
\label{partfree}
\end{align}\normalsize

\marginpar{\vspace{-17mm}
\footnotesize
 Note that $\lag_2\urm{f}=0$ 
if we plug $\phi\drm{2,in}=\phi\drm{2,out}=d/2g$.
\normalsize
 }

Let us check to see if this Lagrangian is correct.
Assume that 
the weight function is a step function: \en{ \win(x)=\step(z) }.
Then the Euler-Lagrange equation yields
\begin{subequations}
\small\begin{align}
 z<0: & \quad \im\partial\dd{+}\phi\drm{1,in}=0, 
\\
 z=0: & \quad \phi\drm{1,out} = \phi\drm{1,in} + d,
\\
 z>0: & \quad \im\partial\dd{+}\phi\drm{1,out}=0,
\end{align}\normalsize
\end{subequations}
which shows that 
the input \en{ \phi\drm{1,in} } is displaced at \en{ z=0 }, 
as expected.

\subsection{Time-varying SU(2) gate}
\label{sec:timevarying}

Let us consider an SU(2) gate\index{SU(2) gate (time-varying)} 
with a time-varying parameter \en{ g=g(t) }.
For later use,
we introduce the following operations:
\begin{subequations}
\label{conv}
\small\begin{align}
 g\ast \phi&\equiv \int_0\uu{t} ds \ g(t-s)\phi(s),
\\
 \psi\dgg \circ g&\equiv \int_0^{\infty} ds \ \psi\dgg(t+s) g(s).
\end{align}\normalsize
\end{subequations}
These are related to each other as
\small\begin{align}
 \int_0^{\infty} dt \ \psi\dgg (g\ast \phi) 
&=
 \int_0^{\infty} dt \ (\psi\dgg \circ g) \phi.
\label{dync}
\end{align}\normalsize
The time-varying SU(2) gate is defined 
by replacing \en{ \lie \to \lie\ast } in (\ref{fboundary}):
\small\begin{align}
 (\phi\drm{in} - \phi\drm{out}) + \lie \ast(\phi\drm{in} + \phi\drm{out})
&=
 0.
\label{tvinout}
\end{align}\normalsize

Let us derive the interaction Lagrangian of this gate 
using Lagrange's method
(because a gauge transformation is not well-defined in this case).
A total Lagrangian is defined as
\small\begin{align}
 L
&=
 \int dz \Bigl[
 \lag^{\win} 
+ 
 \bm{a}\trans \Bigl\{ (\phi\drm{in} - \phi\drm{out}) + \lie \ast(\phi\drm{in} + \phi\drm{out})
\Bigr\}\Bigr],
\label{tvlagconb}
\end{align}\normalsize
where \en{ \bm{a} } is a Lagrange multiplier 
to be determined from the Euler-Lagrange equation
in the same way as before.
However, 
the Lagrangian (\ref{tvlagconb}) is not convenient 
to differentiate with respect to 
\en{ \phi\drm{in} } and \en{ \phi\drm{out} }
because of convolution.
Note that
the Euler-Lagrange equation is the equation of motion 
that minimizes an action 
\small\begin{align}
 S 
&= 
 \int dt \ L(t,z)
=
 \int dt dz \ \lag(t,z).
\end{align}\normalsize
The action \en{ S } is invariant 
if we replace \en{ \ast } to \en{ \circ } in \en{ L } 
because of (\ref{dync}).
Then (\ref{tvlagconb}) is rewritten as
\small\begin{align}
 L
&=
 \int dz \Bigl[
 \lag^{\win}
+ 
 \bm{a}\trans
 \Bigl\{(\phi\drm{in}-\phi\drm{out}) 
+ 
 \circ \ \lie (\phi\drm{in} 
+ 
 \phi\drm{out})
\Bigr\}\Bigr],
\end{align}\normalsize
Now the Euler-Lagrange equation yields
\begin{subequations}
 \small\begin{align}
 \phi\drm{out}\dgg \im\partial\dd{+}\win 
&+ 
 \bm{a}\trans(1-\circ \ \lie) = 0.
\\
 \phi\drm{in}\dgg \im\partial\dd{+}\win  
&+ 
 \bm{a}\trans(1+\circ \ \lie) = 0.
 \end{align}\normalsize
\end{subequations}
from which we get a relevant multiplier as
\small\begin{align}
\hspace{20mm}
 \bm{a}\trans
&=
 -\mm{\phi}\dgg \im\partial\dd{+}\win.
\qquad
\left(
 \mm{\phi}\equiv \frac{\phi\drm{in}+\phi\drm{out}}{2}
\right)
\end{align}\normalsize
A free field Lagrangian \en{ \lag\urm{f} }
and an interaction Lagrangian \en{ \lag\urm{TV} } 
are given as
\begin{subequations}
\label{tvsu2lag}
 \small\begin{align}
 \lag\urm{f}
&=
 \lag^{\win}
-
 \mm{\phi}\dgg (\im\partial\dd{+}\win) (\phi\drm{in}-\phi\drm{out}),
\\
 \lag\urm{TV}
&=
 -2 (\im\partial\dd{+}\win)
 \mm{\phi}\dgg \circ \lie \mm{\phi} 
\\ &=
 -2 (\im\partial\dd{+}\win)
 \mm{\phi}\dgg \lie\ast \mm{\phi}.
 \end{align}\normalsize
\end{subequations}

Let us see this Lagrangian reproduce the input-output relation.
Assume that \en{ \win(x) = \step(z) }.
Then 
the Euler-Lagrange equation at \en{ z=0 } yields
\begin{subequations}
\small\begin{align}
 (\phi\drm{in}-\phi\drm{out}) + \lie \ast (\phi\drm{in}+\phi\drm{out}) &= 0,
\\
 (\phi\drm{in}\dgg-\phi\drm{out}\dgg) - (\phi\drm{in}\dgg+\phi\drm{out}\dgg)\circ \lie &= 0.
\end{align}\normalsize
\end{subequations}
After the Laplace transform, these are written as
\begin{subequations}
\label{tvinoutf}
\small\begin{align}
 \phi\drm{out} 
& = 
 \ffrac{1 + \lie(s)}{1 - \lie(s)} \phi\drm{in},
\\
 \phi\drm{out}\dgg 
& = 
 \phi\drm{in}\dgg \ffrac{1 - \lie(-s)}{1 + \lie(-s)}.
\end{align}\normalsize
\end{subequations}
This is rewritten as
\small\begin{align}
 \Phi\drm{out}
&=
 \tf \Phi\drm{in},
\end{align}\normalsize
where 
\small\begin{align}
 \tf
=
 \A{\tf_{11}}{\tf_{12}}
   {\tf_{21}}{\tf_{22}}
=
 \A{\ffrac{1+\lie}{1-\lie}}{}
   {}{\ffrac{1-\lie\simm}{1+\lie\simm}},
\qquad
 \Phi
=
  \left[
\renewcommand{\arraystretch}{1.2}
    \begin{array}{c} 
        \phi_1 \\   
        \phi_2 \\  \hdashline
        \phi_1\dgg \\      
        \phi_2\dgg
    \end{array}
\renewcommand{\arraystretch}{1}
  \right],
\label{tvgatep}
\end{align}\normalsize
and \en{ \lie\simm(s) \equiv \lie\trans(-s) } 
as defined in (\ref{simm}).
Obviously, 
\en{ \tf } satisfies
\small\begin{align}
 \tf_{11} \tf_{22}\simm
=
 I.
\label{tvpps}
\end{align}\normalsize

Let us examine system theoretical properties of \en{ \tf }.
We first note that \en{ \tf } is expressed as
\small\begin{align}
 \tf
=
 \frac{I+\lie_T}{I-\lie_T},
\end{align}\normalsize
where the reactance matrix \en{ \lie_T } is given as
\small\begin{align}
 \lie_T
\equiv
 \A{\lie}{}
   {}{-\lie\simm}.
\end{align}\normalsize
It follows from (\ref{tvpps}) that 
\small\begin{align}
 \tf \Pi \tf\simm = \Pi,
\label{pz-tv}
\end{align}\normalsize
where
\small\begin{align}
 \Pi 
\equiv 
 \A{}{I}{\pm I}{}.
\end{align}\normalsize
This indicates that 
\en{ \tf } satisfies a certain type of orthogonality
(Section \ref{sec:pzinfreq}).
For the reactance matrix,
it can be written as
\small\begin{align}
 \lie_T \Pi + \Pi \lie_T\simm 
=
 0.
\label{lie-reactance}
\end{align}\normalsize

To see what (\ref{pz-tv}) means,
suppose that \en{ \tf_{11} } has a transmission zero at \en{ s=z }.
According to Definition \ref{def:zero},
there exists a vector \en{ \eta } such that
\small\begin{align}
 \tf_{11}\trans(z)\eta
=
 0.
\end{align}\normalsize
Then it follows from (\ref{tvpps}) that
\small\begin{align}
 \lim\dd{s\to z} \tf_{22}(-s)\eta =\infty,
\end{align}\normalsize
which means that \en{ \tf_{22} } has a pole at \en{ s=-z }.
As a result, 
we have
\small\begin{align}
 \pole(\tf)
=
 -\zero(\tf).
\label{tvpzsym}
\end{align}\normalsize
The poles and transmission zeros are symmetrically placed
in the complex plane,
which is called \textit{pole-zero symmetry}.\index{pole-zero symmetry}
This property is recognized as allpass in systems theory.\index{allpass}
As shown in Chapter \ref{chap:sym},
this is also related to 
the symplectic structure of quantum systems.\index{symplectic}

\newpage

\subsection{Squeezing gate}
\label{sec:sfsq}

\begin{wrapfigure}[0]{r}[53mm]{49mm}
\vspace{-5mm}
\centering
\includegraphics[keepaspectratio,width=30mm]{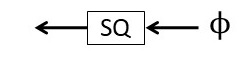}
\caption{
\small
Squeezing gate.
\normalsize
}
\label{fig-sq}
\end{wrapfigure}

Consider a gate with a single input \en{ \phi\drm{in} } 
and single output \en{ \phi\drm{out} }
such that
\small\begin{align}
\hspace{15mm}
 (\phi\drm{in}-\phi\drm{out}) + g(\phi\drm{in}\dgg+\phi\drm{out}\dgg)&=0.
\qquad
 (g\in\mathbb{R})
\label{ssb}
\end{align}\normalsize
Let us define a vector
\small\begin{align}
 \Phi
&\equiv
   \left[
    \begin{array}{l} 
        \phi\\
        \phi\dgg 
    \end{array}
  \right].
\end{align}\normalsize
Then (\ref{ssb}) is written as
\small\begin{align}
 (\Phi\drm{in}-\Phi\drm{out}) 
+ 
 \lie (\Phi\drm{in} + \Phi\drm{out})
&=
 0,
\label{ssb2}
\end{align}\normalsize
where
\small\begin{align}
 \lie = g\sigma\dd{x}= \A{}{g}{g}{}. 
\label{sq-reactance}
\end{align}\normalsize
This gate is non-unitary because \en{ \lie\dgg \not= -\lie }.
Using \en{ \{\sigma\dd{z},\sigma\dd{x}\}=0 },
it is easy to show that 
the non-unitary gauge condition (\ref{nonunicond})
is satisfied. 
If we parameterize \en{ g } as
\small\begin{align}
 g&=\frac{\ex\uu{r}-1}{\ex\uu{r}+1},
\label{gpara}
\end{align}\normalsize
(\ref{ssb2}) is rewritten as
\small\begin{align}
 \Phi\drm{out}
&=
 \frac{1}{1-g^2}
 \A{1+g^2}{2g}{2g}{1+g^2}
 \Phi\drm{in}
=
 \A{\cosh r}{\sinh r}{\sinh r}{\cosh r} 
 \Phi\drm{in}.
\label{ssmx} 
\end{align}\normalsize
This transfer function describes 
squeezing\index{squeezing gate (single field)}.
In fact, 
in the quadrature basis (\ref{quadrature})
\small\begin{align}
\hspace{25mm}
 \qu
\equiv
 \AV{\xi}{\eta}
&=
 \toqu
 \Phi,
\hspace{7mm}
 \toqu
\equiv
 \frac{1}{\sqrt{2}}
 \A{1}{1}{-\im}{\im}.
\label{quadrature-1}
\end{align}\normalsize
(\ref{ssmx} ) is written as
\small\begin{align}
 \qu\drm{out}
&=
 \A{\ex\uu{r}}{}{}{\ex\uu{-r}}
 \qu\drm{in}.
\end{align}\normalsize

In Lagrange's method (\ref{nonunilag}), 
the interaction Lagrangian is given by 
\small\begin{align}
\hspace{42mm}
 \lag\urm{SQ} 
&=
 (\im \partial\dd{+}\win) g
 \left[
 \mm{\phi} \mm{\phi} 
-
 \mm{\phi}\dgg \mm{\phi}\dgg 
 \right].
\quad
 \left(
 \mm{\phi}\equiv \frac{\phi\drm{in}+\phi\drm{out}}{2}
\right)
\label{squeezer}
\end{align}\normalsize
In the gauge theoretical approach (\ref{nonintgauge}),
\small\begin{align}
 \lag\urm{SQ}
&=
 \im \gel g
 \left[ \phi\phi - \phi\dgg\phi\dgg\right].
\end{align}\normalsize
This is the same as 
the degenerate parametric amplifier 
in Example \ref{ex:nonuni1}.

\newpage

\subsection{Cross-squeezing gate}
\label{subsec:twosq}

\begin{wrapfigure}[0]{r}[53mm]{49mm}
\centering
\vspace{1mm}
\includegraphics[keepaspectratio,width=35mm]{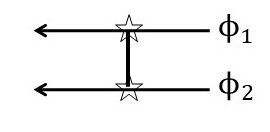}
\caption{
\small
Cross-squeezing gate.
\normalsize
}
\label{fig-ts}
\end{wrapfigure}

Consider a gate with two inputs and two outputs 
as in Figure \ref{fig-ts}.\index{cross-squeezing gate}
Assume that 
the inputs and the outputs are related to each other as 
\begin{subequations}
\label{tfsbc}
\small\begin{align}
 (\phi\drm{1,in} - \phi\drm{1,out}) 
+ 
 g(\phi\drm{2,in}\dgg + \phi\drm{2,out}\dgg) &=0,
\\ 
 (\phi\drm{2,in}\dgg - \phi\drm{2,out}\dgg) 
+ 
 g(\phi\drm{1,in} + \phi\drm{1,out}) &=0.
\qquad
 (g\in\mathbb{R})
\end{align}\normalsize
\end{subequations}
Let us introduce the following notation:
\small\begin{align}
 \Phi\dd{\alpha}
\equiv
 \AV{\phi\dd{\alpha}}{\phi\dd{\alpha}\dgg},
\quad
 \Phi
&\equiv
 \AV{\Phi\drm{1}}{\Phi\drm{2}}
=
 \left[
\begin{array}{l}
 \phi_1 \\
 \phi_1\dgg \\ \hdashline
 \phi_2 \\
 \phi_2\dgg 
\end{array}
 \right].
\end{align}\normalsize
Then (\ref{tfsbc}) is rewritten as
\small\begin{align}
 (\Phi\drm{in} - \Phi\drm{out}) 
+ 
 \lie (\Phi\drm{in} + \Phi\drm{out}) 
&=
 0,
\label{tsbounday}
\end{align}\normalsize
where
\small\begin{align}
 \lie
&\equiv
 \A{}{g\sigma\dd{x}}{g\sigma\dd{x}}{}.
\end{align}\normalsize
This satisfies the non-unitary gauge condition (\ref{nonunicond}).
Using the same parametrization as (\ref{gpara}),
we can further rewrite the input-output relations as
\small\begin{align}
 \left[
\begin{array}{c}
 \xi\drm{1} + \xi\drm{2}\\
 \eta\drm{1} - \eta\drm{2}\\ \hdashline
 \xi\drm{1} - \xi\drm{2}\\
 \eta\drm{1} + \eta\drm{2}
\end{array}
\right]\drm{out}
&=
 \left[
\begin{array}{cc:cc}
 \ex\uu{-r} & & &\\
 & \ex\uu{-r} & &\\ \hdashline
 & & \ex\uu{r} &\\
 & & & \ex\uu{r}
\end{array}
\right]
 \left[
\begin{array}{c}
 \xi\drm{1} + \xi\drm{2}\\
 \eta\drm{1} - \eta\drm{2}\\ \hdashline
 \xi\drm{1} - \xi\drm{2}\\
 \eta\drm{1} + \eta\drm{2}
\end{array}
\right]\drm{in}.
\label{tfsqquad}
\end{align}\normalsize
In this case,
the interaction Lagrangian is written as
\small\begin{align}
 \lag\urm{CQ} 
&= 
 2 (\im\partial\dd{+}\win) g
 \left[
 \mm{\phi}_1 \mm{\phi}_2
 -
 \mm{\phi}_1\dgg \mm{\phi}_2\dgg 
\right].
\label{twosqlag}
\end{align}\normalsize
This form is known as a non-degenerate parametric amplifier.

%\newpage

\subsection{QND gate}
\label{subsec:sum}

\begin{wrapfigure}[0]{r}[53mm]{49mm}
\centering
\vspace{1mm}
\includegraphics[keepaspectratio,width=37mm]{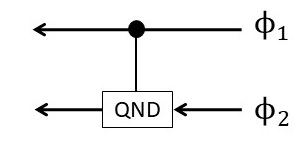}
\caption{
\small
QND gate.
\normalsize
}
\label{fig-sum1}
\end{wrapfigure}

Let us consider a QND gate\index{QND gate} with
a control \en{ \phi_1 } and target \en{ \phi_2 }
as in Figure \ref{fig-sum1}.
This gate is defined by the following input-output relation:
\small\begin{align}
 \left[
\begin{array}{c}
 \xi\drm{1} \\
 \eta\drm{1} \\ \hdashline
 \xi\drm{2} \\
 \eta\drm{2}
\end{array}
\right]\drm{out}
&=
\left[
\begin{array}{cc:cc}
 1& &  & \\
  & 1 & & g\\ \hdashline
 -g & & 1 & \\
  &  & & 1
\end{array}
\right]
\left[
\begin{array}{c}
 \xi\drm{1}\\
 \eta\drm{1}\\ \hdashline
 \xi\drm{2}\\
 \eta\drm{2}
 \end{array}
\right]\drm{in}.
\label{sumbounary} 
\end{align}\normalsize
This can be rewritten as
\small\begin{align}
 \Phi\drm{out}
&=
 \A{I}{\ffrac{g}{2} Q\dd{-} }{-\ffrac{g}{2} Q\dd{+}}{I}
 \Phi\drm{in},
\label{subb2}
\end{align}\normalsize
where we have defined 
\small\begin{align}
% \Phi\drm{out} &\equiv \AV{\Phi\drm{1,out}}{\Phi\drm{2,out}},
%\quad
% \Phi\drm{in} \equiv \AV{\Phi\drm{1,in}}{\Phi\drm{2,in}},
%\quad
 Q\dd{\pm} \equiv I \pm \sigma\dd{x} = \A{1}{\pm 1}{\pm 1}{1}.
\end{align}\normalsize
These matrices satisfy
\small\begin{align}
 Q\dd{-}Q\dd{+} = Q\dd{+}Q\dd{-} =0,
\qquad
 Q\dd{\pm}^2=2Q\dd{\pm}.
\label{abpro}
\end{align}\normalsize

The input-output relation can be rewritten as
\small\begin{align}
 (\Phi\drm{in}-\Phi\drm{out}) 
+ 
 \lie (\Phi\drm{in}+\Phi\drm{out}) 
&=
 0,
\label{sumbc}
\end{align}\normalsize
where
\small\begin{align}
 \lie
&\equiv
 \frac{g}{4} \A{}{Q\dd{-}}{-Q\dd{+}}{}.
\label{qndg}
\end{align}\normalsize
This matrix satisfies 
the non-unitary gauge condition (\ref{nonunicond}).
The interaction Lagrangian is given as
\small\begin{align}
 \lag\urm{QND} 
&= 
 (\partial\dd{+}\win) g
 \mm{\xi}_1 \mm{\eta}_2.
\label{sumlag}
\end{align}\normalsize

\subsection{XX gate}

Let us consider a gate defined by 
the following input-output relation:
\small\begin{align}
\left[
\begin{array}{c}
 \xi\drm{1} \\
 \eta\drm{1}\\ \hdashline
 \xi\drm{2} \\
 \eta\drm{2}
\end{array}
\right]\drm{out}
 &=
\left[
 \begin{array}{cc:cc}
 1 & & & \\
 & 1 & g & \\  \hdashline
 & & 1 & \\
 g & & & 1
 \end{array}
\right]
\left[
\begin{array}{c}
 \xi\drm{1} \\
 \eta\drm{1}\\ \hdashline
 \xi\drm{2} \\
 \eta\drm{2}
\end{array}
\right]\drm{in}.
\label{xxboundary}
\end{align}\normalsize
This gate is called an XX gate\index{XX gate}.
The input-output relation is rewritten as
\small\begin{align}
 (\Phi\drm{in}-\Phi\drm{out})
+
 \lie (\Phi\drm{in}+\Phi\drm{out}) &=0,
\end{align}\normalsize
where 
\small\begin{align}
 \lie
&=
 \frac{\im g}{4} 
 \A{}{Q}
 {Q}{},
\quad
 Q \equiv \sigma\dd{z} +\im\sigma_y = \A{1}{1}{-1}{-1}.
\end{align}\normalsize
Note that
\small\begin{align}
 Q^2=0.
\end{align}\normalsize
The interaction Lagrangian of this gate is given as
\small\begin{align}
 \lag\urm{XX}
&=
 (\partial\dd{+}\win) g \mm{\xi}_1 \mm{\xi}_2.
\end{align}\normalsize

\begin{remark}
For the QND and XX gates,
the reactance matrix satisfies \en{ \lie^2=0 }.
In this case, 
the interaction Lagrangian is not uniquely determined.
For example,
another forms are given as
\begin{subequations}
\small\begin{align}
 \lag^{\mathrm{QND}}
&=
 \frac{1}{2}(\partial\dd{+}\win) g
\bigl[
   \xi_{\mathrm{1,in}} \eta_{\mathrm{2,in}} 
 +  \xi_{\mathrm{1,out}} \eta_{\mathrm{2,out}} 
\bigr],
\\
 \lag^{\mathrm{XX}} 
&=
\frac{1}{2}(\partial\dd{+}\win) g 
\bigl[
 \xi_{\mathrm{2,in}} \xi_{\mathrm{1,out}} 
+
 \xi_{\mathrm{1,in}} \xi_{\mathrm{2,out}}
\bigr].
\end{align}\normalsize
\end{subequations}
\end{remark}

\newpage

\subsection{SU(2) gate and chiral symmetry breaking}
\label{sec:su2chiral}

\begin{wrapfigure}[0]{r}[53mm]{49mm}
\centering
\vspace{-5mm}
\includegraphics[keepaspectratio,width=49mm]{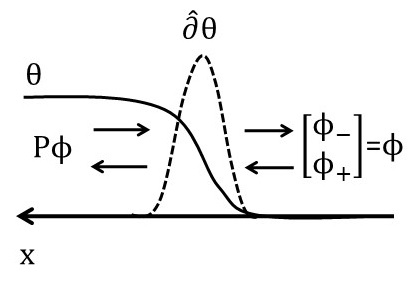}
\caption{
\small
Forward and backward traveling fields interact with each other
through an SU(2) gate.
\normalsize
}
\label{fig-chiral}
\end{wrapfigure}

Here we consider a toy model similar to 
chiral symmetry breaking.\index{chiral symmetry breaking}
As seen in Section \ref{rem:weyl},
the Weyl equation can be decomposed to orthogonal components:
the left- and right-chiral fields \en{ \{\psi_L,\psi_R\} }.
Due to the orthogonality,
the Weyl equation is invariant
under global gauge transformations applied to
\en{ \psi_L } and \en{ \psi_R } independently.
This is called \textit{chiral symmetry}.\index{chiral symmetry}
On the other hand, 
this symmetry does not hold for the Dirac equation
because of the mass term.
Hence chiral symmetry breaking is considered to be
a part of a mass generation mechanism.

We show a similar effect with an SU(2) gate
by choosing suitable parameters.
This toy model is different from 
a well-known model of spontaneous symmetry breaking,
but it is interesting to see 
how massless particles gain masses in the SU(2) gate.

In this model,
the left- and right-chiral fields are replaced to
forward and backward traveling fields \en{ \{\phi\dd{+},\phi\dd{-}\} }
as in Figure \ref{fig-chiral}.
It is described by a Lagrangian
\small\begin{align}
 \lag(\phi,\hat{\partial}\phi)
=
 \phi\dgg \im \hat{\partial}\phi,
\label{beforesu2}
\end{align}\normalsize
where
\small\begin{align}
 \hat{\partial} = 
 \A{\partial\dd{-}}{}
   {}{\partial\dd{+}},
\qquad
 \phi 
&= 
 \AV{\phi\dd{-}}{\phi\dd{+}}.
\end{align}\normalsize
Note that 
this is symmetric under a global gauge transformation
\small\begin{align}
\hspace{20mm}
 \phi\to \phi'
&=
 \A{\ex\uu{\im \alpha}}{}{}{\ex\uu{-\im \alpha}}\phi.
\qquad
 (\alpha\in \mathbb{R})
\label{globalt}
\end{align}\normalsize

Suppose that \en{ \phi } goes through an SU(2) gate.
The Lagrangian is given as %(Section \ref{subsec:su2})
\small\begin{align}
 \lag
&=
 \phi\dgg \im \hat{\partial} \phi
-
 \phi\dgg 2(\im \hat{\partial} \win) \lie \phi,
\label{lagchiral}
\end{align}\normalsize
where 
\small\begin{align}
\hspace{20mm}
 \lie
&\equiv
 \A{0}{g}{-g\uu{*}}{0}.
\qquad
 (g\in\mathbb{C})
\end{align}\normalsize
This Lagrangian is still symmetric under (\ref{globalt})
if \en{ g\to \ex\uu{2\im \alpha}g }.
However, 
if we choose a specific value of \en{ g },
this symmetry is broken.

Let us introduce \en{ 2\times 2 } matrices
\small\begin{align}
 \gamma^0
=
 \A{}{1}{1}{},
\qquad
 \gamma^3
=
 \A{}{1}{-1}{},
\end{align}\normalsize
and
\small\begin{align}
 \widebar{\phi}
\equiv
 \phi\gamma^0,
\qquad
 \fsh{\partial} 
\equiv 
 \gamma^0\partial\dd{t} + \gamma^3\partial\dd{z}
=
 \gamma^0\hat{\partial}.
\end{align}\normalsize
Then (\ref{lagchiral}) is rewritten as
\small\begin{align}
 \lag
&=
 \widebar{\phi} \im \fsh{\partial} \phi
-
 \widebar{\phi} 2(\im \fsh{\partial}\win) \lie \phi.
\end{align}\normalsize

Assume that the SU(2) gate is designed so that 
the coupling constant is purely imaginary 
\en{g=-\im m/2 \ (m\in \mathbb{R})} 
and 
\en{\win(t,z)=ct \ (c\in\mathbb{R})}
in a certain interval.
Then 
an effective Lagrangian inside the SU(2) gate is written as
\small\begin{align}
 \lag
&=
 \widebar{\phi} 
 \bigl[ \im \fsh{\partial} 
-
 cm
\bigr]\phi,
\end{align}\normalsize
which is the same form as 
the Dirac field with a mass parameter \en{cm}.

\newpage

\subsection{SU(2) gate for the Dirac field}
\label{sec:su2dirac}

Quantum gates can be defined for the Dirac field 
in the same way as the forward traveling field.
Note that only unitary gates are well defined for fermions.
Here we consider an SU(2) gate as an example.
\index{SU(2) gate (Dirac)}

In the gauge theoretical approach,
this has been examined in Example \ref{ex:su2gauge}.
Suppose that 
\en{ \psi_1 } and \en{ \psi_2 } go through the SU(2) gate.
The interaction Lagrangian is given by (\ref{su2gaugeintl}):
\small\begin{align}
 \lag\urm{SU(2)}
&=
 -2\im g
\left[
 \widebar{\psi}_1 \fsh{A} \psi_2
-
 \widebar{\psi}_2 \fsh{A} \psi_1
\right].
\label{dsu2l}
\end{align}\normalsize
In the classical limit \en{ A\dd{\mu}\to V(x) },
this interaction Lagrangian is of the same form as the forward traveling case:
\small\begin{align}
 \lag\urm{SU(2)}
&=
 -2\im g
\left[
 \psi\dgg_1\psi_2
-
 \psi\dgg_2\psi_1
\right]V(x).
\label{dsu2l-1}
\end{align}\normalsize

\begin{wrapfigure}[0]{r}[53mm]{49mm}
\centering
\vspace{-6mm}
\includegraphics[keepaspectratio,width=49mm]{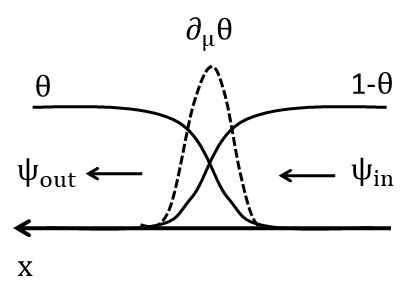}
\caption{
\small
SU(2) gate for the Dirac field.
\normalsize
}
\label{fig-weight2-1}
\end{wrapfigure}

Let us derive this Lagrangian
using Lagrange's method for later use.
A weight function \en{ \win(x) }
is defined as in Figure \ref{fig-weight2-1}.
The two Dirac fields \en{ \{\psi_1,\psi_2\} }
are labeled as 
\en{ \{\psi\drm{1,in},\psi\drm{2,in}\} } in the domain of \en{ 1-\win },
and as 
\en{ \{\psi\drm{1,out},\psi\drm{2,out}\} } in the domain of \en{ \win },
respectively.
Let us introduce the following vector form:
\small\begin{align}
 \psi\drm{in}
\equiv
 \AV{\psi_1}{\psi_2}\drm{in},
\quad
 \psi\drm{out}
\equiv
 \AV{\psi_1}{\psi_2}\drm{out},
\end{align}\normalsize
Then the input-output relation is expressed as
\small\begin{align}
 (\psi\drm{in}-\psi\drm{out})
+
 \lie (\psi\drm{in}+\psi\drm{out}) 
&= 0,
\label{dbound}
\end{align}\normalsize
where
\small\begin{align}
 \lie
&\equiv
 \A{0}{g}{-g}{0}.
\end{align}\normalsize
The Lagrangian is given in the same way as 
the forward traveling field (\ref{unilag}):
\begin{subequations}
\label{diracsu22}
\small\begin{align}
 \lag\urm{SU(2)}
&=
 -2 \widebar{\mm{\psi}} (\im\fsh{\partial}\win) \lie \mm{\psi}.
\\ &=
 -2\im g \left[
 \widebar{\mm{\psi}}_1 (\fsh{\partial}\win) \mm{\psi}_2
-
 \widebar{\mm{\psi}}_2 (\fsh{\partial}\win) \mm{\psi}_1
\right],
\end{align}\normalsize
\end{subequations}
which is the same form as (\ref{dsu2l-1}).

\chapter{Quantum circuits}
\label{chap:circuit}
\thispagestyle{fancy}

We introduce circuits 
by connecting multiple gates (gauge transformations)
through their inputs and outputs in a concatenated way.
The operation of a circuit depends on 
what order to place the gates in.
In general, 
if we swap two gates in a circuit,
it operates differently.
This means that 
the Lagrangian of the circuit is not simply 
the sum of individual gates.
There are extra interactions in the circuit
due to the noncommutativity of the gates.
We investigate this effect in this chapter.

\section{Cascaded SU(2) circuit}
\label{sec:generalsu2}

In this section, 
we consider circuits of two SU(2) gates.
There are two possible cases:
The two gates do commute and do not commute.

\subsection{SU(2)+SU(2): a commutative case}

\begin{wrapfigure}[0]{r}[53mm]{49mm} 
\vspace{-0mm}
\centering
\includegraphics[keepaspectratio,width=49mm]{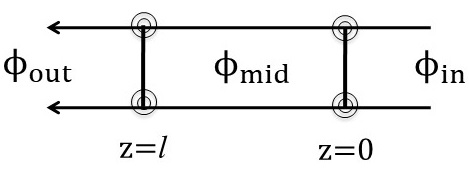}
\caption{
\small
Circuit of two SU(2) gates.
\normalsize
}
\label{fig-2su2}
\end{wrapfigure}

The first example is Figure \ref{fig-2su2}
in which two SU(2) gates are placed at \en{ z=0 } and \en{ z=l }.
In the gauge theoretical approach,
the two gates are described as
\begin{subequations}
\small\begin{align}
z=0&: \hspace{5mm}
 \phi\drm{mid} 
=
 \exp\Bigl(\win_0 2\lie_0\Bigr) 
 \phi\drm{in},
\\
z=l&: \hspace{6mm}
 \phi\drm{out} 
=
 \exp\Bigl(\win\dd{l} 2\lie\dd{l}\Bigr) 
 \phi\drm{mid},
\end{align}\normalsize
\end{subequations}
where \en{ \win\dd{\alpha} } are weight functions and 
\small\begin{align}
\hspace{10mm}
 \lie\dd{\alpha}
=
 \A{}{g\dd{\alpha}}{-g\dd{\alpha}}{}.
\hspace{5mm}
 (g\dd{\alpha}\in\mathbb{R})
\end{align}\normalsize

Note that \en{ [\lie_0, \lie\dd{l}]=0 }.
This means that 
the input-output relation of the circuit is invariant
if we swap the two SU(2) gates.
In fact, 
\begin{subequations}
\label{su2cascom}
\small\begin{align}
 \phi\drm{out}
&=
 \exp\Bigl(\win\dd{l} 2\lie\dd{l}\Bigr) 
 \exp\Bigl(\win_0 2\lie_0\Bigr) 
 \ \phi\drm{in}
\\ &=
 \exp\Bigl(\win\dd{l} 2\lie\dd{l} + \win_0 2\lie_0\Bigr) 
 \ \phi\drm{in}.
\end{align}\normalsize
\end{subequations}
The output is defined where \en{ \win\dd{l}=\win_0=1 }
for which a reactance matrix is given as
\small\begin{align}
 \lie
=
 \lie_0+\lie\dd{l}.
\end{align}\normalsize
This is simply the sum of the two reactance matrices.

\subsection{SU(2)+SU(2): a noncommutative case}

\begin{wrapfigure}[0]{r}[53mm]{49mm} 
\vspace{0mm}
\centering
\includegraphics[keepaspectratio,width=45mm]{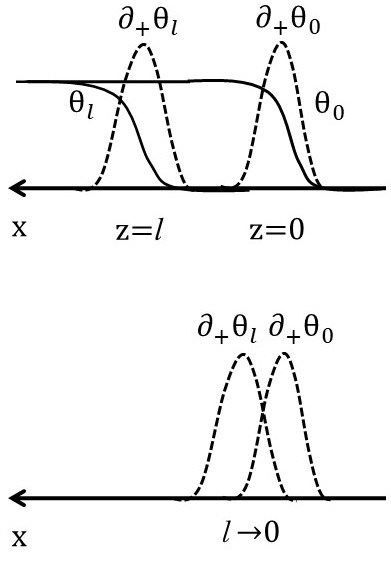}
\caption{
\small
Two weight functions $\win_0$ and $\win_{l}$ 
corresponding to Figure \ref{fig-2su2}.
As $l\to 0$, 
$\partial_+\win_0$ and $\partial_+\win_l$ overlap each other.
If $[\lie_0,\lie_l]\not= 0$,
the extra interaction arises there.
\normalsize
}
\label{fig-cas-over}
\end{wrapfigure}

Let us consider the same configuration as Figure \ref{fig-2su2},
but this time two reactance matrices do not commute:
\begin{subequations}
\small\begin{align}
 \lie\dd{l}
&=
 \frac{1}{4} 
 \A{}{\im g}
   {\im g}{}
=
 \frac{\im g}{2} T\dd{x},
\hspace{6mm}
 (g \in \mathbb{R})
\\
%%%%%%%%%%%%%%%%%%%
 \lie_0
&=
 \frac{1}{4}
 \A{}{g}
   {- g}{}
=
 \frac{\im g}{2} T\dd{y},
\end{align}\normalsize
\end{subequations}
where 
\en{ T\dd{a}\equiv \sigma\dd{a}/2 }
satisfying 
\en{[ T\dd{a},T\dd{b}]=\im \epsilon\dd{abc}T\dd{c} }.
As in (\ref{su2cascom}), 
the circuit is expressed as
\begin{subequations}
\small\begin{align}
 \phi\drm{out}
&=
 \exp\Bigl( \win\dd{l} \im g T\dd{x} \Bigr) 
 \exp\Bigl( \win_0     \im g T\dd{y} \Bigr) 
 \ \phi\drm{in}
\\ &\sim
 \exp\Bigl( \win\uu{a} \im g T\dd{a} \Bigr) 
 \ \phi\drm{in},
\end{align}\normalsize
\end{subequations}
where
\small\begin{align}
 \win^1=\win\dd{l},
\quad
 \win^2=\win_0,
\quad
 \win^3=-\frac{1}{2}\win\dd{l}\win_0 g.
\end{align}\normalsize
Now we have an extra term due to 
\en{ [\lie_0,\lie\dd{l}]\not= 0 }
that results in an extra interaction in the circuit.
This is the same as the general form (\ref{su2ex}).
The output is defined where \en{ \win_0=\win\dd{l}=1 } 
for which a reactance matrix is given as
\small\begin{align}
 \lie
&=
 -\frac{\im g}{8}
 \A{g}{-2(1-\im)}
   {-2(1+\im)}{-g}.
\end{align}\normalsize

\subsection{SU(2)+SU(2)=SU(3): a noncommutative case}
\label{sec:su2su2noncom}

\begin{wrapfigure}[0]{r}[53mm]{49mm} 
\centering
\vspace{-5mm}
\includegraphics[keepaspectratio,width=43mm]{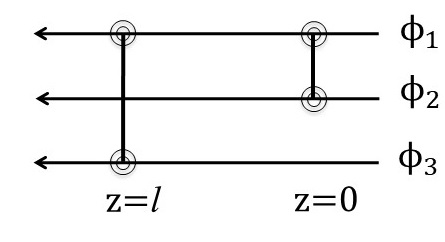}
\caption{
\small
Circuit of two SU(2) gates.
\normalsize
}
\label{fig-su2su2}
\end{wrapfigure}

Another example of noncommutative circuits is Figure \ref{fig-su2su2}
in which two SU(2) gates are concatenated through a single field.
The input-output relation is given by
\small\begin{align}
 \phi\drm{out}
&=
 \B{\ffrac{1-g\dd{l}^2}{1+g\dd{l}^2}}{}{\ffrac{2g\dd{l}}{1+g\dd{l}^2}}
   {}{1}{\rule[0mm]{0mm}{5mm}}
   {\ffrac{-2g\dd{l}}{1+g\dd{l}^2}}{}{\ffrac{1-g\dd{l}^2}{1+g\dd{l}^2}}
 \B{\ffrac{1-g\dd{0}^2}{1+g\dd{0}^2}}{\ffrac{2g\dd{0}}{1+g\dd{0}^2}}{}
   {\ffrac{-2g\dd{0}}{1+g\dd{0}^2}}{\ffrac{1-g\dd{0}^2}{1+g\dd{0}^2}}{\rule[0mm]{0mm}{8mm}}
   {}{}{1}
 \phi\drm{in}.
\quad 
 (g_0,g\dd{l}\in\mathbb{R})
\label{casmat}
\end{align}\normalsize

To obtain \en{ \lie },
let us introduce the basis of su(3) called 
Gell-Mann matrices:\index{Gell-Mann matrices}
\small\begin{align}
 \lambda_1
&\equiv 
 \B{0}{1}{}
   {1}{0}{}
   {}{}{0},
\
 \lambda_2
\equiv 
 \B{0}{-\im}{}
   {\im}{0}{}
   {}{}{0},
\
 \lambda_3
\equiv 
 \B{1}{}{}
   {}{-1}{}
   {}{}{0},
\
 \lambda_4
\equiv 
 \B{0}{}{1}
   {}{0}{}
   {1}{}{0}, 
%%%%%%%%%%%%%%%%%%%%%
\nn\\
 \lambda_5
&\equiv 
 \B{0}{}{-\im}
   {}{0}{}
   {\im}{}{0},
\
 \lambda_6
\equiv 
 \B{0}{}{}
   {}{0}{1}
   {}{1}{0},
\
 \lambda_7
\equiv 
 \B{0}{}{}
   {}{0}{-\im}
   {}{\im}{0},
\
 \lambda_8
\equiv 
\frac{1}{\sqrt{3}}
 \B{1}{}{}
   {}{1}{}
   {}{}{-2}.
\end{align}\normalsize
Then the first gate of (\ref{casmat}) is approximately written as
\small\begin{align}
 \B{1}{2g_0}{}
   {-2g_0}{1}{}
   {}{}{1}
\sim
 \exp\Bigl(\im 2g_0 \win_0 \lambda_2\Bigr),
\label{firstgate}
\end{align}\normalsize
where we have introduced a weight function \en{ \win_0 }.
Likewise, the second gate is given as
\small\begin{align}
 \B{1}{}{2g\dd{l}}
   {}{1}{}
   {-2g\dd{l}}{}{1}
\sim
 \exp\Bigl(\im 2g\dd{l} \win\dd{l} \lambda_5\Bigr).
\label{secondgate}
\end{align}\normalsize
The input-output relation of the circuit is then written as
\begin{subequations}
\label{su3}
\small\begin{align}
 \phi\drm{out}
&=
 \exp\Bigl(\im 2g\dd{l} \win\dd{l} \lambda_5\Bigr) 
 \exp\Bigl(\im 2g_0 \win_0 \lambda_2\Bigr) \ \phi\drm{in}
\\ &\sim
 \exp\Bigl(\im 2g\dd{l} \win\dd{l} \lambda_5 + \im 2g_0 \win_0 \lambda_2
     + \im 2g\dd{l}g_0 \win\dd{l}\win_0 \lambda_7 \Bigr) \ \phi\drm{in}.
\label{su3-1}
\end{align}\normalsize
\end{subequations}
The third term is an extra interaction
resulting from the noncommutativity of the two gates.
In a domain where \en{ \win_0=\win\dd{l}=1 },
this is written as
\small\begin{align}
 \phi\drm{out}
&=
 \exp\bigl(2\lie \bigr) \phi\drm{in},
\end{align}\normalsize
where 
\small\begin{align}
  \lie
&\equiv 
 \B{0}{g_0}{g\dd{l}}
   {-g_0}{0}{g_0g\dd{l}}
   {-g\dd{l}}{-g_0g\dd{l}}{0}.
\end{align}\normalsize

The same result can be obtained 
from the Cayley transform.
The input-output relation of each gate is written as
\begin{subequations}
\label{casb}
\small\begin{align}
 z=0: \quad&
 (\phi\drm{0,in} - \phi\drm{0,out})
+
 \lie_0(\phi\drm{0,in} + \phi\drm{0,out}) =0,
\label{casb1} \\
%%%%%%%%%%%%%%%%%%%%%%%%%%%%%%%%%%
 z=l: \quad&
 (\phi\drm{\ \textit{l},in} - \phi\drm{\ \textit{l},out})
+
 \lie\dd{l}(\phi\drm{\ \textit{l},in} + \phi\drm{\ \textit{l},out}) =0,
\label{casb2}
\end{align}\normalsize
\end{subequations}
where 
\begin{subequations}
\small\begin{align}
 \phi\drm{0,out} 
&= 
 \AV{\phi_1}{\phi_2}\dd{z=0+},
\qquad
 \phi\drm{0,in} 
= 
 \AV{\phi_1}{\phi_2}\dd{z=0-},
%%%%%%%%%%%%%%%%%%%%%%%%%%%%%%%%%%
\\
 \phi\drm{\ \textit{l},out} 
&= 
 \AV{\phi_1}{\phi_3}\dd{z=l+},
\qquad
 \phi\drm{\ \textit{l},in} 
= 
 \AVl{\phi_1}{\phi_3}\dd{z=l-}.
\end{align}\normalsize
\end{subequations}
If \en{ l\to 0 },
the output from the first gate is immediately 
fed into the second gate.
Eliminating the field between the two gate \en{ \phi_1(l>z>0) },
we get the input-output relation of the circuit 
as
\small\begin{align}
 (\phi\drm{in} - \phi\drm{out})
+
 \lie(\phi\drm{in} + \phi\drm{out})
&=0,
\label{casl0}
\end{align}\normalsize
where
\small\begin{align}
 \lie
&\equiv 
 \B{0}{g_0}{g\dd{l}}
   {-g_0}{0}{g_0g\dd{l}}
   {-g\dd{l}}{-g_0g\dd{l}}{0}.
\label{casl0-1}
\end{align}\normalsize

It is worth noting that 
this gate is regarded as an SU(3) gate.
The same configuration of \en{ N }-SU(2) gates 
results in a single SU(N+1) gate.
Furthermore, 
combining with the result of Section \ref{sec:su2su2noncom},
we can construct general SU(N+1) transformations
from SU(2) gates.

\newpage

\section{XX + SU(2) circuit}
\label{sec:xxsu2}

\begin{wrapfigure}[0]{r}[53mm]{49mm} 
\centering
\vspace{-0mm}
\includegraphics[keepaspectratio,width=49mm]{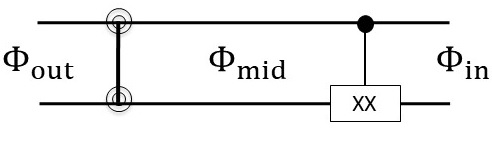}
\caption{
\small
XX + SU(2) circuit
\normalsize
}
\label{fig-xxsu2}
\end{wrapfigure}

Let us consider a circuit in Figure \ref{fig-xxsu2}.
The XX and SU(2) gates are, respectively, defined by
reactance matrices
\begin{subequations}
\small\begin{align}
\hspace{20mm}
 \lie\urm{XX} &= \frac{\im g_0}{4}\A{}{Q}{Q}{},
\qquad
\left( Q = \A{1}{1}{-1}{-1} \right) 
\\
\hspace{20mm}
 \lie\urm{SU(2)} &= g\dd{l}\A{}{I }{-I }{}.
\end{align}\normalsize
\end{subequations}

The cascade configuration is expressed as
\begin{subequations}
\small\begin{align}
 \exp\bigl( 2\lie \bigr)
&=
 \exp \bigl( 2\lie\urm{SU(2)} \bigr)
 \exp \bigl( 2\lie\urm{XX}    \bigr)
\\ &\sim
 \exp \bigl( 2 \lie\urm{SU(2)} + 2 \lie\urm{XX} 
           + 2 [\lie\urm{SU(2)},\lie\urm{XX}] \bigr),
\end{align}\normalsize
\end{subequations}
from which the reactance matrix of the circuit is given as
\small\begin{align}
 \lie
&=
 \lie\urm{SU(2)}
+
 \lie\urm{XX}
+
 \lie\urm{int},
\end{align}\normalsize
where 
\small\begin{align}
 \lie\urm{int}
\equiv
 \frac{\im g_0 g_l}{2}
 \A{Q}{}{}{-Q}.
\end{align}\normalsize

We can obtain the same \en{ \lie\urm{int} } using the Cayley transform.
The two gates are expressed as
\begin{subequations}
\small\begin{align}
 \mbox{SU(2) gate:} \hspace{5mm}
 \bigl( 1-\lie\urm{SU(2)} \bigr) \Phi\drm{out}
&=
 \bigl( 1+\lie\urm{SU(2)} \bigr) \Phi\drm{mid},
\\
%%%%%%%%%
 \mbox{XX gate:} \hspace{24mm}
 \Phi\drm{mid}
&=
 \bigl( 1+2\lie\urm{XX} \bigr) \Phi\drm{in}.
\end{align}\normalsize
\end{subequations}
Eliminating \en{ \Phi\drm{mid} }, we get
\small\begin{align}
 (\Phi\drm{in}-\Phi\drm{out})
+
 \Bigl\{ 
\bigl( \lie\dd{l} + 2\lie_0 - 2\lie_0 \lie\dd{l} \bigr) 
\Phi\drm{in} 
+ 
 \lie\dd{l} \Phi\drm{out}\Bigr\}
&=0.
\end{align}\normalsize
This is the same form as (\ref{tfXY}), so 
it follows from (\ref{pcompare1}) that up to second order
\small\begin{align}
 \lie\urm{int}
=
 \A{ \ffrac{\im g_0 g\dd{l}}{2}Q}
   { \ffrac{\im g_0 g\dd{l}^2}{4}Q}
   { \ffrac{\im g_0 g\dd{l}^2}{4}Q}
   {-\ffrac{\im g_0 g\dd{l}}{2}Q}
\sim
 \frac{\im g_0 g\dd{l}}{2}
 \A{Q}{}{}{-Q}.
\end{align}\normalsize
The resulting interaction Lagrangian is given as
\begin{subequations}
\small\begin{align}
 \lag\urm{XX+SU(2)}
&=
 \lag\urm{XX}
+
 \lag\urm{SU(2)}
+
 \lag\urm{int}, 
\end{align}\normalsize
\end{subequations}
where \en{ \lag\urm{int} } is an extra interaction given as
\small\begin{align}
 \lag\urm{int}
&=
 g_0g\dd{l} \bigl( \mm{\xi}_1^2 - \mm{\xi}_2^2 \bigr).
\label{lagintsu2xx-1}
\end{align}\normalsize

\newpage

\section{D-feedforward and d-feedback}
\label{sec:fffb}

\begin{wrapfigure}[0]{r}[53mm]{49mm} 
\centering
\vspace{-10mm}
\includegraphics[keepaspectratio,width=40mm]{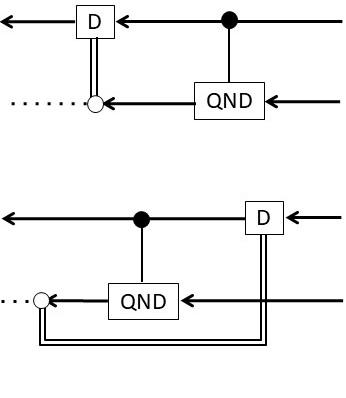}
\caption{
\small
D-feedforward and d-feedback.
\normalsize
}
\label{fig-d-ff-fb}
\end{wrapfigure}

In this section,
we consider two types of circuits 
depicted in Figure \ref{fig-d-ff-fb}.
We assume that 
the displacement parameter is associated with 
the output of the QND gate.
In this case, 
the operations of the circuits depend critically on 
which gate comes first.
If the QND gate is first,
the flow of information is in the forward direction.
This is called \textit{d-feedforward}.\index{d-feedforward}
On the other hand,
the flow of information is in the backward direction
if the displacement gate comes first.
This is called \textit{d-feedback}.\index{d-feedback}
These are simple toy models 
that simulate classical feedforward and feedback processes.

\subsection{QND + SU(2) circuit}

\begin{wrapfigure}[0]{r}[53mm]{49mm} 
\centering
\vspace{20mm}
\includegraphics[keepaspectratio,width=45mm]{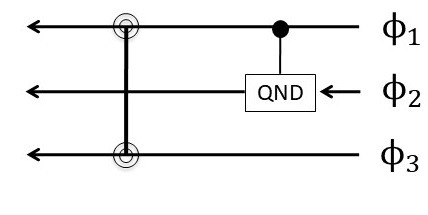}
\caption{
\small
QND+SU(2) circuit.
\normalsize
}
\label{fig-sumsu2}
\end{wrapfigure}

To examine the d-feedforward process,
we start with the cascade of QND and SU(2) gates
in Figure \ref{fig-sumsu2}.
The QND and SU(2) gates are,
respectively, defined by
\small\begin{align}
 \lie\urm{QND}
&\equiv
 \frac{g_0}{4}
 \B{}{Q\dd{-}}{}
   {-Q\dd{+}}{}{}
   {}{}{0},
\qquad
 \lie\urm{SU(2)}
\equiv \
 g\dd{l}
 \B{}{}{I}
   {}{0}{}
   {-I}{}{}.
\end{align}\normalsize
The cascade of these two gates are written as
\begin{subequations}
\small\begin{align}
 \exp\bigl( 2\lie \bigr)
&=
 \exp\bigl( 2\lie\urm{SU(2)} \bigr)
 \exp\bigl( 2\lie\urm{QND}   \bigr)
\\ & \sim 
 \exp
 \bigl(
   2\lie\urm{SU(2)} + 2\lie\urm{QND} + 2\lie\urm{int}
 \bigr),
\end{align}\normalsize
\end{subequations}
where 
\small\begin{align}
 \lie\urm{int}
=
 [\lie\urm{SU(2)}, \lie\urm{QND}]
=
 \frac{g_0g\dd{l}}{4}
 \B{0}{}{}
   {}{}{Q\dd{+}}
   {}{-Q\dd{-}}{}.
\end{align}\normalsize
The resulting interaction Lagrangian is given by
\small\begin{align}
 \lag\urm{QND+SU(2)}
= 
 \lag\urm{QND}
+
 \lag\urm{SU(2)}
+
 \lag\urm{int},
\label{sumsu2}
\end{align}\normalsize
where
\small\begin{align}
 \lag\urm{int}
=
 -\im
 \frac{g_0g\dd{l}}{2} (\mm{\phi}_2\dgg-\mm{\phi}_2)
                      (\mm{\phi}_3\dgg+\mm{\phi}_3).
\end{align}\normalsize

%\newpage

\subsection{QND + displacement circuit}
\label{sec:sumsu2}

\begin{wrapfigure}[0]{r}[53mm]{49mm}
\vspace{-7mm}
\centering
\includegraphics[keepaspectratio,width=45mm]{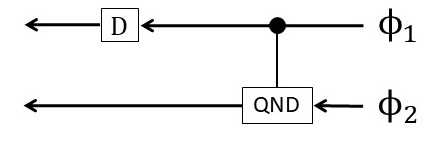}
\caption{
\small
QND + displacement circuit.
\normalsize
}
\label{fig-sumd}
\end{wrapfigure}

The next step toward d-feedforward
is converting the SU(2) gate to a displacement gate
as in Figure \ref{fig-sumd}.
As shown in Section \ref{subsec:dis},
this is done by setting the displacement parameter as
\small\begin{align}
 \mm{\phi}_3 
= 
 \frac{d}{2g\dd{l}}. 
\end{align}\normalsize
Substituting this into (\ref{sumsu2}) yields
\small\begin{align}
 \lag\urm{QND+SU(2)}
\Rightarrow \ 
 \lag\urm{QND+D}
&=
 \lag\urm{QND} 
+
 \lag\urm{D}
- \im 
 \frac{g_0 }{4} (\mm{\phi}_2\dgg-\mm{\phi}_2) (d + d\uu{*} ).
\label{lag-qndd}
\end{align}\normalsize

Let us check to see if this Lagrangian leads to 
a correct input-output relation.
First we express the interaction Lagrangian as
\small\begin{align}
 \lag\urm{QND+D}
=
 -\im  
 \AH{\mm{\Phi}_1\dgg}{\mm{\Phi}_2\dgg}
 \Sigma\dd{z}
\left\{
 \frac{g_0}{4}
 \A{}{Q\dd{-}}{-Q\dd{+}}{}
 \AV{\mm{\Phi}_1}{\mm{\Phi}_2}
+
 \left[
\begin{array}{c}
 d\\
 d\uu{*}\\
 \ffrac{g_0}{4}(d+d\uu{*})\\
 \ffrac{g_0}{4}(d+d\uu{*})
\end{array}
 \right]
 \right\},
\label{lag-qndd2}
\end{align}\normalsize
for which the Euler-Lagrange equation is written as
\small\begin{align}
 \A{1}{ \hspace{-4mm} -\ffrac{g_0}{4}Q\dd{-}}
   {\ffrac{g_0}{4}Q\dd{+}}{ \hspace{-4mm} 1}
 \AV{\Phi_1}{\Phi_2}\drm{out}
&=
 \A{1}{ \hspace{-4mm} \ffrac{g_0}{4}Q\dd{-}}
   {-\ffrac{g_0}{4}Q\dd{+}}{ \hspace{-4mm} 1}
 \AV{\Phi_1}{\Phi_2}\drm{in}
+
 \left[
\begin{array}{c}
 d \\
 d\uu{*} \\
 \ffrac{g_0}{4}(d+d\uu{*}) \\
 \ffrac{g_0}{4}(d+d\uu{*}) 
\end{array}
\right].
\end{align}\normalsize
In the quadrature basis,
this is rewritten as
\small\begin{align}
 \left[
\begin{array}{c}
 \xi_1 \\
 \eta_1 \\
 \xi_2 \\
 \eta_2
\end{array}
\right]\drm{out}
&=
\left[
\begin{array}{cccc}
 1 & & & \\
 & 1 & & g_0 \\
 -g_0 & & 1 & \\
 & & & 1
\end{array}
\right]
 \left[
\begin{array}{c}
 \xi_1\\
 \eta_1\\
 \xi_2\\
 \eta_2
\end{array}
\right]\drm{in}
+
 \left[
\begin{array}{c}
 \sqrt{2}\rea(d) \\
 \sqrt{2}\ima(d) \\
 0 \\
 0
\end{array}
\right],
\label{summx2}
\end{align}\normalsize
in which the QND gate operates first on the input (the first term)
and the displacement gate second (the second term),
as expected.

\begin{wrapfigure}[0]{r}[53mm]{49mm} 
\vspace{-0mm}
\centering
\includegraphics[keepaspectratio,width=45mm,]{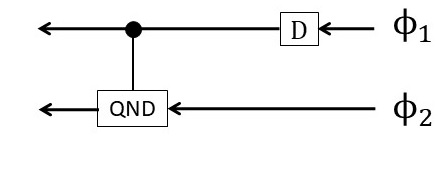}
\caption{
\small
Displacement + QND circuit.
\normalsize
}
\label{fig-dsum}
\end{wrapfigure}

Likewise, 
for Figure \ref{fig-dsum},
an interaction Lagrangian is given by
\small\begin{align}
 \lag\urm{D+QND}
&=
 \lag\urm{QND} 
+
 \lag\urm{D}
+ \im 
 \frac{g_l }{4} (\mm{\phi}_2\dgg-\mm{\phi}_2) (d + d\uu{*} ).
\label{lagsu2sum}
\end{align}\normalsize
Compared to (\ref{lag-qndd}),
a difference is only the sign of the third term,
which results from swapping the two gates
i.e.,
\en{ [\lie\urm{QND},\lie\urm{SU(2)}] = - [\lie\urm{SU(2)},\lie\urm{QND}] }.
In this case, 
the interaction Lagrangian is expressed as
\small\begin{align}
 \lag\urm{D+QND}
=
 -\im  
 \AH{\mm{\Phi}_1\dgg}{\mm{\Phi}_2\dgg}
 \Sigma\dd{z}
\left\{
 \frac{g\dd{l}}{4}
 \A{}{Q\dd{-}}{-Q\dd{+}}{}
 \AV{\mm{\Phi}_1}{\mm{\Phi}_2}
+
 \left[
\begin{array}{c}
 d\\
 d\uu{*}\\
 -\ffrac{g\dd{l}}{4}(d+d\uu{*})\\
 -\ffrac{g\dd{l}}{4}(d+d\uu{*})
\end{array}
 \right]
 \right\},
\label{lag-dqnd}
\end{align}\normalsize
for which 
the Euler-Lagrange equation is written as
\small\begin{align}
 \left[
\begin{array}{c}
 \xi_1 \\
 \eta_1\\
 \xi_2\\
 \eta_2
\end{array}
\right]\drm{out}
&=
\left[
\begin{array}{cccc}
 1 & & & \\
 & 1 & & g\dd{l} \\
 -g\dd{l} & & 1 & \\
 & & & 1
\end{array}
\right]
\left(
 \left[
\begin{array}{c}
 \xi_1   \\
 \eta_1  \\
 \xi_2   \\
 \eta_2
\end{array}
\right]\drm{in}
+
 \left[
\begin{array}{c}
 \sqrt{2}\rea(d) \\
 \sqrt{2}\ima(d) \\
 0 \\
 0 
\end{array}
\right]
\right).
\end{align}\normalsize
This indicates that 
\en{ \phi_1 } is displaced by the first gate
and then 
interacts with \en{ \phi_2 } at the second (QND) gate, 
as expected.

\newpage

\subsection{D-feedforward}
\label{sec:ff}

\begin{wrapfigure}[0]{r}[53mm]{49mm} 
\centering
\vspace{-0mm}
\includegraphics[keepaspectratio,width=45mm]{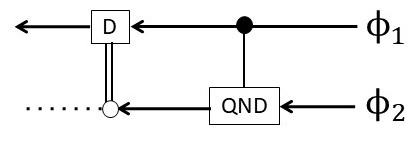}
\caption{
\small
D-feedforward.
\normalsize
}
\label{fig-feedforward}
\end{wrapfigure}

Now let us consider d-feedforward.
Assume that the displacement gate is coupled to
the output of the QND gate so that 
the displacement parameter \en{ d } is written as
\small\begin{align}
\hspace{10mm}
 d
&=
 \frac{k}{2}(\phi\drm{2,out}\dgg + \phi\drm{2,out})
=
 \frac{k}{\sqrt{2}} \xi\drm{2,out}.
\hspace{5mm}
 (k\in\mathbb{R})
\label{ffp}
\end{align}\normalsize

Our purpose 
is to derive the input-output relation of this process,
The first thing we need is to find an interaction Lagrangian.
Substituting (\ref{ffp}) into (\ref{lag-qndd2}),
we get 
\small\begin{align}
 \lag\urm{FF}
=
 -\im  
 \AH{\mm{\Phi}_1\dgg}{\mm{\Phi}_2\dgg}
 \Sigma\dd{z}
\left\{
 \frac{g_0}{4}
 \A{}{Q\dd{-}}{-Q\dd{+}}{}
 \AV{\mm{\Phi}_1}{\mm{\Phi}_2}
+
 \frac{k}{4}
 \A{0}{2 Q\dd{+}}
   {0}{g_0 Q\dd{+}}
 \AV{\Phi_1}{\Phi_2}\drm{out}
 \right\}.
\label{fflag}
\end{align}\normalsize
The last term breaks the input-output symmetry of the Lagrangian,
which reflects the effect of d-feedforward.
The Euler-Lagrange equation is written as
\small\begin{align}
 (\Phi\drm{in}-\Phi\drm{out})
+
 (\lie\drm{in} \Phi\drm{in}+ \lie\drm{out} \Phi\drm{out})
&=0,
\end{align}\normalsize
where
\begin{subequations}
\label{boundff}
\small\begin{align}
% \Phi\drm{in}
%&=
% \AV{\Phi_1}{\Phi_2}\drm{in},
 \lie\drm{in}
&=
 \frac{g_0}{4}
 \A{}{ Q\dd{-}}{- Q\dd{+}}{},
\\
 \lie\drm{out}
&=
 \frac{g_0}{4}
 \A{}{ Q\dd{-}}{- Q\dd{+}}{}
+
 \frac{k}{4}
 \A{0}{2 Q\dd{+}}
   {0}{g_0 Q\dd{+}}.
\end{align}\normalsize
\end{subequations}
In the quadrature basis,
this is rewritten as
\begin{subequations}
\small\begin{align}
 \A{1}{-k}
   {\ffrac{g_0}{2}}{1-\ffrac{g_0k}{2}}
 \AV{\xi_1}{\xi_2}\drm{out}
&=
 \A{1}{0}
   {-\ffrac{g_0}{2}}{1}
 \AV{\xi_1}{\xi_2}\drm{in},
\label{ffquinout1}
%%%%%%%%%%%%%%%%%%%%%%%%%%%%%%
\\
 \A{1}{-\ffrac{g_0}{2}}
   {0}{1}
 \AV{\eta_1}{\eta_2}\drm{out}
&=
 \A{1}{\ffrac{g_0}{2}}
   {0}{1}
 \AV{\eta_1}{\eta_2}\drm{in}.
\end{align}\normalsize 
\end{subequations}
It follows from (\ref{ffquinout1}) that 
\small\begin{align}
 \xi\drm{2,out}
=
 \xi\drm{2,in}-g_0\xi\drm{1,in}.
\end{align}\normalsize
Then the output of the displacement gate 
(the upper line in Figure \ref{fig-feedforward}) is written as
\begin{subequations}
\label{ffio}
\small\begin{align}
 \xi\drm{1,out} &= \xi\drm{1,in} + k \xi\drm{2,out},
\label{ffio1}
\\
 \eta\drm{1,out}&= \eta\drm{1,in} +g_0 \eta\drm{2,in}.
\label{ffio2}
\end{align}\normalsize
\end{subequations}
The \en{ x }-quadrature (real part) \en{ \xi_1 } is displaced 
in proportion to the output \en{ \xi\drm{2,out} }.
%whereas the $y$-quadrature $\eta_1$ is not influenced 
%by the displacement gate.
This is what we expect from the feedforward configuration
in Figure \ref{fig-feedforward}.
On the other hand,
the \en{ y }-quadrature (imaginary part) \en{ \eta\drm{1,out} } 
is not influenced by d-feedforward.
This is also expected 
because the displacement parameter \en{ d } is 
defined to be real in (\ref{ffp}).

\newpage

\subsection{D-feedback}
\label{sec:fbclassical}

\begin{wrapfigure}[0]{r}[53mm]{49mm} 
\centering
\vspace{-0mm}
\includegraphics[keepaspectratio,width=45mm]{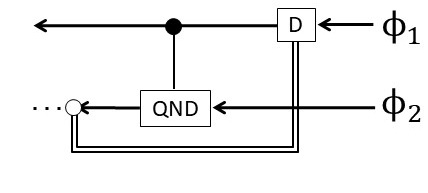}
\caption{
\small
D-feedback.
\normalsize
}
\label{fig-feedback}
\end{wrapfigure}

D-feedback is defined by 
the same displacement parameter as d-feedforward:
\small\begin{align}
\hspace{10mm}
 d
&=
 \frac{k}{2}(\phi\drm{2,out}\dgg + \phi\drm{2,out})
=
 \frac{k}{\sqrt{2}} \xi\drm{2,out}.
\hspace{5mm}
 (k\in\mathbb{R})
\label{ffp2}
\end{align}\normalsize
A difference is that the displacement gate operates first,
as in Figure \ref{fig-feedback}.
This configuration can be regarded as feedback.

We follow the same procedure as d-feedforward
to calculate the input-output relation of this circuit.
The interaction Lagrangian is obtained by
substituting (\ref{ffp2}) into (\ref{lagsu2sum}):
\small\begin{align}
 \lag\urm{FB}
=
 -\im  
 \AH{\mm{\Phi}_1\dgg}{\mm{\Phi}_2\dgg}
 \Sigma\dd{z}
\left\{
 \frac{g\dd{l}}{4}
 \A{}{Q\dd{-}}{-Q\dd{+}}{}
 \AV{\mm{\Phi}_1}{\mm{\Phi}_2}
+
 \frac{k}{4}
 \A{0}{2 Q\dd{+}}
   {0}{-g\dd{l} Q\dd{+}}
 \AV{\Phi_1}{\Phi_2}\drm{out}
 \right\},
\label{fblag}
\end{align}\normalsize
for which
the Euler-Lagrange equation
is written as
\small\begin{align}
 (\Phi\drm{in}-\Phi\drm{out})
+
 (\lie\drm{in} \Phi\drm{in}+ \lie\drm{out} \Phi\drm{out})
&=0,
\end{align}\normalsize
where
\begin{subequations}
\small\begin{align}
 \lie\drm{in}
&=
 \frac{g\dd{l}}{4}
 \A{}{ Q\dd{-}}{- Q\dd{+}}{},
\\ 
 \lie\drm{out}
&=
 \frac{g\dd{l}}{4}
 \A{}{ Q\dd{-}}{- Q\dd{+}}{}
+
 \frac{k}{4}
 \A{0}{ 2 Q\dd{+}}
   {0}{- g\dd{l} Q\dd{+}}.
\label{dfbgint}
\end{align}\normalsize
\end{subequations}
Again,
a difference from the d-feedforward (\ref{boundff})
is only the sign in the (2,2)-element of (\ref{dfbgint}).
This minor difference is critical.
In fact, using an identity
\small\begin{align}
(1-\lie\drm{out})\inv 
&=
 \frac{1}{1+g\dd{l} k}
 \A{1+\ffrac{g\dd{l} k}{2} Q\dd{-}
  + \ffrac{g\dd{l} k}{4}Q\dd{+}}{\ffrac{g\dd{l}}{4}(1+g\dd{l} k)Q\dd{-} 
  + \ffrac{k}{2} Q\dd{+}} 
   {-\ffrac{g\dd{l}}{4}Q\dd{+}}{1+\ffrac{g\dd{l} k}{2} Q\dd{-}},
\end{align}\normalsize
we can rewrite the input-output relation as
\begin{subequations}
\label{fbtfc}
\small\begin{align}
 \xi\drm{1,out}
&=
 \xi\drm{1,in} + \frac{k}{1+g\dd{l} k}(\xi\drm{2,in} - g\dd{l} \xi\drm{1,in}),
\\
 \eta\drm{1,out}
&=
 \eta\drm{1,in} + g\dd{l} \eta\drm{2,in}.
\end{align}\normalsize
\end{subequations}
This is basically the same form as the feedforward (\ref{ffio}).
The difference is the fractional coefficient
that represents
the self-consistent structure of the feedback.
In other words,
the signal travels in the feedback loop infinite times,
which results in
\small\begin{align}
 1+(-g\dd{l}k)+(-g\dd{l}k)^2+\cdots
=
 \frac{1}{1+g\dd{l}k}.
\end{align}\normalsize

This circuit is revisited in Section \ref{sec:circuitss}
where the input-output relations are derived from \textit{S}-matrices.
It will be shown that 
d-feedforward is obtained from
a finite series in the expansion of an \textit{S}-matrix,
whereas d-feedback is expressed by an infinite series
with self-energy \en{ \isel\equiv -g\dd{l}k }.

\newpage

\section{Cross-controlled QND gate}
\label{chap:flipflop}

\begin{wrapfigure}[1]{r}[53mm]{49mm} 
\vspace{-10mm}
\centering
\includegraphics[keepaspectratio,width=50mm]{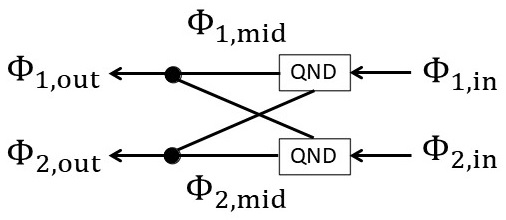}
\caption{Cross-controlled QND gate}
\label{fig-flipflop}
\end{wrapfigure}

Here we introduce another example of feedback.
It is a flip-flop type of circuit
in which two QND gates are connected to each other
as in 
Figure \ref{fig-flipflop}.\index{cross-controlled QND gate}
This is considered as feedback in the sense that 
the states of the QND gates depend on their outputs.

To describe the behavior of this circuit,
all we need is to calculate a reactance matrix \en{\lie}.
The input-output relation is given by
\begin{subequations}
\small\begin{align}
 \AV{\Phi\drm{1,out}}{\Phi\drm{2,mid}}
&=
 \A{1}{\ffrac{g_2}{2}Q\dd{-}}
   {-\ffrac{g_2}{2}Q\dd{+}}{1}
 \AV{\Phi\drm{1,mid}}{\Phi\drm{2,in}},
\\
 \AV{\Phi\drm{2,out}}{\Phi\drm{1,mid}}
&=
 \A{1}{\ffrac{g_1}{2}Q\dd{-}}
   {-\ffrac{g_1}{2}Q\dd{+}}{1}
 \AV{\Phi\drm{2,mid}}{\Phi\drm{1,in}},
\end{align}\normalsize
\end{subequations}
where \en{ Q\dd{\pm} \equiv I\pm \sigma\dd{x} }.
Eliminating \en{ \Phi\drm{1,mid}} and \en{\Phi\drm{2,mid} } yields
\small\begin{align}
 (\Phi\drm{in}-\Phi\drm{out})
+
 (\lie\drm{in} \Phi\drm{in} + \lie\drm{out} \Phi\drm{out})
&=0,
\label{crossbc}
\end{align}\normalsize
where
\begin{subequations}
\small\begin{align}
 \lie\drm{in}
&=
 \A{}{\ffrac{g_2}{2}Q\dd{-}}
   {\ffrac{g_1}{2}Q\dd{-}}{},
\\
 \lie\drm{out}
&=
 \A{}{-\ffrac{g_1}{2}Q\dd{+}}
   {-\ffrac{g_2}{2}Q\dd{+}}{}.
\end{align}\normalsize
\end{subequations}
These satisfy
\small\begin{align}
 \lie\drm{in} \lie\drm{out}=0.
\end{align}\normalsize
For simplicity, 
assume that the circuit is symmetric, i.e., \en{ g_1=g_2=g }.
The reactance matrix is given by (\ref{pcompare1}):
\small\begin{align}
 G
&=
 \frac{g}{4-g^2}
 \A{g\sigma\dd{x}}{-2\sigma\dd{x}}
   {-2\sigma\dd{x}}{g\sigma\dd{x}}.
\end{align}\normalsize
This also has a fractional coefficient 
due to the feedback configuration.

\chapter{Spin 1/2  field and spin gate}
\label{chap:spin}
\thispagestyle{fancy}

Spin \en{ 1/2 }\index{spin $1/2$} 
is introduced as a nonrelativistic field
at a fixed point in space
\en{ \psi(t)=\psi(t,\bm{x}=0) }.
We first consider a relationship 
between a spin Hamiltonian and a gauge transformation.
Then quantum gates are introduced to the spin field.
We also introduce some unconventional applications 
of gauge transformations
to describe noise on qubits such as bit and phase flips.

\section{Spin 1/2  field}

Spin \en{ 1/2 } particles have SU(2) symmetry.
As is well known,
the spin up and down states are invariant 
under rotation around the spin axis.
It is important to note that
the spin rotation symmetry is parameterized 
by only a single variable, i.e.,
the rotation angle around the axis.
From a gauge theoretical perspective,
this is regarded as U(1) gauge symmetry,
as explained in Remark \ref{rem:u2u1}.
In this section,
we consider a spin Hamiltonian and spin rotation gate 
from this point of view.

\subsection{Spin rotation gate}
\label{sec:spingate}

We first note that 
spin energy levels are degenerate in free space.
If we set this level to zero, 
the spin field is described as
\small\begin{align}
 \im \partial\dd{t} \psi 
= 
 0, 
\end{align}\normalsize
where \en{ \psi } is a two-component vector.
The corresponding Lagrangian is written as 
\small\begin{align}
 \lag(\psi,\partial\dd{t}\psi)
=
 \psi\dgg (\im \partial\dd{t}) \psi.
\label{nislag}
\end{align}\normalsize

The spin rotation around its axis is expressed as
\small\begin{align}
 \psi
\to 
 \psi'
=
 \uni(t)\psi,
\label{spingt}
\end{align}\normalsize
where
\small\begin{align}
 \uni(t)
=
 \exp[\win(t) \im g T].
\end{align}\normalsize
\en{ T\in \textrm{su(2)} } is a generator of the rotation.
A gauge field is introduced as
\small\begin{align}
 \ga = \gel \im g T,
\end{align}\normalsize
and the corresponding gauge covariant derivative is given as
\small\begin{align}
 D\dd{t}
=
 \partial\dd{t} - \gel \im g T.
\end{align}\normalsize
In response to the spin rotation (\ref{spingt}),
the gauge field transforms as 
\small\begin{align}
 \gel
\to
 \gel\uu{\prime}
&=
 \gel
 +
 \partial\dd{t}\win,
\end{align}\normalsize
which indicates that the gauge field is 
implemented by the electromagnetic interaction.
A symmetric spin Lagrangian is then given as
\small\begin{align}
 \lag(\psi,D\dd{t}\psi)
=
 \psi\dgg \left(\im\partial\dd{t} + \gel g T \right)\psi.
\end{align}\normalsize

Canonical momentum\index{canonical momentum (spin)} is defined as
\small\begin{align}
 \pi
&\equiv
 \frac{\partial \lag}{\partial(\partial\dd{t}\psi)}
=
 \im\psi\dgg.
\end{align}\normalsize
The Hamiltonian (density) of the spin field is then written as
\small\begin{align}
 \ham
&=
 \pi(\partial\dd{t} \psi) - \lag
=
 -\psi\dgg \left( \gel g T \right) \psi.
\end{align}\normalsize
If we define a \en{ 2\times 2 } Hermitian matrix \en{ H\equiv -\gel g T }, 
the Lagrangian is expressed as
\small\begin{align}
 \lag(\psi,D\dd{t}\psi)
=
 \psi\dgg \left( \im\partial\dd{t} -  H \right)\psi.
\label{rotationgate}
\end{align}\normalsize
This is simply the Schr\"{o}dinger equation.

\subsection{Spin transfer function}

Let us briefly review basic properties of (\ref{rotationgate})
from a field theoretical point of view.
Canonical quantization\index{canonical quantization (spin)} 
is written as
\small\begin{align}
 \{ \psi\dgg_{\alpha}(t),\psi_{\beta}(t') \}
&=
 \delta_{\alpha\beta} \delta(t-t').
\label{spq}
\end{align}\normalsize
Assume that \en{ H } is a constant Hermitian matrix
and \en{ \chi\dd{s} \ (s=\pm 1) } are eigenvectors
\small\begin{align}
\hspace{10mm}
 H\chi\dd{s}&= E\dd{s}\chi\dd{s},
\qquad
 (\chi\dd{r}\dgg \chi\dd{s}=\delta\dd{rs}.)
\end{align}\normalsize
Let us express a spin field as
\small\begin{align}
 \psi
&= 
\sum\dd{s=\pm 1} c\dd{s} \chi\dd{s} \ex\uu{-\im E\dd{s} t}.
\label{spinfield}
\end{align}\normalsize
Using the orthonormality of \en{ \chi\dd{s} },
the coefficients \en{ c\dd{s} } is given by
\small\begin{align}
 c\dd{s}
&= 
 \ex\uu{\im E\dd{s} t} \chi\dd{s}\dgg \psi.
\label{cs}
\end{align}\normalsize
Then the quantization (\ref{spq}) is expressed as
\small\begin{align}
 \{c\dd{r}\dgg, c\dd{s} \}&= \delta\dd{rs}.
\label{c}
\end{align}\normalsize

\newpage

\begin{example}
\label{ex:spin}
\rm
Let us consider a case where \en{ H } is given as
\small\begin{align}
 H&=\A{E\dd{+}}{}{}{E\dd{-}},
\label{exsham}
\end{align}\normalsize
The corresponding eigenstates represent spin up and down modes.
A vacuum state is defined as a state satisfying
the following relations for both modes:
\begin{subequations}
\small\begin{align}
 c\dd{+}\ket{0}&=0,
\\
 c\dd{-}\ket{0}&=0.
\end{align}\normalsize
\end{subequations}
For each mode,
a number operator\index{number operator (spin)} is defined as
\begin{subequations}
\small\begin{align}
 N\dd{+} &\equiv c\dd{+}\dgg c\dd{+},
\\
 N\dd{-} &\equiv c\dd{-}\dgg c\dd{-}.
\end{align}\normalsize
\end{subequations}
Then we have
\begin{subequations}
\small\begin{align}
 N\dd{+} c\dd{+}\dgg\ket{0} &= c\dd{+} \dgg\ket{0},
\\
 N\dd{-} c\dd{-}\dgg\ket{0} &= c\dd{-} \dgg\ket{0}.
\end{align}\normalsize
\end{subequations}
This means that \en{ c\dd{+}\dgg \ (c\dd{-}\dgg) } is
the creation operator of the spin up (down) state:
\begin{subequations}
\small\begin{align}
 c\dd{+}\dgg\ket{0}
&=
 \ket{+},
\\
 c\dd{-}\dgg\ket{0}
&=
 \ket{-}.
\end{align}\normalsize
\end{subequations}
For example, using (\ref{cs}), 
the spin up state is written as
\small\begin{align}
 \ex\uu{\im E\dd{+} t} \ket{+}
=
 \psi\dgg\ket{0} \, \chi\dd{+}.
\end{align}\normalsize

Given a matrix \en{ X }, the corresponding operator is defined by
\small\begin{align}
 \psi\dgg X\psi.
\label{mattoop}
\end{align}\normalsize
For instance, the Hamiltonian is of this form:
\small\begin{align}
 \ham
&\equiv
 \pi(\partial\dd{t}\psi) + (\partial\dd{t}\psi\dgg)\pi\dgg - \lag
=
 \psi\dgg H \psi.
\end{align}\normalsize
It is not difficult to see that \en{ \ket{+} } is actually an eigenstate of 
the Hamiltonian:
\begin{subequations}
\small\begin{align}
 \psi\dgg H\psi\ket{+}
&=
 \sum\dd{s} \chi\dd{s} \dgg H \chi\dd{s} N\dd{s} \ket{+}
\\ &=
 \chi\dd{+}\dgg H \chi\dd{+}
 \ket{+}
\\ &=
 E\dd{+} \ket{+}.
\end{align}\normalsize
\end{subequations}
Likewise, 
the lowering operator is defined as 
\en{ \psi\dgg \sigma\dd{-} \psi },
where \en{ \sigma\dd{-} } is 
the lowering matrix defined in Remark \ref{rem:pauli}.
This operator converts \en{ \ket{+} } to
\begin{subequations}
\small\begin{align}
 \bigl(\psi\dgg \sigma\dd{-} \psi \bigr)
 \ex\uu{\im E\dd{+} t} \ket{+} 
&=
 \sum\dd{s,r} 
 \ex\uu{\im E\dd{+} t}
 \ex\uu{-\im (E\dd{r}-E\dd{s})t} \
 \chi\dd{s}\dgg \sigma\dd{-} \chi\dd{r}
\
 c\dd{s} \dgg c\dd{r} \ket{+}
\\ &=
 \sum\dd{s} \ex\uu{ \im E\dd{s} t} \
 \chi\dd{s} \dgg \sigma\dd{-} \chi\dd{+}
\
 c\dd{s} \dgg \ket{0}
\\ &=
 \ex\uu{ \im E\dd{-} t} \ket{-},
\end{align}\normalsize
\end{subequations}
which is the spin down state, 
as expected.
\qed
\end{example}

\newpage

\begin{definition}
\label{def:spintf}
A spin transfer function\index{transfer function (spin)} \en{ Y } is defined as
\small\begin{align}
 \ipg\dd{\psi|\psi}(t)
&\equiv
 \wick{\psi(t)}{\psi\dgg(0)}.
\end{align}\normalsize
\end{definition}

\begin{theorem}
\label{thm:spint}
The spin transfer function is expressed in the frequency domain as
\small\begin{align}
 \ipg\dd{\psi|\psi}(\omega) 
= 
 \wick{\psi}{\psi\dgg}(\omega)
&=
 \lim\dd{\epsilon\to 0\dd{+}}\frac{\im}{\omega-H+\im\epsilon}
\label{spintf} 
\end{align}\normalsize
\end{theorem}
\begin{proof}
It follows from the expression (\ref{spinfield}) that
\small\begin{align}
 \wick{\psi(t)}{\psi\dgg(0)}
&=
 \sum\dd{s} \step(t) \ex\uu{-\im E\dd{s} t }\chi\dd{s} \chi\dd{s}\dgg.
\label{stf1}
\end{align}\normalsize
This is a causal function.
The Laplace transform yields
\small\begin{align}
 \ipg\dd{\psi|\psi}(s)
&=
 \frac{1}{s+\im H},
\end{align}\normalsize
where the region of convergence is \en{ \rea(s)>0 }.
As we did in Section \ref{sec:p=0},
we can rewrite it as a stable function 
\small\begin{align}
\hspace{15mm}
 \ipg\dd{\psi|\psi}(s)
&=
 \frac{1}{s+\im H + \epsilon}.
\qquad
 (\epsilon >0)
\end{align}\normalsize
By setting \en{ s=-\im\omega },
we get (\ref{spintf}).
\end{proof}

%\newpage

\subsection{Controlled unitary gate}

\begin{wrapfigure}[0]{r}[53mm]{49mm} 
\vspace{-0mm}
\centering
\includegraphics[keepaspectratio,width=37mm]{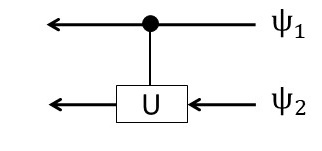}
\caption{
\small
Controlled unitary gate.
\normalsize
}
\label{fig-cU}
\end{wrapfigure}

Let us consider a controlled unitary gate\index{controlled unitary gate}
in Figure \ref{fig-cU}.
A unitary operator \en{ \uni } acts on the lower line \en{ \psi_2 } (a target)
only when the upper line \en{ \psi_1 } (a control) is in \en{ \chi\dd{+} }.
The input-output relation is written as
\small\begin{align}
 \psi\drm{out}
&=
\Bigl(
 \chi\dd{+}\chi\dd{+}\dgg \otimes \uni
\ + \ 
 \chi\dd{-}\chi\dd{-}\dgg \otimes I 
\Bigr) \
 \psi\drm{in},
\qquad
\left(
 \psi \equiv \psi_1 \otimes \psi_2.
\right)
\label{cuinout}
\end{align}\normalsize

We need to find a reactance matrix \en{ \lie }
to calculate the interaction Lagrangian of this gate.
Let us introduce \en{ \lie_T } 
through the Cayley transform\index{Cayley transform}
as
\small\begin{align}
 \uni =  \frac{1+\lie_T}{1-\lie_T}.
\end{align}\normalsize
Then (\ref{cuinout}) is rewritten as
\small\begin{align}
 \left(
  \psi\drm{in} - \psi\drm{out}
 \right)
+
 \left(
  \lie\drm{in} \psi\drm{in} + \lie\drm{out} \psi\drm{out}
 \right)
&=
 0,
\quad
\kakkon{ \hspace{1mm} \lie\drm{in}= (\chi\dd{+}\chi\dd{+}\dgg 
  - \chi\dd{-}\chi\dd{-}\dgg)\otimes \lie_T,}
{\lie\drm{out}= I\otimes \lie_T.}
\end{align}\normalsize
Note that this is the same form as (\ref{tfXY}).
Using (\ref{pcompare1}),
we have
\small\begin{align}
 \lie
=
 \chi\dd{+}\chi\dd{+}\dgg \otimes \lie_T.
\end{align}\normalsize
Substituting \en{ \lie } into (\ref{unilag}),
we get
\small\begin{align}
 L\urm{cU}
&=
 - 2
 \mm{\psi}\dgg  (\im\partial\dd{t} \win) 
 \Bigl(\chi\dd{+}\chi\dd{+}\dgg \otimes \lie_T \Bigr)
 \mm{\psi}.
\end{align}\normalsize

\newpage

\section{Stochastic gauge transformations}

In this section, 
we introduce a different interpretation to the gauge field
and develop unconventional applications of the gauge theory.

Suppose that spin \en{ 1/2 } particles \en{ \psi } are placed 
in a noisy environment.
In quantum theory,
the evolution of \en{ \psi } is described by 
completely positive maps.\index{completely positive map}
There are two representations for it.
One is unitary evolution in an extended Hilbert space
\en{ \hil\dd{\gel} \otimes \hil\dd{\psi} }.
The other is the Kraus operators in \en{ \hil\dd{\psi} } alone.

The two representations are equivalent.
To see how they are related to each other,
consider a unitary operator \en{ \uni } 
on a density matrix \en{\ket{\zeta}\bra{\zeta}\otimes \rho} :
\small\begin{align}
 \ket{\zeta}\bra{\zeta}\otimes \rho
\to
 \uni\left(\ket{\zeta}\bra{\zeta}\otimes \rho\right)\uni\dgg.
\label{extunitary}
\end{align}\normalsize
Suppose that 
\en{\{ \, \ket{n} \mid  n=1,\cdots,N \}} 
is a basis of \en{ \hil\dd{\gel}  }. 
Taking a partial trace over \en{ \hil\dd{\gel} },
we get
\begin{subequations}
\label{kraus}
\small\begin{align}
 \rho
\to
 \rho'
&\equiv
 \pTr{\gel} \uni
 \left(\ket{\zeta}\bra{\zeta}\otimes \rho\right)
 \uni\dgg
\\ &
=
 \sum_{n=1}^N V\dd{n} \, \rho \, V\dd{n}\dgg,
\end{align}\normalsize
\end{subequations}
where \en{ V\dd{n} } are given by
\small\begin{align}
  V\dd{n} \equiv \bra{n}\uni\ket{\zeta}.
\end{align}\normalsize
These are operators on \en{ \hil\dd{\psi} }
and called \textit{Kraus operators}.\index{Kraus operator}
In the Heisenberg picture,
the map (\ref{kraus}) is expressed as
\small\begin{align}
X \to X' =
 \sum_{n=1}^{N} V\dd{n}\dgg X V\dd{n}
\equiv
 \maths{K}(X).
\label{krausum}
\end{align}\normalsize
This satisfies a \textit{unitary condition}\index{unitary condition}
\small\begin{align}
 \maths{K}(I) =  I. 
\label{kunicon}
\end{align}\normalsize

Conversely,
given Kraus operators \en{ V\dd{n} },
there exists a pair \en{ \{\uni,\ket{\zeta}\} } 
satisfying (\ref{kraus}).
For example, 
given \en{ \{V_1,V_2\} },
one can formally construct \en{ \{\uni,\ket{\zeta}\} } as
\small\begin{align}
 \pTr{\gel}
 \A{V_1}{\star}{V_2}{\star}
\left(
 \A{1}{0}{0}{0}\otimes \rho
\right)
 \A{V_1\dgg}{V_2}{\star}{\star}
=
 \sum_{n=1}^2 V\dd{n} \, \rho \, V\dd{n}\dgg,
\end{align}\normalsize
where the matrix elements \en{ \star } are 
determined to satisfy the unitary condition (\ref{kunicon}).
Obviously, 
there are many choices of \en{ \{\uni,\ket{\zeta}\} }.
We need to analyze actual physical systems 
to find a feasible realization.

In this section,
we regard a gauge field as a source of noise
and 
develop a gauge theoretical method to find \en{ \{\uni,\ket{\zeta}\} }
for a specific type of noise channels.
As a first step,
we reconsider the case of \en{ N=1 } in (\ref{krausum}) 
from a gauge theoretical perspective.
However, 
there is a problem.
In general, 
the gauge field is not quantized in a standard way.

\subsection{Gauge symmetry revisited}
\label{sec:gsrevisit}

Suppose that a spin field \en{ \psi } goes through a gate 
represented by a gauge field \en{ \gel }.
In general, 
the quantization of a gauge field is not straightforward
because 
Yang-Mills theory is a system with primary and secondary constraints.
In the case of U(1),
for example,
a gauge field has four coordinates
described by Maxwell's four equations.
However, 
these equations are not independent due to 
the gauge invariance
and 
we can define only three momenta for the four coordinates.
Then canonical quantization is not well defined.

Here we do not quantize \en{ \gel } explicitly.
Instead,
we assume eigenvectors of the `operator' \en{ \gel }
in a Hilbert space \en{ \hil\dd{\gel} }.
Recall that 
the weight function \en{ \win } serves as a switch
to turn the gate on and off.
\en{ \partial\dd{t}\win } is identically zero
if the gate is off,
Denoted by \en{ \ket{0} }
is an eigenvector corresponding to this situation:
\en{ \gel dt\ket{0}=0 }.
Likewise, 
there is a state \en{ \ket{\win} }
such that 
\en{ \gel dt \ket{\win}=-d\win\ket{\win} }.

Let us review the gauge theory in a simple setting
where both spin and gauge fields are static
without interactions.
The spin Lagrangian is written as
\small\begin{align}
 \lag(\psi,\partial\dd{t}\psi)
=
 \psi\dgg (\im \partial\dd{t}) \psi.
\label{nislag-1}
\end{align}\normalsize
Suppose a (unitary) gauge transformation of the form
\small\begin{align}
\hspace{25mm}
 \psi
\to
 \psi'
=
 V \psi,
\qquad
\left( V \equiv \exp\left[\win \lie \right],
\quad
 \lie\dgg = -\lie.\right)
\label{spgtrans}
\end{align}\normalsize
Then the spin Lagrangian transforms as
\small\begin{align}
 \psi\dgg (\im \partial\dd{t}) \psi
\to
 \psi\dgg \maths{K} (\im \partial\dd{t}) \psi
=
 \psi\dgg \im \left[\partial\dd{t} 
+ 
 (\partial\dd{t} \win) \lie \right] \psi,
\label{spintrans-1}
\end{align}\normalsize
where 
\en{ \maths{K}(\bullet) \equiv V\dgg \bullet V }.
As in Section \ref{sec:gsym},
a gauge field is introduced to cancel the second term 
of (\ref{spintrans-1}) out:
\small\begin{align}
 D\dd{t}
\equiv
 \partial\dd{t} - \gel \lie.
\label{spingcd}
\end{align}\normalsize
The spin Lagrangian is rewritten as 
\small\begin{align}
 \psi\dgg (\im D\dd{t}) \psi
=
 \psi\dgg (\im \partial\dd{t}) \psi
+
 \lag\urm{int},
\end{align}\normalsize
where \en{ \lag\urm{int} } is an interaction Lagrangian
given as
\small\begin{align}
 \lag\urm{int}
=
 -\im \gel \otimes \psi\dgg \lie \psi.
\end{align}\normalsize

Let us consider the time evolution of \en{ \psi } 
under \en{ \lag\urm{int} }.
This is usually done in the interaction picture,
but it is not necessary here
because \en{ \psi } is static in free space as in (\ref{nislag-1}).
An infinitesimal time evolution operator is written as
\small\begin{align}
 \uni
=
 \exp\left[- \im  \lag\urm{int} dt \right]
\sim
 I 
- \gel dt \otimes \psi\dgg \lie \psi.
\end{align}\normalsize
Assume that a density matrix is initially 
prepared in \en{\ket{\win}\bra{\win}\otimes \rho}.
Using the anticommutation relation (\ref{spq}),
we have
\begin{subequations}
\small\begin{align}
 \Tr \uni \Bigl( \ket{\win}\bra{\win}\otimes \rho \Bigr) \uni\dgg
     \Bigl(I \otimes \psi \Bigr)
&= 
 \Tr \rho \Bigl(\psi - \bra{\win}\gel dt\ket{\win} \, \lie\psi \Bigr) 
\\ &\sim 
 \Tr \rho \, \exp\Bigl[(d\win) \, \lie\Bigr] \psi.
\label{spgtrans-2}
\end{align}\normalsize
\end{subequations}
This is the same form as the gauge transformation (\ref{spgtrans}).
It is also easy to see that 
\small\begin{align}
 \pTr{\gel} \uni \Bigl( \ket{\win}\bra{\win}\otimes \rho \Bigr) \uni\dgg
&\sim
 \exp\Bigl[(d\win) \, \psi\dgg \lie \psi\Bigr] \,
 \rho 
 \exp\Bigl[(d\win) \, \psi\dgg \lie\dgg \psi\Bigr],
\end{align}\normalsize
which corresponds to \en{ N=1 } in (\ref{kraus}).

Unlike \en{ \psi },
the gauge field \en{ A } is static under the interaction Lagrangian.
Its transformation follows from the gauge principle.
In our notation, 
it is written as
\small\begin{align}
 \psi\dgg(\im D\dd{t})\psi
=
 \psi\dgg\maths{K}(\im D\dd{t}')\psi.
\label{gpsto-0}
\end{align}\normalsize
This is satisfied if 
\en{ D\dd{t} = \maths{K}(D\dd{t}') },
which can be rewritten as
\small\begin{align}
 \gel dt
\to
 \gel' dt
&=
 \gel dt + d\win.
\label{gpsto-1}
\end{align}\normalsize
It is easy to see that 
\small\begin{align}
 \gel' \ket{\win}=0.
\label{reset}
\end{align}\normalsize
This means that 
gauge particles (or the quantum gate) 
initially prepared in \en{ \ket{\win} } 
to generate the interaction
are reset (absorbed) after the interaction.

\subsection{Stochastic gauge transformation}
\label{sec:stogtrans}

Let us consider a transformation of the form
\small\begin{align}
 \lag(\psi, \partial\dd{t} \psi)
\to
 \sum\dd{n} p\dd{n} \lag(\uni\dd{n}\psi,\partial\dd{t} \uni\dd{n}\psi),
\label{probgt}
\end{align}\normalsize
where \en{ \sum\dd{n} p\dd{n}=1 } 
and \en{ \uni\dd{n} } are unitary operators given as
\small\begin{align}
\hspace{27mm}
 \uni\dd{n}(t)
=
 \exp[\win\dd{n}(t) \lie\dd{n}],
\qquad
 \left(\lie\dd{n}\dgg = -\lie\dd{n}.\right)
\label{stoun}
\end{align}\normalsize
This is a process in which 
\en{ \uni\dd{n} } acts on \en{ \psi } 
with a probability of \en{ p\dd{n} }.
Let us define 
\small\begin{align}
 V\dd{n} \equiv \sqrt{p\dd{n}}\uni\dd{n}.
\end{align}\normalsize
These are regarded as Kraus operators on \en{ \hil\dd{\psi} }.
Note that the Lagrangian is globally symmetric
if all \en{ \win\dd{n} } are constant.
Let us introduce 
\small\begin{align}
 \maths{K}(\bullet) \equiv \sum\dd{n} V\dd{n}\dgg \bullet V\dd{n}.
\end{align}\normalsize
Then (\ref{probgt}) is written as
\small\begin{align}
 \psi\dgg (\im \partial\dd{t}) \psi
\to 
 \psi\dgg \maths{K}(\im \partial\dd{t}) \psi.
\label{probgt2}
\end{align}\normalsize

Our purpose is to find \en{ \{\uni, \ \ket{\zeta}\} } 
satisfying (\ref{kraus}) for the Kraus operators \en{ V\dd{n} }.
To find \en{ \uni }, 
we rewrite (\ref{probgt2}) as
\small\begin{align}
 \psi\dgg \maths{K}(\im \partial\dd{t}) \psi.
&=
 \psi\dgg \im \Bigl[\partial\dd{t} 
         + \sum\dd{n} p\dd{n}(\partial\dd{t} \win\dd{n}) \lie\dd{n} \Bigr] \psi.
\end{align}\normalsize
As in (\ref{spingcd}), 
we introduce gauge fields \en{ \gel\dd{n} } 
and a covariant derivative as
\small\begin{align}
 D\dd{t}
\equiv
 \partial\dd{t} - \sum\dd{n} p\dd{n} \gel\dd{n} \lie\dd{n}.
\label{spingcd-1}
\end{align}\normalsize
Then the Lagrangian is written as
\small\begin{align}
 \lag(\psi, D\dd{t} \psi)
=
 \lag(\psi, \partial\dd{t} \psi)
+
  \lag\urm{int},
\end{align}\normalsize
where 
\small\begin{align}
 \lag\urm{int}
=
 - \im \sum\dd{n} p\dd{n} \gel\dd{n} \otimes \psi\dgg \lie\dd{n} \psi.
\end{align}\normalsize
Note that 
the gauge fields \en{ \{\gel\dd{n} \mid n=1,2,\cdots\} } 
are understood as
\en{ \{\cdots I \otimes \gel\dd{n} \otimes I \cdots\} }.
The transformation of the gauge fields \en{ \gel\dd{n} \to \gel'\dd{n} }
follows from the gauge principle
\small\begin{align}
 \psi\dgg(\im D\dd{t})\psi
=
 \psi\dgg \maths{K}(\im D\dd{t}') \psi.
\label{gpsto}
\end{align}\normalsize
This is satisfied if
\small\begin{align}
 \gel\dd{n}' \maths{K}(\lie\dd{n}) dt
=
 \gel\dd{n} \lie\dd{n} dt
+
 \left(d \win\dd{n} \right) \lie\dd{n}. 
\label{probgft}
\end{align}\normalsize
This form is different from a standard expression
because \en{ \maths{K}\inv  } cannot be defined.
However, it can be simplified in a specific case.
Assume that 
nothing happens to the spin with a probability of \en{ p_1 },
and a transformation occurs with a probability of \en{ p_2 }.
In this case, \en{ \lie_1=0 } %and $\lie_2\not= 0$
and \en{ \maths{K}(\lie_2)= \lie_2 }.
Then (\ref{probgft}) is expressed as
\small\begin{align}
 \gel_2 dt
\to
 \gel_2' dt
=
 \gel_2 dt
+
 d \win_2.
\label{probgft2}
\end{align}\normalsize

For the other component of the pair, \en{ \ket{\zeta} },
we assume the following form:
\small\begin{align}
 \ket{\zeta}
=
 \sum\dd{n}\sqrt{p\dd{n}}\Bigl(
 \cdots \ket{0}\otimes \ket{\win\dd{n}} \otimes \ket{0} \cdots
 \Bigr),
\label{zeta}
\end{align}\normalsize
where \en{ \ket{0} } and \en{ \ket{\win\dd{n}} } are defined 
in the same way as Section \ref{sec:gsrevisit}.
Assume that 
\en{ \ket{0} } and \en{ \ket{\win\dd{n}} } are orthogonal.
In general, 
\en{ \ket{\zeta} } is an entangled state.
For an infinitesimal time evolution operator
\small\begin{align}
 \uni
=
 \exp\left( 
 - \sum\dd{n} p\dd{n} \gel\dd{n}dt \otimes \psi\dgg \lie\dd{n} \psi 
 \right),
\end{align}\normalsize
we have
\small\begin{align}
 \pTr{\gel}  \uni 
 \Bigl( \ket{\zeta}\bra{\zeta}\otimes \rho \Bigr) 
 \uni\dgg
&=
 \sum\dd{n} p\dd{n} \uni\dd{n} \, \rho \, \uni\dd{n}\dgg,
\label{stoU}
\end{align}\normalsize
where
\small\begin{align}
 \uni\dd{n} 
\equiv 
 \exp\left[  (d\win\dd{n}) \psi\dgg \lie\dd{n} \psi  \right],
\end{align}\normalsize
which corresponds to (\ref{stoun}), 
as expected.
It also follows from (\ref{probgft}) that 
\small\begin{align}
\hspace{20mm}
 \gel\dd{n}' \ket{\zeta} \ \maths{K}(\lie\dd{n}) 
&=
 0,
\qquad
 \forall n.
\label{gelzero}
\end{align}\normalsize
This corresponds to (\ref{reset}), 
i.e.,
the gauge field is reset after the interaction.

\begin{example}
\label{ex:stogauge}
 \rm
There are two well-known types 
of noise on qubits:\index{bit flip}\index{phase flip}
\begin{subequations}
\small\begin{align}
\mbox{\normalsize bit flips: \small}& \quad
 V_1 = \sqrt{p_1} \A{1}{}{}{1},
\quad
 V_2 = \sqrt{p_2} \A{}{1}{1}{},
\\
\mbox{\normalsize phase flips: \small}& \quad
 V_1 = \sqrt{p_1} \A{1}{}{}{1},
\quad
 V_2 = \sqrt{p_2} \A{1}{}{}{-1}.
\end{align}\normalsize
\end{subequations}
In both cases, Kraus operators are of the form
\small\begin{align}
 V_1 
=
 \sqrt{p_1} I,
\qquad
 V_2
=
 \sqrt{p_2}\exp\left[\win_2 \lie_2\right],
\end{align}\normalsize
where \en{ \win_2 } is a scalar 
and \en{ \lie_2\in \mathrm{u(2)} }.
One possible choice 
of the pair \en{ \{\lag\urm{int}, \ket{\zeta}\} }
is given as
\begin{subequations}
\small\begin{align}
 \lag\urm{int}
&=
 -\im I \otimes p_2 \gel_2 \otimes \psi\dgg \lie_2 \psi,
\\
 \ket{\zeta}
&=
 \sqrt{p_1} \, \ket{\win_1}\otimes \ket{0} 
 + \sqrt{p_2} \, \ket{0} \otimes \ket{\win_2}.
\end{align}\normalsize
\end{subequations}
It follows from (\ref{probgft2}) that
\small\begin{align}
 \gel_2'dt
=
 \gel_2 dt + d\win_2,
\end{align}\normalsize
which satisfies 
\en{ \gel_2'\ket{\zeta}=0 }
as expected from (\ref{gelzero}).
\qed
\end{example}

\chapter{Quantum systems via feedback}
\label{chap:feedback}
\thispagestyle{fancy}

In Chapter \ref{chap:1},
we have seen that 
locality is the key to distinguishing between systems and signals.
For example,
quantum gates are regarded as (static) systems
defined by local gauge transformations.
Another example is 
a closed-loop field \en{ \mas } (Section \ref{sec:ltf})
that is a field defined locally in a finite interval 
under periodic boundary conditions.
In this chapter, 
we combine these two to make dynamical systems.

\section{Unitary systems}
\label{sec:sysfor}

\begin{wrapfigure}[0]{r}[53mm]{49mm} 
\centering
\vspace{-7mm}
\includegraphics[keepaspectratio,width=38mm]{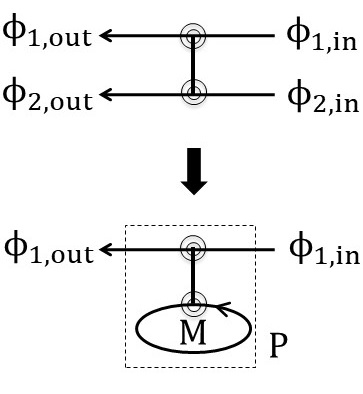}
\caption{
\small
SU(2) gate (upper)
and 
SU(2) system (lower).
\normalsize
}
\label{fig-ssu2-1}
\end{wrapfigure}

In this section, 
we illustrate a basic procedure 
to make a dynamical system from a quantum gate.
Let us consider a unitary gate as in Figure \ref{fig-ssu2-1}.
If one of the outputs, say \en{ \phi\drm{2,out} },
is coherently fed back to the corresponding input \en{ \phi\drm{2,in} },
a closed-loop field \en{ \mas } is created as
\small\begin{align}
\phi\drm{2,out} = \phi\drm{2,in} &\equiv \mas.
\label{feedback}
\end{align}\normalsize
The resulting system \en{ \tf } is called 
a unitary system\index{unitary system}.

Let us calculate the input-output relation of \en{ \tf }.
Our first step is to find a Lagrangian.
Assume that the unitary gate is placed at \en{ z=0 }.
Before closing the loop, the Lagrangian 
is given by (\ref{unilag}):
\begin{subequations}
\small\begin{align}
 \lag\urm{f}
&=
 \lag_1\urm{f} +  \lag_2\urm{f} 
\quad
\\
 \lag\urm{int}
&=
 - 
  2 \im\delta(z) \mm{\phi}\dgg \, \lie \, \mm{\phi},
\end{align}\normalsize
\end{subequations}
If we substitute (\ref{feedback}) into this Lagrangian,
it is modified as
\small\begin{align}
 \mm{\phi}
=
 \AV{\mm{\phi}_1}
    {\mm{\phi}_2}
&\Rightarrow
 \AV{\mm{\phi}_1}
    {\mas}
\equiv
 \phi\dd{\mas},
\label{fieldvector}
\end{align}\normalsize
and
\small\begin{align}
 \lag_2\urm{f}
 \  &\Rightarrow \
 \im \delta(z)  \mas\dgg\dot{\mas}.
\end{align}\normalsize
As a result, 
the free field and interaction Lagrangians of \en{ \tf } are,
respectively,  given as
\begin{subequations}
\label{fham}
\small\begin{align}
 \lag\urm{f}_{\mas}
= & \,
 \lag_1\urm{f}
+ 
 \im \delta(z)
 \mas\dgg \dot{\mas},
%%%%%%%%%%%%%%%%
\\
 \lag\urm{SU(2)}_{\mas}
= & 
 -2\im \delta(z)
 \phi\dd{\mas}\dgg \, \lie \, \phi\dd{\mas}.
\end{align}\normalsize
\end{subequations}

\marginpar{\vspace{-57mm}
\footnotesize
$\mm{\phi}_2\equiv\ffrac{\phi\drm{2,out}+\phi\drm{2,in}}{2}
\xrightarrow{(\ref{feedback})}\mas$.
\normalsize
 }

\subsection{SU(2) system}
\label{onoffshell}

Let us consider the input-output relation of the unitary system
for some examples.
We start with the SU(2) gate.\index{SU(2) system (forward traveling)}
The reactance matrix is given as
\small\begin{align}
 \lie
=
 \A{}{g}{-g\uu{*}}{}.
\end{align}\normalsize
From the Lagrangian (\ref{fham}) and the Euler-Lagrange equation,
the input-output relation is written as
\begin{subequations}
\label{stateq}
\small\begin{align}
 \frac{d}{dt}{\AV{\mas}{\mas\dgg}}
&=
 \A{-2|g|^2}{}
   {}{-2|g|^2}
 \AV{\mas}{\mas\dgg}
+
 \A{-2g\uu{*}}{}{}{-2g} 
 \AV{\phi_1}{\phi_1\dgg}\drm{in},
\\
 \AV{\phi_1}{\phi_1\dgg}\drm{out}
&= 
\hspace{11mm}
 \A{2g}{}{}{2g\uu{*}}
 \AV{\mas}{\mas\dgg}
\hspace{12mm}
+
 \A{1}{}{}{1}
 \AV{\phi_1}{\phi_1\dgg}\drm{in}.
\end{align}\normalsize
\end{subequations}
In the frequency domain,
this is expressed as (Section \ref{sec:tf-sys})
\small\begin{align}
 P(s)
&= 
 \dtf{\nA{-2|g|^2}{}{}{-2|g|^2}}
    {\nA{-2g\uu{*}}{}{}{-2g}}
    {\nA{2g}{}{}{2g\uu{*}}}
    {\nA{1}{}{}{1}}
=
\A{\ffrac{s-2|g|^2}{s+2|g|^2}}{}
  {}{\ffrac{s-2|g|^2}{s+2|g|^2}}.
\label{su2d}
\end{align}\normalsize

Note that 
the SU(2) system is a special case of the time-varying quantum gate
introduced in Section \ref{sec:timevarying}.
The forward traveling field propagates
in the closed loop \en{ \mas } infinite times,
which results in a memory effect
described as a time-varying reactance matrix \en{ \lie=\lie(t) }.
In fact, 
(\ref{su2d}) is also obtained by substituting an `integrator'
\small\begin{align}
 \lie(s) = -\frac{2|g|^2}{s}
\end{align}\normalsize
into the time-varying quantum gate (\ref{tvgatep})
\small\begin{align}
 \tf
=
 \A{\ffrac{1+\lie}{1-\lie}}{}
   {}{\ffrac{1-\lie\simm}{1+\lie\simm}}.
\end{align}\normalsize

\marginpar{\vspace{5mm}
\small
Note that 
$ \pole=-\zero $ holds only for transfer functions
between signals \textit{on shell}
such as $\phi\drm{in}$ and $\phi\drm{out}$.
If we consider a transfer function for $\mas$,
$ \pole=-\zero $ is not satisfied
because $\mas$ is an internal field and \textit{off shell}.
This is also related to canonical quantization.
See Chapter \ref{chap:sym} for details.
\normalsize
}

As shown in Section \ref{sec:timevarying},
the time-varying quantum gate has 
the pole-zero symmetry\index{pole-zero symmetry}
\small\begin{align}
 \pole(P)=-\zero(P).
\label{pzsym}
\end{align}\normalsize
It is easy to see this from (\ref{su2d}):
\begin{subequations}
\label{su2po}
\small\begin{align}
 \mbox{\normalsize poles: \small}& 
\quad \pole(P) = -2|g|^2,
\\
 \mbox{\normalsize zeros: \small}& 
\quad \zero(P) = +2|g|^2.
\end{align}\normalsize 
\end{subequations}

\begin{remark}
The state equation (\ref{stateq}) is 
the same as the quantum stochastic differential equation
(\ref{qsdecav}) in Appendix \ref{app:qsde}.
This system is known as an optical cavity
in quantum optics.
In this case, the forward traveling field describes
optical lasers and the $\mathrm{SU(2)}$ gate corresponds to 
a beam splitter.
\end{remark}

\newpage

\subsection{Time-varying SU(2) system}
\label{sec:dyn}

As another example of the unitary system,
let us consider 
the time-varying gate\index{SU(2) system (time-varying)}
introduced in Section \ref{sec:timevarying}.
Before closing the loop,
the Lagrangian is given by (\ref{tvsu2lag}):
\begin{subequations}
 \small\begin{align}
 \lag\urm{f}
&=
 \lag_1\urm{f} +  \lag_2\urm{f},
\\
 \lag\urm{TV}
&=
 -2 (\im\partial\dd{+}\win)
 \mm{\phi}\dgg \circ \lie \mm{\phi} 
\\ &=
 -2 (\im\partial\dd{+}\win)
 \mm{\phi}\dgg \lie\ast \mm{\phi}.
 \end{align}\normalsize
\end{subequations}
After closing the loop,
the Lagrangian is written as
\begin{subequations}
\label{tvsystem}
 \small\begin{align}
 \lag_{\mas}\urm{f}
&=
 \lag_1\urm{f}
+
 \im \delta(z)  \mas\dgg \dot{\mas},
\\
 \lag_{\mas}\urm{TV}
&=
 -2 (\im\partial\dd{+}\win)
 \phi\dd{\mas}\dgg \circ \lie \phi\dd{\mas} 
\\ &=
 -2 (\im\partial\dd{+}\win)
 \phi\dd{\mas}\dgg \lie\ast \phi\dd{\mas}.
 \end{align}\normalsize
\end{subequations}

Let us calculate the input-output relation of this system.
The reactance matrix is given by
\small\begin{align}
 \lie(t)
=
 \A{}{g(t)}{-g\uu{*}(t)}{}.
\end{align}\normalsize
The interaction Lagrangian is then written as
\begin{subequations}
 \small\begin{align}
 \lag\urm{SU(2)}_{\mas}
&=
-2\im \delta(z) \left[
 \phi\dd{\mas}\dgg  g\ast \mas 
- 
 \mas\dgg g\uu{*}\ast\phi\dd{\mas}
\right]
\\ &=
-2\im \delta(z) \left[
 \phi\dd{\mas}\dgg \circ g \ \mas 
- 
 \mas\dgg \circ g\uu{*} \ \phi\dd{\mas}
\right],
 \end{align}\normalsize
\end{subequations}
from which we have
\begin{subequations}
\small\begin{align}
 \dot{\mas} 
&=
 -2g\ast \mas\circ g\uu{*} -2g\uu{*} \ast \phi\drm{1,in},
\\
 \phi\drm{1,out} 
&= 
 2\mas\circ g +\phi\drm{1,in},
\end{align}\normalsize
and
\small\begin{align}
 \dot{\mas}\dgg 
&= 
 -2g\uu{*} \ast \mas\dgg\circ g -2\phi\drm{1,in}\dgg\circ g,
\\
 \phi\drm{1,out}\dgg 
&= 
 2g\uu{*} \ast \mas\dgg+\phi\drm{1,in}\dgg.
\end{align}\normalsize
\end{subequations}
After the Laplace transform,
these are expressed as
\small\begin{align}
 \AV{\phi_1}{\phi_1\dgg}\drm{out}
&=
 \A{\ffrac{s-2g(s)g\uu{*}(-s)}{s+2g(s)g\uu{*}(-s)}}{}
   {}{\ffrac{s-2g\uu{*}(s)g(-s)}{s+2g\uu{*}(s)g(-s)}}
  \AV{\phi_1}{\phi_1\dgg}\drm{in}.
\label{tvtf}
\end{align}\normalsize
It is not difficult to see the pole-zero symmetry
\en{ \pole=-\zero },\index{pole-zero symmetry}
even though
\en{ \pole } and \en{ \zero }
are not obtained explicitly,
This will be revisited in Chapter \ref{chap:sym}.

\newpage

\section{Non-unitary systems}

Non-unitary systems\index{non-unitary system}
can be defined from non-unitary gates
in the same way as the unitary systems.
Before closing the loop,
the Lagrangian of a non-unitary gate is given by (\ref{nonunilag}):
\begin{subequations}
\small\begin{align}
 \lag\urm{f}
&=
 \lag_1\urm{f} +  \lag_2\urm{f},
\\
 \lag\urm{int}
&=
 - 
  \mm{\Phi}\dgg \Sigma\dd{z} (\im \partial\dd{+} \win) \lie \mm{\Phi}.
\end{align}\normalsize
\end{subequations}
After closing the loop, 
the field vector is modified as
\small\begin{align}
 \mm{\Phi}
=
   \left[
    \begin{array}{c} 
        \mm{\phi}_1 \\
        \mm{\phi}_1\dgg \\ \hdashline
        \mm{\phi}_2 \\
        \mm{\phi}_2\dgg
    \end{array}
  \right]
\Rightarrow
   \left[
    \begin{array}{c} 
        \mm{\phi}_1 \\
        \mm{\phi}_1\dgg \\ \hdashline
        \mas \\
        \mas\dgg
    \end{array}
  \right]
\equiv
 \Phi\dd{\mas}
\end{align}\normalsize
The resulting free field and interaction Lagrangians are,
respectively, 
given as
\begin{subequations}
\label{nonunisysint}
\small\begin{align}
 \lag\urm{f}_{\mas}
&=
 \lag_1\urm{f}
+
 \im \delta(z)  \mas\dgg \dot{\mas},
\\
 \lag\dd{\mas}\urm{int}
&=
 - 
 \Phi\dd{\mas}\dgg \Sigma\dd{z} 
 (\im \partial\dd{+} \win) 
 \lie \Phi\dd{\mas}.
\end{align}\normalsize
\end{subequations}

\subsection{System via cross-squeezing}
\label{sec:systfsq}

\begin{wrapfigure}[0]{r}[53mm]{49mm} 
\centering
\vspace{-0mm}
\includegraphics[keepaspectratio,width=31mm]{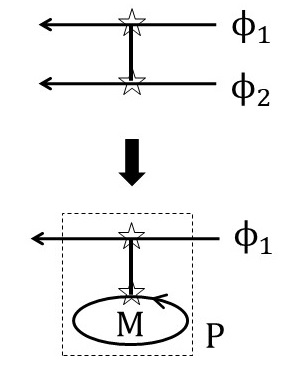}
\caption{
\small
(Upper) cross-squeezing gate.
(Lower) cross-squeezing system.
\normalsize
}
\label{fig-tsqsys-1}
\end{wrapfigure}

As an example,
let us consider the cross-squeezing gate\index{cross-squeezing gate}
\small\begin{align}
 \lie
=
 \A{}{g\sigma_x}
   {g\sigma_x}{}.
\end{align}\normalsize
In this case, the interaction Lagrangian is given by
\small\begin{align}
 \lag\urm{CQ}_{\mas}
&= 
 2 \im \delta(z) g \Bigl[
 \mm{\phi}_1 \mas - \mm{\phi}_1\dgg \mas\dgg 
\Bigr].
\end{align}\normalsize
The Euler-Lagrange equation yields a state equation
\begin{subequations}
\small\begin{align}
 \frac{d}{dt}\AV{\mas}{\mas\dgg}
&=
 \A{2g^2}{}{}{2g^2}\AV{\mas}{\mas\dgg}
+
 \A{}{2g}{2g}{}
 \AV{\phi_1}{\phi_1\dgg}\drm{in},
\\
 \AV{\phi_1}{\phi_1\dgg}\drm{out}
&=
 \hspace{3mm} \A{}{2g}{2g}{}\AV{\mas}{\mas\dgg}
+
 \hspace{5mm} \A{1}{}{}{1} 
 \AV{\phi_1}{\phi_1\dgg}\drm{in}.
\end{align}\normalsize
\end{subequations}

The transfer function of this system is given by
\small\begin{align}
 P(s)
&=
 \dtf{\nA{2g^2}{}{}{2g^2}}{\nA{}{2g}{2g}{}}
     {\nA{}{2g}{2g}{}}{\nA{1}{}{}{1}}
=
 \A{\ffrac{s+2g^2}{s-2g^2}}{}{}{\ffrac{s+2g^2}{s-2g^2}},
\label{sys2sq}
\end{align}\normalsize
which is the inverse of the SU(2) system (\ref{su2d}).
As a result,
\begin{subequations}
\small\begin{align}
\mbox{\normalsize poles: \small}& \quad \pole =2g^2,
\\
\mbox{\normalsize zeros: \small}& \quad \zero =-2g^2.
\end{align}\normalsize
\end{subequations}
Again the system possesses 
the pole-zero symmetry\index{pole-zero symmetry}
\en{ \pole = -\zero }.

\newpage

\subsection{QND and XX systems}
\label{sec:sumsys}

\begin{wrapfigure}[0]{r}[53mm]{49mm} 
\centering
\vspace{-5mm}
\includegraphics[keepaspectratio,width=35mm]{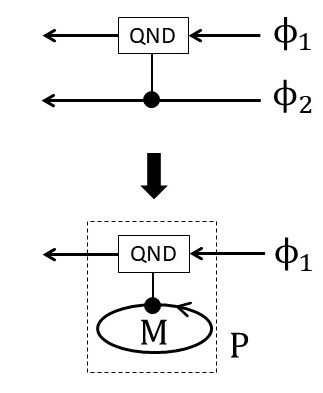}
\caption{
\small
(Upper) QND gate. 
(Lower) QND system.
\normalsize
}
\label{fig-sumsys-1}
\end{wrapfigure}

Let us consider the QND gate\index{QND gate}
defined by a reactance matrix
\small\begin{align}
\hspace{20mm}
 \lie
=
 \frac{g}{4}
 \A{}{-Q\dd{+}}
   {Q\dd{-}}{},
\hspace{5mm}
 \left(
 Q\dd{\pm}
=
 \A{1}{\pm 1}
   {\pm 1}{1}
\right)
\end{align}\normalsize
where we have swapped the control and target just 
for a computational purpose.
The interaction Lagrangian is given as
\begin{subequations}
\label{sumsyslag}
\small\begin{align}
 \lag\urm{QND}_{\mas} 
&= 
 \frac{1}{2} \im \delta(z) g
 (\mm{\phi}_1\dgg - \mm{\phi}_1) (\mas\dgg+\mas)
\\ &=
 \delta(z) \ g \mm{\eta}_1 \mas\dd{\xi}.
\end{align}\normalsize
\end{subequations}
for which the Euler-Lagrange equation is written as
\begin{subequations}
\small\begin{align}
 \AV{\dot{\mas}}{\dot{\mas}\dgg}
&=
\hspace{40mm}  
 \A{g/2}{-g/2}
   {-g/2}{g/2} 
 \AV{\phi_1}{\phi_1\dgg}\drm{in},
\\
 \AV{\phi_1}{\phi_1\dgg}\drm{out}
&=
 \A{-g/2}{-g/2}
   {-g/2}{-g/2}
 \AV{\mas}{\mas\dgg}
+ \hspace{23mm} 
 \AV{\phi_1}{\phi_1\dgg}\drm{in}.
\end{align}\normalsize
\end{subequations}
In the quadrature basis,
this is rewritten as
\begin{subequations}
\label{sumsysio}
\small\begin{align}
 \AV{\dot{\mas}\dd{\xi}}{\dot{\mas}\dd{\eta}}
&=
\hspace{37mm} \A{0}{0}{0}{g}
 \AV{\xi_1}{\eta_1}\drm{in},
\\
 \AV{\xi_1}{\eta_1}\drm{out}
&=
 \A{-g}{0}{0}{0}\AV{\mas\dd{\xi}}{\mas\dd{\eta}}
+ \hspace{17mm}
 \AV{\xi_1}{\eta_1}\drm{in},
\end{align}\normalsize
\end{subequations}
which clearly shows that the system is 
neither controllable nor observable.
\index{controllable (systems theory)}\index{observable (systems theory)}
In particular, 
\en{ \mas\dd{\xi} } is decoupled from the inputs.
This type of interaction 
is said to be \textit{non-demolition}.\index{non-demolition}

%\newpage

The same characteristics can be observed 
in the XX gate\index{XX gate} defined by a reactance matrix
\small\begin{align}
\hspace{20mm}
 \lie
=
 \frac{\im g}{4}
 \A{}{Q}
   {Q}{}.
\hspace{5mm}
 \left(Q
=
 \A{1}{1}
   {-1}{-1}\right)
\end{align}\normalsize
The interaction Lagrangian is written as
\small\begin{align}
% L &= \int dz [\lag\urm{f}_{\mas} + \lag\urm{XX}_{\mas}],
%\\
% \lag\urm{f}_{\mas}
%&=
% \step(z) \, \phi_4\dgg \im\partial\dd{+} \phi_4 
%+
% \step(-z) \, \phi_1\dgg \im\partial\dd{+} \phi_1 
 \lag\urm{XX}_{\mas} 
&= 
 \frac{1}{2}\delta(z) g
 (\mm{\phi}_1\dgg + \mm{\phi}_1) (\mas\dgg+\mas).
\end{align}\normalsize
In the quadrature basis,
the input-output relation is given as
\begin{subequations}
\label{xxsystem}
\small\begin{align}
 \AV{\dot{\mas}\dd{\xi}}{\dot{\mas}\dd{\eta}}
&=
\hspace{32mm} \A{0}{0}{g}{0}
  \AV{\xi_1}{\eta_1}\drm{in},
\\
 \AV{\xi_1}{\eta_1}\drm{out}
&=
 \A{0}{0}{g}{0}\AV{\mas\dd{\xi}}{\mas\dd{\eta}}
+ \hspace{15mm}
  \AV{\xi_1}{\eta_1}\drm{in}.
\end{align}\normalsize
\end{subequations}
Again,
this system is neither controllable nor observable.
\index{controllable (systems theory)}\index{observable (systems theory)}

\newpage

\section{SU(2) Dirac system}
\label{sec:diracsys}

Consider an SU(2) system 
for the Dirac field.\index{SU(2) system (Dirac)}
The procedure is the same as the forward traveling field.
Before closing a loop,
the Lagrangian of the SU(2) gate is given in (\ref{diracsu22}):
\begin{subequations}
\label{diracsu22-s}
\small\begin{align}
 \lag\urm{f}
&=
 \lag_1\urm{f} +  \lag_2\urm{f}
\\
 \lag\urm{SU(2)}
&=
 -2 \widebar{\mm{\psi}}(\im\fsh{\partial}\win) \, \lie \, \mm{\psi},
\end{align}\normalsize
\end{subequations}
where
\small\begin{align}
 \lie
=
 \A{}{g}{-g\uu{*}}{}.
\end{align}\normalsize
After closing the loop,
the field vector is modified as
\small\begin{align}
 \mm{\psi}
=
 \AV{\mm{\psi}_1}{\mm{\psi}_2}
\Rightarrow
 \AV{\mm{\psi}_1}{\mas}
\equiv
 \psi\dd{\mas}.
\label{dircdef}
\end{align}\normalsize
The free field Lagrangian 
\en{ \bigl(\phi_2 }-component of \en{ \lag^0\bigr) }
is also modified as
\small\begin{align}
% \step( z) \, \widebar{\psi}\drm{2,out}(\im \fsh{\partial}-m)\drm{2,out}
%+
% \step(-z) \, \widebar{\psi}\drm{2,in}(\im \fsh{\partial}-m)\psi\drm{2,in}
 \lag_2\urm{f}
\Rightarrow
 \delta(z)[\im \mas\dgg \dot{\mas}- m \mas\dgg \beta \mas].
\end{align}\normalsize
As a result, 
we get
\begin{subequations}
\small\begin{align}
  \lag\urm{f}_{\mas}
&=
 \lag_1\urm{f}
+
 \im \delta(z)
 \Bigl[
 \mas\dgg \dot{\mas}
+ 
 \im m \mas\dgg \beta \mas
\Bigr],
\\
%%%%%%%%%%%%%%%%%%%%%%%%%%%%
 \lag\urm{SU(2)}_{\mas}
&=
 -2\im\delta(z)
 \psi\dd{\mas}\dgg \alpha\uu{z} \lie \psi\dd{\mas}.
\\ &=
 -2\im\delta(z) 
\Bigl[
 \mm{\psi}\dgg \alpha\uu{z} g \mas
-
 \mas\dgg \alpha\uu{z} g\uu{*} \mm{\psi}
\Bigr].
\label{dsu2int}
\end{align}\normalsize
\end{subequations}

The input-output relation of this system is obtained 
via the Euler-Lagrange equation as
\begin{subequations}
\small\begin{align}
 \dot{\mas}
&=
 (-2|g|^2 \alpha\uu{z} -\im m \beta)\mas -2g\uu{*} \alpha^z \psi\drm{1,in},
\\
 \psi\drm{1,out}
&= 
 \hspace{25mm} 2g\mas+ \hspace{10mm} \psi\drm{1,in}.
\end{align}\normalsize
\end{subequations}
After the Laplace transform,
the system is expressed by a transfer function
\small\begin{align}
 P
=
 \dtf{-2|g|^2 \alpha\uu{z} -\im m \beta}{-2g\uu{*}\alpha\uu{z}}{2g}{1}.
\label{dtf} 
\end{align}\normalsize

To examine this transfer function,
assume that \en{ m=0 }.
In the Weyl basis (\ref{weylbasis}),
\small\begin{align}
 \alpha\uu{z}
&=
 \A{-1}{0}{0}{1}\otimes \sigma\uu{z},
\end{align}\normalsize
hence (\ref{dtf}) is written as
\small\begin{align}
 P(s)
&=
   \left[
    \begin{array}{cc:cc} 
        \ffrac{s+2|g|^2}{s-2|g|^2} & & & \\      
        & \ffrac{s-2|g|^2}{s+2|g|^2} & & \\ \hdashline
        & & \ffrac{s-2|g|^2}{s+2|g|^2} & \\
        & & & \ffrac{s+2|g|^2}{s-2|g|^2}
    \end{array}
  \right],
\label{weyltfd}
\end{align}\normalsize
where the top-left and bottom-right blocks are
the left- and right-chiral subspaces, 
respectively.
This transfer function is 
stable in the (2,2)- and (3,3)-elements,
and 
unstable in the (1,1)- and (4,4)-elements.
These elements correspond to
the forward and backward traveling components 
of the Weyl equation (\ref{chiralweyl}), 
respectively.

It is also interesting to examine \en{ \tf } 
in the Dirac basis.
The eigenvectors of the Dirac equation 
have been given in (\ref{eigenw}).
In the current setting, 
the positive-energy eigenvectors are written as
\small\begin{align}
 w\uu{(+)}\dd{\uparrow}
&\sim
 \AV{1}
    {\sigma\uu{z}}
 \otimes \AV{1}{0}
=
 \left[
 \begin{array}{c}
  1\\
  0\\
  1\\
  0
 \end{array}
\right],
\qquad
 w\uu{(+)}\dd{\downarrow}
\sim
 \AV{1}
    {\sigma\uu{z}}
 \otimes \AV{0}{1}
=
 \left[
 \begin{array}{c}
  0\\
  1\\
  0\\
  -1
 \end{array}
\right],
\end{align}\normalsize
where the suffixes \en{ \{\uparrow, \downarrow\} } represent 
positive and negative helicity.
Likewise,
the negative-energy eigenvectors are
\small\begin{align}
 w^{(-)}\dd{\uparrow}
&\sim
 \AV{-\sigma\uu{z}}
    {1}
 \otimes \AV{1}{0}
=
 \left[
 \begin{array}{c}
  -1\\
  0\\
  1\\
  0
 \end{array}
\right],
\qquad
 w\uu{(-)}\dd{\downarrow}
\sim
 \AV{-\sigma\uu{z}}
    {1}
 \otimes \AV{0}{1}
=
 \left[
 \begin{array}{c}
  0\\
  1\\
  0\\
  1
 \end{array}
\right].
\end{align}\normalsize
In the Dirac basis,
\small\begin{align}
 \alpha\uu{z}
&=
 \A{}{1}{1}{} \otimes \sigma\uu{z},
\end{align}\normalsize
hence 
\begin{subequations}
\small\begin{align}
 P(s)
&=
 \frac{s-2|g|^2\alpha\uu{z}}{s+2|g|^2\alpha\uu{z}}
\\ &=
 \frac{1}{s^2-(2|g|^2)^2}
   \left[
    \begin{array}{cc:cc} 
        s & & -2|g|^2 & \\      
        & s & & 2|g|^2 \\ \hdashline
        -2|g|^2 & & s & \\
        &  2|g|^2& & s
    \end{array}
  \right]\uu{2},
\end{align}\normalsize
\end{subequations}
from which it is easy to see that 
\begin{subequations}
\label{diractfd}
\small\begin{align}
 P \ w\uu{(+)}\dd{\uparrow(\downarrow)}
&=
 \frac{s-2|g|^2}{s+2|g|^2} \ w\uu{(+)}\dd{\uparrow(\downarrow)},
\\
 P \ w\uu{(-)}\dd{\uparrow(\downarrow)}
&=
 \frac{s+2|g|^2}{s-2|g|^2} \ w\uu{(-)}\dd{\uparrow(\downarrow)}.
\end{align}\normalsize
\end{subequations}
The system is stable in the positive-energy subspace,
whereas it is unstable in the negative-energy subspace.

This fermionic system has
the pole-zero symmetry as well.\index{pole-zero symmetry}
From (\ref{weyltfd}) or (\ref{diractfd}),
the poles and transmission zeros are given as
\begin{subequations}
\small\begin{align}
\mbox{\normalsize poles: \small}& \
 \pole = \{2|g|^2,-2|g|^2\},
\\
\mbox{\normalsize zeros: \small}& \
 \zero = \{2|g|^2,-2|g|^2\},
\end{align}\normalsize
\end{subequations}
which satisfy
\small\begin{align}
 \pole = -\zero.
\end{align}\normalsize

\chapter{Interconnections of systems}
\label{chap:inter}
\thispagestyle{fancy}

We introduce two types of interconnections of systems:
cascade and feedback connections.
In the cascade case, 
multiple systems interact with each other unidirectionally,
whereas
they interact bi-directionally in the feedback connections.

\section{Cascade connections}

\begin{wrapfigure}[0]{r}[53mm]{49mm} 
\centering
\vspace{-5mm}
\includegraphics[keepaspectratio,width=37mm]{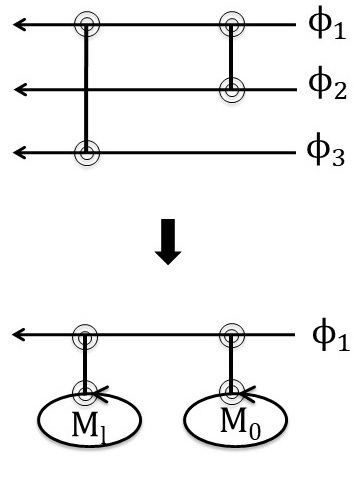}
\caption{
\small
Cascade of SU(2) systems.
\normalsize
}
\label{fig-cassys-1}
\end{wrapfigure}

The Lagrangians of cascaded systems are obtained 
in the same way as single systems.
As an example,
let us consider two SU(2) systems
as in Figure \ref{fig-cassys-1}.
The Lagrangian before closing the loops
has been given in Section \ref{sec:su2su2noncom}:
\begin{subequations}
\small\begin{align}
 \lag\urm{f}
&=
 \sum_{\alpha=1,2,3}\lag\dd{\alpha}\urm{f}
\\
 \lag\urm{int}
&=
 - 
  2 \im\delta(z) \mm{\phi}\dgg \, \lie \, \mm{\phi},
\end{align}\normalsize
\end{subequations}
where 
\small\begin{align}
 \mm{\phi}
=
 \BV{\mm{\phi}_1}
    {\mm{\phi}_2}
    {\mm{\phi}_3},
\qquad
  \lie
\equiv 
 \B{0}{g_0}{g\dd{l}}
   {-g_0}{0}{g_0g\dd{l}}
   {-g\dd{l}}{-g_0g\dd{l}}{0}.
\end{align}\normalsize
The two systems \en{ \mas_0 } and \en{ \mas\dd{l} } are 
defined as
\small\begin{align}
 \mm{\phi}
\Rightarrow
 \phi\dd{\mas}
\equiv
 \BV{\mm{\phi}_1}
    {\mas_0}
    {\mas\dd{l}}.
\end{align}\normalsize
The resulting Lagrangian is given as
\begin{subequations}
\label{ln0lag}
\small\begin{align}
 \lag\urm{f}_{\mas}
&= 
 \lag_1\urm{f}
+ 
 \im \delta(z)
\Bigl[
 \mas_0\dgg \dot{\mas}_0 
+
 \mas\dd{l}\dgg \dot{\mas}\dd{l}
\Bigr],
%%%%%%%%%%%%%%%%%%%%%%%%%
\\
 \lag\urm{SU(2)+SU(2)}\dd{\mas_0 + \mas\dd{l}} 
&=
 -2\im\delta(z)
 \phi\dd{\mas}\dgg \, \lie \, \phi\dd{\mas}
\\ &=
 -2\im\left[
 \mm{\phi}_1\dgg (g_0 \mas_0 + g\dd{l}\mas\dd{l})
+
 g_0g\dd{l} \mas_0\dgg  \mas\dd{l}
-\hc
\right].
\label{su2su2lag}
\end{align}\normalsize
\end{subequations}

\newpage

The input-output relation of this cascade system
is given by the Euler-Lagrange equation
\begin{subequations}
\small\begin{align}
 \frac{d}{dt}\AV{\mas_0}{\mas\dd{l}}
&=
 \A{-2g_0^2}{0}{-4g_0g\dd{l}}{-2g\dd{l}^2}
 \AV{\mas_0}{\mas\dd{l}}
+
 \AV{-2g_0}{-2g\dd{l}}
 \phi\drm{1,in},
\\
 \phi\drm{1,out}
&=
\hspace{9mm} \AH{2g_0}{2g\dd{l}} \AV{\mas_0}{\mas\dd{l}}
+
\hspace{16mm}  \phi\drm{1,in}.
\end{align}\normalsize
\end{subequations}
In the frequency domain,
we have
\begin{subequations}
\label{su2su2}
\small\begin{align}
 \phi\drm{1,out}
&=
 \dtf{\nA{-2g_0^2}{0}{-4g_0g\dd{l}}{-2g\dd{l}^2}}
    {\nAV{-2g_0}{-2g\dd{l}}}
    {\nAH{2g_0}{2g\dd{l}}}
    {1}
 \phi\drm{1,in}
\\ &=
 \dtf{-2g\dd{l}^2}{-2g\dd{l}}{2g\dd{l}}{1}
 \dtf{-2g_0^2}{-2g_0}{2g_0}{1}
 \phi\drm{1,in},
\end{align}\normalsize
\end{subequations}
where we have used (\ref{2.13}) in the second line.
This is the cascade of two SU(2) systems, 
as expected.

\section{Feedback connections}
\label{sec:fc}

\begin{wrapfigure}[0]{r}[53mm]{49mm} 
\centering
\vspace{-20mm}
\includegraphics[keepaspectratio,width=28mm]{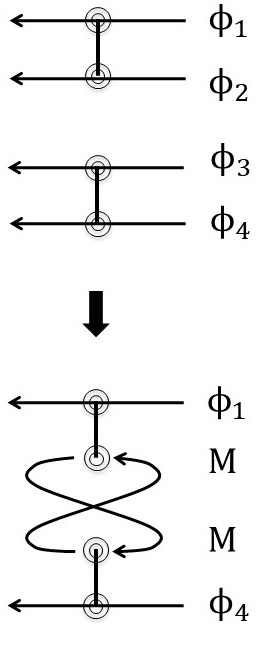}
\caption{
\small
(Upper) two SU(2) gates.
(Lower) feedback connection.
\normalsize
}
\label{fig-fbconnection1}
\end{wrapfigure}

In this section,
we introduce feedback connections of multiple quantum gates.
Unlike the cascade connection,
the flow of signals is bidirectional in this case.
We first consider a single feedback connection of two gates,
and then, 
examine the continuum limit of multiple feedback connections.

\subsection{Single feedback connection}
\label{sec:singlefb}

\begin{wrapfigure}[0]{r}[53mm]{49mm} 
\centering
\vspace{45mm}
\includegraphics[keepaspectratio,width=40mm]{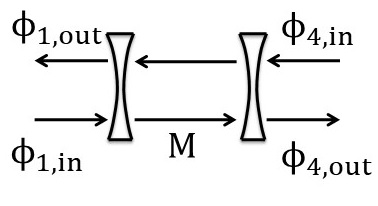}
\caption{
\small
Chain-scattering representation of the feedback connection
(see Appendix \ref{csr}.)
\normalsize
}
\label{fig-fbconnection1-c}
\end{wrapfigure}

Suppose that two SU(2) gates are located at \en{ z=0 }
as in Figure \ref{fig-fbconnection1}.
This is described by 
\begin{subequations}
\small\begin{align}
 \lag\urm{f}
&=
 \lag_1\urm{f} +  \lag_2\urm{f} +  \lag_3\urm{f} +  \lag_4\urm{f},
\\
 \lag\urm{int}
&=
 - 
  2 \im\delta(z) \mm{\phi}\dgg \, \lie \, \mm{\phi},
\end{align}\normalsize
\end{subequations}
where
\small\begin{align}
 \phi
=
 \left[
  \begin{array}{c} 
    \phi_1 \\
    \phi_2 \\ \hdashline
    \phi_3 \\
    \phi_4
    \end{array}
 \right],
\hspace{5mm}
 \lie 
=
   \left[
    \begin{array}{cc:cc} 
         &  g_0 &  & \\      
       -g_0\uu{*} &  & & \\ \hdashline
         & &  & g_1 \\
         & &  -g_1\uu{*}& 
    \end{array}
  \right]. 
\end{align}\normalsize

A feedback connection 
is defined by making a closed loop \en{ \mas } as
\small\begin{align}
 \mas
\equiv
 \phi\drm{2,in}  = \phi\drm{3,in} = 
 \phi\drm{2,out} = \phi\drm{3,out}.
\end{align}\normalsize
This leads to the modification of the field vector
\small\begin{align}
 \mm{\phi}
=
 \left[
  \begin{array}{c} 
    \mm{\phi}_1 \\
    \mm{\phi}_2 \\ \hdashline
    \mm{\phi}_3 \\
    \mm{\phi}_4
    \end{array}
 \right]
\Rightarrow
 \left[
  \begin{array}{c} 
    \mm{\phi}_1 \\
    \mas \\ \hdashline
    \mas \\
    \mm{\phi}_4
    \end{array}
 \right].
\label{fbcondef}
\end{align}\normalsize
Note that this vector is redundant.
In fact, 
\begin{subequations}
\small\begin{align}
\left[
  \begin{array}{cc:cc} 
    \mm{\phi}_1\dgg & \mas\dgg & \mas\dgg & \mm{\phi}_4\dgg
    \end{array}
 \right]
   \left[
    \begin{array}{cc:cc} 
         &  g_0 &  & \\      
       -g_0\uu{*} &  & & \\ \hdashline
         & &  & g_1 \\
         & &  -g_1\uu{*}& 
    \end{array}
  \right]
 \left[
  \begin{array}{c} 
    \mm{\phi}_1 \\
    \mas \\ \hdashline
    \mas \\
    \mm{\phi}_4
    \end{array}
 \right]&
%%%%%%%%%%%%%%%%%%%%%%%%%%
%%%%%%%%%%%%%%%%%%%%%%%%%%
\\ =
\left[
  \begin{array}{ccc} 
    \mm{\phi}_1\dgg & \mas\dgg & \mm{\phi}_4\dgg
    \end{array}
 \right]
   \left[
    \begin{array}{ccc} 
         &  g_0 &   \\      
       -g_0\uu{*} & 0  & g_1 \\ 
          &  -g_1\uu{*} & 
    \end{array}
  \right]
 \left[
  \begin{array}{c} 
    \mm{\phi}_1 \\
    \mas \\
    \mm{\phi}_4
    \end{array}
 \right]&.
\end{align}\normalsize
\end{subequations}
Let us define 
\small\begin{align}
 \phi\dd{\mas}
&\equiv
 \BV{\mm{\phi}_1}
    {\mas}
    {\mm{\phi}_4},
\quad
 \lie_T
\equiv
 \B{}{g_0}{}
   {-g_0\uu{*}}{0}{g_1}
   {}{-g_1\uu{*}}{}.
\end{align}\normalsize
Then the feedback connection is described by a Lagrangian
\begin{subequations}
\label{mimol}
\small\begin{align}
 \lag\urm{f}
&=
 \lag_1\urm{f} + \lag_4\urm{f}
+
 \im \delta(z)
 \mas\dgg \dot{\mas}
%%%%%%%%%%%%%%%%%
\\ 
 \lag\urm{SU(2)}\dd{\mas} 
&=
 -2 \im \delta(z)
 \phi\dd{\mas}\dgg \, \lie_T \, \phi\dd{\mas}
\\ &=
 -2 \im \delta(z)
\left\{
 \AH{\mm{\phi}_1\dgg}{\mm{\phi}_4\dgg} 
 \AV{g_0}{-g_1\uu{*}}
 \mas
+
 \mas\dgg
 \AH{-g_0\uu{*}}{g_1}
 \AV{\mm{\phi}_1}{\mm{\phi}_4} 
\right\}.
\end{align}\normalsize
\end{subequations}

\begin{wrapfigure}[0]{r}[53mm]{49mm} 
\centering
\vspace{-30mm}
\includegraphics[keepaspectratio,width=28mm]{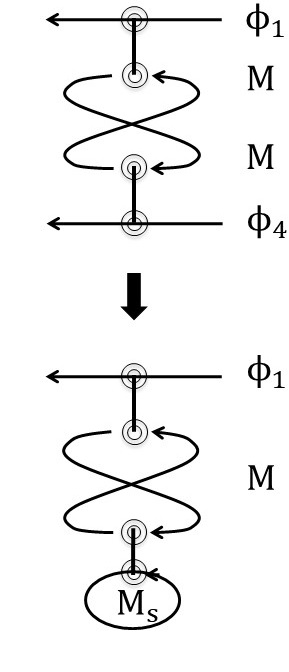}
\caption{
\small
Termination.
\normalsize
}
\label{fig-fbconnection2}
\end{wrapfigure}

The input-output relation of this system
is given by the Euler Lagrange equation as
\begin{subequations}
\small\begin{align}
 \dot{\mas}
&= 
 \left(-2|g_0|^2 -2|g_1|^2 \right)\mas
 + 
 \AH{-2g_0\uu{*}}{2g_1} 
 \AV{\phi_1}
    {\phi_4}\drm{in},
\\ 
 \AV{\phi_1}
    {\phi_4}\drm{out}
&=
\hspace{13mm}
 \AV{2g_0}
    {-2g_1\uu{*}}
 \mas
+ 
\hspace{26mm}
 \AV{\phi_1}
    {\phi_4}\drm{in}.
\end{align}\normalsize
\end{subequations}
In the frequency domain, this is expressed as
\small\begin{align}
 \AV{\phi_1}
    {\phi_4}\drm{out}
&=
 \dtf{-2|g_0|^2-2|g_1|^2}{\nAH{-2g_0\uu{*}}{2g_1}}
    {\nAV{2g_0}{-2g_1\uu{*}}}{\nA{1}{\quad 0}{0}{\quad 1}}
 \AV{\phi_1}
    {\phi_4}\drm{in}.
\label{mimotf}
\end{align}\normalsize

\begin{wrapfigure}[0]{r}[53mm]{49mm} 
\centering
\vspace{-5mm}
\includegraphics[keepaspectratio,width=40mm]{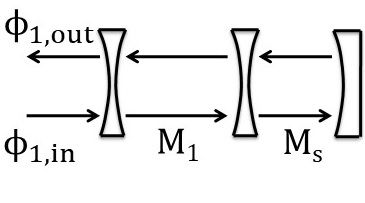}
\caption{
\small
Chain-scattering representation of the termination
(see Appendix \ref{csr}.)
\normalsize
}
\label{fig-fbconnection2-c}
\end{wrapfigure}

It is also possible to add an extra closed loop 
as in Figure \ref{fig-fbconnection2}.
It is defined by
\small\begin{align}
 \mas\drm{s}
\equiv
 \phi\drm{4,out} = \phi\drm{4,in},
\end{align}\normalsize
This is called a termination in systems theory.
In this case,
the interaction Lagrangian is obtained 
by defining \en{ \phi\dd{\mas} } as 
\small\begin{align}
 \phi\dd{\mas}
&\equiv
 \BV{\mm{\phi}_1}
    {\mas}
    {\mas\dd{s}}.
\end{align}\normalsize

\newpage

\subsection{A continuum limit of feedback connections}
\label{sec:casfeedback}

\begin{wrapfigure}[0]{r}[53mm]{49mm} 
\centering
\vspace{-5mm}
\includegraphics[keepaspectratio,width=27mm]{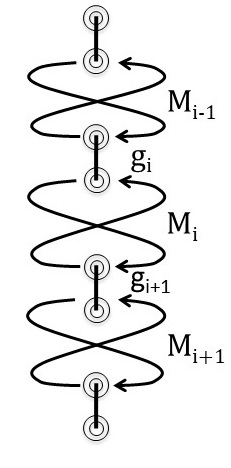}
\caption{
\small
Cascade of feedback connections.
\normalsize
}
\label{fig-fbconnection-lim}
\end{wrapfigure}

Suppose that 
\en{ (N+1) }-SU(2) gates are connected 
as in Figure \ref{fig-fbconnection-lim}.
Before closing loops,
the Lagrangian is given by 
\begin{subequations}
\small\begin{align}
 \lag\urm{f}
&=
 \sum\dd{\alpha=0}\uu{N}
 \lag\dd{\alpha}\urm{f}
\\
 \lag\urm{int}
&=
 - 
  2 \im \mm{\phi}\dgg \, \lie \, \mm{\phi},
\end{align}\normalsize
\end{subequations}
where
\small\begin{align}
 \phi
=
 \left[
  \begin{array}{c} 
    \phi_0 \\
    \phi_1 \\ \hdashline
    \vdots \\
    \vdots \\ \hdashline
    \phi\dd{2N} \\
    \phi\dd{2N+1}
    \end{array}
 \right],
\hspace{5mm}
 \lie 
=
   \left[
    \begin{array}{cc:cc:cc} 
         &  g_0 &  & & &\\      
       -g_0 &  & 0 & & & \\ \hdashline
         & 0 & \ddots & & & \\
         & & & \ddots & 0 & \\ \hdashline
         & & & 0 &  & g_N \\
         & & & &  -g_N & 
    \end{array}
  \right]. 
\end{align}\normalsize
where \en{ g\dd{n} } are real.
The feedback connection is defined 
in the same way as the single feedback connection:
\small\begin{align}
  \mm{\phi}
=
 \left[
  \begin{array}{c} 
    \mm{\phi}_0 \\
    \mm{\phi}_1 \\ \hdashline
    \mm{\phi}_2 \\
    \vdots \\
    \mm{\phi}_{2N-1} \\ \hdashline
    \mm{\phi}\dd{2N} \\
    \mm{\phi}\dd{2N+1}
    \end{array}
 \right]
\Rightarrow
 \left[
  \begin{array}{c} 
    \mm{\phi}_0 \\
    \mas_1 \\ \hdashline
    \mas_1 \\
    \vdots \\
    \mas_N \\ \hdashline
    \mas_N \\
    \mm{\phi}_{2N+1}
    \end{array}
 \right].
\end{align}\normalsize

\begin{wrapfigure}[0]{r}[53mm]{49mm} 
\centering
\vspace{-40mm}
\includegraphics[keepaspectratio,width=47mm]{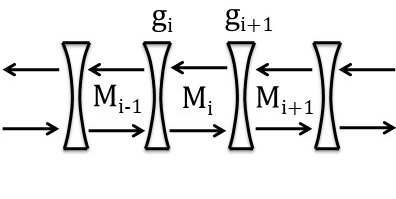}
\caption{
\small
Chain-scattering representation of 
the Cascade of feedback connections
(see Appendix \ref{csr}.)
\normalsize
}
\label{fig-fbconnection-lim-c}
\end{wrapfigure}

Again, this vector is redundant.
Let us define
\small\begin{align}
 \mas
\equiv
 \left[
 \begin{array}{c}
  \mas_1 \\
  \mas_2 \\
  \vdots \\
  \mas_N
 \end{array}
 \right].
\end{align}\normalsize
Then the feedback loops are described by a Lagrangian
\begin{subequations}
\small\begin{align}
 \lag\urm{f}
&=
 \mas\dgg \im\partial\dd{t}\mas,
\\
 \lag\urm{SU(2)}\dd{\mas}
&=
 -2\im \mas\dgg T \mas,
\label{casfeedlag}
\end{align}\normalsize
\end{subequations}
where \en{ T } is a Toeplitz matrix of the form
\small\begin{align}
 T
&=
\left[
  \begin{array}{ccccc}
     0 &  g_1  &       &        & \\ 
  -g_1 &  0    & g_2   &        & \\ 
       & -g_2  & 0     & \ddots & \\ 
       &       & \ddots& \ddots & g_{N-1}\\
       &       &       & -g_{N-1} & 0
  \end{array}
\right].
\end{align}\normalsize

Now let us consider a continuum limit.
The interaction Lagrangian (\ref{casfeedlag}) is rewritten as
\begin{subequations}
\small\begin{align}
  \lag\urm{SU(2)}\dd{\mas}
=&
 -2\im \sum\dd{n}
 g\dd{n}
(
 \mas\dd{n}\dgg \mas\dd{n+1}
-
 \mas\dd{n+1}\dgg \mas\dd{n}
)
\\ \to&
 -2\im \int
 g(z)
\left[
 \mas\dgg(z) \mas(z+dz)
-
 \mas\dgg(z+dz) \mas(z)
\right]
\\ =&
 -2\im \int
 g
\,
 W(\mas\dgg,\mas) \, dz
\end{align}\normalsize
\end{subequations}
where \en{ W } is a Wronskian defined as
\small\begin{align}
 W(a,b)
&=
 a(\partial\dd{z} b)-(\partial\dd{z} a) b.
\end{align}\normalsize
As a result,
the loops are described by a Lagrangian
\begin{subequations}
\label{splag}
\small\begin{align}
 \lag 
&=
 \int dz \
\Bigl[
 \mas\dgg \im\partial\dd{t} \mas - 2\im g W(\mas\dgg,\mas)
\Bigr]
\\ &=
 \int dz \
\Bigl[
 \mas\dgg \im(\partial\dd{t}-4g \partial\dd{z}) \mas - \mas\dgg 2(\im \partial\dd{z}g)  \mas 
\Bigr],
\end{align}\normalsize
\end{subequations}
where we have used integration by parts in the second line.
It follows from the Euler Lagrange equation that 
\small\begin{align}
 (\partial\dd{t}-4g\partial\dd{z})\mas
&= 2(\partial\dd{z} g) \mas,
\end{align}\normalsize
which describes the traveling field \en{ \mas } in a potential 
created by the non-uniform SU(2) coupling parameter \en{ g(z) }.
If \en{ g } is a negative (positive) constant,
this is simply a forward (backward) traveling equation.

\chapter*{}

\begin{center}
\begin{Large}
\begin{bfseries}
\textsf{Part III}
\\
\vspace{10mm}
\textsf{Quantum formulation of gates and systems}
\end{bfseries}
\end{Large}
\end{center}

\chapter{Quantum gates via \textit{S}-matrices}
\label{chap:fqf}
\thispagestyle{fancy}

In the classical approach,
the input-output relations %of quantum gates and systems 
have been obtained from the Euler-Lagrange equation.
In this chapter,
we show the same input-output relations
using \textit{S}-matrices. %in a quantum mechanical way.
This approach is useful 
to calculate transfer functions between arbitrary fields.
For example,
we have calculated an input-output relation in a classical way as
\small\begin{align}
 \phi\drm{out} = \tf \phi\drm{in}.
\label{inoutintro}
\end{align}\normalsize
When we consider feedback,
we need a transfer function 
from \en{ \phi\drm{out} } to \en{ \phi\drm{in} } 
(\en{ \phi\drm{out} \to \phi\drm{in} }):
\small\begin{align}
 \phi\drm{in} = X \phi\drm{out}.
% \phi_1(t+0\dd{+})=X\phi_4(t).
\end{align}\normalsize
How can we express \en{ X }?
It is not simply given by \en{ X=\tf\inv }.
It is important to note that 
\en { \tf } is a transfer function
\en{ \phi\drm{out} \gets \phi\drm{in} }
under the condition that 
a signal propagates 
\en{ \phi\drm{out} \gets \phi\drm{in} }.
Quantum systems are energy-preserving, or time-reversal,
hence if a signal propagates 
\en{ \phi\drm{out} \to \phi\drm{in} },
we have
\small\begin{align}
 \phi\drm{in} = \tf \phi\drm{out}.
\end{align}\normalsize
However,
what we want to know is 
a transfer function 
\en{ \phi\drm{out} \to \phi\drm{in} }
when a signal propagates 
\en{ \phi\drm{out} \gets \phi\drm{in} }.
A correct form of \en{ X } is obtained 
from \textit{S}-matrices
under an appropriate definition of transfer functions.

\section{Preliminaries}
\label{sec:pre}

We first introduce mathematical tools
such as Wick's theorem and the Feynman digram,
and then,
illustrate a procedure 
to calculate transfer functions.

\subsection{\textit{S}-matrices}

For quantum gates, circuits and systems, 
Lagrangians have been given as
\small\begin{align}
 L(\phi)
&=
 L\urm{f}(\phi) + L\urm{int}(\phi).
\end{align}\normalsize
\en{ L\urm{f} } is a free field Lagrangian 
that is always the same form 
for which  transfer functions are known.
\en{ L\urm{int}(\phi) } is an interaction Lagrangian 
that is treated as a perturbation.
To calculate transfer functions under \en{ L\urm{int} },
we introduce 
the \textit{interaction picture}\index{interaction picture}
that is a technique to express 
\en{ L\urm{int} } as a function of \en{ L\urm{f} }.

\newpage

Consider an operator of the form \en{ A(x)=A(\phi(x)) }.
Denoted by \en{ A\urm{X} } is 
its evolution under a Lagrangian \en{ L\urm{X} }.
For example,
\en{ \phi\urm{f} } is the field operator
evolving under the free field Lagrangian \en{ L\urm{f} }.
In the interaction picture, 
the evolution of \en{ A } is written as
\small\begin{align}
 A\urm{f + int}(x)
&=
 \uni\dgg(t,t_0) \ A\urm{f}(x) \ \uni(t,t_0),
\label{intuu-1}
\end{align}\normalsize
where \en{ \uni } obeys
\small\begin{align}
 \partial\dd{t} \uni(t,t_0)
&=
 \im L\urm{int}(\phi\urm{f}) \ \uni(t,t_0).
\label{intuu}
\end{align}\normalsize
Note that 
the RHS of (\ref{intuu-1}) is written only by \en{ \phi\urm{f} },
which means that 
\en{ A\urm{f + int} } can be expressed 
as a function of \en{ \phi\urm{f} } 
by expanding \en{ \uni }.
Since \en{ L\urm{int} } is a function of time
\en{ L\urm{int}(t) \equiv \int d\bm{x} \lag\urm{int}(\phi\urm{f}(x)) },
the evolution operator \en{ \uni } is expressed as
\small\begin{align}
 \uni(t,t_0)
=& \
 1 + \frac{\im}{1!}\int\dd{t_0}\uu{t} dt_1 \ L\urm{int}(t_1)
+ \frac{(\im)^2}{2!}\int\dd{t_0}\uu{t} dt_1 \int\dd{t_0}\uu{t_1} dt_2 \
 L\urm{int}(t_1) L\urm{int}(t_2)
+\cdots
%%%%%%%%%%%%%%%%%%%%%%%%%%%%%%%%%%
\nn\\ =& \
 1+\frac{\im}{1!}\int\dd{t_0}\uu{t} dt_1 \ L\urm{int}(t_1)
+ \frac{(\im)^2}{2!}\int\dd{t_0}\uu{t} dt_1 \int\dd{t_0}\uu{t} dt_2 \
 \T \left[L\urm{int}(t_1) L\urm{int}(t_2) \right]
+\cdots
%%%%%%%%%%%%%%%%%%%%%%%%%%%%%%%%%%
\nn\\ \equiv & \ 
 \T \exp \left[ \im \int\dd{t_0}\uu{t} dt' L\urm{int}(t') \right],
\label{spand}
\end{align}\normalsize
where \en{ \T } is 
the time-ordering operator\index{time-ordering operator}.
As \en{ t_0\to -\infty } and \en{ t\to \infty },
this operator is said to be 
an \textit{S}-matrix\index{\textit{S}-matrix}:
\small\begin{align}
 S
\equiv
 \uni(\infty,-\infty)
&=
 \T \exp \left[
 \im \int d^4x \ \lag\urm{int}
 \right].
\end{align}\normalsize

\subsection{Transfer function}

A transfer function under \en{ L\urm{f}+L\urm{int} }
is defined in the same way as Definition \ref{def:tf}:
\small\begin{align}
 \ipg\dd{A|B}\urm{f+int}(x_4,x_1)
\equiv
 \bra{\bm{0}} \T A\urm{f+int}(x_4) 
                 B\urm{f+int $\ddagger$}(x_1)\ket{\bm{0}}.
\label{deftflint}
\end{align}\normalsize
where \en{ \ket{\bm{0}} } is 
the vacuum state of \en{ L\urm{f}+L\urm{int} }.
Our purpose here is to rewrite (\ref{deftflint}) 
as a function of \en{ \phi\urm{f} }.

We first note that 
\en{ \ket{\bm{0}} } is different 
from the vacuum state \en{ \ket{0} } of \en{ L\urm{f} }
because different Lagrangians have different lowest energy states.
They are related to each other as
\small\begin{align}
 \ket{\bm{0}}
&\equiv
 \frac{\uni(0,-\infty)}{\sqrt{\mzero{S}}}\ket{0},
\label{vacint}
\end{align}\normalsize
where 
\en{ \uni(\infty,0)\uni(0,-\infty)=S }
has been used for normalization.
Using (\ref{intuu-1},\ \ref{vacint}), 
we can express the transfer function as
\begin{subequations}
\label{ptf}
\small\begin{align}
 \ipg\dd{A|B}\urm{f+int}(x_4,x_1)
&=
 \frac{
 \mzero{ \T \
 \uni(\infty,t_4) \ A\urm{f}(x_4) \ \uni(t_4,t_1) \ B\urm{f $\ddagger$}(x_1) \ \uni(t_1,-\infty)
 }}
 {\mzero{S}}
\\ &=
 \frac{\mzero{\T A\urm{f}_4 B\urm{f $\ddagger$}_1 S}}{\mzero{S}},
\end{align}\normalsize
\end{subequations}
which is a function of \en{ \phi\urm{f} }, 
as desired.
To express the RHS
with \en{ \mzero{ \T \phi\urm{f} \phi\urm{f $\ddagger$}} },
we need one more step, i.e., Wick's theorem.

\subsection{Wick's theorem}

Wick's theorem\index{Wick's theorem}
is used to extract contractions
from time-ordered products of operators.
For arbitrary operators \en{ A_1,\cdots, A\dd{n} },
we have the following identity:
\small\begin{align}
 \T (A_1A_2\cdots A\dd{n})
&= \
 :A_1A_2\cdots A\dd{n}: 
\nn\\ &\hspace{3.5mm} 
+
 \contraction{}{\phi}{\hspace{3mm}}{\hspace{0mm}}
 A_1A_2
:A_3\cdots A\dd{n}:
+ \
 \contraction{}{\phi}{\hspace{3mm}}{\hspace{0mm}}
 A_1A_3
:A_2\cdots A\dd{n}:
+\cdots
\nn\\ &\hspace{-5mm} +
 \contraction{}{\phi}{\hspace{3mm}}{\hspace{0mm}}
 A_1A_2
 \contraction{}{\phi}{\hspace{3mm}}{\hspace{0mm}}
 A_3A_4
:A_5\cdots A\dd{n}:
+ \ 
 \contraction{}{\phi}{\hspace{3mm}}{\hspace{0mm}}
 A_1A_3
 \contraction{}{\phi}{\hspace{3mm}}{\hspace{0mm}}
 A_2A_4
:A_5\cdots A\dd{n}: 
+\cdots,
\end{align}\normalsize
where \en{ : \ : } represents 
the normal ordering operation\index{normal ordering operation}
in which 
all \en{ \phi\dgg } are placed to the left of all \en{ \phi } 
in the product.
For example,
\en{  :\phi_1 \phi\dgg_2 \phi_3: \
=
 \phi\dgg_2 \phi_1 \phi_3,
 }
where 
the order of \en{ \phi_1=\phi(x_1) } and \en{ \phi_3=\phi(x_3) } 
does not matter because they commute with each other.
%It is important to note that 

The vacuum expectation of the normal-ordered operators is always zero:
\small\begin{align}
 \mzero{: A_1A_2\cdots A\dd{n}:}&=0.
\end{align}\normalsize
Hence only completely contracted terms remain 
after taking the vacuum expectation of the time-ordered product:
\small\begin{align}
 \mzero{\T  (A_1\cdots A\dd{n}) }
&=
 \contraction{}{\phi}{\hspace{3mm}}{\hspace{0mm}}
 A_1A_2
 \contraction{}{\phi}{\hspace{3mm}}{\hspace{0mm}}
 A_3A_4
\cdots
 \contraction{}{\phi}{\hspace{6mm}}{\hspace{0mm}}
 A_{n-1}A\dd{n}
+
 \cdots. 
\end{align}\normalsize

\subsection{Feynman diagram}
\label{sec:feydiag}

Feynman diagrams\index{Feynman diagram} 
are used to visualize contractions.
First, 
a four-position \en{ x } is represented by a vertex.
For forward traveling fields,
\en{ \phi(x) } and \en{ \phi\dgg(x) } are, 
respectively, 
depicted by incoming and outgoing dashed arrows,
and they can be connected to each other,
which results in a contraction.
These are summarized as follows:
\begin{subequations}
\small\begin{align}
 \phi(x) \ &= \quad  \bullet \hspace{-1mm} \dashleftarrow
\\
 \phi\dgg(x) \ &= \quad \ \dashleftarrow \hspace{-1mm}\bullet
\\
 \ipg\dd{\phi | \phi}
=
 \wick{\phi}{\phi\dgg}
=
 \contraction{}{\phi}{\hspace{1mm}}{\hspace{0mm}}
 \phi\phi\dgg
\ &= \quad
 \bullet \hspace{-1mm}  \dashleftarrow \hspace{-1mm} \bullet
\end{align}\normalsize
\end{subequations}

We do not prove here,
but 
we need to consider only connected diagrams
of the numerator in (\ref{ptf})
because disconnected diagrams are canceled out 
by the denominator.

\subsection{Summary}
\label{sec:summary}

Eventually, 
\en{ \ipg\urm{f+int} } is obtained
as a function of \en{ \ipg\urm{f} } as follows:
\bee
\setlength{\itemsep}{0mm} 
\item Calculate all free field (unperturbed) transfer functions 
      \en{ \ipg\urm{f} }.
\item Expand the \textit{S}-matrix in the numerator of (\ref{ptf}).
\item Draw Feynman diagrams for each term of the expansion.
\item Select connected diagrams.
\item Plug \en{ \ipg\urm{f} } obtained in Step-1 
      into the selected diagrams.
\ee
In the next section,
we demonstrate this procedure for quantum gates in detail.

\newpage

\section{Transfer functions of unitary gates}

As an exercise of the preceding section,
quantum gates are examined.
We assume that all gates are \textit{lumped}
for which gauge fields are static.
In this case, 
the interaction Lagrangian is quadratic
and the input-output relation
results in the same form as the gauge transformation.
Recall that 
unitary gates are described by \en{ \phi } alone,
whereas \en{ \Phi } is necessary for non-unitary gates:
\begin{subequations}
\small\begin{align}
 \phi\ddgg
&=
 \phi\dgg,
\\
 \Phi\ddgg
&=
 \AV{\phi}{\phi\dgg}\ddgg
=
 \AH{\phi\dgg}{-\phi}
=
 \Phi\dgg \sigma\dd{z}.
\end{align}\normalsize
\end{subequations}
We consider the unitary case first.
According to Section \ref{sec:summary},
we start with the calculation of free field transfer functions.

\subsection{Free field transfer functions}
\label{sec:unp}

\begin{wrapfigure}[0]{r}[53mm]{49mm} 
\vspace{0mm}
\centering
\includegraphics[keepaspectratio,width=37mm]{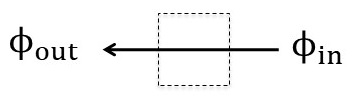}
\caption{
\small
Input and output of a quantum gate placed at $z$.
If the gate is empty, the input and the output are, respectively, defined by 
$\phi_{\textrm{in}}=\phi(z-0)$
and $\phi_{\textrm{out}}=\phi(z+0)$.
\normalsize
}
\label{fig-no-operation-1}
\end{wrapfigure}

Consider a lumped gate 
that is localized at a single point in space
as in Figure \ref{fig-no-operation-1}.
The input is a forward traveling field before 
entering the gate,
and the output is the same field
after coming out of the gate.

\begin{lemma}
\label{lem:freefield}
Assume that a single-input and single-output gate is empty.
The (free field) transfer functions of this gate are given by
\begin{subequations}
\small\begin{align}
 \ipg\dd{\phi\dd{\mathrm{in}}|\phi\dd{\mathrm{in}}} 
&=
 \wick{\phi\dd{\mathrm{in}}}{\phi\dd{\mathrm{in}}\dgg}
=
1,
\\
 \ipg\dd{\phi\dd{\mathrm{out}}|\phi\dd{\mathrm{in}}} 
&=
 \wick{\phi\dd{\mathrm{out}}}{\phi\dd{\mathrm{in}}\dgg}
=
1,
\\
 \ipg\dd{\phi\dd{\mathrm{in}}|\phi\dd{\mathrm{out}}} 
&=
 \wick{\phi\dd{\mathrm{in}}}{\phi\dd{\mathrm{out}}\dgg}
=
 -1,
\\
 \ipg\dd{\phi\dd{\mathrm{out}}|\phi\dd{\mathrm{out}}} 
&=
 \wick{\phi\dd{\mathrm{out}}}{\phi\dd{\mathrm{out}}\dgg}
=
1.
\end{align}\normalsize
\end{subequations}
\end{lemma}

\begin{proof}
Basically these relations result from Section \ref{sec:tf}.
It follows from (\ref{fthrou-1}) that 
the transfer function at a fixed point
can be written as
\small\begin{align}
 \ipg\dd{\phi\dd{\mathrm{in}}|\phi\dd{\mathrm{in}}} 
=
 \ipg\dd{\phi\dd{\mathrm{out}}|\phi\dd{\mathrm{out}}} =1.
\label{ngtf11}
\end{align}\normalsize
If no operations act on the field, 
the output \en{ \phi\drm{out} } is equivalent 
to the input \en{ \phi\drm{in} } 
in the forward direction of time:
\small\begin{align}
 \phi\drm{out}(t)
=
 \lim\dd{\epsilon\to 0\dd{+}} \phi\drm{in}(t+\epsilon).
\end{align}\normalsize
As a result, 
the transfer function 
from \en{ \phi\drm{in} } to \en{ \phi\drm{out} }
is the same as (\ref{ngtf11}):
\small\begin{align}
 \ipg\dd{\phi\drm{out} | \phi\drm{in}}
&=
 \ipg\dd{\phi\drm{in} | \phi\drm{in}}
=
 1.
\end{align}\normalsize
It also follows from (\ref{fthrou-2}),
or the asymmetry shown in Section \ref{lem:asymmetry}, 
that
\small\begin{align}
  \ipg\dd{\phi\drm{in}|\phi\drm{out}} 
=
 -1,
\label{unpas}
\end{align}\normalsize
which completes the assertion.
\end{proof}

\newpage

If we express an input-output relation as
\small\begin{align}
 \phi\dd{\alpha}
=
 \ipg\dd{\phi\dd{\alpha}|\phi\dd{\beta}} \,
 \phi\dd{\beta},
\label{inoutcont}
\end{align}\normalsize
we get two contradictory relations
from Lemma \ref{lem:freefield}:
\begin{subequations}
\small\begin{align}
 \phi\dd{\mathrm{out}}
&=
 \phi\dd{\mathrm{in}},
\label{cont-1} \\
 \phi\dd{\mathrm{in}}
&=
 -\phi\dd{\mathrm{out}}.
\label{cont-2}
\end{align}\normalsize
\end{subequations}
Note that 
\en{ \ipg\dd{\phi\dd{\alpha}|\phi\dd{\beta}} }
is not simply a transfer function 
from right to left: \en{ \phi\dd{\alpha}\gets\phi\dd{\beta} }.
It is a `conditional' transfer function.
For example,
\begin{subequations}
\begin{align}
 (\ref{cont-1}):
 \mbox{transfer function }
 \phi\drm{out} \gets \phi\drm{in}
 \mbox{ as the signal propagates }
 \phi\drm{out} \gets \phi\drm{in}.
\nn \\
 (\ref{cont-2}):
 \mbox{transfer function }
\phi\drm{in} \gets \phi\drm{out} 
 \mbox{ as the signal propagates }
 \phi\drm{out} \gets \phi\drm{in}.
\nn
\end{align}
\end{subequations}
Accordingly,
a transfer function \en{ \phi\dd{\beta}\gets\phi\dd{\alpha} }
is not simply given by 
\en{ (\ipg\dd{\phi\dd{\alpha}|\phi\dd{\beta}})\inv }.
This will be important when we consider feedback.

\begin{lemma}
\label{lem:mmtrans}
For 
\small\begin{align}
\mm{\phi}
\equiv
 \frac{\phi\dd{\mathrm{in}} + \phi\dd{\mathrm{out}}}{2},
\end{align}\normalsize
free field transfer functions are given as
\begin{subequations}
\label{outmmin}
\small\begin{align}
 \ipg\dd{\mm{\phi} |\phi\dd{\mathrm{in}}} 
&=
 \wick{\mm{\phi}}{\phi\dd{\mathrm{in}}\dgg}
=
1,
\\
 \ipg\dd{\phi\dd{\mathrm{out}}|\mm{\phi}} 
&=
 \wick{\phi\dd{\mathrm{out}}}{\mm{\phi}\dgg}
=
1,
\\
\ipg\dd{\mm{\phi} | \mm{\phi}}
&=
 \wick{\mm{\phi}}{\mm{\phi}\dgg}
=
 \frac{1}{2},
\\
 \ipg\dd{\mm{\phi} |\phi\dd{\mathrm{out}}} 
&=
 \wick{\mm{\phi}}{\phi\dd{\mathrm{out}}\dgg}
=0,
\\
 \ipg\dd{\phi\dd{\mathrm{in}}|\mm{\phi}} 
&=
 \wick{\phi\dd{\mathrm{in}}}{\mm{\phi}\dgg}
=
0.
\end{align}\normalsize
\end{subequations}
It follows from the asymmetry of the forward traveling field that 
\begin{subequations}
\label{outmmin2}
\small\begin{align}
 \ipg\dd{\mm{\phi}\dgg |\phi\dd{\mathrm{in}}\dgg } 
&=
 \wick{\mm{\phi}\dgg}{-\phi\dd{\mathrm{in}}}
=
1,
\\
 \ipg\dd{\phi\dd{\mathrm{out}}\dgg | \mm{\phi}\dgg} 
&=
 \wick{\phi\dd{\mathrm{out}}\dgg}{-\mm{\phi}}
=
1,
\\
\ipg\dd{\mm{\phi}\dgg | \mm{\phi}\dgg}
&=
 \wick{\mm{\phi}\dgg}{-\mm{\phi}}
=
 \frac{1}{2},
\\
 \ipg\dd{\mm{\phi}\dgg |\phi\dd{\mathrm{out}}\dgg } 
&=
 \wick{\mm{\phi}\dgg}{ -\phi\dd{\mathrm{out}}}
=0,
\\
 \ipg\dd{\phi\dd{\mathrm{in}}\dgg | \mm{\phi}\dgg } 
&=
 \wick{\phi\dd{\mathrm{in}}\dgg}{-\mm{\phi}}
=
0.
\end{align}\normalsize
\end{subequations}
A relationship between (\ref{outmmin}) and (\ref{outmmin2})
is expressed as
\small\begin{align}
 \wick{A\dgg}{B}
=
 -\wick{A}{B\dgg}\uu{*}.
\end{align}\normalsize
\end{lemma}

\begin{wrapfigure}[0]{r}[53mm]{49mm} 
\vspace{-110mm}
\centering
\includegraphics[keepaspectratio,width=48mm]{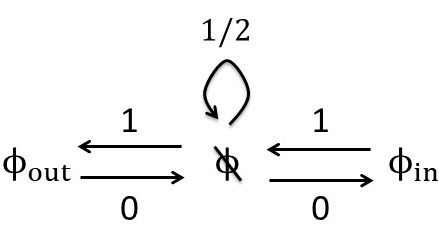}
\caption{
\small
Free field transfer functions.
\normalsize
}
\label{fig-interpretation}
\end{wrapfigure}

This lemma is interpreted as Figure \ref{fig-interpretation}.
Note that 
\en{ \mm{\phi} } is regarded as a `mean field' of 
the input \en{ \phi\dd{\mathrm{in}} } and 
the output \en{ \phi\dd{\mathrm{out}} }.
The free field transfer functions are unity 
along the forward direction,
whereas they are zero in the backward direction.
This is a consequence of the unidirectionality 
of the forward traveling field.

\newpage

\subsection{The input-output relations of unitary gates}
\label{sec:qsu2}

For unitary gates,
the interaction Lagrangian is given by (\ref{unilag-2}):
\small\begin{align}
 \im\lag\urm{int}
&=
 \mm{\phi}\dgg  (2\lie) \, \mm{\phi}.
\label{unintlag}
\end{align}\normalsize
Our purpose is to calculate a contraction
\small\begin{align}
 \wick{A}{B\dgg}\urm{int}
\equiv
 \frac{\mzero{\T A B\dgg S}}{\mzero{S}}.
\label{tf41}
\end{align}\normalsize
Assume that \en{ A } and \en{ B } are functions of \en{ \phi }.
According to Section \ref{sec:summary},
we expand the \textit{S}-matrix in the numerator and select 
connected diagrams.
Since (\ref{unintlag}) is quadratic, 
the expansion is simply given by
\begin{subequations}
\label{ft}
\small\begin{align}
 &
 \wick{A}{B\dgg} 
\label{sut0}\\ +&
%%%%%%%%%%%%%%%
 \frac{1}{1!}\int dx_2 \ 
 \mzero{\T \, A \  \im\lag\urm{int}(x_2) \ B\dgg }
\label{sut1}\\ +&
%%%%%%%%%%%%%%%
 \frac{1}{2!}\int dx_3 dx_2 \ 
 \mzero{\T \, A \,
 \im \lag\urm{int}(x_3) \ \im\lag\urm{int}(x_2) \
 B\dgg }
\label{sut2} +\cdots.
\end{align}\normalsize
\end{subequations}

Let us take a look at each term in detail.
The first term (\ref{sut0}) is the zeroth order term,
which is a direct-through effect
\en{ A\gets B }.
The next term (\ref{sut1}) is the first order correction  
that is given by 
\small\begin{align}
 \int dx_2 \ 
\contraction{}{\phi}{\hspace{2mm}}{} 
 A \
\contraction{}{\phi\hspace{2mm}}{\hspace{10mm}}{} 
 \im \lag\urm{int}(x_2)
 B\dgg
=
 \wick{A}{\mm{\phi}\dgg} 
 (2 \lie) 
 \wick{\mm{\phi}}{B\dgg}.
\label{su2gfirst}
\end{align}\normalsize
The next correction term is (\ref{sut2})
in which there are two possible contractions:
\begin{subequations}
\small\begin{align}
& \contraction{}{\phi}{\hspace{2mm}}{} 
 A \
\contraction{}{\phi\hspace{2mm}}{\hspace{11mm}}{} 
 \im\lag\urm{int}(x_3) \
\contraction{}{\phi\hspace{3mm}}{\hspace{9mm}}{} 
 \im\lag\urm{int}(x_2)
 B\dgg,
%%%%%%%%%%%%%%%%%%%%%%
\\
& \contraction[2ex]{}{\phi}{\hspace{19mm}}{} 
 A \
\contraction{}{\phi\hspace{3mm}}{\hspace{10mm}}{} 
\contraction[3ex]{}{\phi\hspace{0mm}}{\hspace{26mm}}{} 
 \im\lag\urm{int}(x_3) \
 \im\lag\urm{int}(x_2)
 B\dgg.
\end{align}\normalsize
\end{subequations}
This is expressed by a factor \en{ 2! }.
The second order correction is therefore given as
\small\begin{align}
 2!\frac{1}{2!}  \int  dx_3 & dx_2 \ 
\contraction{}{A}{\hspace{1mm}}{} 
 A \
\contraction{}{\im\lag}{\hspace{11mm}}{} 
 \im\lag\urm{int}(x_3) \
\contraction{}{\im\lag}{\hspace{11mm}}{} 
 \im\lag\urm{int}(x_2)
 B\dgg
%\\ &=
=
 \wick{A}{\mm{\phi}\dgg} 
 (2 \lie) 
 \wick{\mm{\phi}}{\mm{\phi}\dgg} 
 (2 \lie) 
 \wick{\mm{\phi}}{B\dgg}.
\label{2ndf}
\end{align}\normalsize

Likewise, the third order correction is given by 
\small\begin{align}
\hspace{3mm}
 \wick{A}{\mm{\phi}\dgg} 
 (2 \lie) 
 \wick{\mm{\phi}}{\mm{\phi}\dgg} 
 (2 \lie) 
 \wick{\mm{\phi}}{\mm{\phi}\dgg} 
 (2 \lie) 
 \wick{\mm{\phi}}{B\dgg}.
\end{align}\normalsize
As a result,
we have
\begin{subequations}
\label{su2gatetf}
\small\begin{align}
 \wick{A}{B\dgg}\urm{int}
&=
 \wick{A}{B\dgg} 
+
 \wick{A}{\mm{\phi}\dgg}  (2\lie) 
\Bigl\{
 I
+
 \wick{\mm{\phi}}{\mm{\phi}\dgg} (2\lie)
+\cdots
\Bigr\}
 \wick{\mm{\phi}}{B\dgg}
\\ &\sim
 \wick{A}{B\dgg}  
+
 \wick{A}{\mm{\phi}\dgg}  (2\lie) 
\Bigl\{
 I
-
 \wick{\mm{\phi}}{\mm{\phi}\dgg}  (2\lie)
\Bigr\}\inv
 \wick{\mm{\phi}}{B\dgg}.
\end{align}\normalsize
\end{subequations}
Let us consider examples of this transfer function.

\newpage

\subsection{SU(2) gate}

\begin{wrapfigure}[0]{r}[53mm]{49mm}
\vspace{60mm}
\centering
\includegraphics[keepaspectratio,width=39mm]{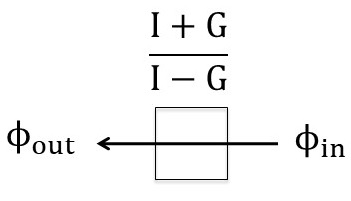}
\caption{
\small
Transfer function 
$\phi_{\mathrm{out}}\gets\phi_{\mathrm{in}}$.
\normalsize
}
\label{fig-tf1}
\end{wrapfigure}

As an example of the unitary gate,
let us consider an SU(2) gate for which 
\small\begin{align}
 \lie
=
 \A{}{g}{-g\uu{\ast}}{}.
\end{align}\normalsize
To calculate the input-output relation,
set 
\small\begin{align}
 A
=
 \phi\drm{out}
\equiv
 \AV{\phi_1}{\phi_2}\drm{out},
\quad
 B
=
 \phi\drm{in}
\equiv
 \AV{\phi_1}{\phi_2}\drm{in}.
\end{align}\normalsize
Using the free field transfer functions 
given in Lemma \ref{lem:mmtrans},
we get
\begin{subequations}
\label{su2iov}
\small\begin{align}
  \ipg\dd{\phi\drm{out}|\phi\drm{in}}\urm{SU(2)}
&=
 \wick{\phi\drm{out}}{\phi\drm{in}\ddgg}\urm{SU(2)}
\\ &=
 \wick{\phi\drm{out}}{\phi\drm{in}\dgg}\urm{SU(2)}
\\ &=
 I 
+ 
 2\lie
\left(
 I - \lie
\right)\inv \hspace{7mm} \because (\ref{su2gatetf}) 
\\ &=
 \frac{I+\lie}{I-\lie}
\\ &=
 \frac{1}{1+|g|^2}
 \A{1-|g|^2}{2g}{-2g\uu{*}}{1-|g|^2}.
\end{align}\normalsize
\end{subequations}
This is the definition of the SU(2) gate (\ref{sincos}).

\subsection{Other components of the SU(2) gate}

\begin{wrapfigure}[0]{r}[53mm]{49mm}
\vspace{-10mm}
\centering
\includegraphics[keepaspectratio,width=37mm]{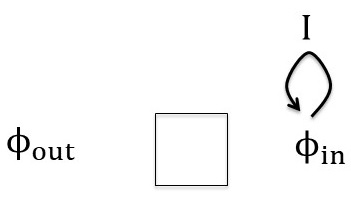}
\caption{
\small
Transfer function 
$\phi_{\mathrm{in}}\gets\phi_{\mathrm{in}}$.
\normalsize
}
\label{fig-tf2}
\end{wrapfigure}

The advantage of the \textit{S}-matrix approach 
is that we can calculate transfer functions
between arbitrary fields.
This cannot be done 
with the Euler-Lagrange equation in a classical fashion.

\begin{wrapfigure}[0]{r}[53mm]{49mm}
\vspace{30mm}
\centering
\includegraphics[keepaspectratio,width=37mm]{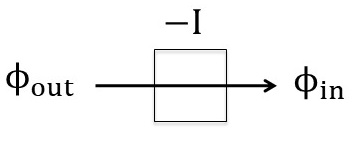}
\caption{
\small
Transfer function 
$\phi_{\mathrm{in}}\gets\phi_{\mathrm{out}}$.
\normalsize
}
\label{fig-tf3}
\end{wrapfigure}

Let us consider a transfer function 
from the input to itself (\en{ \phi\drm{in}\gets\phi\drm{in} }),
in which case
\small\begin{align}
 A=B=\phi\drm{in}.
\end{align}\normalsize
Obviously, the SU(2) gate does not involve in this process
as in Figure \ref{fig-tf2}.
It is expected that 
we simply get the free field transfer function
from the input to itself.
In fact, since 
\small\begin{align}
 \wick{\phi\drm{in}}{\mm{\phi}\dgg} =0,
\end{align}\normalsize
(\ref{su2gatetf}) is written as
\small\begin{align}
 \wick{\phi\drm{in}}{\phi\drm{in}\dgg}\urm{SU(2)}
&=
 \wick{\phi\drm{in}}{\phi\drm{in}\dgg}
=
 I.
\end{align}\normalsize

Likewise,
a transfer function 
from the output to the input 
(\en{ \phi\drm{in}\gets\phi\drm{out} })
is given as
\small\begin{align}
% \ipg\dd{\phi\drm{in}|\phi\drm{out}}\urm{SU(2)}
 \wick{\phi\drm{in}}{ \phi\drm{out}\dgg}\urm{SU(2)}
&=
 -I,
\end{align}\normalsize
in which \en{ \lie } does not appear.
(Compare Figure \ref{fig-tf1} to Figure \ref{fig-tf3}.)
This means that 
the SU(2) gate does not operate in the backward direction in time.

\newpage

It is also interesting to see 
a transfer function 
\en{ \phi\drm{out}\dgg \gets \phi\drm{in}\dgg }.
This is expressed by 
\small\begin{align}
 \ipg\dd{\phi\drm{out}\dgg | \phi\drm{in}\dgg}\urm{SU(2)}
=
 -\wick{\phi\drm{out}\dgg}{ \phi\drm{in}}\urm{SU(2)}.
\end{align}\normalsize
To avoid notational complexity,
let us calculate each element
\small\begin{align}
\hspace{27mm}
 \wick{\phi\dd{\alpha,\textrm{out}}\dgg}
      {\phi\dd{\beta,\textrm{in}}}\urm{SU(2)},
\hspace{5mm}
\left(
 \alpha,\beta = 1,2.
\right)
\end{align}\normalsize
The \textit{S}-matrix expansion is given as
\begin{subequations}
\small\begin{align}
 \wick{\phi\dd{\alpha,\textrm{out}}\dgg}
      {\phi\dd{\beta,\textrm{in}}}\urm{SU(2)}
%%%%%%%%%%%%%%%%%%%%%%%%%%
\ &= \
\contraction{}{\phi\hspace{0mm}}{\hspace{9mm}}{} 
 \phi\dd{\alpha,\textrm{out}}\dgg
 \phi\dd{\beta,\textrm{in}}
\ + \ 
\contraction[2ex]{}{\phi\hspace{0mm}}{\hspace{26mm}}{} 
 \phi\dd{\alpha,\textrm{out}}\dgg \
\contraction{}{\phi\hspace{0mm}}{\hspace{21mm}}{} 
 \mm{\phi}\dd{\mu}\dgg 
 (2\lie\dd{\mu\nu}) 
 \mm{\phi}\dd{\nu} \ 
 \phi\dd{\beta,\textrm{in}}
\ + \cdots
%%%%%%%%%%%%%%%%%%%%%%%%%%
\\ \ &= \
- \delta\dd{\alpha\beta}
\hspace{11mm} + \
 \Bigl(-\delta\dd{\alpha\mu}\Bigr) \
 2\lie\dd{\mu\nu}
 \Bigl(-\delta\dd{\mu\beta}\Bigr) 
\ + \ \cdots
%%%%%%%%%%%%%%%%%%%%%%%%%%
\\ \ &= \
 - \, I + (2\lie\trans)  - (2\lie\trans) \lie\trans +\cdots
%%%%%%%%%%%%%%%%%%%%%%%%%%
\\ \ &= \
 -\frac{I-\lie\trans}{I+\lie\trans}.
\end{align}\normalsize
\end{subequations}
Note that \en{ \lie\trans=-\lie\uu{*} } for unitary gates.
Compared to (\ref{su2iov}), we get
\begin{subequations}
\small\begin{align}
\hspace{-9mm}
 \wick{\phi\dd{\alpha,\textrm{out}}\dgg}
      {\phi\dd{\beta,\textrm{in}}}\urm{SU(2)}
\ &= \
 -\frac{I+\lie\uu{*}}{I-\lie\uu{*}}
%%%%%%%%%%%%%%%%%%%%%%%%%%
\\ \ &= \
 -\left( \frac{I+\lie}{I-\lie} \right)\uu{*}
\ = \
 -\left( 
 \wick{\phi\dd{\alpha,\textrm{out}}}
      {\phi\dd{\beta,\textrm{in}}\dgg}\urm{SU(2)}
\right)\uu{*},
\end{align}\normalsize
\end{subequations}
which represents the asymmetry of the forward traveling field.
These results are summarized as follows:

\begin{theorem}
\label{prop:bsfull}
For the SU(2) gate, contractions are given as
\begin{subequations}
\small\begin{align}
 \wick{\phi\dd{\mathrm{in}}}{\phi\dd{\mathrm{in}}\dgg}^{\mathrm{SU(2)}}
&=
 I,
\\%%%%%%%%%%%%%%%%%%%%%%%%
 \wick{\phi\dd{\mathrm{out}}}{\phi\dd{\mathrm{in}}\dgg}^{\mathrm{SU(2)}}
&=
 \tf_2
\equiv
 \ffrac{1}{1+|g|^2}
 \A{1-|g|^2}{2g}
   {-2g\uu{*}}{1-|g|^2},
\\%%%%%%%%%%%%%%%%%%%%%%%%
 \wick{\phi\dd{\mathrm{in}}}{\phi\dd{\mathrm{out}}\dgg}^{\mathrm{SU(2)}}
&=
 -I,
\\%%%%%%%%%%%%%%%%%%%%%%%%
 \wick{\phi\dd{\mathrm{out}}}{\phi\dd{\mathrm{out}}\dgg}^{\mathrm{SU(2)}}
&=
 I.
\end{align}\normalsize
\end{subequations}
Also,
\begin{subequations}
\small\begin{align}
 \wick{\mm{\phi}}{\phi\dd{\mathrm{in}}\dgg}^{\mathrm{SU(2)}}
&=
 \frac{1}{2}(I+P_2)
=
 \ffrac{1}{1+|g|^2}
 \A{1}{g}
   {-g\uu{*}}{1},
\\%%%%%%%%%%%%%%%%%%%%%%%%
  \wick{\mm{\phi}}{\mm{\phi}\dgg}^{\mathrm{SU(2)}}
&=
 \frac{1}{4}(I+P_2),
\\%%%%%%%%%%%%%%%%%%%%%%%%
 \wick{\phi\dd{\mathrm{out}}}{\mm{\phi}\dgg}^{\mathrm{SU(2)}}
&=
 \frac{1}{2}(I+P_2).
\end{align}\normalsize
\end{subequations}
These contractions satisfy
\small\begin{align}
 \wick{A\dgg}{B}\uu{\mathrm{SU(2)}}
=
 -\left(\wick{A}{B\dgg}\uu{\mathrm{SU(2)}}\right)\uu{*}.
\end{align}\normalsize
\end{theorem}

\newpage

\section{Transfer functions of non-unitary gates}

\begin{wrapfigure}[0]{r}[53mm]{49mm} 
\vspace{100mm}
\centering
\includegraphics[keepaspectratio,width=48mm]{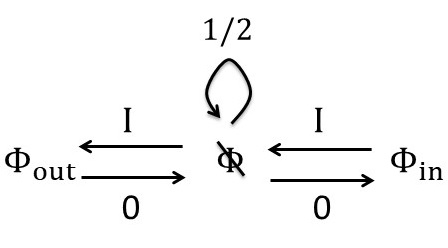}
\caption{
\small
Free field transfer functions.
\normalsize
}
\label{fig-interpretation2}
\end{wrapfigure}

In the case of the non-unitary gate,
we use the vector form
\small\begin{align}
 \Phi
\equiv
 \AV{\phi}{\phi\dgg}.
\label{vectorform}
\end{align}\normalsize
Note that 
\begin{subequations}
\label{phiddgg}
\small\begin{align}
 \Phi\dgg
&\equiv
 \AH{\phi\dgg}{(\phi\dgg)\dgg}
=
 \AH{\phi\dgg}{\phi},
\\
 \Phi\ddgg
&\equiv
 \AH{\phi\ddgg}{(\phi\dgg)\ddgg}
=
 \AH{\phi\dgg}{-\phi}
=
 \Phi\dgg \sigma_z.
\end{align}\normalsize
\end{subequations}
As in the unitary case,
we start with free field transfer functions.

\begin{lemma}
\label{lem:freevec}
The free field transfer functions of \en{ \Phi } are given by
\begin{subequations}
\label{vecinout}
\small\begin{align}
 \wick{\Phi\dd{\mathrm{out}}}{\Phi\dd{\mathrm{in}}\dgg}
& =
 \sigma_z,
\\
 \wick{\Phi\dd{\mathrm{in}}}{\Phi\dd{\mathrm{out}}\dgg}
& =
 -\sigma_z.
\end{align}\normalsize
\end{subequations}
We also have 
\begin{subequations}
\label{freevec}
\small\begin{align}
 \wick{\mm{\Phi}}{\Phi\dd{\mathrm{in}}\dgg}
&=
 \sigma_z,
%%%%%%%%%%%%%%%%%%%%%
\\
 \wick{\Phi\dd{\mathrm{out}}}{\mm{\Phi}\dgg}
&=
 \sigma_z,%\A{1}{}{}{-1},
%%%%%%%%%%%%%%%%%%%%%
\\
 \wick{\mm{\Phi}}{\mm{\Phi}\dgg}
&=
 \frac{1}{2}\sigma_z.
\\ %%%%%%%%%%%%%%
 \wick{\mm{\Phi}}{\Phi\dd{\mathrm{out}}\dgg}
&=
 0,
%%%%%%%%%%%%%%%%%%%%%
\\
 \wick{\Phi\dd{\mathrm{in}}}{\mm{\Phi}\dgg}
&=
 0.
\end{align}\normalsize
\end{subequations}
\end{lemma}

\begin{remark}
From (\ref{phiddgg}), the transfer function turns out to be
\begin{subequations}
\small\begin{align}
 \ipg\dd{\mm{\Phi} | \Phi\dd{\mathrm{in}}}
&=
 \wick{\mm{\Phi}}{\Phi\dd{\mathrm{in}}\ddgg}
\\ &=
 \wick{\mm{\Phi}}{\Phi\dd{\mathrm{in}}\dgg}\sigma_z
=
 I.
\end{align}\normalsize
\end{subequations}
 This is also interpreted as Figure \ref{fig-interpretation2}.
\end{remark}

\subsection{The input-output relations of non-unitary gates}

For non-unitary gates,
the interaction Lagrangian is given by (\ref{squeezer}):
\small\begin{align}
 \im \lag\urm{int}
=
 \mm{\Phi}\dgg \bigl(\Sigma\dd{z} \lie \bigr) \, \mm{\Phi}.
\end{align}\normalsize
A significant difference from the unitary case is that 
the interaction Lagrangian involves the square of \en{ \mm{\phi} }.
For example,
in the case of a single field,
\small\begin{align}
 \im \lag\urm{int}
=
 \AH{\mm{\phi}\dgg}{\mm{\phi}}
\left[
 \begin{array}{rr}
   \lie_{11} &  \lie_{12} \\
  -\lie_{21} & -\lie_{22}
 \end{array}
\right]
 \AV{\mm{\phi}}{\mm{\phi}\dgg}.
\end{align}\normalsize
For the diagonal elements,
the contraction is written as
\small\begin{align}
 \lie_{11}
 \contraction{\hspace{6mm}}{\phi}{\hspace{3mm}}{} 
 (\mm{\phi}\dgg \mm{\phi}) \
 \phi\drm{in}\dgg.
\end{align}\normalsize
However, for the off-diagonal elements,
there are two choices:
\small\begin{align}
 \lie_{21}
 \contraction{\hspace{4mm}}{\phi}{\hspace{4mm}}{} 
 (\mm{\phi} \mm{\phi}) \
 \phi\drm{in}\dgg
\qquad \mbox{and} \qquad
 \lie_{21}
 \contraction{\hspace{1mm}}{\phi}{\hspace{7mm}}{} 
 (\mm{\phi} \mm{\phi}) \
 \phi\drm{in}\dgg,
\end{align}\normalsize
which results in \en{ 2! }.
Accordingly,
\en{ \lie } is modified as
\begin{subequations}
\small\begin{align}
 \A{\lie_{11}}{\lie_{12}}
   {\lie_{21}}{\lie_{22}}
\to
 \left[
 \begin{array}{rr}
    \lie_{11} & 2! \lie_{12} \\
 2! \lie_{21} &    \lie_{22}
 \end{array}
 \right].
\end{align}\normalsize
\end{subequations}
This can be generalized to the multi-variable case.
The interaction Lagrangian is expressed as
\small\begin{align}
 \im \lag\urm{int}
=
 \mm{\Phi}\dgg \bigl(\Sigma\dd{z} \widetilde{\lie} \bigr) \, \mm{\Phi},
\end{align}\normalsize
where
\small\begin{align}
 \widetilde{\lie}
\equiv
 \diag(\lie) 
+ 
 2! \bigl[ \lie - \diag(\lie) \bigr].
\end{align}\normalsize

Let us consider a transfer function 
\en{ \Phi\drm{out}\gets \Phi\drm{in} }.
The expansion of the \textit{S}-matrix is 
the same as the unitary case:
\begin{subequations}
\small\begin{align}
%&
 \wick{\Phi\drm{out}}{\Phi\drm{in}\dgg}\urm{int}
=&
 \wick{\Phi\drm{out}}{\Phi\drm{in}\dgg}
\\ & 
+
%%%%%%%%%%%%%%%
 \wick{\Phi\drm{out}}{\mm{\Phi}\dgg} 
 \bigl( \Sigma\dd{z}\widetilde{\lie} \bigr) 
\Bigl\{
 I
-
 \wick{\mm{\Phi}}{\mm{\Phi}\dgg} 
 \bigl( \Sigma\dd{z}\widetilde{\lie} \bigr)
\Bigr\}\inv
 \wick{\mm{\Phi}}{\Phi\drm{in}\dgg}.
\label{cont-nonuni}
\end{align}\normalsize
\end{subequations}
Using Lemma \ref{lem:freevec},
we have
\begin{subequations}
\label{cont-nonuni-2}
\small\begin{align}
 \ipg\dd{\Phi\drm{out} | \Phi\drm{in}}\urm{int}
&=
 \wick{\Phi\drm{out}}{\Phi\drm{in}\ddgg}\urm{int}
\\ &=
 \wick{\Phi\drm{out}}{\Phi\drm{in}\dgg}\urm{int} \ \Sigma\dd{z}
=
 \frac{I+\widetilde{\lie}/2}{I-\widetilde{\lie}/2}.
\end{align}\normalsize
\end{subequations}

\subsection{Squeezing gate}

For the squeezing gate,\index{squeezing gate (single field)}
\en{ \lie } is given by (\ref{sq-reactance}):
\small\begin{align}
  \lie =  \A{}{g}{g}{}.
\end{align}\normalsize
In this case,
\small\begin{align}
 \widetilde{\lie}
=
 \A{}{2g}{2g}{}
=
 2\lie.
\end{align}\normalsize
From (\ref{cont-nonuni-2}),
the transfer function \en{ \Phi\drm{out}\gets\Phi\drm{in} }
is written as
\small\begin{align}
  \ipg\dd{\Phi\drm{out}|\Phi\drm{in}}\urm{SQ}
&=
 \frac{1+\lie}{1-\lie}
=
  \frac{1}{1-g^2}
 \A{1+g^2}{2g}{2g}{1+g^2}.
\label{ssmx-s}
\end{align}\normalsize
This is the same as (\ref{ssmx}).

\newpage

\subsection{QND gate and XX gate}

The reactance matrix of the QND gate is given by (\ref{qndg}):
\small\begin{align}
\hspace{27mm}
 \lie
&\equiv
 \frac{g}{4} \A{}{Q\dd{-}}{-Q\dd{+}}{},
\quad
 (Q\dd{\pm}\equiv
 I\pm\sigma\dd{x})
\end{align}\normalsize
for which
\small\begin{align}
\hspace{-9mm}
 \widetilde{\lie}
=
 2\lie.
\end{align}\normalsize

Let us consider a transfer function
\small\begin{align}
 \ipg\dd{\qu\drm{out}|\qu\drm{in}}\urm{QND}
=
 \wick{\qu\drm{out}}{\qu\drm{in}\ddgg}\urm{QND},
\end{align}\normalsize
where \en{ \qu } is the quadrature basis defined as
\small\begin{align}
\hspace{10mm}
 \qu
\equiv
 \AV{\xi}{\eta}
=
 \toqu
 \Phi,
\qquad
 \toqu
 \equiv
 \frac{1}{\sqrt{2}}
 \A{1}{1}
   {-\im}{\im}.
\label{defqu}
\end{align}\normalsize
Note that 
the \en{ \qu\dgg } and \en{ \qu\ddgg } 
are different (Section \ref{sec:tf}):
\begin{subequations}
\label{qudif}
\small\begin{align}
 \qu\dgg
&=
 \AH{\xi}{\eta}.
\\ 
 \qu\ddgg
&=
 \AH{-\im\eta}{\im\xi}
=
 \qu\dgg (-\sigma_y).
\end{align}\normalsize
\end{subequations}
The transfer function is therefore rewritten as
\small\begin{align}
 \ipg\dd{\qu\drm{out}|\qu\drm{in}}\urm{QND}
=
 \toqu \
 \wick{\Phi\drm{out}}{\Phi\drm{in}\dgg}\urm{QND} \
 \toqu\dgg 
 \left(-\Sigma\dd{y}\right),
\end{align}\normalsize
which can be calculated from (\ref{cont-nonuni}).
It is worth noting that 
the second and higher order corrections are zero 
for the QND gate because 
\small\begin{align}
 \lie^2
=
 0.
\end{align}\normalsize
Hence we have
\begin{subequations}
\small\begin{align}
 \wick{\Phi\drm{out}}{\Phi\drm{in}\dgg}\urm{QND}
&= 
 \wick{\Phi\drm{out}}{\Phi\drm{in}\dgg}
+
%%%%%%%%%%%%%%%
 \wick{\Phi\drm{out}}{\mm{\Phi}\dgg} 
 \bigl( \Sigma\dd{z}\widetilde{\lie} \bigr)
 \wick{\mm{\Phi}}{\Phi\drm{in}\dgg}
\\ &=
 \A{I}{\ffrac{g}{2}Q\dd{-}}
   {-\ffrac{g}{2}Q\dd{+}}{I}
 \Sigma\dd{z},
\\
 \therefore
  \wick{\qu\dd{\mathrm{out}} }{\qu\dd{\mathrm{in}}\dgg}\uu{\mathrm{QND}}
&=
   \left[
    \begin{array}{c:c} 
       -\sigma_y &  \nA {}{0}{-\im g}{} \\ \hdashline
       \nA{}{-\im g}{0}{} & -\sigma_y
    \end{array}
  \right],
\\ 
 \therefore
 \ipg\dd{\qu\dd{\mathrm{out}} | \qu\dd{\mathrm{in}}}\urm{QND}
&=
\left[
\begin{array}{cc:cc}
 1& & & \\
  & 1 & 0 & g\\ \hdashline
 -g & 0 & 1 & \\
 &  & & 1
\end{array}
\right].
\end{align}\normalsize
\end{subequations}
This transfer function is the same as (\ref{sumbounary}).

\marginpar{\vspace{-45mm}
\footnotesize
 The Pauli matrices: \\ \\
 $\sigma_x=\A{}{1}{1}{},$ \\
 $\sigma_y=\A{}{-\im}{\im}{},$ \\
 $\sigma_z=\A{1}{}{}{-1}.$ \\
\normalsize
 }

\newpage

For later use,
we summarize other components of transfer functions.

\begin{theorem}
\label{thm:sumlist}
For the QND gate, 
the contractions are given as
\begin{subequations}
\small\begin{align}
  \wick{\qu\dd{\mathrm{in}} }{\qu\dd{\mathrm{in}}\dgg}\uu{\mathrm{QND}}
&=
   \left[
    \begin{array}{c:c} 
        -\sigma_y &  \\ \hdashline
                  & -\sigma_y
    \end{array}
  \right],
\\
  \wick{\qu\dd{\mathrm{out}} }{\qu\dd{\mathrm{in}}\dgg}\uu{\mathrm{QND}}
&=
   \left[
    \begin{array}{c:c} 
       -\sigma_y &  \nA {}{0}{-\im g}{} \\ \hdashline
       \nA{}{-\im g}{0}{} & -\sigma_y
    \end{array}
  \right],
\\ \nn\\
 \wick{\qu\dd{\mathrm{out}}}{\mm{\qu}\dgg}\uu{\mathrm{QND}}
&=
 \wick{\mm{\qu}}{\qu\dd{\mathrm{in}}\dgg}\uu{\mathrm{QND}}
=
   \left[
    \begin{array}{c:c} 
       -\sigma_y &  \nA {}{0}{-\im \ffrac{g}{2}}{}\rule[0mm]{0mm}{8mm} \\ \hdashline
       \nA{}{-\im \ffrac{g}{2}}{0}{} \rule[0mm]{0mm}{8mm} & -\sigma_y
    \end{array}
  \right],
\\ \nn\\
  \wick{\qu\dd{\mathrm{in}} }{\qu\dd{\mathrm{out}}\dgg}\uu{\mathrm{QND}}
&=
   \left[
    \begin{array}{c:c} 
        \sigma_y &  \\ \hdashline
                  & \sigma_y
    \end{array}
  \right],
\\
  \wick{\qu\dd{\mathrm{out}} }{\qu\dd{\mathrm{out}}\dgg}\uu{\mathrm{QND}}
&=
   \left[
    \begin{array}{c:c} 
        -\sigma_y &  \\ \hdashline
                  & -\sigma_y
    \end{array}
  \right].
\end{align}\normalsize
\end{subequations}
Likewise,
for the XX gate,
\small\begin{align}
  \wick{\qu\dd{\mathrm{in}} }{\qu\dd{\mathrm{in}}\dgg}\uu{\mathrm{XX}}
&=
   \left[
    \begin{array}{c:c} 
        -\sigma_y &  \\ \hdashline
                  & -\sigma_y
    \end{array}
  \right],
\\
  \wick{\qu\dd{\mathrm{out}} }{\qu\dd{\mathrm{in}}\dgg}\uu{\mathrm{XX}}
&=
   \left[
    \begin{array}{c:c} 
       -\sigma_y &  \nA {0}{}{}{\im g} \\ \hdashline
       \nA{0}{}{}{\im g} & -\sigma_y
    \end{array}
  \right],
\\ \nn\\
 \wick{\qu\dd{\mathrm{out}}}{\mm{\qu}\dgg}\uu{\mathrm{XX}}
&=
 \wick{\mm{\qu}}{\qu\dd{\mathrm{in}}\dgg}\uu{\mathrm{XX}}
=
   \left[
    \begin{array}{c:c} 
       -\sigma_y &  \nA {0}{}{}{\im \ffrac{g}{2}}\rule[0mm]{0mm}{8mm} \\ \hdashline
       \nA{0}{}{}{\im \ffrac{g}{2}} \rule[0mm]{0mm}{8mm} & -\sigma_y
    \end{array}
  \right].
\\ \nn\\
  \wick{\qu\dd{\mathrm{in}} }{\qu\dd{\mathrm{out}}\dgg}\uu{\mathrm{XX}}
&=
   \left[
    \begin{array}{c:c} 
        \sigma_y &  \\ \hdashline
                  & \sigma_y
    \end{array}
  \right],
\\
  \wick{\qu\dd{\mathrm{out}} }{\qu\dd{\mathrm{out}}\dgg}\uu{\mathrm{XX}}
&=
   \left[
    \begin{array}{c:c} 
        -\sigma_y &  \\ \hdashline
                  & -\sigma_y
    \end{array}
  \right].
\end{align}\normalsize
\end{theorem}

\newpage

\section{Transfer functions of circuits}
\label{sec:circuitss}

For d-feedforward and d-feedback,
(\ref{cont-nonuni}) is not directly applicable.
Here we split the interaction Lagrangians into
two parts
and apply Theorem \ref{thm:sumlist}.

\subsection{D-feedforward}
\label{sec:ffq}

\begin{wrapfigure}[0]{r}[53mm]{49mm} 
\vspace{-12mm}
\centering
\includegraphics[keepaspectratio,width=49mm]{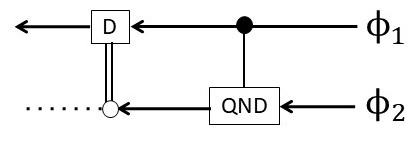}
\caption{
\small
D-feedforward.
\normalsize
}
\label{fig-feedforward-2}
\end{wrapfigure}

The interaction Lagrangian of d-feedforward 
has been given in (\ref{fflag}):
\begin{subequations}
\label{dffinlag2}
\small\begin{align}
 \lag\urm{FF}
&=
 \lag\urm{QND} + \lag\urm{int},
\\
% \lag\urm{QND}
%&= \
% :\frac{g}{4} (\eta_2+\eta_3)(\xi_1+\xi_4):,
%\\
 \lag\urm{int}
&= \
:
( 2 \mm{\eta}_1 + g \mm{\eta}_2 )
\left(-\frac{k}{2}\right)
 \xi\drm{2,out} : ,
\end{align}\normalsize
\end{subequations}
where \en{ \lag\urm{QND} } is 
the interaction Lagrangian of the QND gate,
for which the transfer functions have been obtained 
in Theorem \ref{thm:sumlist}.
Hence 
we %regard $\lag\urm{QND}$ as a free field Lagrangian and 
regard \en{ \lag\urm{int} } as a perturbation
and expand the corresponding \textit{S}-matrix.

Here we consider 
how the output of the upper line in Figure \ref{fig-feedforward-2}
\small\begin{align}
 \AVl{\xi\drm{1,out}}{\eta\drm{1,out}}
\end{align}\normalsize
is related to the inputs.
Let us start with 
a transfer function \en{ \xi\drm{1,out}\gets\xi\drm{1,in} }:
\begin{subequations}
\small\begin{align}
 \ipg\dd{\xi\drm{1,out} | \xi\drm{1,in}}\urm{FF}
&=
 \wick{\xi\drm{1,out}}{\xi\drm{1,in}\ddgg}\urm{FF}
\\ &=
 \wick{\xi\drm{1,out}}{-\im \eta\drm{1,in}}\urm{FF}.
\end{align}\normalsize
\end{subequations}
Expanding the \textit{S}-matrix 
and using Theorem \ref{thm:sumlist}, 
we have 
\begin{subequations}
\small\begin{align}
\left(
 0\urm{th}\mbox{ order of } 
 \ipg\dd{\xi\drm{1,out} | \xi\drm{1,in}}\urm{FF}\right)
&=
 \wick{\xi\drm{1,out}}{-\im \eta\drm{1,in}}\urm{QND}
=
 1,
\\ 
\left(
 1\urm{st}\mbox{ order of } 
 \ipg\dd{\xi\drm{1,out} | \xi\drm{1,in}}\urm{FF}\right)
&=
\contraction{}{\phi a}{\hspace{16mm}}{\hspace{0mm}}
 \xi\drm{1,out} : (2\mm{\eta}_1 + g\mm{\eta}_2 )
 \left(-\im\frac{k}{2}\right)
\contraction{}{\phi a}{\hspace{12mm}}{\hspace{0mm}}
 \xi\drm{2,out} : -\im \eta\drm{1,in}
\nn \\ &= 
\underbrace{
 \wick{\xi\drm{1,out}}{2\mm{\eta}_1 }\urm{QND}
}\dd{\hspace{5mm} =2\im}
 \left(-\im\frac{k}{2}\right)
\underbrace{
 \wick{\xi\drm{2,out}}{-\im \eta\drm{1,in}}\urm{QND}
}\dd{\hspace{5mm} =-g}
\nn\\ &=
 -gk.
\\
\left(
 2\urm{nd} \mbox{ order of } 
 \ipg\dd{\xi\drm{1,out} | \xi\drm{1,in}}\urm{FF}\right)
& 
 \nn\\ &  \hspace{-30mm} =
 \wick{\xi\drm{1,out}}{2\mm{\eta}_1 }\urm{QND}
 \left(-\im\frac{k}{2}\right)
\underbrace{
 \wick{\xi\drm{2,out}}{2\mm{\eta}_1 + g\mm{\eta}_2}\urm{QND}
}\dd{ \hspace{5mm} =-\im g + \im g }
 \left(-\im\frac{k}{2}\right)
 \wick{\xi\drm{2,out}}{-\im \eta\drm{1,in}}\urm{QND}
 \nn\\ &  \hspace{-30mm} =
 0.
\end{align}\normalsize
\end{subequations}
Likewise, higher order corrections are zero.
As a result, we get
\small\begin{align}
 \ipg\dd{\xi\drm{1,out} |\xi\drm{1,in}}\urm{FF}
&=
 1-gk.
\end{align}\normalsize

In the case of 
a transfer function \en{ \xi\drm{1,out}\gets \xi\drm{2,in} },
only the first order correction is nonzero:
\small\begin{align}
\ipg\dd{\xi\drm{1,out} |\xi\drm{2,in}}\urm{FF}
&=
\underbrace{
 \wick{\xi\drm{1,out}}{2\mm{\eta}_1 }\urm{QND}
}\dd{\hspace{5mm} =2\im}
 \left(-\im\frac{k}{2}\right)
\underbrace{
 \wick{\xi\drm{2,out}}{-\im\eta\drm{2,in}}\urm{QND}
}\dd{\hspace{5mm} =1}
=
 k.
\end{align}\normalsize
It is also easy to show that 
\begin{subequations}
\small\begin{align}
 \ipg\dd{\eta\drm{1,out} | \eta\drm{1,in}}\urm{FF} 
&= 
 \ipg\dd{\eta\drm{1,out} | \eta\drm{1,in}}\urm{QND} = 1,
\\
 \ipg\dd{\eta\drm{1,out} | \eta\drm{2,in}}\urm{FF} 
&= 
 \ipg\dd{\eta\drm{1,out} | \eta\drm{2,in}}\urm{QND} = g.
\end{align}\normalsize
\end{subequations}
Using these results, 
we can express the input-output relation as
\small\begin{align}
\AV{\xi_1}{\eta_1}\drm{out} 
&=
\left[
 \begin{array}{cc:cc}
  (1-gk) & 0 & k & 0 \\ \hdashline
  0 & 1 & g & 0
 \end{array}
\right]
\left[
\begin{array}{c}
 \xi_1 \\
 \eta_1 \\ \hdashline
 \xi_2 \\
 \eta_2
\end{array}
\right]\drm{in},
\label{sffio}
\end{align}\normalsize
which is the same as the classical result (\ref{ffio}).

\subsection{D-feedback}
\label{sec:fbq}

\begin{wrapfigure}[0]{r}[53mm]{49mm} 
\vspace{-12mm}
\centering
\includegraphics[keepaspectratio,width=49mm]{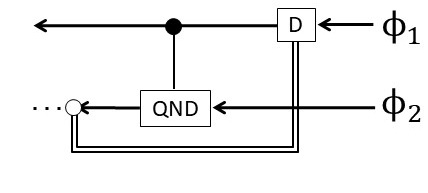}
\caption{
\small
D-feedback.
\normalsize
}
\end{wrapfigure}

The interaction Lagrangian of d-feedback 
has been given in (\ref{fblag}):
\begin{subequations}
\label{dfbinlag2}
\small\begin{align}
 \lag\urm{FB}
&=
 \lag\urm{QND} + \lag\urm{int},
\\
 \lag\urm{int}
&= \
:
( 2 \mm{\eta}_1 - g \mm{\eta}_2 )
\left(-\frac{k}{2}\right)
 \xi\drm{2,out} : ,
\label{dfbinlag2-2}
\end{align}\normalsize
\end{subequations}
Compared to (\ref{dffinlag2}),
a difference is only the sign of the second term in (\ref{dfbinlag2-2}).
This leads to a significant difference 
in the expansion of the \textit{S}-matrix.

Again, we consider 
how the output of the upper line in Figure \ref{fig-feedforward-2}
\small\begin{align}
 \AVl{\xi\drm{1,out}}{\eta\drm{1,out}}
\end{align}\normalsize
is related to the inputs.
Let us start with 
a transfer function \en{ \xi\drm{1,out}\gets\xi\drm{1,in} }:
\begin{subequations}
\small\begin{align}
\left(
 0\urm{th}\mbox{ order of } 
 \ipg\dd{\xi\drm{1,out} | \xi\drm{1,in}}\urm{FB}\right)
&=
 \wick{\xi\drm{1,out}}{-\im \eta\drm{1,in}}\urm{QND}
=
 1,
%%%%%%%%%%%%%%%%%%%%%%%%%%%%%%%%%%%%%%%
\\ 
\left(
 1\urm{st}\mbox{ order of } 
 \ipg\dd{\xi\drm{1,out} | \xi\drm{1,in}}\urm{FB}\right)
&=
\underbrace{
 \wick{\xi\drm{1,out}}{2\mm{\eta}_1 }\urm{QND}
}\dd{\hspace{5mm} =2\im}
 \left(-\im\frac{k}{2}\right)
\underbrace{
 \wick{\xi\drm{2,out}}{-\im \eta\drm{1,in}}\urm{QND}
}\dd{\hspace{5mm} =-g}
\nn\\ &=
 k(-g),
%%%%%%%%%%%%%%%%%%%%%%%%%%%%%%%%%%%%%%%
\\
\left(
 2\urm{nd}\mbox{ order of } 
 \ipg\dd{\xi\drm{1,out} | \xi\drm{1,in}}\urm{FB}\right)
& 
 \nn\\ &  \hspace{-30mm} =
 \wick{\xi\drm{1,out}}{2\mm{\eta}_1 }\urm{QND}
 \left(-\im\frac{k}{2}\right)
\underbrace{
 \wick{\xi\drm{2,out}}{2\mm{\eta}_1 - g\mm{\eta}_2}\urm{QND}
 \left(-\im\frac{k}{2}\right)
}\dd{ \hspace{9mm} \equiv \isel }
 \wick{\xi\drm{2,out}}{-\im \eta\drm{1,in}}\urm{QND}
 \nn\\ &  \hspace{-30mm} =
 k  (\isel) (-g),
\end{align}\normalsize
\end{subequations}
where 
\small\begin{align}
 \isel
\equiv
\underbrace{
 \wick{\xi\drm{2,out}}{2\mm{\eta}_1 - g\mm{\eta}_2}\urm{QND}
}\dd{ \hspace{5mm} =-\im g - \im g }
 \left(-\im\frac{k}{2}\right)
=
 -gk.
\end{align}\normalsize
Note that 
\en{ \isel=0 } for d-feedforward % of the preceding section
because of the different sign.
Likewise, 
the third order correction is given as \en{ k(\isel)^2(-g) }.
As a result, 
we have
\begin{subequations}
\small\begin{align}
 \ipg\dd{\xi\drm{1,out} | \xi\drm{1,in}}\urm{FB}
&=
 1 + k\left\{ 1+ (\isel) + (\isel)^2 + \cdots \right\}(-g)
\\ &\sim
 1 + k \frac{1}{1-\isel} (-g).
\end{align}\normalsize
\end{subequations}
A transfer function \en{ \xi\drm{1,out}\gets \xi\drm{2,in} }
is obtained in the same way:
\begin{subequations}
\small\begin{align}
\left(
 1\urm{st}\mbox{ order of } 
 \ipg\dd{\xi\drm{1,out} | \xi\drm{1,in}}\urm{FB}\right)
&=
\underbrace{
 \wick{\xi\drm{1,out}}{2\mm{\eta}_1 }\urm{QND}
}\dd{\hspace{5mm} =2\im}
 \left(-\im\frac{k}{2}\right)
\underbrace{
 \wick{\xi\drm{2,out}}{-\im\eta\drm{2,in}}\urm{QND}
}\dd{\hspace{5mm} =1}.
%%%%%%%%%%%%%%%%%%%%%%%%%%%%%%%%%%%%%%%
\\
\left(
 2\urm{nd}\mbox{ order of } 
 \ipg\dd{\xi\drm{1,out} | \xi\drm{1,in}}\urm{FB}\right)
& 
 \nn\\ &  \hspace{-30mm} =
 \wick{\xi\drm{1,out}}{2\mm{\eta}_1 }\urm{QND}
 \left(-\im\frac{k}{2}\right)
\underbrace{
 \wick{\xi\drm{2,out}}{2\mm{\eta}_1 - g\mm{\eta}_2}\urm{QND}
 \left(-\im\frac{k}{2}\right)
}\dd{ \hspace{9mm} \equiv \isel }
 \wick{\xi\drm{2,out}}{-\im\eta\drm{2,in}}\urm{QND}.
\end{align}\normalsize
\end{subequations}
This results in
\small\begin{align}
 \ipg\dd{\xi\drm{1,out} | \xi\drm{2,in} }\urm{FB}
&=
 k\frac{1}{1-\isel}.
\end{align}\normalsize

Other transfer functions do not involve \en{ \isel } 
hence 
they are the same as d-feedforward.
As a result,
the input-output relation of the d-feedback is written as
\small\begin{align}
\AV{\xi_1}{\eta_1}\drm{out} 
&=
\left[
 \begin{array}{cc:cc}
  \ffrac{1}{1+gk} & 0 & \ffrac{k}{1+gk} & 0 \\ \hdashline
  0 & 1 & g & 0
 \end{array}
\right]
\left[
\begin{array}{c}
 \xi_1 \\
 \eta_1 \\ \hdashline
 \xi_2 \\
 \eta_2
\end{array}
\right]\drm{in},
\label{sfbio}
\end{align}\normalsize
which is the same as the classical result (\ref{fbtfc}).

\begin{remark}
The difference between feedforward and feedback
is the fractional function
\small\begin{align}
 \frac{1}{1+\isel}.
\end{align}\normalsize
In fact, 
the results of d-feedforward is obtained by 
setting \en{ \isel=0 } in d-feedback.
This factor represents the effect that 
the signal propagates in the feedback loop infinite times.
In the case of d-feedforward,
there are no loops and hence \en{ \isel=0 }.
We will examine feedback in detail 
in Chapter \ref{chap:additional}
where \en{ \isel } turns out to be self-energy.
\end{remark}

\newpage

\section{Concluding remarks: Feynman diagrams}

So far, 
we have not drawn any Feynman diagrams
because all the examples we have considered were linear
and 
it is not beneficial to draw diagrams for linear interactions.
However, 
it is still educational to show it 
because there is a unique rule for the forward traveling field
due to its unidirectionality.

\begin{wrapfigure}[0]{r}[53mm]{49mm} 
\vspace{-17mm}
\centering
\includegraphics[keepaspectratio,width=40mm]{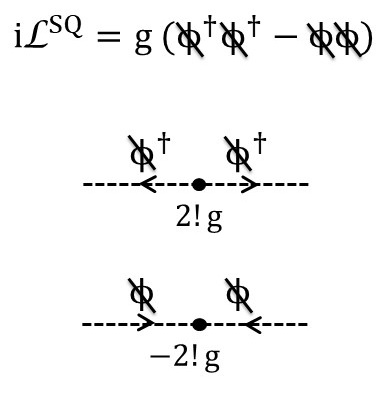}
\caption{
\small
Elementary diagrams corresponding to
the first and second term of (\ref{sstfdiag}).
\normalsize
}
\label{fig-ss-ele}
\end{wrapfigure}

Here 
we consider the squeezing gate.\index{squeezing gate (single field)}
The interaction Lagrangian is given as
\small\begin{align}
 \im \lag\urm{SQ}
&=
 (2!g) \left[ \mm{\phi}\dgg \mm{\phi}\dgg  - \mm{\phi} \mm{\phi} \right].
\label{sstfdiag}
\end{align}\normalsize
The corresponding diagrams are depicted in Figure \ref{fig-ss-ele}.
\en{ \mm{\phi} } and \en{ \mm{\phi}\dgg } are
represented by incoming and outgoing arrows, respectively.
Let us consider a transfer function 
\small\begin{align}
 \ipg\dd{\phi\drm{out} | \phi\drm{in}}\urm{SQ}
=
 \wick{\phi\drm{out}}{\phi\drm{in}\dgg}\urm{SQ}.
\end{align}\normalsize
The zeroth order term is 
a direct-through diagram \en{ \phi\drm{out}\gets \phi\drm{in}\dgg },
which is depicted as

\vspace{-1mm}
\begin{figure}[H]
\centering
\includegraphics[keepaspectratio,width=49mm]{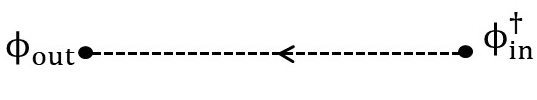}
\end{figure}
\vspace{-1mm}

\noindent
This is a free field transfer function
\en{ \wick{\phi\drm{out}}{\phi\drm{in}\dgg }= 1 }.
The first order correction is depicted as

\vspace{-2mm}
\begin{figure}[H]
\centering
\includegraphics[keepaspectratio,width=49mm]{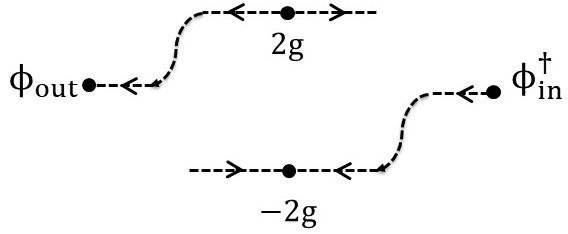}
\end{figure}
\vspace{-3mm}

\noindent
\en{ \phi\drm{out} } is not connected to \en{ \phi\drm{in}\dgg }
through the interaction Lagrangian.
This diagram is disconnected and hence ignored.
The second order correction is depicted as

\vspace{0mm}
\begin{figure}[H]
\centering
\includegraphics[keepaspectratio,width=75mm]{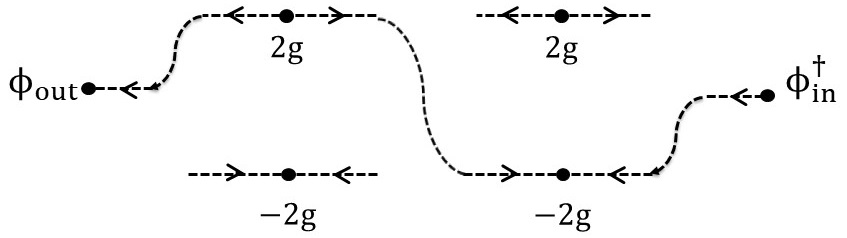}
\end{figure}
\vspace{-3mm}

\begin{wrapfigure}[0]{r}[53mm]{49mm} 
\vspace{9mm}
\centering
\includegraphics[keepaspectratio,width=48mm]{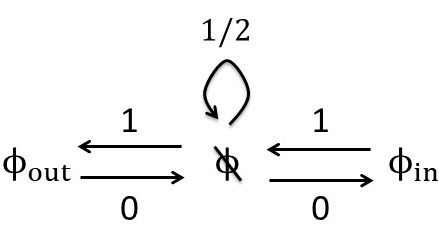}
\caption{
\small
Free field transfer functions.
\normalsize
}
\label{fig-interpretation-2}
\end{wrapfigure}

\noindent
This is a connected diagram.
Note that the middle arrow is pointing backward.
As explained in Section \ref{lem:asymmetry},
we need to reverse it 
using the asymmetry of the forward traveling field.
This is re-depicted as

\vspace{-3mm}
\begin{figure}[H]
\centering
\includegraphics[keepaspectratio,width=50mm]{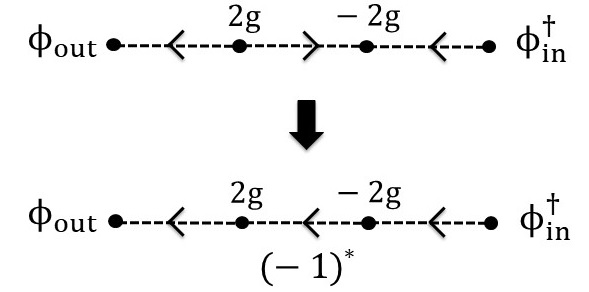}
\end{figure}
\vspace{-2mm}

\noindent
where \en{ (-1)\uu{*} } represents 
reversing the the arrow.
Recall that 
the free field transfer functions (Lemma \ref{lem:mmtrans}) 
are given as Figure \ref{fig-interpretation-2}.
Then this diagram is given as
\small\begin{align}
 (1) \ (2g) \ \left(-\frac{1}{2}\right)\uu{*} (-2g) \ (1) = 2g^2. 
\end{align}\normalsize
This is equivalent to the following calculation:
\begin{subequations}
\small\begin{align}
 \wick{\phi\drm{out}}{\mm{\phi}}
 \ (2g) 
\underbrace{
 \wick{\mm{\phi}\dgg}{ \mm{\phi} }
}\dd{= - \wick{\mm{\phi}}{\mm{\phi}\dgg}\uu{*}
}
 (-2g) \
 \wick{\mm{\phi}}{\phi\drm{in}\dgg}
 =
 2g^2.
\end{align}\normalsize
\end{subequations}
Likewise, the fourth order correction is given by \en{ 2g^4 }.
As a result, 
\small\begin{align}
 \ipg\dd{\phi\drm{out} | \phi\drm{in}}\urm{SQ}
&=
 1 + 2g^2(1+g^2+\cdots)
=
 \frac{1+g^2}{1-g^2},
\end{align}\normalsize
which is the same as the (1,1)-element of (\ref{ssmx-s}),
as expected.
It is important to note that 
the correct result is obtained 
by putting all arrows in the forward direction.

\chapter{Quantum systems via \textit{S}-matrices}
\label{sec:tfs}
\thispagestyle{fancy}

In this chapter,
we derive the transfer functions of quantum systems 
using \textit{S}-matrices 
as we have done for quantum gates in the preceding chapter.
The advantage of the \textit{S}-matrix approach is that 
various techniques developed in field theory 
can be employed to investigate nonlinear dynamics.
For example, 
third-order nonlinear interactions are examined 
with the \en{ \phi^4 } model of the scalar field 
in Chapter \ref{chap:nonlinear}.
Fermion-boson interactions are analyzed 
in the same way as Yukawa's interaction
in Chapter \ref{chap:intspin}.
Gravitational wave detection are considered
using SU(2) systems in Chapter \ref{chap:gravity}.
All of these applications are based on the formulation of this chapter.

\section{SU(2) system}
\label{sec:su2system}

\subsection{Free field transfer functions}

\begin{wrapfigure}[0]{r}[53mm]{49mm} 
\vspace{-60mm}
\centering
\includegraphics[keepaspectratio,width=27mm]{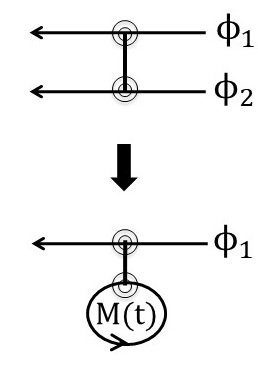}
\caption{
\small
Classical formulation of the SU(2) system.
\normalsize
}
\label{fig-cavity-cl}
\end{wrapfigure}

In the classical formulation (Chapter \ref{chap:feedback}),
the SU(2) system was defined by connecting 
the output \en{ \phi\drm{2,out} } and the input \en{ \phi\drm{2,in} } 
across the SU(2) gate as in Figure \ref{fig-cavity-cl}.
In the \textit{S}-matrix approach,
we interpret it in a different way.
The system is defined by the interaction between
a closed-loop field \en{ \mas } and 
a free field \en{ \phi } through the SU(2) gate,
as in Figure \ref{fig-cavity}.

\begin{wrapfigure}[0]{r}[53mm]{49mm} 
\vspace{-15mm}
\centering
\includegraphics[keepaspectratio,width=25mm]{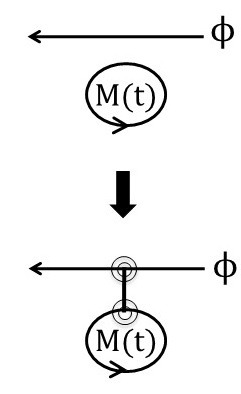}
\caption{
\small
\textit{S}-matrix approach to the SU(2) system.
\normalsize
}
\label{fig-cavity}
\end{wrapfigure}

To calculate the input-output relation of the system,
we need free field transfer functions (contractions) for a field vector
\small\begin{align}
 \phi\dd{\mas}
\equiv
 \AV{\mm{\phi}}{\mas}.
\end{align}\normalsize

\begin{lemma}
\label{lem:tfbefore}
 For \en{ \phi\dd{\mas} },
 free field contractions are given as
\begin{subequations}
\small\begin{align}
 \wick{\phi\dd{\mas}}{\phi\dd{\mathrm{in}}\dgg}
&=
 \AV{1}{0},
& %\hspace{24mm}
 \wick{\phi\dd{\mas}\dgg}{\phi\dd{\mathrm{in}}}
&=
 \AH{-1}{0},
%%%%%%%%%%%%%%%%%%%%%%%%%%%%%%%%%
\\
 \wick{\phi\dd{\mathrm{out}}}{\phi\dd{\mas}\dgg}
&=
 \AH{1}{0},
& %\hspace{18mm}
 \wick{\phi\dd{\mathrm{out}}\dgg}{\phi\dd{\mas}}
&=
 \AV{-1}{0},
%%%%%%%%%%%%%%%%%%%%%%%%%%%%%%%%%
\\
 \wick{\phi\dd{\mas}}{\phi\dd{\mas}\dgg}
&=
 \A{\ffrac{1}{2}}{}
   {}{\wick{\mas}{\mas\dgg}},
& %\hspace{2mm}
 \wick{\phi\dd{\mas}\dgg}{\phi\dd{\mas}}
&=
 -\A{\ffrac{1}{2}}{}
   {}{\wick{\mas}{\mas\dgg}},
%%%%%%%%%%%%%%%%%%%%%%%%%%%%%%%%%
\\
 \wick{\phi\dd{\mas}}{\phi\dd{\mathrm{out}}\dgg}
&=
 0,
& %\hspace{32mm}
 \wick{\phi\dd{\mas}\dgg}{\phi\dd{\mathrm{out}}}
&=
 0,
%%%%%%%%%%%%%%%%%%%%%%%%%%%%%%%%%
\\
 \wick{\phi\dd{\mathrm{in}}}{\phi\dd{\mas}\dgg}
&=
 0,
& %\hspace{34mm}
 \wick{\phi\dd{\mathrm{in}}\dgg}{\phi\dd{\mas}}
&=
 0,
\end{align}\normalsize
\end{subequations}
where \en{ \wick{\mas}{\mas\dgg} } is given in Section \ref{sec:ltf} as
\small\begin{align}
 \wick{\mas}{\mas\dgg} 
=
\kakkon{\ffrac{1}{s}, \hspace{19mm} 
        \mbox{\normalsize(single mode)\small}}
       {\ffrac{1}{2}\ffrac{1+\ex\uu{-sl}}{1-\ex\uu{-sl}}. \hspace{4mm}
        \mbox{\normalsize(infinite mode)\small}}
\label{ftf}
\end{align}\normalsize
\end{lemma}
\begin{proof}
 All of these can be obtained from Lemma \ref{lem:mmtrans}.
For example,
\begin{subequations}
\small\begin{align}
 \wick{\phi\dd{\mathrm{out}}}{\phi\dd{\mas}\dgg}
&=
 \AH{ \wick{\phi\drm{out}} {\mm{\phi}\dgg}}
    { \wick{\phi\drm{out}} {\mas\dgg}}
\\ &=
 \AH{1}{0},
\end{align}\normalsize
\end{subequations}
where \en{ 0 } in the second element results from 
the fact that 
\en{ \phi } and \en{ \mas } are independent in free space.
\end{proof}

\subsection{Transfer function $\phi_{\mathrm{out}} \gets \phi_{\mathrm{in}}$}
\label{sec:1to4}

For unitary systems,
the interaction Lagrangian is given by (\ref{fham}):
\small\begin{align}
 \im \lag\urm{int}_{\mas}
&=
  \phi\dd{\mas}\dgg (2\lie) \phi\dd{\mas},
\label{fham-1-2}
\end{align}\normalsize
The transfer functions of quantum systems are 
obtained in the same way as quantum gates (\ref{su2gatetf}).
In the case of a transfer function \en{ \phi\drm{out} \gets \phi\drm{in} },
we have
\begin{subequations}
\label{su2sysgen}
\small\begin{align}
 \ipg\dd{\phi\drm{out} | \phi\drm{in}}\urm{int}
&=
 \wick{\phi\drm{out}}{\phi\drm{in}\dgg}\urm{int}
\\ &=
 \wick{\phi\drm{out}}{\phi\drm{in}\dgg}
+
 \wick{\phi\drm{out}}{\phi\dd{\mas}\dgg} (2\lie) 
 \Bigl\{
 1- \wick{\phi\dd{\mas}}{\phi\dd{\mas}\dgg}  (2\lie) 
\Bigr\}\inv
 \wick{\phi\dd{\mas}}{\phi\drm{in}\dgg}.
\end{align}\normalsize
\end{subequations}

For the SU(2) system,
\small\begin{align}
 \lie
=
 \A{}{g}
   {-g\uu{*}}{}.
\end{align}\normalsize
Using Lemma \ref{lem:tfbefore}, 
we get
\begin{subequations}
\label{41ms-2}
\small\begin{align}
 \wick{\phi\drm{out}}{\phi\drm{in}\dgg}\urm{SU(2)}
& =
 \frac{1-2|g|^2\wick{\mas}{\mas\dgg}}{1+2|g|^2\wick{\mas}{\mas\dgg}}
\label{41ms-2-1} \\ &=
\kakkon{
 \ffrac{s-2|g|^2}{s+2|g|^2}
=
 \dtf{-2|g|^2}{-2g\uu{*}}{2g}{1},
\hspace{3mm} \mbox{(single mode)}}
 {
 \ffrac{1-|g|^2 - \ex\uu{-sl}(1+|g|^2)}
       {1+|g|^2 - \ex\uu{-sl}(1-|g|^2)}.
\hspace{15mm} \mbox{(infinite mode)}}
\end{align}\normalsize
\end{subequations}
The single-mode transfer function is 
the same as the classical result (\ref{su2d}).
Note that 
the single-mode system is also obtained 
as a special case of the infinite mode system
\small\begin{align}
 \ffrac{1-|g|^2 - \ex\uu{-sl}(1+|g|^2)}
       {1+|g|^2 - \ex\uu{-sl}(1-|g|^2)}
\ \sim \
 \frac{s - (2|g|^2/l)}{s + (2|g|^2/l)},
\hspace{10mm}
 (l \ll 1).
\end{align}\normalsize

\subsection{Transfer function $\phi_{\mathrm{out}}\dgg \gets \phi_{\mathrm{in}}\dgg$}

Let us consider a transfer function
\begin{subequations}
\label{su2sdag}
\small\begin{align}
\ipg\dd{\phi\drm{out}\dgg|\phi\drm{in}\dgg}\urm{int}
&=
 \wick{\phi\drm{out}\dgg}{(\phi\drm{in}\dgg)\ddgg}\urm{int}
\\ &=
 - \wick{\phi\drm{out}\dgg}{\phi\drm{in}}\urm{int}.
\end{align}\normalsize
\end{subequations}
Basically this is calculated in the same way as the preceding case:
\begin{subequations}
\small\begin{align}
\left(
 0\urm{th}\mbox{order of }
 \wick{\phi\drm{out}\dgg}{\phi\drm{in}}\urm{int}
\right)
&= 
 \wick{\phi\drm{out}\dgg}{\phi\drm{in}}
=
 -1.
%%%%%%%%%%%%%%%%%%%%%%%%%%%%%%%%%%%%%%%%%
\\ 
 \left(
 1\urm{st}\mbox{ order of }
 \wick{\phi\drm{out}\dgg}{\phi\drm{in}}\urm{int}
\right)
&=
\contraction{}{\phi}{\hspace{21mm}}{\hspace{0mm}}
 \phi\drm{out}\dgg \ (
\contraction[2ex]{}{\phi}{\hspace{19mm}}{\hspace{0mm}}
\phi\dd{\mas}\dgg \ 2\lie \ \phi\dd{\mas}) \
 \phi\drm{in}
\\ &=
 \wick{\phi\drm{out}\dgg}{\phi\dd{\mas}}\trans  (2\lie) \
 \wick{\phi\dd{\mas}\dgg}{\phi\drm{in}}\trans,
%%%%%%%%%%%%%%%%%%%%%%%%%%%%%%%%%%%%%%%%%
\\ 
 \left(
 2\urm{nd}\mbox{ order of }
 \wick{\phi\drm{out}\dgg}{\phi\drm{in}}\urm{int}
\right)
\nn \\ & \hspace{-15mm} =
 2!\frac{1}{2!} 
\contraction{}{\phi}{\hspace{21mm}}{\hspace{0mm}}
 \phi\drm{out}\dgg \ ( 
\contraction[2ex]{}{\phi}{\hspace{35mm}}{\hspace{0mm}}
\phi\dd{\mas}\dgg \ 2\lie \ \phi\dd{\mas}) \ (
\contraction{}{\phi}{\hspace{21mm}}{\hspace{0mm}}
 \phi\dd{\mas}\dgg (2\lie) \phi\dd{\mas} ) \
 \phi\drm{in}
\\ & \hspace{-15mm} =
 \wick{\phi\drm{out}\dgg}{\phi\dd{\mas}}\trans  (2\lie) \
 \wick{\phi\dd{\mas}\dgg}{\phi\dd{\mas}}  \ (2\lie) \ 
 \wick{\phi\dd{\mas}\dgg}{\phi\drm{in}}\trans.
\end{align}\normalsize
\end{subequations}
As a result,
\small\begin{align}
 \wick{\phi\drm{out}\dgg}{\phi\drm{in}}\urm{int}
 = &
%%%%%%%%%%%%%%%%%%%%%%%
 \wick{\phi\drm{out}\dgg}{\phi\drm{in}}
\nn \\ & +   
 \wick{\phi\drm{out}\dgg}{\phi\dd{\mas}}\trans  (2\lie) 
\Bigl\{
 1-\wick{\phi\dd{\mas}\dgg}{\phi\dd{\mas}} (2G)
\Bigr\}\inv
  \wick{\phi\dd{\mas}\dgg}{\phi\drm{in}}\trans.
\label{su2tfdag}
\end{align}\normalsize

For the SU(2) system
\small\begin{align}
 \lie
=
 \A{}{g}
   {-g\uu{*}}{},
\end{align}\normalsize
(\ref{su2tfdag}) is written as
\small\begin{align}
 \wick{\phi\drm{out}\dgg}{\phi\drm{in}}\urm{SU(2)}
& =
  - \, \frac{1-2|g|^2\wick{\mas}{\mas\dgg}}{1+2|g|^2\wick{\mas}{\mas\dgg}}.
\end{align}\normalsize
Compared to (\ref{41ms-2-1}),
we get the same asymmetry as the quantum gates:
\small\begin{align}
 \wick{\phi\drm{out}\dgg}{\phi\drm{in}}\urm{SU(2)}
=
  - \,  \wick{\phi\drm{out}}{\phi\drm{in}\dgg}\urm{SU(2)}.
\end{align}\normalsize
The transfer function (\ref{su2sdag}) is therefore given as
\begin{subequations}
\small\begin{align}
\ipg\dd{\phi\drm{out}\dgg|\phi\drm{in}\dgg}\urm{SU(2)}
&=
 - \wick{\phi\drm{out}\dgg}{\phi\drm{in}}\urm{SU(2)}
\\ &=
  \wick{\phi\drm{out}}{\phi\drm{in}\dgg}\urm{SU(2)}
\\ &=
 \ipg\dd{\phi\drm{out}|\phi\drm{in}}\urm{SU(2)}.
\end{align}\normalsize
\end{subequations}
For unitary systems,
the transfer functions 
\en{ (\phi\drm{out}     \gets \phi\drm{in}) } and 
\en{ (\phi\drm{out}\dgg \gets \phi\drm{in}\dgg) } 
are always equivalent due to this asymmetry.

\newpage

\subsection{Summary}
\label{sec:su2summ}

For the single-mode SU(2) system,
transfer functions are given as follows:

\begin{theorem}
\label{thm:tfbs}
The contractions of the single-mode SU(2) system are given as
\small\begin{align}
\renewcommand{\arraystretch}{1.3}
\begin{array}{ll}
 \wick{\phi\dd{\mathrm{out}}}{\phi\dd{\mathrm{in}}\dgg}^{\mathrm{SU(2)}}
=
 \dtf{-2|g|^2}{-2g\uu{*}}{2g}{1},
&
 \wick{\phi\dd{\mathrm{out}}\dgg}{\phi\dd{\mathrm{in}}}^{\mathrm{SU(2)}}
=
 -\dtf{-2|g|^2}{-2g\uu{*}}{2g}{1},
%%%%%%%%%%%%%%%%%%%%%%%%%%%%%%%%%
\\
 \wick{\mas}{\phi\dd{\mathrm{in}}\dgg}^{\mathrm{SU(2)}}
=
 \dtf{-2|g|^2}{-2g\uu{*}}{1}{0},
&
 \wick{\mas\dgg}{\phi\dd{\mathrm{in}}}^{\mathrm{SU(2)}}
 =
 \dtf{-2|g|^2}{1}{2g}{0},
\rule[0mm]{0mm}{9mm}
%%%%%%%%%%%%%%%%%%%%%%%%%%%%%%%%%
\\
 \wick{\phi\dd{\mathrm{in}}}{\phi\dd{\mathrm{in}}\dgg}^{\mathrm{SU(2)}}
=
 1,
&
 \wick{\phi\dd{\mathrm{in}}\dgg}{\phi\dd{\mathrm{in}}}^{\mathrm{SU(2)}}
=
-1,
\\ \\
%%%%%%%%%%%%%%%%%%%%%%%%%%%%%%%%%
%%%%%%%%%%%%%%%%%%%%%%%%%%%%%%%%%
 \wick{\phi\dd{\mathrm{out}}}{\mas\dgg}^{\mathrm{SU(2)}}
=
 \dtf{-2|g|^2}{1}{2g}{0},
&
 \wick{\phi\dd{\mathrm{out}}\dgg}{\mas}^{\mathrm{SU(2)}}
=
 \dtf{-2|g|^2}{-2g\uu{*}}{1}{0},
%%%%%%%%%%%%%%%%%%%%%%%%%%%%%%%%%
\\
 \wick{\mas}{\mas\dgg}^{\mathrm{SU(2)}}
=
 \dtf{-2|g|^2}{1}{1}{0},
&
 \wick{\mas\dgg}{\mas}^{\mathrm{SU(2)}}
=
-\dtf{-2|g|^2}{1}{1}{0},
\rule[0mm]{0mm}{9mm}
%%%%%%%%%%%%%%%%%%%%%%%%%%%%%%%%%
\\
 \wick{\phi\dd{\mathrm{in}}}{\mas\dgg}^{\mathrm{SU(2)}}
=
 0,
&
 \wick{\phi\dd{\mathrm{in}}\dgg}{\mas}^{\mathrm{SU(2)}}
=
 0,
\\ \\
%%%%%%%%%%%%%%%%%%%%%%%%%%%%%%%%%
%%%%%%%%%%%%%%%%%%%%%%%%%%%%%%%%%
 \wick{\phi\dd{\mathrm{out}}}{\phi\dd{\mathrm{out}}\dgg}^{\mathrm{SU(2)}}
=
 1,
&
 \wick{\phi\dd{\mathrm{out}}\dgg}{\phi\dd{\mathrm{out}}}^{\mathrm{SU(2)}}
=
 -1,
%%%%%%%%%%%%%%%%%%%%%%%%%%%%%%%%%
\\
 \wick{\mas}{\phi\dd{\mathrm{out}}\dgg}^{\mathrm{SU(2)}}
=
 0,
&
 \wick{\mas\dgg}{\phi\dd{\mathrm{out}}}^{\mathrm{SU(2)}}
=
 0,
%%%%%%%%%%%%%%%%%%%%%%%%%%%%%%%%%
\\
 \wick{\phi\dd{\mathrm{in}}}{\phi\dd{\mathrm{out}}\dgg}^{\mathrm{SU(2)}}
=
 -1,
&
 \wick{\phi\dd{\mathrm{in}}\dgg}{\phi\dd{\mathrm{out}}}^{\mathrm{SU(2)}}
=
 1.
\end{array}
\renewcommand{\arraystretch}{1}
\label{tfbs}
\end{align}\normalsize
These satisfy
\small\begin{align}
 \wick{A\dgg}{B}\uu{\mathrm{SU(2)}}
=
 -\left(\wick{A}{B\dgg}\uu{\mathrm{SU(2)}}\right)\uu{*}.
\end{align}\normalsize
\end{theorem}

\begin{remark}
In this section, 
we have demonstrated single-mode and infinite-mode systems.
A finite-multimode system can be obtained in the same way.
Consider 
\small\begin{align}
 \wick{\mas}{\mas\dgg}
&=
 \frac{1}{l}\sum_{n=1}^N \frac{1}{s+\im \omega_n},
\end{align}\normalsize
where \en{ \omega_n=2\pi n/l }.
Substituting this into (\ref{41ms-2}),
we get the transfer function of the finite-multimode system as
\small\begin{align}
% \ipg\dd{\phi\dd{\mathrm{out}}|\phi\dd{\mathrm{in}}}\uu{\mathrm{SU(2)}}(s)
\wick{\phi\dd{\mathrm{out}}}{\phi\dd{\mathrm{in}}}\uu{\mathrm{SU(2)}}
&=
 \dtf{-2|g|^2\bm{F}\bm{F}\dgg-\im\Omega}{-2g\uu{*}\bm{F}}
    {2g\bm{F}\dgg}{1},
\label{2mode}
\end{align}\normalsize
where
\small\begin{align}
 \bm{F}&=\BV{1/l}{\vdots}{1/l},
\qquad
 \Omega=\B{\omega_1}{}{}{}{\ddots}{}{}{}{\omega_N}.
\end{align}\normalsize
\end{remark}

\newpage

\section{Time-varying SU(2) system}

The interaction Lagrangian of 
the time-varying SU(2) system\index{SU(2) system (time-varying)}
is given as (\ref{tvsystem}):
\begin{subequations}
 \small\begin{align}
 \im \lag_{\mas}\urm{TV}
&=
 \mm{\phi}_{\mas}\dgg \circ (2\lie) \ \mm{\phi}_{\mas} 
\\ &=
 \mm{\phi}_{\mas}\dgg \ (2\lie)\ast \mm{\phi}_{\mas},
 \end{align}\normalsize
\end{subequations}
where
\small\begin{align}
 \lie(t)
\equiv
 \A{}{g(t)}
   {-g\uu{*}(t)}{}.
\end{align}\normalsize
Here we rewrite this as
 \small\begin{align}
\im \lag_{\mas}\urm{TV}
&=
 \mm{\phi}_{\mas}\dgg  (2\lie) \ \mm{\phi}_{\mas},
 \end{align}\normalsize
where
\small\begin{align}
 \lie
\equiv
 \A{}{g \, \ast}
   {\circ (-g\uu{*})}{}.
\end{align}\normalsize
Then the transfer function is expressed 
in the same form as the SU(2) system:
\begin{subequations}
\small\begin{align}
 \ipg\dd{\phi\drm{out} | \phi\drm{in}}\urm{int}
&=
  \wick{\phi\drm{out}}{\phi\drm{in}\dgg}
+
 \wick{\phi\drm{out}}{\phi\dd{\mas}\dgg}  (2\lie) 
 \wick{\phi\dd{\mas}}{\phi\drm{in}\dgg}
\\ & \hspace{23mm} +
 \wick{\phi\drm{out}}{\phi\dd{\mas}\dgg}  (2\lie) 
 \wick{\phi\dd{\mas}}{\phi\dd{\mas}\dgg}  (2\lie) 
 \wick{\phi\dd{\mas}}{\phi\drm{in}\dgg}
\label{tvsu22nd}
\\ & \hspace{23mm} + \cdots.
\end{align}\normalsize
\end{subequations}

Let us calculate each term carefully.
The second order correction (\ref{tvsu22nd}) is written as
\small\begin{align}
 & \AH{1}{0}
 \A{}{2g \, \ast}
   {\circ (-2g\uu{*})}{}
 \A{\ffrac{1}{2}}{}
   {}{\wick{\mas}{\mas\dgg}}
 \A{}{2g \, \ast}
   {\circ (-2g\uu{*})}{}
 \AV{1}{0},
\end{align}\normalsize
where we have used Lemma \ref{lem:tfbefore}.
Higher order terms are given as
\begin{subequations}
\small\begin{align}
 2\urm{nd} \mbox{\normalsize order: \small}&
 2 \left\{ g\ast \wick{\mas}{\mas\dgg} \circ (-2g\uu{*})\right\},
\\
 4\urm{th} \mbox{\normalsize order: \small}&
 2 \left\{ g\ast \wick{\mas}{\mas\dgg} \circ (-2g\uu{*})\right\}^2.
\end{align}\normalsize
\end{subequations}
As a result, 
the transfer function is given by
\small\begin{align}
 \ipg\dd{\phi\drm{out} | \phi\drm{in}}\urm{TV}
&=
 \frac{1+  g\ast \wick{\mas}{\mas\dgg} \circ (-2g\uu{*}) }
      {1-  g\ast \wick{\mas}{\mas\dgg} \circ (-2g\uu{*}) }
\end{align}\normalsize
In the frequency domain, this is expressed as
\small\begin{align}
 \ipg\dd{\phi\drm{out} | \phi\drm{in}}\urm{TV}(s)
&=
\kakkon{\ffrac{s-2g(s)g\uu{*}(-s)}{s+2g(s)g\uu{*}(-s)}, 
 \hspace{42mm} (\mbox{single mode})}
 {\ffrac{1 - g(s)g\uu{*}(-s) - \ex\uu{-sl}\bigl[1 + g(s)g\uu{*}(-s)\bigr]}
 {1 + g(s)g\uu{*}(-s) - \ex\uu{-sl}\bigl[1 - g(s)g\uu{*}(-s)\bigr]}, 
        \quad (\mbox{infinite mode})\rule[0mm]{0mm}{8mm}}
\end{align}\normalsize
which is consistent with the classical result (\ref{tvtf}).

\newpage

\section{Feedback connection}

\begin{wrapfigure}[0]{r}[53mm]{49mm} 
\vspace{-5mm}
\centering
\includegraphics[keepaspectratio,width=35mm]{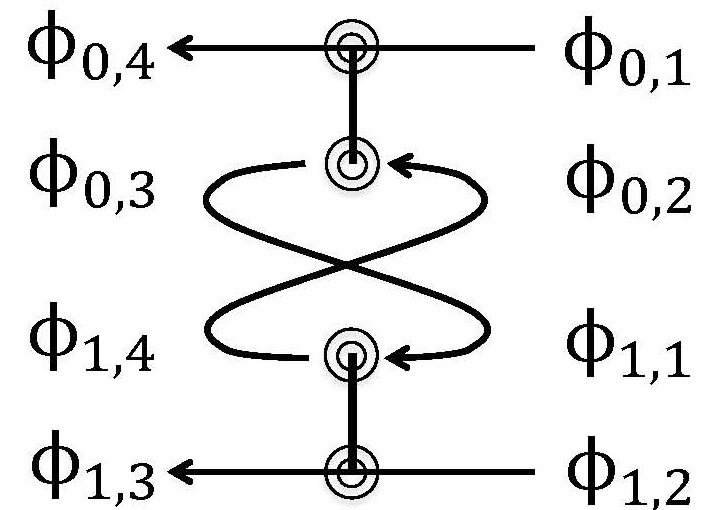}
\caption{
\small
Feedback connection.
\normalsize
}
\label{fig-cassys1-1}
\end{wrapfigure}

As another example of our approach,
let us consider the feedback connection
in Section \ref{sec:singlefb}.
The interaction Lagrangian is given by
\small\begin{align}
 \im \lag\urm{SU(2)}\dd{\mas} 
&=
 \phi\dd{\mas}\dgg \, 2\lie_T \, \phi\dd{\mas},
\end{align}\normalsize
where
\small\begin{align}
 \phi\dd{\mas}
&\equiv
 \BV{\mm{\phi}_1}
    {\mas}
    {\mm{\phi}_4},
\quad
 \lie_T
\equiv
 \B{}{g_0}{}
   {-g_0\uu{*}}{0}{g_1}
   {}{-g_1\uu{*}}{}.
\end{align}\normalsize
In this case, 
the input and the output are defined as
\small\begin{align}
 \phi\drm{out}=\AV{\phi_1}{\phi_4}\drm{out}
\quad
 \phi\drm{in}=\AV{\phi_1}{\phi_4}\drm{in},
\end{align}\normalsize
for which
\begin{subequations}
\small\begin{align}
 \wick{\phi\drm{out}}{\phi\dd{\mas}\dgg}
&=
 \left[
\begin{array}{ccc}
 1 & 0 & 0 \\
 0 & 0 & 1
\end{array}
\right],
\\
 \wick{\phi\dd{\mas}}{\phi\drm{in}\dgg}
&=
 \left[
\begin{array}{cc}
 1 & 0  \\
 0 & 0 \\ 
 0 & 1
\end{array}
\right],
\\
 \wick{\phi\dd{\mas}}{\phi\dd{\mas}\dgg}
&=
 \B{\ffrac{1}{2}}{}{}
   {}{\hspace{-2mm} \wick{\mas}{\mas\dgg}}{}
   {}{}{\hspace{-2mm} \ffrac{1}{2}}.
\end{align}\normalsize
\end{subequations}
Then (\ref{su2sysgen}) is written as
\small\begin{align}
 \wick{\phi\drm{out}}{\phi\drm{in}\dgg}
&=
 \frac{1}{1 + 2(|g_0|^2 + |g_1|^2) \wick{\mas}{\mas\dgg}}
\nn\\ & \hspace{5mm} \times
\renewcommand{\arraystretch}{1.5}
 \A{1 - 2 (|g_0|^2 + |g_1|^2) \wick{\mas}{\mas\dgg} }
   {4g_0 g_1 \wick{\mas}{\mas\dgg} }
   {4g_0\uu{*}g_1\uu{*} \wick{\mas}{\mas\dgg}}
   {1 - 2(|g_0|^2 + |g_1|^2) \wick{\mas}{\mas\dgg}}.
\renewcommand{\arraystretch}{1}
\end{align}\normalsize

\begin{wrapfigure}[0]{r}[53mm]{49mm} 
\centering
\vspace{-30mm}
\includegraphics[keepaspectratio,width=28mm]{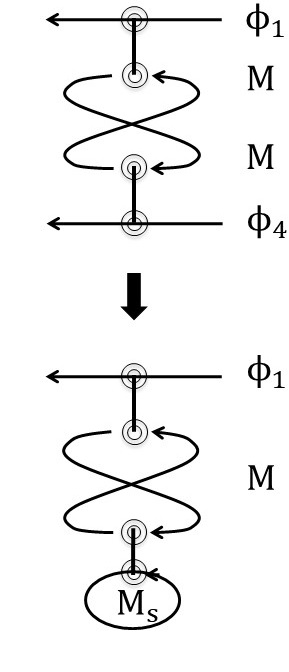}
\caption{
\small
Termination
\normalsize
}
\label{fig-fbconnection2-2}
\end{wrapfigure}

If there is a termination as in Figure \ref{fig-fbconnection2-2},
the transfer function is obtained by replacing
\small\begin{align}
 \phi\dd{\mas}
&\equiv
 \BV{\mm{\phi}_1}
    {\mas}
    {\mas\dd{s}},
\quad
 \wick{\phi\dd{\mas}}{\phi\dd{\mas}\dgg}
=
 \B{\ffrac{1}{2}}{}{}
   {}{ \wick{\mas}{\mas\dgg}}{}
   {}{}{ \wick{\mas\dd{s}}{\mas\dd{s}}},
\end{align}\normalsize
which results in
\small\begin{align}
 \wick{\phi\drm{out}}{\phi\drm{in}\dgg}
&=
 \frac{1 
    - 2|g_0|^2\wick{\mas}{\mas\dgg} 
    + 4|g_1|^2\wick{\mas}{\mas\dgg}\wick{\mas\dd{s}}{\mas\dd{s}\dgg}}
      {1 
    + 2|g_0|^2\wick{\mas}{\mas\dgg} 
    + 4|g_1|^2\wick{\mas}{\mas\dgg}\wick{\mas\dd{s}}{\mas\dd{s}\dgg}}.
\end{align}\normalsize

%\newpage

\section{SU(2) Dirac system}
\label{diracsyssm}

Our approach can be applied to 
the SU(2) Dirac system as well\index{SU(2) system (Dirac)}.
The interaction Lagrangian has been given in (\ref{dsu2int}):
\small\begin{align}
 \im \lag\urm{SU(2)}_{\mas}
&=
 \widebar{\psi}_{\mas} \gamma\uu{z} (2\lie) \ \psi_{\mas},
\end{align}\normalsize
where
\small\begin{align}
 \psi_{\mas}
=
 \AV{\mm{\psi}}{\mas},
\quad
 \lie
=
 \A{}{g}{-g\uu{*}}{}.
\end{align}\normalsize

\begin{lemma}
 The free field transfer functions of the Dirac field
are given by
\begin{subequations}
\label{lem:ffdirac}
\small\begin{align}
 \wick{\psi\dd{\mathrm{in}}}{\widebar{\psi}\dd{\mathrm{in}}}
&=
 -\gamma\uu{z},
\\
 \wick{\psi\dd{\mathrm{out}}}{\widebar{\psi}\dd{\mathrm{in}}}
&=
 -\gamma\uu{z},
\label{lem:ffdirac-2}
\\
 \wick{\psi\dd{\mathrm{in}}}{\widebar{\psi}\dd{\mathrm{out}}}
&=
 \gamma\uu{z},
\label{lem:ffdirac-3}
\\
 \wick{\psi\dd{\mathrm{out}}}{\widebar{\psi}\dd{\mathrm{out}}}
&=
 -\gamma\uu{z},
\\
 \wick{\mas}{\widebar{\mas}}
&=
 (s\gamma^0 + \im m)\inv.
\end{align}\normalsize
\end{subequations}
\end{lemma}

These follow from the definition in Section \ref{sec:dfde}.
Note that (\ref{lem:ffdirac-2}, \ref{lem:ffdirac-3})
indicate the same asymmetry as the forward traveling field.
This lemma leads to
\begin{subequations}
\label{skk}
\small\begin{align}
 \wick{\psi\dd{\mas}}{\widebar{\psi}\drm{in}}
&=
 \AV{-\gamma\uu{z}}{0}, 
\\
 \wick{\psi\drm{out}}{\widebar{\psi}\dd{\mas}}
&=
 \AH{-\gamma\uu{z}}{0},
\\
 \wick{\psi\dd{\mas}}{\widebar{\psi}\dd{\mas}}
&=
 \A{-\ffrac{\gamma\uu{z}}{2}}{}
   {}{\wick{\mas}{\widebar{\mas}}},
\\
 \wick{\psi\dd{\mas}}{\widebar{\psi}\drm{out}}
&=
 0,
\\
 \wick{\psi\drm{in}}{\widebar{\psi}\dd{\mas}}
&=
 0.
\end{align}\normalsize
\end{subequations}

The input-output relation of a Dirac system
is given in the same way as (\ref{su2sysgen}):
\begin{subequations}
\small\begin{align}
 \ipg\dd{\psi\drm{out} | \psi\drm{in}}\urm{int}
&=
 \wick{\psi\drm{out}}{\widebar{\psi}\drm{in}}\urm{int}
\\ &=
 \wick{\psi\drm{out}}{\widebar{\psi}\drm{in}}
+
 \wick{\psi\drm{out}}{\widebar{\psi}\dd{\mas}}  \gamma\uu{z} (2\lie) 
 \left\{
 1- \wick{\psi\dd{\mas}}{\widebar{\psi}\dd{\mas}}  \gamma\uu{z} (2\lie) 
\right\}\inv
 \wick{\psi\dd{\mas}}{\widebar{\psi}\drm{in}}.
\end{align}\normalsize
\end{subequations}
Using (\ref{skk}), 
we have
\begin{subequations}
 \small\begin{align}
 \wick{\psi\drm{out}}{\widebar{\psi}\drm{in}}
&=
 \frac{1-2|g|^2\wick{\mas}{\widebar{\mas}}\gamma\uu{z}}
      {1+2|g|^2\wick{\mas}{\widebar{\mas}}\gamma\uu{z}}
 (-\gamma\uu{z})
\\ &=
 \dtf{-2|g|^2 \alpha\uu{z} -\im m\beta}{-2g\uu{*}\alpha\uu{z}}
    {2g}{1}
(-\gamma\uu{z}),
\end{align}\normalsize
\end{subequations}
which is the same form as the classical result (\ref{dtf}).

\chapter{Extra interactions in SU(2) systems}
\label{chap:additional}
\thispagestyle{fancy}

Here we consider 
linear/nonlinear interactions
placed in the SU(2) system.
We first discuss a relationship between
feedback and the Dyson equation.
(This has been briefly examined for the example of the d-feedback
in Section \ref{sec:circuitss}.)
Then two linear examples are demonstrated.
Nonlinear examples will be considered in Chapter \ref{chap:nonlinear}.

\section{Feedback and the Dyson equation}

Suppose a Lagrangian of the form
\small\begin{align}
 \lag\urm{f}+ \lag\urm{SU(2)}_{\mas} + \lag(\mas,\mas\dgg),
\end{align}\normalsize
in which each term describes
\begin{subequations}
\small\begin{align}
 \lag\urm{f} :& \quad
 \mbox{\normalsize free field,\small}
\\
 \lag\urm{SU(2)}_{\mas} :& \quad
 \mbox{\normalsize SU(2) system,\small}
\\ 
 \lag(\mas,\mas\dgg) :& \quad
 \mbox{\normalsize additional linear/nonlinear interaction.\small} 
\end{align}\normalsize
\end{subequations}
We have considered 
the input-output relation of the `empty' SU(2) system 
by regarding \en{ \lag\urm{SU(2)}_{\mas} } as a perturbation
in Chapter \ref{sec:tfs}.
Now we know all transfer functions for the SU(2) system.
Here \en{ \lag(\mas,\mas\dgg) } is regarded as a perturbation
to examine the effect of the additional interactions.

\subsection{Transfer function \en{ \mas \gets \mas } revisited}
\label{sec:dyson}

Before proceeding to deal with \en{ \lag(\mas,\mas\dgg) },
let us reconsider the transfer function \en{ \mas \gets \mas }
of the empty SU(2) system in detail.
The interaction Lagrangian is of the form 
\small\begin{align}
 \im \lag\urm{SU(2)}_{\mas}
&=
 \phi\dd{\mas}\dgg (2\lie) \ \phi\dd{\mas},
\end{align}\normalsize
where
\small\begin{align}
 \phi\dd{\mas}
\equiv
 \AV{\mm{\phi}}{\mas}.
\end{align}\normalsize

\newpage

\noindent
As in the preceding chapter,
the transfer function \en{ \mas \gets \mas } is expressed as
\begin{subequations}
\small\begin{align}
\hspace{-5mm}
 \wick{\mas}{\mas\dgg}\urm{SU(2)}
=
  \wick{\mas}{\mas\dgg}
&+
 \wick{\mas}{\phi\dd{\mas}\dgg} (2\lie) 
 \wick{\phi\dd{\mas}}{\mas\dgg}
\\ & +
 \wick{\mas}{\phi\dd{\mas}\dgg} (2\lie) 
 \wick{\phi\dd{\mas}}{\phi\dd{\mas}\dgg}  (2\lie) 
 \wick{\phi\dd{\mas}}{\mas\dgg}
\label{su2dyson-2} 
\\ &  + \cdots,
\end{align}\normalsize
\end{subequations}
where
\small\begin{align}
 \wick{\mas}{\phi\dd{\mas}\dgg}
=
 \AH{0}{\wick{\mas}{\mas\dgg}},
\qquad
 \wick{\phi\dd{\mas}}{\mas\dgg}
=
 \AV{0}{\wick{\mas}{\mas\dgg}}.
\end{align}\normalsize
In this case, 
the first order correction is zero.
The second order term (\ref{su2dyson-2}) is written as
\begin{subequations}
\small\begin{align}
 \wick{\mas}{\mas\dgg} 
(
\underbrace{ 
 -2|g|^2
 }\dd{\hspace{5mm} \equiv -\isel}
)
\wick{\mas}{\mas\dgg}
&=
 \ipg\dd{\mas|\mas} \ (-\isel) \ \ipg\dd{\mas|\mas}. 
\end{align}\normalsize
\end{subequations}
The fourth order term is obtained in the same way.
These non-zero terms are expressed as the following diagrams:

\vspace{0mm}
\begin{figure}[H]
\centering
\includegraphics[keepaspectratio,width=115mm]{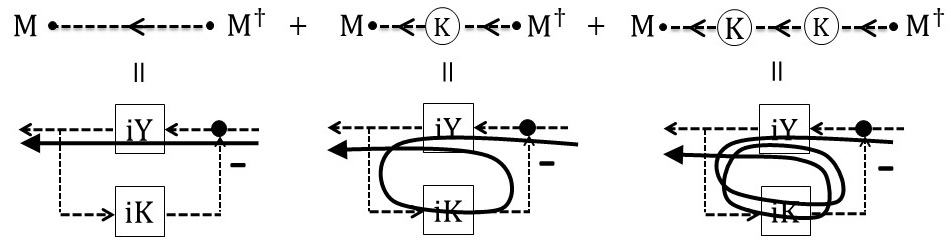}
\end{figure}
\vspace{0mm}

As discussed in Chapter \ref{chap:1},
the expansion of the \textit{S}-matrix is expressed 
as (negative) feedback.
The zeroth order term is the direct through effect
from the input to the output as in the left diagram above.
The second order correction is depicted in the middle diagram
where the signal travels in the feedback loop once.
In the case of the fourth order term,
the signal travels twice.
As a result,
the transfer function \en{ \mas\gets\mas } is given by 
\small\begin{align}
 \ipg\urm{SU(2)}_{\mas|\mas} 
&=
 \Bigl[ 1+ \ipg_{\mas|\mas}\isel \Bigr]\inv  \ipg_{\mas|\mas}.
\label{pdyson2}
\end{align}\normalsize
This is called a feedback transfer function in systems theory.
Alternatively, it is expressed as
\small\begin{align}
 \pg\urm{SU(2)}_{\mas|\mas} 
&=
 \pg_{\mas|\mas} 
+ 
 \pg_{\mas|\mas} \ K \ \pg\urm{SU(2)}_{\mas|\mas}.
\label{pdyson3}
\end{align}\normalsize
In physics, 
this is known as the Dyson equation in which\index{Dyson equation}
\en{ \sel } is called \textit{self-energy}.\index{self-energy}
From (\ref{pdyson2}) and (\ref{pdyson3}),
we have the following correspondence:
\index{bare propagator}\index{dressed propagator}\index{self-energy}

\vspace{0mm}
\begin{table}[H]
\centering
\renewcommand{\arraystretch}{1.8}
 \begin{tabular}{r|l|l} 
  & in physics & in systems theory \\ \hline \hline
    \en{ \pg_{\mas|\mas} } &  bare propagator & plant 
\\ \hline
    \en{ \sel } & self-energy & controller 
\\ \hline
    \en{ \pg\urm{SU(2)}_{\mas|\mas} } & dressed propagator 
                              & feedback transfer function 
 \end{tabular}
\renewcommand{\arraystretch}{1}
\end{table}

\if0
\begin{subequations}
\label{sedef}
\small\begin{align}
 \pg\dd{\mas|\mas}: & \quad \mbox{bare propagator \ (plant)}
\\
 \pg^{SU(2)}\dd{\mas|\mas}: & \quad \mbox{dressed propagator \ (feedback transfer function)}
\\
 \sel: & \quad \mbox{self-energy \ (controller)}
\end{align}\normalsize
\end{subequations}
\fi

\newpage

\subsection{Second perturbation: $\lag(\mas,\mas\dgg)$}
\label{sec:dyson2}

\marginpar{\vspace{8mm}
\footnotesize
 From Theorem \ref{thm:tfbs}, \\ \\
 $ \ipg\dd{\phi\drm{out} |\phi\drm{in} }\urm{SU(2)}
   = \ffrac{s-2|g|^2}{s+2|g|^2},$ \\
 $ \ \ \ \ipg\dd{\mas | \mas }\urm{SU(2)}
   = \ffrac{1}{s+2|g|^2}.$ 
\normalsize
 }

Now let us consider a transfer function
\en{ \phi\drm{out} \gets\phi\drm{in} }
under the additional Lagrangian \en{ \lag }.
The zeroth order term is the direct-through effect,
which is the empty SU(2) system.
For later use, 
we express it as 
\begin{subequations}
\label{kzeroth-l}
\small\begin{align}
 \Bigl( 0\urm{th} \mbox{ order of } 
 \ipg\dd{\phi\drm{out} |\phi\drm{in} }\urm{SU(2)+L}
\Bigr) 
&=  
 \ipg\dd{\phi\drm{out} |\phi\drm{in} }\urm{SU(2)}
\\ &=
 1-4|g|^2  \ipg_{\mas|\mas}\urm{SU(2)},
\end{align}\normalsize
\end{subequations}
where we have used Theorem \ref{thm:tfbs} 
in the second line.
The next order correction involves the following three processes:
\bee
\setlength{\itemsep}{-1mm} 
\item 
 The particles propagate from the input field \en{ \phi\drm{in} } 
 to the SU(2) system \en{ \mas }.
 This is described by a transfer function 
 \en{ \ipg\dd{\mas |\phi\drm{in}}\urm{SU(2)} };
\item 
 The particles interact with each other through the interaction \en{ \lag },
 which is represented by self-energy \en{ \sel^\lag };
\item 
 After the interaction,
 the particles propagate away from the system \en{ \mas } 
 to the output field \en{ \phi\drm{out} },
 which is described by \en{ \ipg\dd{\phi\drm{out} |\mas}\urm{SU(2)} }.
\ee

In higher order corrections,
the second process repeats in the same way as the negative feedback
described in the preceding section.
Accordingly,
the transfer function \en{ \phi\drm{out} \gets \phi\drm{in} } 
is expressed as

\vspace{0mm}
\begin{figure}[H]
\hspace{0mm}
\centering
\includegraphics[keepaspectratio,width=139mm]{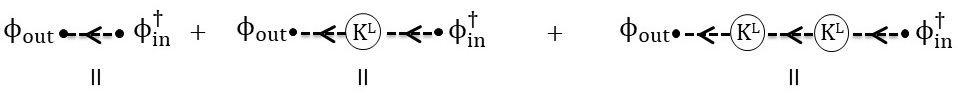}
\end{figure}
\vspace{-13mm}

\small\begin{align}
\hspace{7mm}
 \ipg\dd{\phi\drm{out} | \phi\drm{in} }\urm{SU(2)}
\hspace{10mm}
 \ipg\dd{\phi\drm{out} | \mas}\urm{SU(2)} \, 
 \left(-\isel^L\right) \, 
 \ipg\dd{\mas|\phi\drm{in}}\urm{SU(2)}
\hspace{10mm}
 \ipg\dd{\phi\drm{out} |\mas}\urm{SU(2)} \, 
 \left(-\isel^L\right) \,
 \ipg_{\mas|\mas}\urm{SU(2)} \, 
 \left(-\isel^L\right) \, \ipg\dd{\mas|\phi\drm{in}}^{SU(2)}
\nn\\
\nn\\
\label{dysonext}
\end{align}\normalsize

\vspace{0mm}

\noindent
Using Theorem \ref{thm:tfbs} again,
we can express this series as
\begin{subequations}
 \label{41dyson0} 
\small\begin{align}
 \ipg\dd{\phi\drm{out} | \phi\drm{in} }\urm{SU(2)+L}
&=
1+
(-4|g|^2)
\underbrace{
\left[
 1- \ipg_{\mas|\mas}\urm{SU(2)} \ \isel\urm{L}
 +\cdots
\right]
 \ipg_{\mas|\mas}\urm{SU(2)} }
\dd{\textrm{the Dyson equation (\ref{pdyson2})}}
\\ &=
\left[
 1+ \ipg_{\mas|\mas}\urm{SU(2)} \ \isel\urm{L}
\right]\inv 
\left[
 \ipg\dd{\phi\drm{out} | \phi\drm{in}}\urm{SU(2)} 
+ 
 \ipg_{\mas|\mas}\urm{SU(2)} \ \isel\urm{L}
\right]
\\ &=
\kakkon{\ffrac{s-2|g|^2+\isel^L(s)}{s+2|g|^2+\isel^\lag(s)}, 
 \hspace{52mm}
 (\mbox{single mode})}
{\ffrac{(1+\ex\uu{-sl})\isel^L(s) - 2(|g|^2-1)-2\ex\uu{-sl}(|g|^2+1)}
       {(1+\ex\uu{-sl})\isel^L(s) + 2(|g|^2+1)+2\ex\uu{-sl}(|g|^2-1)}.
\hspace{5mm}
 (\mbox{infinite mode})
\rule[0mm]{0mm}{9mm}}
 \label{41dyson} 
\end{align}\normalsize
\end{subequations}
where we have used (\ref{ftf}) in the last line.

\begin{remark}
 Note that there are examples to which the Dyson equation 
is not applicable.
For example,
self-energy was zero 
for d-feedforward in Section \ref{sec:ffq}.
\end{remark}

\newpage

\section{Linear examples}
\label{sec:exlinear}

If \en{ \lag } is linear, \en{ \isel\urm{L} } is constant.
Then the single-mode transfer function is simplified as
\small\begin{align}
 \ipg\dd{\phi\drm{out} | \phi\drm{in} }\urm{SU(2)+L}
&=
 \dtf{-2|g|^2-\isel^L}{-2g\uu{*}}{2g}{1}.
\label{constl}
\end{align}\normalsize
We shall show two examples of this case.

\subsection{Detuning}
\label{detuning}

Consider a case where 
the additional Lagrangian is given as\index{detuning}
\small\begin{align}
 \im \lag\urm{DT}
&= 
 - \im \Omega \mas\dgg \mas.
\qquad
 (\Omega\in\mathbb{R})
\end{align}\normalsize
Here we calculate 
a transfer function \en{ \Phi\drm{out} \gets \Phi\drm{in} }
\small\begin{align}
 \ipg\dd{\Phi\drm{out}|\Phi\drm{in}}\urm{SU(2)+DT} 
&=
 \wick{\Phi\drm{out}}{\Phi\drm{in}\ddgg}\urm{SU(2)+DT},
\end{align}\normalsize
where
\small\begin{align}
 \Phi\drm{out}
=
 \AV{\phi\drm{out}}{\phi\drm{out}\dgg},
\quad
 \Phi\drm{in}\ddgg 
=
 \AH{\phi\drm{in}\dgg}{-\phi\drm{in}}.
\end{align}\normalsize

The quadratic Lagrangian \en{ \lag\urm{DT} } 
results in linear dynamics. 
The corresponding self-energy
can be found from the first order correction,
which is written as
\begin{subequations}
\small\begin{align}
\contraction{}{\Phi}{\hspace{6mm}}{\hspace{0mm}}
 \Phi\drm{out}
 \  \im  
\contraction{}{\phi}{\hspace{7mm}}{\hspace{0mm}}
 \lag\urm{DT} \
 \Phi\drm{in}\ddgg
&=
 \A{-\im \Omega \
    \contraction{}{\phi}{\hspace{6mm}}{\hspace{0mm}}
     \phi\drm{out} \mas\dgg
    \contraction{}{\phi}{\hspace{3mm}}{\hspace{0mm}}
      \mas \phi\drm{in}\dgg}
   {0}
   {0} 
   {-\im \Omega \
    \contraction[2ex]{}{\phi}{\hspace{11mm}}{\hspace{0mm}}
     \phi\drm{out}\dgg 
    \contraction{}{\mas}{\hspace{10mm}}{\hspace{0mm}}
     \mas\dgg  \mas (-\phi\drm{in})}
\\ &=
 \ipg\dd{\phi\drm{out}|\mas}\urm{SU(2)} 
 \A{-\im \Omega}{}
   {}{\im\Omega}
 \ipg\dd{\mas| \phi\drm{in}}\urm{SU(2)}.
\end{align}\normalsize
\end{subequations}
Compared to the first order correction in (\ref{dysonext}),
the self-energy is defined as
\small\begin{align}
 \isel\urm{DT}
&=
 \im \Omega \sigma\dd{z}.
\end{align}\normalsize
As a result, 
the transfer function is given as
\begin{subequations}
\small\begin{align}
 \ipg_{\Phi\drm{out}|\Phi\drm{in}}\urm{SU(2)+DT}
&=
 \dtf{-2|g|^2-\im\Omega \sigma\dd{z}}{-2g\uu{*}}{2g}{I}.
\end{align}\normalsize
\end{subequations}

It is easy to see that 
the poles and transmission zeros of this system are given as
\begin{subequations}
\small\begin{align}
 \pole 
&= 
 \{-2|g|^2-\im\Omega, \ -2|g|^2+\im\Omega \},
\\
 \zero 
&= 
 \{ \hspace{3mm} 2|g|^2+\im\Omega, \hspace{3mm} 2|g|^2-\im\Omega \},
\end{align}\normalsize
\end{subequations}
which satisfy the pole-zero symmetry\index{pole-zero symmetry}
\small\begin{align}
 \pole = -\zero.
\end{align}\normalsize

\subsection{Squeezing in the SU(2) system}
\label{sec:sq}

Consider an interaction Lagrangian
\small\begin{align}
 \hspace{10mm}
 \im \lag\urm{SQ}
&= 
 \frac{r}{2} (\mas\dgg \mas\dgg-\mas\mas).
\qquad
 (r\in\mathbb{R})
\end{align}\normalsize
Let us calculate 
a transfer function \en{ \Phi\drm{out}\gets\Phi\drm{in} }
in the same way as the preceding example.
The first order correction is written as
\small\begin{align}
\contraction{}{\phi}{\hspace{6mm}}{\hspace{0mm}}
 \Phi\drm{out}
 \  \im  
\contraction{}{\phi}{\hspace{7mm}}{\hspace{0mm}}
 \lag\urm{SQ} \
 \Phi\drm{in}\ddgg
&=
 \ipg\dd{\phi\drm{out}|\mas}\urm{SU(2)} \
 \frac{r}{|g|^2}
 \A{0}{g^2}
   {g^{\ast 2}}{0}
 \ipg\dd{\mas| \phi\drm{in}}\urm{SU(2)},
\end{align}\normalsize
where we have used Theorem \ref{thm:tfbs}.
Compared to the first order correction in (\ref{dysonext}),
the self-energy is defined as
\small\begin{align}
 \isel\urm{SQ} 
&= 
 -\frac{r}{|g|^2}
 \A{0}{g^2}
   {g^{\ast 2}}{0}.
\end{align}\normalsize
Then (\ref{constl}) is written as
\small\begin{align}
 \ipg\dd{\Phi\drm{out} | \Phi\drm{in}}\urm{SU(2)+SQ}
&=
 \dtf{-2|g|^2-\isel\urm{SQ}}{-2g\uu{*}}{2g}{I}. 
\end{align}\normalsize

%\newpage

In the quadrature basis
\small\begin{align}
 \qu
=
 \AV{\xi}{\eta}
= 
 \toqu\Phi,
\quad
\left(
 \toqu
 \equiv
 \frac{1}{\sqrt{2}}
 \A{1}{1}
   {-\im}{\im}
\right)
\end{align}\normalsize
this transfer function is expressed as
\small\begin{align}
 \toqu \  \ipg\dd{\Phi\drm{out}|\Phi\drm{in}}\urm{SU(2)+SQ} \ \toqu\inv 
&=
 \A{\ffrac{s^2-(2g^2+r)(2g^{*2}+r)}
         {(s+2|g|^2+r)(s+2|g|^2-r)}}
   {\ffrac{2\im r(g^2-g^{*2})}
         {(s+2|g|^2+r)(s+2|g|^2-r)}}
   {\ffrac{-2\im r(g^2-g^{*2})}
         {(s+2|g|^2+r)(s+2|g|^2-r)}}
   {\ffrac{s^2-(2g^2-r)(2g^{*2}-r)}
         {(s+2|g|^2+r)(s+2|g|^2-r)} \rule[0mm]{0mm}{9mm}}.
\end{align}\normalsize
If \en{ g } is real,
\small\begin{align}
 \toqu \  \ipg\dd{\Phi\drm{out}|\Phi\drm{in}}\urm{SU(2)+SQ} \ \toqu\inv 
&=
 \A{\ffrac{s-2g^2-r}{s+2g^2-r}}{}
   {}{\ffrac{s-2g^2+r}{s+2g^2+r}},
\end{align}\normalsize
for which poles and transmission zeros are given as
\begin{subequations}
 \small\begin{align}
 \pole&=\left\{-2g^2-r, \ -2g^2+r \right\},
\\
 \zero&=\left\{\hspace{3mm} 2g^2+r, \hspace{3mm} 2g^2-r\right\}.
 \end{align}\normalsize
\end{subequations}
This satisfies the pole-zero symmetry\index{pole-zero symmetry}
\small\begin{align}
 \pole=-\zero.
\end{align}\normalsize

\newpage

\section{Concluding remarks: Feynman diagrams}

\begin{wrapfigure}[0]{r}[53mm]{49mm} 
\vspace{-0mm}
\centering
\includegraphics[keepaspectratio,width=30mm]{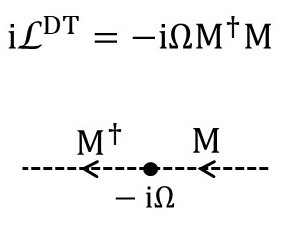}
\caption{
\small
$\im\lag^{\textrm{DT}}$ 
and the corresponding diagram.
$\mas$ and $\mas\dgg$ are 
the incoming and outgoing arrows, respectively.
\normalsize
}
\label{fig-dt-lag}
\end{wrapfigure}

It is a good exercise 
to show the same transfer functions 
using Feynman diagrams again.
Let us consider the detuning.
The Lagrangian \en{ \im\lag\urm{DT} } is 
depicted in Figure \ref{fig-dt-lag}.
To find self-energy,
we calculate the first order correction of the transfer function
\small\begin{align}
 \ipg\dd{\Phi\drm{out}|\Phi\drm{in}}\urm{SU(2)+DT} 
&=
 \A{\ipg\dd{\phi\drm{out} |\phi\drm{in} }\urm{SU(2)+DT}}
   {\ipg\dd{\phi\drm{out} |\phi\drm{in}\dgg}\urm{SU(2)+DT}}
   {\ipg\dd{\phi\drm{out}\dgg |\phi\drm{in}}\urm{SU(2)+DT}}
   {\ipg\dd{\phi\drm{out}\dgg |\phi\drm{in}\dgg}\urm{SU(2)+DT} 
    \rule[0mm]{0mm}{8mm}}.
\label{dt-tf}
\end{align}\normalsize

The off-diagonal elements are disconnected diagrams and hence zero.
For instance,
the (1,2)-element is written as
\small\begin{align}
 \ipg\dd{\phi\drm{out} |\phi\drm{in}\dgg}\urm{SU(2)+DT}
=
 \wick{\phi\drm{out}}{-\phi\drm{in}}\urm{SU(2)+DT},
\end{align}\normalsize
which is depicted as

\vspace{-3mm}
\begin{figure}[H]
\centering
\includegraphics[keepaspectratio,width=47mm]{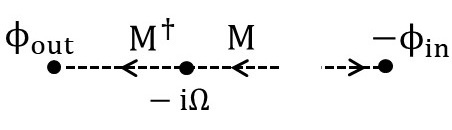}
\end{figure}
\vspace{-3mm}

\noindent
This is disconnected because 
\en{ \mas } and \en{ -\phi\drm{in} } are pointing 
in opposite directions.

For the (1,1)-element,
the first order correction is given as

\vspace{-2mm}
\begin{figure}[H]
\centering
\includegraphics[keepaspectratio,width=41mm]{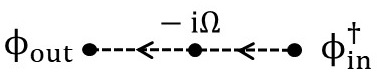}
\end{figure}
\vspace{-10mm}

\small\begin{align}
 =\ipg\dd{\phi\drm{out} | \mas}\urm{SU(2)} \
  (-\im\Omega) \
  \ipg\dd{\mas|\phi\drm{in} }\urm{SU(2)}.
\end{align}\normalsize
The (2,2)-element is written as
\en{ \wick{\phi\drm{out}\dgg}{-\phi\drm{in}}\urm{SU(2)+DT} }
to which 
the first order correction is given as

\vspace{-4mm}
\begin{figure}[H]
\centering
\includegraphics[keepaspectratio,width=41mm]{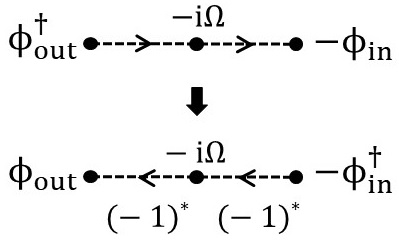}
\end{figure}
\vspace{-9mm}

\small\begin{align}
= 
 \left(-\ipg\dd{\phi\drm{out} |\mas}\urm{SU(2)}\right)\uu{*} 
\ (-\im \Omega)
\ \left(\ipg\dd{\mas|\phi\drm{in} }\urm{SU(2)}\right)\uu{*},
\end{align}\normalsize
where 
we have used the asymmetric property (Theorem \ref{thm:tfbs})
to reverse the arrows to the forward direction. 
This is represented by \en{ (-1)\uu{*} } in the diagram above.

As a result,
the first order correction is written as
\small\begin{align}
  \ipg\dd{\phi\drm{out} | \mas}\urm{SU(2)}
  \A{-\im\Omega}{}{}{\im\Omega} 
  \ipg\dd{\mas | \phi\drm{in}}\urm{SU(2)}.
\end{align}\normalsize
Compared to the same order term of (\ref{dysonext}),
the self-energy is given by
\small\begin{align}
 \isel^{DT}&=\im \Omega \sigma\dd{z}.
\end{align}\normalsize

\chapter{Pole-zero symmetry}
\label{chap:sym}
\thispagestyle{fancy}

\marginpar{\vspace{1mm}
\small
As mentioned in Section \ref{onoffshell},
$ \pole=-\zero $ holds only for transfer functions
between signals \textit{on shell}.
\normalsize
}

This chapter is dedicated to investigating 
the pole-zero symmetry \en{ \pole=-\zero }.\index{pole-zero symmetry}
This has been demonstrated in some examples.
We first show that 
the symmetry results from canonical quantization
for bosons and fermions.
Then its system theoretical characterizations 
are given both in the time domain and in the frequency domain.
We also examine how self-energy is influenced by  
the symmetry.
This is used to analyze nonlinear interactions 
in Chapter \ref{chap:nonlinear}.

\section{Characterization in the frequency-domain}
\label{sec:pzinfreq}

Let us consider a system \en{ \tf } 
with an input \en{ u } and output \en{ z }:
\small\begin{align}
 z= \tf \, u,
\label{frequa}
\end{align}\normalsize
where
\small\begin{align}
% z= \tf \, u;
%\qquad 
 z=\AV{z_1}{z_2},
\quad
 \tf = \A{\tf_{11}}{\tf_{12}}{\tf_{21}}{\tf_{22}},
\quad
 u= \AV{u_1}{u_2}.
\end{align}\normalsize
There are two things to note in this expression.
First, to describe quantum systems, 
the vectors \en{ u } and \en{ z } have to include 
all inputs and outputs so that \en{ \tf } is square.
Otherwise 
gauge transformations are not well-defined.
In particular, 
\en{ u } and \en{ z } have to include 
all conjugate variables 
for non-unitary gauge transformations.

Second, 
the components of \en{ u } and \en{ z } can be in any order,
but it is convenient 
to choose \en{ u_{2}= u_{1}\ddgg } and \en{ z_{2}= z_{1}\ddgg },
where the double dagger is elementwise.
This form is said to be
\textit{canonical}\index{canonical}.
(This is also known as 
\textit{symplectic} for bosons.)\index{symplectic}
All systems that we have considered so far 
are expressed in this way.
For example, 
in most cases, 
we have used the following form:
\small\begin{align}
 u_1=\BV{\phi_1}
        {\phi_2\rule[0mm]{0mm}{4mm} }
        {\vdots}\drm{in},
\quad
 u_2=\BV{\phi_1\dgg}
        {\phi_2\dgg\rule[0mm]{0mm}{4mm} }
        {\vdots}\drm{in},
%\right\}.
\end{align}\normalsize
which is canonical.
Another example is in the quadrature basis as
\small\begin{align}
 u_1=\BV{\xi_1}
        {\xi_2\rule[0mm]{0mm}{4mm}}
        {\vdots}\drm{in},
\quad
 u_2=\BV{-\im\eta_1\dgg}
        {-\im\eta_2\dgg\rule[0mm]{0mm}{4mm}}
        {\vdots}\drm{in}.
\end{align}\normalsize

%\newpage

A reason for introducing the canonical form
is to simplify the commutation and anticommutation relations
defined as
\small\begin{align}
 [u\dd{1,\alpha},u\dd{2,\beta}]\dd{\pm}
&\equiv 
 u\dd{1,\alpha} \, u\dd{2,\beta} 
\pm 
 u\dd{2,\beta}  \, u\dd{1,\alpha},
\end{align}\normalsize 
which is understood as a matrix.
If \en{ u } is canonical, %as in (\ref{mimoinex}),
this is expressed as
\small\begin{align}
 [u,u\simm]\dd{\pm}
\equiv
\renewcommand{\arraystretch}{1.5}
 \A{\bigl[u_1,u_1\simm\bigr]\dd{\pm}}
   {\bigl[u_1,u_2\simm\bigr]\dd{\pm}}
   {\bigl[u_2,u_1\simm\bigr]\dd{\pm}}
   {\bigl[u_2,u_2\simm\bigr]\dd{\pm}}
\renewcommand{\arraystretch}{1}
=
 \A{}{I}{\pm I}{},
\hspace{3mm}
 \kakkon{+: \mbox{fermion,}}
        {-: \mbox{boson,}}
\label{canonicalpi}
\end{align}\normalsize
where \en{ u\simm(s)=u\trans(-s) } 
as defined in (\ref{simm}).

\begin{definition}
 For a nonsingular matrix \en{ X },
\en{ \tf } is said to be \en{ X }-orthogonal
if 
\small\begin{align}
  \tf X \tf\simm = X.
\end{align}\normalsize
\end{definition}

\begin{lemma}
\label{thm1}
Consider a system \en{\tf} of the form (\ref{frequa}).
The input \en{ u } is not necessarily canonical.
Define a matrix \en{ \Pi } as 
\small\begin{align}
 \Pi
\equiv
 [u,u\simm]\dd{\pm}.
\label{geguu}
\end{align}\normalsize
The output \en{ z } satisfies 
the same relation as (\ref{geguu})
if and only if 
\en{ \tf } is \en{ \Pi }-orthogonal:\index{\en{ \Pi }-orthogonal}
\small\begin{align}
 \tf \Pi \tf\simm&=\Pi.
\label{gqc}
\end{align}\normalsize
 \end{lemma}
%%%%%%%%%%%%%%%%%%%%%%%%%%%%%%%%%%%%%%%%%%%%%%%
\begin{proof}
The matrix \en{ [z,z\simm]\dd{\pm} } can be written as
\small\begin{align}
 [z,z\simm]\dd{\pm}
&=
 [\tf u,(\tf u)\simm]
\nn\\ &=
 \tf [u,u\simm] \tf\simm 
\nn\\ &=
 \tf \Pi \tf\simm.
\end{align}\normalsize
As a result,
\small\begin{align}
 [z,z\simm]\dd{\pm}
=
 \Pi \quad \Leftrightarrow  
\quad  
 \tf \Pi \tf\simm=\Pi,
\end{align}\normalsize
which establishes the assertion.
\end{proof}

\begin{corollary}
\label{cor:diag}
Consider a diagonal system 
\small\begin{align}
 \AV{z_1}{z_2}
&=
 \A{\tf_{11}}{}{}{\tf_{22}}
 \AV{u_1}{u_2}.
\label{diag}
\end{align}\normalsize
Assume that the input \en{ u } is canonical.
Then the output \en{ z } is canonical if and only if 
\small\begin{align}
 \tf_{11} \tf_{22}\simm &= I,
\qquad
\mbox{(for both fermions and bosons.)}
\label{qc}
\end{align}\normalsize
\end{corollary}
%%%%%%%%%%%%%%%%%%%%%%%%%%%%%%
\begin{proof}
For a canonical form,
the \en{ \Pi }-orthogonality is written as
\small\begin{align}
 \A{\tf_{11}}{}{}{\tf_{22}}
 \A{}{I}{\pm I}{}  
 \A{\tf_{11}\simm}{}{}{\tf_{22}\simm}
&=
 \A{}{I}{\pm I}{},
\end{align}\normalsize 
which is equivalent to (\ref{qc}).
\end{proof}

\newpage

\begin{example}
 \rm
Let us consider the QND gate described in Section \ref{subsec:sum}:
%The system has been defined by (\ref{sumbounary}):
\small\begin{align}
 \left[
\begin{array}{c}
 \xi_1 \\
 \eta_1 \\ \hdashline
 \xi_2 \\
 \eta_2
\end{array}
\right]\drm{out}
&=
\left[
\begin{array}{cc:cc}
 1& &  & \\
  & 1 & & g\\ \hdashline
 -g & & 1 & \\
  &  & & 1
\end{array}
\right]
\left[
\begin{array}{c}
 \xi_1\\
 \eta_1\\ \hdashline
 \xi_2\\
 \eta_2
 \end{array}
\right]\drm{in}.
\label{exnonconju}
\end{align}\normalsize
It is easy to see that this system is \en{ \Pi }-orthogonal, 
\en{ \tf \Pi \tf\simm = \Pi }, 
where
\small\begin{align}
 \Pi
&=
\left[
\begin{array}{cc:cc}
 & \im &  & \\
 -\im &  & & \\ \hdashline
  & &  & \im \\
  &  & -\im & 
\end{array}
\right].
%\qquad  \because \ [\xi,\eta]=\im.
\end{align}\normalsize
Note that (\ref{exnonconju}) is not canonical.
Let us rewrite it in a canonical form.
In general, 
there are many expressions for it.
For example,
\small\begin{align}
 \left[
\begin{array}{c}
 \xi_1 \\
 \xi_2 \\ \hdashline
-\im \eta_1\\
-\im \eta_2
\end{array}
\right]\drm{out}
=
\left[
\begin{array}{cc:cc}
 1& &  & \\
 -g & 1 & & \\ \hdashline
  & & 1 & g\\
  &  & & 1
\end{array}
\right]
\left[
\begin{array}{c}
 \xi_1 \\
 \xi_2 \\ \hdashline
-\im \eta_1\\
-\im \eta_2
 \end{array}
\right]\drm{in}.
\end{align}\normalsize
This is block diagonal
and hence \en{ \tf_{11} \tf_{22}\simm = I }.
\qed
\end{example}

\begin{example}
\rm
Let us consider the SU(2) system 
with detuning in Section \ref{detuning}:
\small\begin{align}
 \AVl{\phi}{\phi\dgg}\drm{out}
&=
 \A{\ffrac{s-2|g|^2-\im\Omega}{s+2|g|^2-\im\Omega}}{}
   {}{\ffrac{s-2|g|^2+\im\Omega}{s+2|g|^2+\im\Omega}}
 \AVl{\phi}{\phi\dgg}\drm{in}.
\end{align}\normalsize
Since this is also canonical and diagonal,
\en{ \tf_{11} \tf_{22}\simm = I }.
\qed
\end{example}

\begin{example}
\label{ex:timesysex}
 \rm
The time-varying SU(2) system is given by (\ref{tvtf}):
\small\begin{align}
 \AV{\phi}{\phi\dgg}\drm{out}
&=
 \A{\ffrac{s-2g(s)g\uu{*}(-s)}{s+2g(s)g\uu{*}(-s)}}{}
   {}{\ffrac{s-2g\uu{*}(s)g(-s)}{s+2g\uu{*}(s)g(-s)}}
 \AV{\phi}{\phi\dgg}\drm{in}.
\label{tvtf-1}
\end{align}\normalsize
This is canonical and satisfies \en{ \tf_{11} \tf_{22}\simm = I }.
Since this system is diagonal,
poles \en{ p } and transmission zeros \en{ z } are easily found from 
\begin{subequations}
\small\begin{align}
[\mbox{Denominator = 0}] \Rightarrow
 &\kakkon{p+2g(p)g\uu{*}(-p)=0,}
         {p+2g\uu{*}(p)g(-p)=0,}
\\ 
[\mbox{Numerator = 0}] \Rightarrow
 &\kakkon{z-2g(z)g\uu{*}(-z)=0,}
         {z-2g\uu{*}(z)g(-z)=0,}
\end{align}\normalsize
\end{subequations}
which indicates that the poles \en{ \pole } and transmission zeros \en{ \zero }
are related to each other as \en{ \pole=-\zero }.
\qed
\end{example}

\newpage

\section{Characterization in the time-domain}
\label{sec:ss}

Example \ref{ex:timesysex}
indicates a certain connection between
\en{ \pole=-\zero } and \en{ \tf \Pi \tf\simm = \Pi }.
To investigate this,
we consider the time-domain characterization of 
\en{ \tf\Pi \tf\simm = \Pi }
in this section.
Assume that the input \en{ u } and the output \en{ z } are canonical 
and \en{ \tf } has a minimal realization
\small\begin{align}
 \tf = \dtf{A}{B}{C}{D}.
\label{gsr}
\end{align}\normalsize

\begin{lemma}
If a system of the form (\ref{gsr}) is \en{ \Pi }-orthogonal 
\en{ \tf \Pi \tf\simm=\Pi },
there exist non-singular matrices \en{ V } and \en{ W } such that
\begin{subequations}
\small\begin{align}
 A V + VA\trans + B\Pi B\trans &=0,
\\
 W A + A\trans W + C\trans \Pi\trans C &= 0,
\\
 VW &= I.
\end{align}\normalsize
\end{subequations}
\end{lemma}

\begin{proof}
Note that \en{ \Pi\Pi\trans=I }.
The \en{ \Pi }-orthogonality is rewritten as
\small\begin{align}
 \Pi \tf\simm \Pi\trans = \tf\inv .
\label{gegei}
\end{align}\normalsize
Using Section \ref{sec:circuits},
we can express the both sides as
\small\begin{align}
 \Pi P\simm \Pi\trans
=&
 \dtf{-A\trans}{C\trans\Pi\trans}{-\Pi B\trans}{\Pi D\trans\Pi\trans},
\renewcommand{\arraystretch}{1}
\\
 P\inv =&
 \dtf{A-BD\inv C}
    {BD\inv }
    {-D\inv C}
    {D\inv }.
\end{align}\normalsize
These two minimal realizations are equivalent.
From Lemma \ref{lem:twomin},
there exists a similarity transformation \en{ T } such that
\begin{subequations} 
\label{gsimilarity}
\small\begin{align}
 -A\trans &= T\inv  (A-BD\inv  C)T,
\label{gsimilarity1}\\
 C\trans\Pi\trans &= T\inv  B D\inv ,
\label{gsimilarity2}\\
 \Pi B\trans &= D\inv  CT.
\label{gsimilarity3}
\end{align}\normalsize
\end{subequations}
In particular, 
(\ref{gsimilarity2},\ \ref{gsimilarity3}) are rewritten as
\begin{subequations}
\small\begin{align}
 BD\inv  
&=
 TC\trans\Pi\trans,
\\
 CT
&=
 D\Pi B\trans.
\end{align}\normalsize
\end{subequations}
Substituting these into (\ref{gsimilarity1}) yields
\begin{subequations}
\small\begin{align}
 AT + TA\trans - B\Pi B\trans  &= 0,
\\
 T\inv  A+A\trans T\inv  - C\trans \Pi\trans C &= 0.
\end{align}\normalsize
\end{subequations}
Setting \en{ T=-V } and \en{ T\inv =-W } establishes the assertion.
\end{proof}

\newpage

\begin{lemma}
\label{thm:gss}
A minimal realization (\ref{gsr}) is \en{ \Pi }-orthogonal 
if and only if there exists \en{ V } such that
\begin{subequations}
\label{gns}
 \small\begin{align}
 AV + VA\trans + B\Pi B\trans &=0,
\label{gns1} \\
 CV + D\Pi B\trans &=0,
\label{gns2} \\
  D\Pi D\trans &= \Pi,
\label{gns3}
 \end{align}\normalsize
\end{subequations}
and 
\small\begin{align}
 V\trans = \pm V 
\qquad 
 \kakkon{+: \mbox{fermion,}}
        {-: \mbox{boson.}}
\end{align}\normalsize
\end{lemma}

\begin{proof}
The necessity is obvious from (\ref{gsimilarity}). 
To show the sufficiency, 
define \en{ T } as
\small\begin{align}
 T=\A{1}{V}{}{1}, 
\qquad
 T\inv =\A{1}{-V}{}{1}.
\end{align}\normalsize
Then we have
\begin{subequations}
\small\begin{align}
 P\Pi P\simm
 &=
 \dtf{A}{B}{C}{D}\Pi
 \dtf{-A\trans}{C\trans}{-B\trans}{D\trans}
%%%%%%%%%%%%%%%%%%%%%%%%%%%%%%%%%%%
\\ &=
 \dtf{\nA{A}{-B\Pi B\trans}{0}{-A\trans}}
     {\nAV{B\Pi D\trans}{C\trans}}
     {\nAH{C}{-D\Pi B\trans}}
     {D\Pi D\trans}
\hspace{18mm}
\because \mbox{Section }\ref{sec:circuits}
%%%%%%%%%%%%%%%%%%%%%%%%%%%%%%%%%%%
\\ &=
\left[
  \begin{array}{c|c}
    T\A{A}{-B\Pi B\trans}{0}{-A\trans}T\inv 
&  
   T\AV{B\Pi D\trans}{C\trans} 
  \\ \hline
     \AH{C}{-D\Pi B\trans}T\inv 
& 
   D\Pi D\trans
 \rule[0mm]{0mm}{5mm}
  \end{array}
\right]
\qquad
\because (\ref{simtrans})
%%%%%%%%%%%%%%%%%%%%%%%%%%%%%%%%%%%
\\ &=
\left[
  \begin{array}{cc|c}
     A & -AV-VA\trans -B\Pi B\trans & VC\trans +B\Pi D\trans \\
     0 & -A\trans & C\trans \\ \hline
     C & -CV-D\Pi B\trans & D\Pi D\trans
 \rule[0mm]{0mm}{4mm}
  \end{array}
\right].
\end{align}\normalsize
\end{subequations}
Note that 
\small\begin{align}
 \Pi\trans = \pm \Pi 
\qquad 
 \kakkon{+: \mbox{fermion,}}
        {-: \mbox{boson.}}
\end{align}\normalsize
Thus, (\ref{gns2}) is rewritten as
\small\begin{align}
  VC\trans + B\Pi D\trans &=0,
\end{align}\normalsize
As a result, 
we have
\small\begin{align}
 P\Pi P\simm
 &=
\left[
  \begin{array}{cc|c}
      A & 0  & 0 \\
      0 & -A\trans & C\trans \\ \hline
     C & 0 & \Pi \rule[0mm]{0mm}{4mm}
  \end{array}
\right]
=
 \Pi,
\end{align}\normalsize
which establishes the assertion.
\end{proof}

\newpage

\begin{corollary}
Consider a diagonal system 
\small\begin{align}
 \tf
=
 \A{\tf_{11}}{}{}{\tf_{22}}
\end{align}\normalsize
with minimal realizations
\begin{subequations}
\label{sp}
\small\begin{align}
 \tf_{11}=&\dtf{A_{11}}{B_{11}}{C_{11}}{D_{11}},
\\
 \tf_{22}=&\dtf{A_{22}}{B_{22}}{C_{22}}{D_{22}}.
\end{align}\normalsize
\end{subequations}
This is \en{ \Pi }-orthogonal \en{ \tf \Pi \tf\simm=\Pi }
if and only if
there exists \en{ V } such that
\begin{subequations}
\label{ns}
 \small\begin{align}
 A_{11} V + V A_{22}\trans + B_{11} B_{22}\trans &=0,
\\
 V C_{22}\trans + B_{11} D_{22}\trans &=0,
\\
 C_{11} V + D_{11}B_{22}\trans &=0,
\\
 D_{11} D_{22}\trans &= I.
 \end{align}\normalsize
\end{subequations}
\end{corollary}
\begin{proof}
The assertion is readily proved by replacing \en{ V } to 
\en{ \A{0}{V}{\pm V\trans}{0} }
in (\ref{gns}).
\end{proof}

\begin{example}
\rm
Lemma \ref{thm:gss}
indicates that 
there do not exist \textit{minimal} bosonic systems with only a single variable.
To see this, suppose a system of the form 
\small\begin{align}
 \AV{z_1}{z_2}
&=
 \tf \AV{u_1}{u_2},
\end{align}\normalsize
where 
\small\begin{align}
\hspace{10mm}
 \tf
=
 \dtf{a}
    {\nAH{b_1}{b_2}}
    {\nAV{c_1}{c_2}}
    {D},
\quad
 (a\not= 0.)
\end{align}\normalsize
The \en{ D }-matrix is a direct-through term 
from the input to the output,
which corresponds to 
the zeroth order term of the \textit{S}-matrix.
Hence \en{ D } is non-singular: \en{ \det D\not= 0 }.
We first note that 
\small\begin{align}
 B\Pi B\trans 
=
 \AH{b_1}{b_2}
 \A{}{1}{-1}{}
 \AV{b_1}{b_2}
=0,
\end{align}\normalsize
which implies \en{ V=0 } in (\ref{gns1}).
Then it follows from (\ref{gns2}) that 
\small\begin{align}
 \AV{b_1}{b_2} = 0.
\end{align}\normalsize
As a result, 
the system is not controllable,
i.e.,
it is not minimal.
We have seen this situation 
in the QND and XX (non-unitary) systems
in Section \ref{sec:sumsys}.
\qed
\end{example}

\begin{remark}
 Unlike bosons,
there can be nontrivial single-variable systems for fermions.
In fact, 
\small\begin{align}
 B\Pi B\trans 
=
 2 b_1b_2.
\end{align}\normalsize
If \en{ b_1b_2\not= 0 },
then \en{ V=-b_1b_2/a } in (\ref{gns1}).
It turns out from (\ref{gns2}) that
\small\begin{align}
 C=\frac{a}{b_1b_2}D \Pi B\trans.
\end{align}\normalsize
In this case, the system can be minimal.
However, 
it is unphysical
because there are no non-unitary systems for fermions.
\end{remark}

\begin{example}
\rm
\label{exp:2va}
Let us show a parameterization of stable bosonic systems
with two variables
using Lemma \ref{thm:gss}.
Suppose that the system has a minimal realization as
\small\begin{align}
  \AV{z_1}{z_2}&=
 \dtf{A}{B}{C}{D}
 \AV{u_1}{u_2},
\label{ss2va}
\end{align}\normalsize
where all components are \en{ 2\times 2 } matrices.
Assume that \en{ A } and \en{ D } are diagonalized as 
\begin{subequations}
\small\begin{align}
 T\inv AT &=\mbox{diag}(a_1,a_2),
\\
 S\inv DS &=\mbox{diag}(d_1,d_2).
\end{align}\normalsize
\end{subequations}
Then the system is expressed as
\small\begin{align}
 S\AV{z_2}{z_2}&=
 \dtf{\nA{a_1}{}{}{a_2}}
    {TBS\inv }
    {SCT\inv }
    {\nA{d_1}{}{}{d_2}}
  S\AV{u_1}{u_2}.
\end{align}\normalsize
Note that 
\en{ S\AV{u_1}{u_2} } and \en{ S\AV{z_1}{z_2} }
satisfy the same commutation relation
as \en{ \AV{u_1}{u_2} } and \en{ \AV{z_1}{z_2} },
respectively.
Redefining \en{ TBS\inv \to B } and \en{ SCT\inv \to C },
we have
\small\begin{align}
\tf&=  \dtf{\nA{a_1}{}{}{a_2}}
    {B}
    {C}
    {\nA{d_1}{}{}{d_2}}.
\label{tfs}
\end{align}\normalsize
Then it follows from (\ref{gns}) that
\small\begin{align}
 C=\frac{a_1+a_2}{-\det B} D \Pi B\trans \Pi, \qquad
 D=\A{d_1}{}{}{1/d_1}.
\end{align}\normalsize
% \tf=\frac{-\det B}{a_1+a_2}\Pi,
%  \A{-b_{22}}{b_{12}}{b_{21}}{-b_{11}}.
This shows a relationship between 
%fluctuation and dissipation.
the diffusion and drift coefficients.
For example, the SU(2) system (\ref{su2d}) is of this form.
\qed
\end{example}

\newpage

\section{Pole-zero symmetry}
\label{sec:symmetric}

Here we prove the pole-zero symmetry.
Denoted by \en{ \pole(\tf) } and \en{ \zero(\tf) } are
the poles and transmission zeros of \en{ \tf },
respectively.

\begin{theorem}
If \en{ \tf } is \en{ \Pi }-orthogonal \en{ \tf\Pi \tf\simm=\Pi }, 
then
\en{ \pole(\tf) = -\zero(\tf) }.
\end{theorem}

\begin{proof}
Assume that \en{ \tf } has a minimal realization
\small\begin{align}
 \tf=\dtf{A}{B}{C}{D}. 
\end{align}\normalsize
 Let \en{ z } be a transmission zero of \en{ \tf }.
 From Definition \ref{def:zero}, there exist \en{ \xi,\eta } such that
%that satisfy (\ref{zero}):
\begin{subequations}
\small\begin{align}
 (-z+A\trans)\xi + C\trans\eta &=0,
\label{zr1}\\
 B\trans \xi + D\trans \eta &=0.
\label{zr2}
\end{align}\normalsize 
\end{subequations}
Premultiplication by \en{ B\Pi } yields
\small\begin{align}
 B\Pi B\trans\xi + B\Pi D\trans\eta &= 0.
\end{align}\normalsize
From (\ref{gns}),
this turns out to be 
\small\begin{align}
 \left[(z+A)V+V(-z+A\trans)\right]\xi + VC\trans\eta &=0.
\end{align}\normalsize
Using (\ref{zr1}), we get 
\small\begin{align}
 (z+A)V\xi &=0.
\end{align}\normalsize
Since \en{ V } is nonsingular,
\en{ -z } is an eigenvalue of \en{ A }.
Hence \en{ \pole\supset -\zero }.

Let \en{ p } be a pole of \en{ \tf }, 
i.e., 
there exists a vector \en{ \zeta }
 such that 
\small\begin{align}
 A\zeta&=p\zeta.
\end{align}\normalsize
Note that \en{ D } is invertible because of (\ref{gns3}).
For an arbitrary vector \en{ \xi }, 
define \en{ \eta } as
\small\begin{align}
 B\trans\xi + D\trans\eta = 0.
\label{pz1}
\end{align}\normalsize
Premultiplying by \en{ B\Pi } and using (\ref{gns}), 
we can rewrite this as
\small\begin{align}
 (AV+VA\trans)\xi + VC\trans\eta &=0.
\label{pz2}
\end{align}\normalsize
Since \en{ \xi } is arbitrary,
we take \en{ \xi=V\inv \zeta }.
Then (\ref{pz1}) and (\ref{pz2}) are rewritten as
\begin{subequations}
\small\begin{align}
 (p + VA\trans V\inv )\zeta + VC\trans\eta&= 0,
\\
 B\trans V\inv \zeta + D\trans \eta &=0,
\end{align}\normalsize
\end{subequations}
from which we get
\small\begin{align}
 \AH{\zeta\trans}{\eta\trans}
 \A{-p-V\transinv AV\trans}{-V\transinv B}{-CV\trans}{-D}
&=
 0.
\end{align}\normalsize
By Definition \ref{def:zero},
\en{ -p } is a transmission zero of the 
following transfer function:
\small\begin{align}
 \dtf{V\transinv AV\trans}{V\transinv B}{CV\trans}{D}
&=
 \dtf{A}{B}{C}{D},
\end{align}\normalsize
where we have used (\ref{simtrans}).
Hence \en{ \pole \subset -\zero }.
This establishes \en{ \pole = -\zero }.
\end{proof}

%\newpage

\section{Pole-zero symmetry for nonlinear interactions}

In this section, 
we consider the pole-zero symmetry 
for the SU(2) system with additional interactions
both in the time domain and in the frequency domain.
As in Section \ref{sec:dyson},
the effect of additional interactions is represented by 
self-energy \en{ \isel }.
The pole-zero symmetry is then expressed as conditions on \en{ \isel }.

\subsection{Frequency-domain characterization}
\label{sec:fredchark}

Let us consider an SU(2) system with additional
interactions described by a Lagrangian \en{ L }.
Denoted by \en{ \isel } is self-energy resulting from \en{ L }.
The transfer function of this system is given 
as in (\ref{41dyson0}):
\begin{subequations}
\label{const}
\small\begin{align}
 \tf 
\equiv
 \ipg\dd{\Phi\drm{out} | \Phi\drm{in}}\urm{SU(2)+L}
&=
 I 
+
 \left(-4|g|^2\right)
\left[
 I+ \ipg_{\mas|\mas}\urm{SU(2)} \isel
\right]\inv 
 \ipg_{\mas|\mas}\urm{SU(2)}
\label{omm} \\ &=
 \left[s+2|g|^2+\isel \right]\inv 
 \left[s-2|g|^2+\isel \right],
\label{const-3}
\end{align}\normalsize
\end{subequations}

\begin{lemma}
 \label{thm:pzsk}
A system of the form (\ref{const}) is \en{ \Pi }-orthogonal
if and only if 
\small\begin{align}
 (\isel)\Pi
+ 
 \Pi (\isel)\simm   
&=
 0.
\label{ek}
\end{align}\normalsize
\end{lemma}
%%%%%%%%%%%%%%%%%%%%%%%%%%%%%%%%%%%%%%%%%
\begin{proof}
It follows from (\ref{const}) that
\small\begin{align}
 \tf\simm
&=
 \left[-s-2|g|^2+(\isel)\simm \right]\left[-s+2|g|^2 +(\isel)\simm\right]\inv .
\end{align}\normalsize
Substituting this into \en{  \tf \Pi \tf\simm = \Pi }
completes the assertion.
\end{proof}
%%%%%%%%%%%%%%%%%%%%%%%%%%%%%%%%%%%%%%%%%
\begin{remark}
 Compared to (\ref{lie-reactance}),
\en{ \isel } corresponds to a reactance matrix. %a Lie algebra.
\end{remark}
%%%%%%%%%%%%%%%%%%%%%%%%%%%%%%%%%%%%%%%%%
\begin{corollary}
\label{cor:ek}
Suppose that \en{ \sel } is of the form
\small\begin{align}
\isel&=\A{\isel_{11}}{\isel_{12}}
         {\isel_{21}}{\isel_{22}}.
\end{align}\normalsize
 The system (\ref{const}) 
possesses the pole-zero symmetry
if \en{ \sel } satisfies
\begin{subequations}
\label{ekdyn}
\small\begin{align}
 \sel_{11}(s)&=-\sel_{22}(-s),
\label{ekdyn1}\\
 \sel_{12}(s)&=\sel_{12}(-s),
\label{ekdyn2}\\
 \sel_{21}(s)&=\sel_{21}(-s).
\label{ekdyn3}
\end{align}\normalsize
\end{subequations}
\end{corollary}

\begin{example}
\rm
Let us check to see if (\ref{ekdyn}) holds for linear interactions.
In the case of detuning in Section \ref{detuning},
\small\begin{align}
 \isel
&= 
 \im \Omega \sigma\dd{z}
= 
 \A{\im\Omega}{}{}{-\im\Omega}
\end{align}\normalsize
obviously satisfies (\ref{ekdyn}).
In the case of squeezing in Section \ref{sec:sq},
\small\begin{align}
 \isel
&= 
 -r\sigma\dd{x}
=
 \A{0}{-r}{-r}{0},
\end{align}\normalsize
which also satisfies (\ref{ekdyn}).
\qed
\end{example}

\subsection{Time-domain characterization}

To see the pole-zero symmetry for nonlinear interactions
in the time domain,
we first consider the state space realization of
the transfer function \en{ \Phi\drm{out}\gets\Phi\drm{in} }.

\begin{lemma}
 \label{lem:step}
Consider an SU(2) system
\small\begin{align}
 \ipg\dd{\Phi\dd{\mathrm{out}}|\Phi\dd{\mathrm{in}}}^{\mathrm{SU(2)}}
&=
 \dtf{-2|g|^2 }{-2g\uu{*} }{2g }{I}.
\end{align}\normalsize
If an additional interaction \en{ L } represented by self-energy
\small\begin{align}
\isel&=\dtf{A}{B}{C}{D}
\end{align}\normalsize
is placed in the SU(2) system,
then
\small\begin{align}
 \tf
\equiv
 \ipg\dd{\Phi\dd{\mathrm{out}}|\Phi\dd{\mathrm{in}}}^{\mathrm{SU(2)+L}}
&=
 \dtf{
\begin{array}{cc}
 -D-2|g|^2  & C  \\
 -B      & A 
\end{array}
}{
\begin{array}{c}
 -2g\uu{*}  \\
 0
\end{array}
}{
 \hspace{3mm} 2g \hspace{9mm} 0 
}{
 I
}.
\label{thm1315}
\end{align}\normalsize
\end{lemma}
%%%%%%%%%%%%%%%%%%%%%%%%%%%%%%%%%%%%%%%%%%%
\begin{proof}
We prove the assertion using (\ref{omm}):
\small\begin{align}
  \ipg_{\Phi\drm{out}|\Phi\drm{in}}\urm{SU(2)+L} 
&=
 I 
+
\left(-4|g|^2\right)
\left[
 I+ \ipg_{\mas|\mas}\urm{SU(2)} \isel
\right]\inv 
 \ipg_{\mas|\mas}\urm{SU(2)}.
\label{omm-1}
\end{align}\normalsize
Let us rewrite this in the time domain 
using Section \ref{sec:circuits}.
It follows from (\ref{2.13}) that 
\begin{subequations}
\small\begin{align}
 \ipg_{\mas|\mas}\urm{SU(2)} \isel &=
 \dtf{-2|g|^2}{I}{I}{0}\dtf{A}{B}{C}{D}
\\ &=
 \dtf{\nA{-2|g|^2}{C}{0}{A}}
    {\nAV{D}{B}}
    {\hspace{2mm} I\hspace{7mm}0}{0}.
\end{align}\normalsize
\end{subequations}
Furthermore, using (\ref{2.15}), we get
\begin{subequations}
\small\begin{align}
 \Bigl[1+\ipg\urm{SU(2)}_{\mas|\mas} \isel \Bigr]\inv  
&=
 \dtf{\A{-2|g|^2}{C}{0}{A}-\AV{D}{B}\AH{I}{0}}
    {\nAV{D}{B}}
    {\hspace{4mm}-1 \hspace{9mm}0}
    {I}
\\ &=
 \dtf{\nA{-D-2|g|^2}{C}{-B}{A}}
    {\nAV{D}{B}}
    {\hspace{4mm}-I \hspace{9mm}0}
    {I}.
\end{align}\normalsize
\end{subequations}
Hence (\ref{omm-1}) is expressed as
\begin{subequations}
\small\begin{align}
 \ipg_{\Phi\drm{out}|\Phi\drm{in}}\urm{SU(2)+L} 
&=
 I+
 \dtf{\nA{-D-2|g|^2}{C}{-B}{A}}
    {\nAV{D}{B}}
    {\hspace{4mm}-I \hspace{9mm}0}
    {I}
 \dtf{-2|g|^2}{-2g\uu{*}}{2g}{0}
\\&=
 \dtf{
 \begin{array}{ccc}
  -D-2|g|^2  & C & 2gD \\
  -B       & A & 2gB \\
  0        & 0 & -2|g|^2
 \end{array}
}{
 \begin{array}{c}
  0 \\
  0 \\
  -2g\uu{*}
 \end{array}
}{
 \hspace{3mm} -I \hspace{10mm} 0 \hspace{6mm} 2g
}{
 I
}.
\end{align}\normalsize
\end{subequations}
Note that 
there is an unobservable mode that can be eliminated
in the same way as Example \ref{ex:similarity}.
To see this,
recall that 
the transfer function is invariant under a similarity
transformation, see (\ref{simtrans}).
Choose \en{ T } in (\ref{simtrans}) as
\small\begin{align}
 T&=
 \B{I}{0}{2g}
   {0}{I}{0}
   {0}{0}{I},
\qquad
 T\inv =
 \B{I}{0}{-2g}
   {0}{I}{0}
   {0}{0}{I}.
\end{align}\normalsize
Then we get
\begin{subequations}
\small\begin{align}
\ipg_{\Phi\drm{out}|\Phi\drm{in}}\urm{SU(2)+L} 
&=
 \dtf{
T\inv \left[
 \begin{array}{ccc}
  -D-2g^2  & C & 2gD \\
  -B       & A & 2gB \\
  0        & 0 & -2|g|^2
 \end{array}
\right]T
}{
T\inv \left[
 \begin{array}{c}
  0 \\
  0 \\
  -2g\uu{*}
 \end{array}
\right]
}{ \hspace{6mm}
\left[
 \hspace{3mm} -I \hspace{12mm} 0 \hspace{9mm} 2g
\right]T
\rule[0mm]{0mm}{5mm}
}{
 I
}
\\ &=
 \dtf{
\begin{array}{ccc}
 -D-2|g|^2 & C & 0 \\
 -B      & A & 0 \\
 0       & 0 & -2|g|^2
\end{array}
}{
\begin{array}{c}
 4|g|^2 \\
 0 \\
 -2g\uu{*}
\end{array}
}{
 -I \hspace{11mm} 0 \hspace{7mm} 0
}{
 I
}
\\ &=
 \dtf{
\begin{array}{cc}
 -D-2|g|^2 & C  \\
 -B      & A 
\end{array}
}{
\begin{array}{c}
 -2g\uu{*} \\
 0
\end{array}
}{
 \hspace{3mm} 2g \hspace{9mm} 0 
}{
 I
},
\label{kstate}
\end{align}\normalsize
\end{subequations}
where an unobservable mode has been eliminated 
in the last line.
\end{proof}

\begin{theorem}
\label{thm:stateonk}
Suppose that \en{ \sel } has a minimal realization
\small\begin{align}
\isel&=\dtf{A}{B}{C}{D}.
\end{align}\normalsize
A system of the form (\ref{const}) possesses the pole-zero symmetry 
if and only if
\begin{subequations}
\label{pzk}
\small\begin{align}
 A\Pi +\Pi A\trans &=0,
\\
 B\Pi -\Pi C\trans &=0,
\\
 D\Pi +\Pi D\trans &=0.
\end{align}\normalsize
\end{subequations}
\end{theorem}

\begin{proof}
We first note that 
(\ref{thm1315}) can be rewritten as
\small\begin{align}
\ipg_{\Phi\drm{out}|\Phi\drm{in}}\urm{SU(2)+L}
&=
  \dtf{
\begin{array}{cc}
 -D-2|g|^2 & C  \\
 -B      & A 
\end{array}
}{
\begin{array}{c}
 -2|g| \\
 0
\end{array}
}{
 \hspace{3mm} 2|g| \hspace{9mm} 0 
}{
 I
}.
\end{align}\normalsize
Hence Lemma \ref{thm:gss} can be written as
\begin{subequations}
\small\begin{align}
& \hspace{-3mm} 
 \A{-D-2|g|^2}{C}{-B}{A} V + V \A{-D\trans -2|g|^2}{-B\trans}{C\trans}{A\trans}
 + \AV{-2|g|}{0}\Pi \AH{-2|g|}{0}
=0,
 \\
& \hspace{-3mm} 
\AH{2|g|}{0} V + \Pi\AH{-2|g|}{0}=0.
\end{align}\normalsize
\end{subequations}
The second equation indicates that \en{ V=\Pi }.
The assertion is achieved by substituting it into the first equation.
\end{proof}

\begin{remark}
The sufficiency can be easily shown from 
Lemma \ref{thm:pzsk}.
\end{remark}

\chapter{Nonlinear interactions}
\label{chap:nonlinear}
\thispagestyle{fancy}

As an example of the pole-zero symmetry,
we consider 
a third-order nonlinear interaction.\index{third-order nonlinearity}
This is similar to a well-known self-interacting 
\en{ \phi^4 } theory of a scalar field.
A difference is that 
the nonlinear interaction is now placed in an SU(2) system.
We examine self-energy \en{ \isel } 
up to second order.
\en{ \isel } is a constant matrix 
in the first order (linear) approximation.
In the second order approximation,
\en{ \isel } exhibits instability,
as expected from the pole-zero symmetry
of the preceding chapter.
This produces bistability in the SU(2) system.

\section{Third-order nonlinear interaction}

\subsection{Interaction Lagrangian}

Let us consider a Lagrangian of the form
\small\begin{align}
  \lag = \lag\urm{f} +\lag\urm{int},
\end{align}\normalsize
where 
\begin{subequations}
\small\begin{align}
 \lag\urm{f} 
&\equiv 
 \lag\urm{SU(2)}_{\mas}+\lag\urm{DT},
\\
 \lag\urm{int} 
&\equiv 
 \ffrac{\lambda}{2} \mas\dgg \mas\dgg \mas \mas.
\end{align}\normalsize
\end{subequations}
\en{ \lag\urm{f} } describes linear interactions for 
an SU(2) system with detuning (Section \ref{detuning}).
We have already obtained 
the transfer functions of this linear component.
\en{ \lag\urm{int} } is a nonlinear interaction
placed in the SU(2) system.
Our purpose is to calculate a transfer function
\en{ \Phi\drm{out}\gets\Phi\drm{in} }:
\small\begin{align}
 \AVl{\phi}{\phi\dgg}\drm{out}
=
 \ipg\dd{\Phi\drm{out} | \Phi\drm{in}}\urm{f+int}
 \AVl{\phi}{\phi\dgg}\drm{in},
\end{align}\normalsize
where
\small\begin{align}
 \ipg\dd{\Phi\drm{out} | \Phi\drm{in}}\urm{f+int}
\equiv
 \A{\ipg\dd{\phi\drm{out} | \phi\drm{in}}\urm{f+int}}
   {\ipg\dd{\phi\drm{out} | \phi\drm{in}\dgg}\urm{f+int}}
   {\ipg\dd{\phi\drm{out}\dgg | \phi\drm{in}}\urm{f+int}}
   {\ipg\dd{\phi\drm{out}\dgg | \phi\drm{in}\dgg}\urm{f+int}
   \rule[0mm]{0mm}{5mm}}.
\label{1stmat}
\end{align}\normalsize
This is done by finding self-energy 
corresponding to the nonlinear interaction
in the first and second order approximations.

\subsection{Self-energy}

For nonlinear interactions,
self-energy depends on the order of approximation.
The expansion of the \textit{S}-matrix is written as
\begin{subequations}
\small\begin{align}
& \wick{\Phi\drm{out}}{\Phi\drm{in}\ddgg}\urm{f}
\\ & + 
%%%%%%%%%%%%%%%
 \frac{1}{1!}\int dx_2 \ 
 \mzero{\T \, \Phi\drm{out} \  \im\lag\urm{int}(x_2) \ \Phi\drm{in}\ddgg }
\\ & + 
%%%%%%%%%%%%%%%
 \frac{1}{2!}\int dx_3 dx_2 \ 
 \mzero{\T \, \Phi\drm{out} \,
 \im \lag\urm{int}(x_3) \ \im\lag\urm{int}(x_2) \
 \Phi\drm{in}\ddgg } 
\ + \ \cdots.
\end{align}\normalsize
\end{subequations}
The \en{ n }th order term is expressed as
\en{ \ipg\dd{ \phi\drm{out} | \mas}\urm{\ f} \,
  \bigl(\isel\dd{n} \bigr) \, \ipg\dd{\mas|\phi\drm{in}}\urm{\ f} }.
Then up to second order, 
\small\begin{align}
 \ipg\dd{\Phi\drm{out} | \Phi\drm{in}}\urm{f+int}
&\sim
 \ipg\dd{\phi\drm{out} | \phi\drm{in} }\urm{\ f}
-
 \ipg\dd{\phi\drm{out} | \mas}\urm{\ f} \
\Bigl[ \isel_1 + \isel_2 \Bigr] \
 \ipg\dd{\mas| \phi\drm{in}}\urm{\ f},
\end{align}\normalsize
from which self-energy \en{ \isel } is defined as
\small\begin{align}
 \isel
\equiv
 \isel_1 + \isel_2.
\end{align}\normalsize
This is used for the expression (\ref{41dyson0}):
\small\begin{align}
 \ipg\dd{\Phi\drm{out} | \Phi\drm{in}}\urm{f+int}
&=
\left[
 1+ \ipg_{\mas|\mas}\urm{\ f} \ \isel
\right]\inv 
\left[
 \ipg\dd{\phi\drm{out} | \phi\drm{in}}\urm{\ f} 
+ 
 \ipg_{\mas|\mas}\urm{\ f} \ \isel
\right].
\label{41dyson-s}
\end{align}\normalsize

\subsection{Feynman rules}

Each \en{ \isel_n } is calculated
according to the following Feynman rules:\index{Feynman rule}
\bee
\setlength{\itemsep}{0mm} 
\item
 Draw connected diagrams 
 and define momentum to be conserved at each vertex;
\item 
 All arrows on the outside of loops need to be put 
 in the forward (causal) direction using the asymmetry of
 Theorem \ref{thm:tfbs}
 to assure a conversion between the Laplace and Fourier transforms;
\item 
 Integrate with respect to the momentum at the vertices.
\ee
Note that Rule 3 was not necessary for linear interactions.

\subsection{Nonzero vacuum expectation}

We assume that the system \en{ \mas } has a nonzero vacuum expectation
because 
\en{ \lag\urm{int} } is a `wine-bottle' type of potential.
In other words, 
\en{ \mas } is displaced in the phase space as
\small\begin{align}
\hspace{20mm}
 \mas
&\to 
 \ves + \mas, \qquad \ves \in \mathbb{C}.
\end{align}\normalsize
In the frequency domain, 
\begin{subequations}
\small\begin{align}
 \contraction{}{\phi\hspace{1mm}}{\hspace{2mm}}{}
 \mas \mas\dgg (\omega)
&=
 2\pi |\ves|^2 \delta(\omega) + \ipg_{\mas|\mas}\urm{\ f}(\omega),
\label{vacmmd}\\
 \contraction{}{\phi\hspace{1mm}}{\hspace{2mm}}{}
 \mas \mas (\omega)
&=
 2\pi \ves^2 \delta(\omega).
\label{vacmm}
\end{align}\normalsize
\end{subequations}
Note that 
\en{ \delta(\omega) } is regarded as a disconnected edge.

\newpage

\section{First order approximation}
\label{sec:1stnon}

\begin{wrapfigure}[0]{r}[53mm]{49mm}
\vspace{-0mm}
\centering
\includegraphics[keepaspectratio,width=47mm]{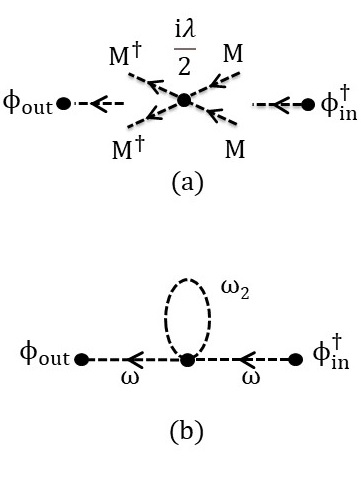}
\caption{
\small
(a) Three elementary diagrams for
 $\phi_{\mathrm{out}}$, 
 $(\im\lambda/2)\mas\dgg \mas\dgg \mas \mas$ and 
 $\phi_{\mathrm{in}}\dgg$.
(b) First order correction 
 of $\ipg_{\phi_{\mathrm{out}}|\phi_{\mathrm{in}}}^{\mathrm{f+int}}$.
 The input $\phi_{\mathrm{in}}\dgg$ is connected to one of $\mas$'s, 
 which results in $\ipg_{\mas|\phi_{\mathrm{in}}}^{\mathrm{f}}(\omega)$. 
 There are two choices of $\mas$ 
 and hence 
 we get the factor $2!$ in (\ref{1phi1}). 
 Likewise, 
 $\phi_{\mathrm{out}}$ is connected to one of $\mas\dgg$'s, 
 which results in $\ipg_{\phi_{\mathrm{out}}|\mas}^{\mathrm{f}}(\omega)$.
 The remaining $\mas$ and $\mas\dgg$ are connected to each other as 
 $\contraction{}{\phi\hspace{1mm}}{\hspace{2mm}}{}
 \mas \mas\dgg (\omega_2)$, 
 which is expressed by the loop in the Feynman diagram.
 The momentum is conserving at the vertex, i.e., 
 $\omega+\omega_2-\omega-\omega_2=0$.
\normalsize
}
\label{fig-phi-1st}
\end{wrapfigure}

Let us start with the (1.1)-element 
of the transfer function (\ref{1stmat}).
In the first order approximation,
the corresponding Feynman diagram is depicted 
in Figure \ref{fig-phi-1st}(b).
According to the Feynman rules,
it is written in the frequency domain as
\small\begin{align}
& (2!)^2 \ \im\frac{\lambda}{2} 
 \int \frac{d\omega_2}{2\pi} \
 \ipg\dd{\phi\drm{out} | \mas}\urm{\ f}(\omega) 
\
\contraction{}{\phi\hspace{1mm}}{a}{}
 \mas\mas\dgg
(\omega_2)
\
 \ipg\dd{\mas|\phi\drm{in}}\urm{\ f}(\omega)
%%%%%%%%%%%%%%%%%%%%%%%%%%%%%%%%%%%%%%%%%%%%%
\nn\\ &\hspace{23mm} =
 \ipg\dd{\phi\drm{out} | \mas}\urm{\ f}(\omega) 
\left[
 \im 2\lambda \left(|\ves|^2+\frac{1}{2}\right)
\right]
 \ipg\dd{\mas| \phi\drm{in}}\urm{\ f}(\omega),
\label{1phi1}
\end{align}\normalsize
where we have used (\ref{vacmmd}).
The other matrix elements are given
in the same way:
\begin{subequations}
\label{1phi1other}
\small\begin{align}
& \hspace{-7mm}
\ipg\dd{\phi\drm{out} |\phi\drm{in}\dgg}\urm{f+int} 
: 
 \im\lambda
\contraction{}{\phi\hspace{0mm}}{\hspace{9mm}}{}
 \phi\drm{out} : \mas\dgg 
\contraction{}{\mas}{\hspace{16mm}}{a}
 \mas\dgg 
\contraction[2ex]{}{\mas}{\hspace{1mm}}{a}
\mas\mas : (-\phi\drm{in}) \cdot 2!
=
  \ipg\dd{\phi\drm{out} | \mas}\urm{\ f}
  \left[\im \lambda \ves^2 \right]
  \ipg\dd{\mas|\phi\drm{in}}\urm{\ f},
%%%%%%%%%%%%%%%%%%%%%%%%%
\\
& \hspace{-7mm}
\ipg\dd{\phi\drm{out}\dgg | \phi\drm{in}}\urm{f+int}  
: 
 \im\lambda
\contraction{}{\phi}{\hspace{20mm}}{}
 \phi\drm{out}\dgg : 
\contraction[2ex]{}{\mas}{\hspace{2mm}}{a}
 \mas\dgg \mas\dgg \mas 
\contraction{}{\mas}{\hspace{5mm}}{}
 \mas : \phi\drm{in}\dgg \cdot 2!
=
 \ipg\dd{\phi\drm{out} | \mas}\urm{\ f} 
 \left[-\im \lambda \ves^{*2} \right]
 \ipg\dd{\mas|\phi\drm{in}}\urm{\ f},
%%%%%%%%%%%%%%%%%%%%%%%%%
\\
& \hspace{-7mm}
\ipg\dd{\phi\drm{out}\dgg | \phi\drm{in}\dgg}\urm{f+int}  
: 
 \im \lambda
\contraction{}{\phi}{\hspace{22mm}}{}
 \phi_4\dgg :  
\contraction[2ex]{}{\mas}{\hspace{23mm}}{}
 \mas\dgg 
\contraction[3ex]{}{\mas}{\hspace{2mm}}{a}
 \mas\dgg \mas \mas : (-\phi\drm{in}) \cdot (2!)^2 
=
 \ipg\dd{\phi\drm{out} | \mas}\urm{\ f}
 \left[-\im 2\lambda 
 \left(|\ves|^2+\frac{1}{2}\right) \right] 
 \ipg\dd{\mas|\phi\drm{in} }\urm{\ f}.
\end{align}\normalsize
\end{subequations}
Compared to (\ref{dysonext}),
the first order correction is written as
\small\begin{align}
 \ipg\dd{\phi\drm{out} | \mas}\urm{\ f} \
 \bigl(-\isel_1 \bigr) \
 \ipg\dd{\mas|\phi\drm{in} }\urm{\ f},
\end{align}\normalsize
where 
\small\begin{align}
 \isel_1&=
 \A{-\im 2\lambda \left(|\ves|^2+\ffrac{1}{2}\right)}
   {-\im\lambda \ves^2}
   {\im \lambda \ves^{*2}}
   {\im 2\lambda\left(|\ves|^2+\ffrac{1}{2}\right)}.
\label{k1}
\end{align}\normalsize

\begin{remark}
The same result can be obtained 
in the time domain as we did before.
For example, 
the LHS of (\ref{1phi1}) is written as
\begin{subequations}
\small\begin{align}
& \im 2 \lambda
 \int dt_2 \
\contraction{}{\phi}{\hspace{9mm}}{}
 \phi\dd{\mathrm{out}} :\mas\dgg_2 
\contraction{}{\mas}{\hspace{4mm}}{}
 \mas_2\dgg \mas_2
\contraction{}{\mas}{\hspace{5mm}}{}
 \mas_2 :\phi\dd{\mathrm{in}}\dgg
%%%%%%%%%%%%%%%%%%%%%%%%%%%%%%%%%%%%%%%%%
\\ = \ &
 \im 2 \lambda
 \int dt_2
 \int \frac{d\omega}{2\pi} 
 \ex\uu{-\im \omega(t\dd{\mathrm{out}} - t\dd{2})} \
 \ipg\dd{\phi\dd{\mathrm{out}} |\mas}\uu{\mathrm{f}}(\omega)
\\ &\hspace{17mm}
 \int \frac{d\omega_2}{2\pi} 
 \ex\uu{-\im \omega\dd{2} (t\dd{2} - t\dd{2})} 
\left[
 2\pi |\ves|^2 \delta(\omega_2) + \ipg\dd{\mas|\mas}\uu{\mathrm{f}}(\omega_2)
\right]
\\ &\hspace{17mm}
 \int \frac{d\omega_3}{2\pi} 
 \ex\uu{-\im \omega\dd{3} (t\dd{2} - t\dd{\mathrm{in}})} \
 \ipg\dd{\mas|\phi\dd{\mathrm{in}} }\uu{\mathrm{f}}(\omega_3)
%%%%%%%%%%%%%%%%%%%%%%%%%%%%%%%%%%%%%%%%%
\\ = \ &
\im 2\lambda
 \int \frac{d\omega}{2\pi}
 \ex\uu{-\im \omega(t\dd{\mathrm{out}} - t\dd{\mathrm{in}})} \
 \ipg\dd{\phi\dd{\mathrm{out}} |M}\uu{\mathrm{f}}(\omega) 
 \left[ |\ves|^2+\frac{1}{2} \right]
 \ipg\dd{\mas|\phi\dd{\mathrm{in}} }\uu{\mathrm{f}}(\omega).
\end{align}\normalsize
\end{subequations}
This is equivalent to the RHS of (\ref{1phi1}) after the Fourier transform.
\end{remark}

\newpage

In the first order approximation,
the self-energy \en{ \isel_1 } is a constant (static) matrix.
The transfer function is then given in the same form as (\ref{constl}):
\begin{subequations}
\label{1stdyson}
\small\begin{align}
 \ipg_{\Phi\drm{out} | \Phi\drm{in}}\urm{f+int}
&= 
\left[
 s+2|g|^2+\im\Omega \sigma\dd{z} + \isel_1
\right]\inv 
 \left[
 s-2|g|^2+\im\Omega \sigma\dd{z} + \isel_1
\right] 
\\ &=
 \dtf{-2|g|^2-\im\Omega \sigma\dd{z} -\isel_1}{-2g\uu{*}}{2g}{I},
\label{1stdyson2}
\end{align}\normalsize
\end{subequations}

It is not difficult to see 
the pole-zero symmetry \en{ \pole=-\zero }\index{pole-zero symmetry}
in this transfer function.
Let us set
\small\begin{align}
 \isel
&= 
 \im \Omega\sigma\dd{z}+\isel_1.
\end{align}\normalsize
Then this satisfies Lemma \ref{thm:pzsk}:
\small\begin{align}
 (\isel)\Pi + \Pi(\isel)\simm &=0.
\end{align}\normalsize

It is also possible to calculate
the poles \en{ \pole } and transmission zeros \en{ \zero } explicitly.
The pole \en{ p } is defined by 
the eigenvalues of the \en{ A }-matrix (the upper-left corner)
of (\ref{1stdyson2}):
\small\begin{align}
&\hspace{-4mm}
\det \A{p+2|g|^2+\im\Omega-\im 2\lambda \Bigl(|\ves|^2+\ffrac{1}{2}\Bigr)}
       {-\im \lambda \ves^2}
       {\im \lambda \ves^{\ast 2}}
       {p+2|g|^2-\im\Omega+\im 2\lambda \Bigl(|\ves|^2+\ffrac{1}{2}\Bigr)}
=0,
\end{align}\normalsize
from which 
\small\begin{align}
 p&=
 -2|g|^2 
\pm
 \sqrt{-(3\lambda|\ves|^2
        + \lambda-\Omega)(\lambda|\ves|^2+\lambda-\Omega)}.
\end{align}\normalsize
According to Definition \ref{def:zero},
the transmission zero \en{ z} is given by
\small\begin{align}
& \hspace{-4mm}
\det 
\left[
\begin{array}{cccc}
 z+2|g|^2+\im\Omega-\im 2\lambda\Bigl(|\ves|^2+\ffrac{1}{2}\Bigr) 
& 
 -\im \lambda \ves^2 & 2g\uu{*} & 0 \\
 \im \lambda \ves^{\ast 2} 
& \hspace{-9mm}
 z+2|g|^2-\im\Omega+\im 2\lambda\Bigl(|\ves|^2+\ffrac{1}{2}\Bigr) 
  & 0 & 2g\uu{*} \\
 -2g & 0 & -1 & 0 \\
 0 & -2g & 0 & -1 
\end{array}
 \right]
=0,
\end{align}\normalsize
which yields 
\small\begin{align}
 z
&=
 2|g|^2\pm
 \sqrt{-(3\lambda|\ves|^2+\lambda-\Omega)(\lambda|\ves|^2+\lambda-\Omega)}.
\end{align}\normalsize
As a result, \en{ \pole = -\zero }.

\begin{remark}
If \en{ \Omega } satisfies
\small\begin{align}
 (3\lambda|\ves|^2+\lambda-\Omega)(\lambda|\ves|^2+\lambda-\Omega)<0,
\end{align}\normalsize
it is possible to locate a pole near \en{ s=0 }.
Then the system exhibits the same effect as squeezing.
%as mentioned in Remark \ref{rem:sq}. 
\end{remark}

\newpage

\section{Second order approximation}
\label{sec:secorapp}

\begin{wrapfigure}[0]{r}[53mm]{49mm} 
\vspace{-15mm}
\centering
\includegraphics[keepaspectratio,width=47mm]{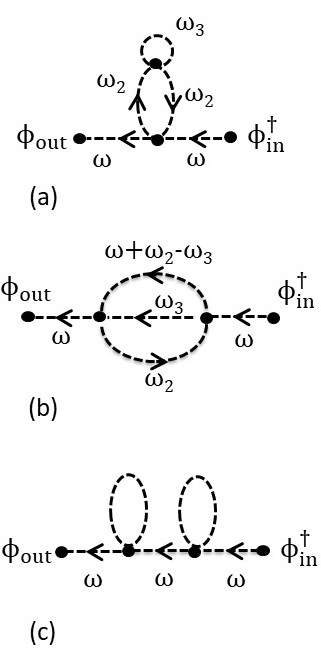}
\caption{
\small
Three diagrams for the second order correction.
\normalsize
}
\label{fig-2phi1-1}
\end{wrapfigure}

There are three different diagrams 
for the second order correction,
as in Figure \ref{fig-2phi1-1}.
Note that (c) is not necessary
because it is the cascade of the first order correction
and included in (\ref{1stdyson}) 
as a part of the Dyson series.
For simplicity, 
we assume that 
\en{ \Omega=0 } and \en{ g\in\mathbb{R} }.

\subsection{(1,1)-element of the transfer function (\ref{1stmat})}

\begin{wrapfigure}[0]{r}[53mm]{49mm} 
\vspace{80mm}
\centering
\includegraphics[keepaspectratio,width=42mm]{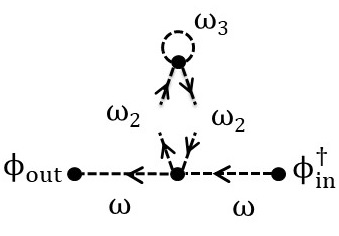}
\caption{
\small
Disconnected diagram corresponding to 
Figure \ref{fig-2phi1-1}(a).
\normalsize
}
\label{fig-disconnect-1}
\end{wrapfigure}

The (1,1)-element, 
\en{ \ipg\dd{\phi\drm{out} | \phi\drm{in}}\urm{f+int} },
is written as
\small\begin{align}
 \ipg\dd{\phi\drm{out} | \mas}\urm{\ f}(s)
\
\left[
- \isel_{11}\uu{(a)} - \isel_{11}\uu{(b)}(s)
\right]
\
 \ipg\dd{\mas|\phi\drm{in} }\urm{\ f}(s).
\end{align}\normalsize
where
\en{ K_{11}\uu{(a)} } and \en{ K_{11}\uu{(b)} } are 
self-energy corresponding to
Figure \ref{fig-2phi1-1}(a) and (b), 
respectively.
Applying the Feynman rules to Figure \ref{fig-2phi1-1}(a), 
we get \en{ \isel_{11}\uu{(a)}  } as
\begin{subequations}
\label{dis1}
\small\begin{align}
- \isel_{11}\uu{(a)} 
&= 
 \frac{2^5}{2!} \Bigl(\frac{\im\lambda}{2}\Bigr)^2
 \int \frac{d\omega_3}{2\pi} \frac{d\omega_2}{2\pi}
\contraction{}{\mas}{\hspace{2mm}}{}
 \mas \mas\dgg (\omega_3)
\left[
\contraction{}{\mas}{\hspace{2mm}}{}
 \mas \mas\dgg (\omega_2)
\right]^2
\\ &=
 (2\im\lambda)^2 
 \int \frac{d\omega_3}{2\pi}
\left[
 2\pi |\ves|^2 \delta(\omega_3) + \frac{\im}{\omega_3+\im 2g^2}
\right]
\\ & \hspace{15mm}
 \int \frac{d\omega_2}{2\pi}
\left[
 2\pi |m|^2 \delta(\omega_2) + \frac{\im}{\omega_2+\im 2g^2}
\right]^2
%\label{dis1-1}\\ &=
\\ &=
 (2\im\lambda)^2 
\left(
 |\ves|^2+\frac{1}{2}
\right)
 \frac{|\ves|^2}{g^2}.
\end{align}\normalsize
\end{subequations}

\begin{wrapfigure}[0]{r}[53mm]{49mm} 
\vspace{50mm}
\centering
\includegraphics[keepaspectratio,width=47mm]{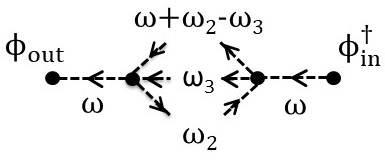}
\caption{
\small
Disconnected diagram corresponding to 
Figure \ref{fig-2phi1-1}(b).
\normalsize
}
\label{fig-disconnect-2}
\end{wrapfigure}

\noindent
We have ignored 
\en{ \delta(\omega_3)[\delta(\omega_2)]^2 } 
in the last line
because 
the \en{ \delta }-function is regarded as a disconnected edge,
as in Figure \ref{fig-disconnect-1}.
Likewise, 
it follows from Figure \ref{fig-2phi1-1}(b) that 
\small\begin{align}
- \isel_{11}\uu{(b)}(\omega)
&=
 \frac{2^5}{2!} \Bigl(\frac{\im\lambda}{2}\Bigr)^2 
 \int \frac{d\omega_2}{2\pi}\frac{d\omega_3}{2\pi} \
\contraction{}{\mas}{\hspace{2mm}}{}
 \mas \mas\dgg (\omega+\omega_2-\omega_3)
\contraction{}{\mas}{\hspace{2mm}}{}
 \mas \mas\dgg (\omega_3)
\contraction{}{\mas}{\hspace{2mm}}{}
 \mas \mas\dgg (\omega_2)
\nn\\ &=
 (2\im\lambda)^2
 \int \frac{d\omega_2}{2\pi}\frac{d\omega_3}{2\pi}
\left[
 2\pi |\ves|^2 \delta(\omega+\omega_2-\omega_3) 
   + \frac{\im}{\omega+\omega_2-\omega_3+\im 2 g^2}
\right]
\nn\\ &\hspace{24mm}
\left[
 2\pi |\ves|^2 \delta(\omega_3) + \frac{\im}{\omega_3+\im 2g^2}
\right]
\left[
 2\pi |m|^2 \delta(\omega_2) + \frac{\im}{\omega_2+\im 2g^2}
\right]
\nn\\ &=
 (2\im\lambda)^2
\left(
 \frac{2\im |\ves|^4}{\omega+\im 2g^2}
-
 \frac{\im |\ves|^4}{\omega-\im 2g^2}
+
 \frac{\im |\ves|^2}{\omega+\im 4g^2}
\right).
\label{dis2}
\end{align}\normalsize
Again 
\en{ \delta(\omega+\omega_2-\omega_3) \, \delta(\omega_2) \, \delta(\omega_3) } 
has been ignored 
because it is disconnected as in Figure \ref{fig-disconnect-2}.
As a result,
we get 
\begin{subequations}
\label{sel1120}
\small\begin{align}
- \isel_{11}\uu{(a)} 
&= 
 (2\im\lambda)^2 
\left(
 |\ves|^2+\frac{1}{2}
\right)
 \frac{|\ves|^2}{g^2},
\\
- \isel_{11}\uu{(b)}(s)
&=
 (2\im\lambda)^2
\left(
 \frac{2 |m|^4}{s+ 2g^2}
-
 \frac{|m|^4}{s- 2g^2}
+
 \frac{ |m|^2}{s+ 4g^2}
\right).
\label{sel112}
\end{align}\normalsize
\end{subequations}

\subsection{(2,2)-element of the transfer function (\ref{1stmat})}

The (2,2)-element,
\en{ \ipg\dd{\phi\drm{out}\dgg | \phi\drm{in}\dgg}\urm{f+int} },
is calculated in the same way as the (1,1)-element:
\small\begin{align}
 \ipg\dd{\phi\drm{out} |\mas}\urm{\ f} (s)
\Bigl[
 \isel_{11}\uu{(a)} + \isel_{11}\uu{(b)}(-s)
\Bigr]
\ipg\dd{\mas|\phi\drm{in} }\urm{\ f} (s),
\end{align}\normalsize
where 
\en{ \isel_{11}\uu{(a)} } and 
\en{ \isel_{11}\uu{(b)} } are 
the same self-energy as (\ref{sel1120}).

\subsection{Off-diagonal elements of the transfer function (\ref{1stmat})}

\begin{wrapfigure}[0]{r}[53mm]{49mm} 
\vspace{-20mm}
\centering
\includegraphics[keepaspectratio,width=47mm]{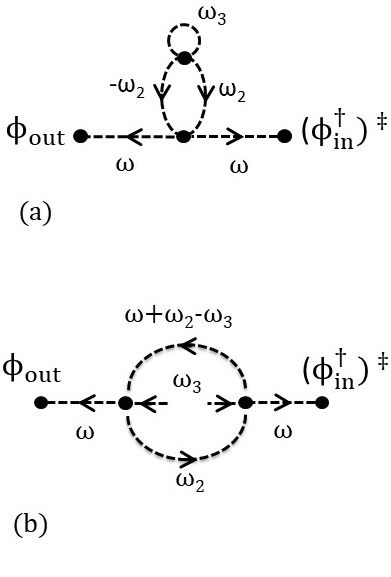}
\caption{
\small
(a)
The backward arrow (left to right) 
needs to be reversed using Theorem \ref{thm:tfbs}, i.e.,
$\ipg_{\mas\dgg | \phi_{\mathrm{in}}\dgg}^{\mathrm{f}}
=\wick{\mas\dgg}{(\phi_{\mathrm{in}}\dgg)\ddgg}^{\mathrm{f}}
=(\wick{\mas}{\phi_{\mathrm{in}}\dgg}^{\mathrm{f}})^{*}
=\ipg_{\mas|\phi_{\mathrm{in}}}^{\mathrm{f}}$,
where we have used the fact that $g$ is real in the last equality.
The loop on the top is a vacuum expectation
$\contraction{}{\mas}{\hspace{2mm}}{}
\mas\mas(\omega_3)=2\pi \ves^2\delta(\omega_3)$
as in (\ref{vacmm}).
(b)
There are two arrows in the middle loop that cannot be connected
to each other,
which corresponds to a disconnected edge 
$\contraction{}{\mas}{\hspace{2mm}}{}
\mas\mas(\omega_3)=2\pi \ves^2\delta(\omega)$.
\normalsize
}
\label{fig-2phi2-2}
\end{wrapfigure}

The (1,2)-element,
\en{ \ipg\dd{\phi\drm{out} | \phi\drm{in}\dgg}\urm{f+int} },
is expressed as
\small\begin{align}
 \ipg\dd{\phi\drm{out} | \mas}\urm{\ f}(s)
\
\Bigl[
- \isel_{12}\uu{(a)}
- \isel_{12}\uu{(b)}(s)
\Bigr]
\
\ipg\dd{\mas|\phi\drm{in} }\urm{\ f}(s),
\end{align}\normalsize
where
\en{ \isel_{12}\uu{(a)} } and \en{ \isel_{12}\uu{(b)} } are from 
Figure \ref{fig-2phi2-2}(a) and (b), respectively.

It is straightforward to calculate \en{ \isel_{12}\uu{(a)} }
from Figure \ref{fig-2phi2-2}(a):
\begin{subequations}
\small\begin{align}
- \isel_{12}\uu{(a)}
&=
 \frac{2^3}{2!} \Bigl( \frac{\im\lambda}{2}\Bigr)^2
 \int \frac{d\omega_3}{2\pi}\frac{d\omega_2}{2\pi}
\contraction{}{\phi}{\hspace{3mm}}{}
 \mas\mas(\omega_3)
\contraction{}{\phi}{\hspace{3mm}}{}
 \mas \mas\dgg (\omega_2)
\contraction{}{\phi}{\hspace{3mm}}{}
 \mas \mas\dgg (-\omega_2)
\\ &=
 (\im \lambda \ves)^2
\left(
 \frac{|\ves|^2}{g^2} + \frac{1}{4g^2}
\right).
\end{align}\normalsize
\end{subequations}
Figure \ref{fig-2phi2-2}(b) is a bit tricky because 
there are two arrows that cannot be connected to each other.
This part has to be replaced to the vacuum expectation (\ref{vacmm}):
\small\begin{align}
 \contraction{}{\mas}{\hspace{2mm}}{}
 \mas \mas (\omega_3)
=
 2\pi \ves^2 \delta(\omega_3).
\end{align}\normalsize
Then we get
\small\begin{align}
- \isel_{12}\uu{(b)}(\omega)
&=
 \frac{2^5}{2!} \Bigl( \frac{\im\lambda}{2}\Bigr)^2 
 \int \frac{d\omega_2}{2\pi}\frac{d\omega_3}{2\pi}
\contraction{}{\mas}{\hspace{2mm}}{}
 \mas \mas\dgg (\omega+\omega_2-\omega_3)
\contraction{}{\mas}{\hspace{2mm}}{}
 \mas \mas (\omega_3)
\contraction{}{\mas}{\hspace{2mm}}{}
 \mas \mas\dgg (\omega_2)
\nn\\ &=
 (2\im \lambda \ves)^2 
 \int \frac{d\omega_2}{2\pi} 
\contraction{}{\mas}{\hspace{2mm}}{}
 \mas \mas\dgg (\omega+\omega_2)
\contraction{}{\mas}{\hspace{2mm}}{}
 \mas \mas\dgg (\omega_2)
\nn\\ &=
 (2\im \lambda \ves)^2 
\left(
 \frac{\im |\ves|^2}{\omega+\im 2g^2} 
-
 \frac{\im |\ves|^2}{\omega-\im 2g^2}
\right),
\label{int12}
\end{align}\normalsize
As a result,
we have 
\begin{subequations}
\label{selofab}
\small\begin{align}
- \isel_{12}\uu{(a)}
&=
 (\im \lambda \ves)^2
\left(
 \frac{|\ves|^2}{g^2} + \frac{1}{4g^2}
\right),
\\
- \isel_{12}\uu{(b)}(s)
&=
 (2\im \lambda \ves)^2 
\left(
 \frac{|\ves|^2}{s+ 2g^2} 
-
 \frac{|\ves|^2}{s- 2g^2}
\right).
\end{align}\normalsize
\end{subequations}

Likewise, 
the (2,1)-element is obtained as
\small\begin{align}
 \ipg\dd{\phi\drm{out} | \mas}\urm{\ f} (s)
\frac{\ves^{*2}}{\ves^2}
\Bigl[
 \isel_{12}\uu{(a)}
+\isel_{12}\uu{(b)}(s)
\Bigr]
\ipg\dd{\mas|\phi\drm{in} }\urm{\ f}(s),
\end{align}\normalsize
where \en{ \isel_{12}\uu{(a)} } and \en{ \isel_{12}\uu{(b)}(s) }
are the same self-energy as (\ref{selofab}).

\section{Pole-zero symmetry}

Eventually,
the self-energy is given by 
\small\begin{align}
 \isel(s) = \isel_1 + \isel_2(s),
\end{align}\normalsize
where \en{ \isel_1 } is the first order correction 
given in (\ref{k1}):
\small\begin{align}
 \isel_1&=
 \A{-\im 2\lambda \left(|\ves|^2+\ffrac{1}{2}\right)}
   {-\im\lambda \ves^2}
   {\im \lambda \ves^{*2}}
   {\im 2\lambda\left(|\ves|^2+\ffrac{1}{2}\right)}.
\end{align}\normalsize
\en{ \isel_2 } is the second order correction 
that has been calculated in the preceding section:
\small\begin{align}
 \isel_2 (s)
&=
 \A{ \isel_{11}\uu{(a)}+\isel_{11}\uu{(b)}(s)}
   { \isel_{12}\uu{(a)}+\isel_{12}\uu{(b)}(s)}
   {-\ffrac{\ves\uu{*2}}{\ves^2}\left(\isel_{12}\uu{(a)}
                               +\isel_{12}\uu{(b)}(s)\right)}
   {-\isel_{11}\uu{(a)}-\isel_{11}\uu{(b)}(-s)}.
\label{se2nd}
\end{align}\normalsize
The transfer function \en{ \Phi\drm{out}\gets\Phi\drm{in} }
is obtained by substituting this self-energy
into (\ref{41dyson-s}).
Using Theorem \ref{thm:tfbs},
it is expressed as
\small\begin{align}
 \ipg\dd{\Phi\drm{out} | \Phi\drm{in}}\urm{f+int}
&=
 \left[ s+2g^2 +\isel(s) \right]\inv 
 \left[ s-2g^2 +\isel(s) \right].
\label{pzchecktf}
\end{align}\normalsize

This transfer function is well-defined as a quantum system
only when 
it satisfies the pole-zero symmetry.\index{pole-zero symmetry}
To see this,
we use Corollary \ref{cor:ek}.
Note that (\ref{pzchecktf}) is the same form as (\ref{const-3}).
Let us express the self-energy as
\small\begin{align}
 \isel
= \isel_1 + \isel_2
&=
 \A{\isel_{11}}{\isel_{12}}
   {\isel_{21}}{\isel_{22}}.
\end{align}\normalsize
From Corollary \ref{cor:ek},
the pole-zero symmetry is satisfied if
\begin{subequations}
\label{eknon}
\small\begin{align}
 \sel_{11}(s)&=-\sel_{22}(-s),
\label{eknon1}\\
 \sel_{12}(s)&= \sel_{12}(-s),
\label{eknon2}\\
 \sel_{21}(s)&= \sel_{21}(-s).
\label{eknon3}
\end{align}\normalsize
\end{subequations}
The first order correction \en{ \isel_1 } satisfies these conditions.
The diagonal elements of the second order correction \en{ \isel_2 }
satisfy (\ref{eknon1}).
For the off-diagonal elements,
we note that they are of the form
\small\begin{align}
\left(
 \frac{1}{s+2g^2}-\frac{1}{s-2g^2}
\right).
\end{align}\normalsize
This satisfies (\ref{eknon2}) and (\ref{eknon3}).
As a result, 
\en{ \isel_2 } also satisfies the all conditions of (\ref{eknon}),
which completes the assertion.

Obviously, 
the unstable modes are essential in this proof.
In other words,
self-energy is inevitably unstable for nonlinear interactions.

\newpage

\section{Second order correlation function}

In this section,
we consider the following type of correlation function: 
\small\begin{align}
 \langle 
  \phi(x\drm{in}) \phi(x\drm{out})   
  \phi\dgg(x\drm{out}) \phi\dgg(x\drm{in})
 \rangle
=
 \langle 
  \phi\drm{in} \phi\drm{out}   
  \phi\drm{out}\dgg \phi\drm{in}\dgg
 \rangle,
\label{2correlation}
\end{align}\normalsize
where \en{ \phi(x\dd{\alpha})=\phi\dd{\alpha} }.
This is a special case of 
a four-point transfer function\index{transfer function (four-point)}
\small\begin{align}
 \ipg\dd{\phi_4 \phi_3\dgg|\phi_2 \phi_1\dgg}\urm{f+int}
&=
 \wick{\phi_4 \phi_3\dgg}{\phi_2\dgg \{-\phi_1\}}\urm{f+int}.
\label{4ptf}
\end{align}\normalsize
This is calculated 
by expanding the \textit{S}-matrix in the same way
as the (two-point) transfer functions.
Here we consider first and second order approximations.

\subsection{First order approximation}

\begin{wrapfigure}[0]{r}[53mm]{49mm}
\vspace{-5mm}
\centering
\includegraphics[keepaspectratio,width=48mm]{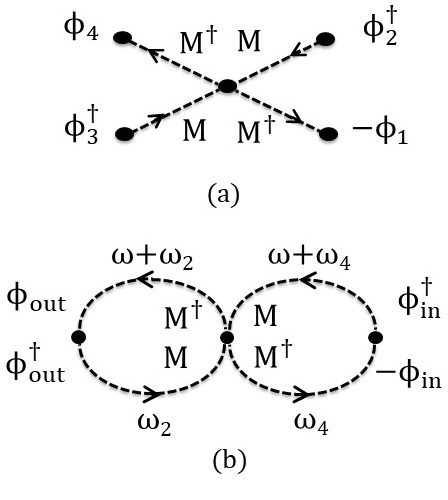}
\caption{
\small
First order correction.
\normalsize
}
\label{fig-correlation-1st}
\end{wrapfigure}

The first order approximation of (\ref{4ptf})
is depicted in Figure \ref{fig-correlation-1st}(a).
In the case of (\ref{2correlation}),
\begin{subequations}
\small\begin{align}
 x_1 & =x_2=x\drm{in}, 
\\
 x_3 & =x_4=x\drm{out},
\end{align}\normalsize
\end{subequations}
for which the diagram is re-depicted 
as Figure \ref{fig-correlation-1st}(b).
In the frequency domain,
this is written as
\begin{subequations}
\label{cor1}
\small\begin{align}
 &
 (2!)^2 \, \im \frac{\lambda}{2}
 \int\frac{d\omega_2}{2\pi} 
\contraction{}{\phi}{\hspace{6mm}}{}
 \phi\drm{out} \mas\dgg(\omega+\omega_2) \
%%%%%
\contraction{}{\mas}{\hspace{2mm}}{}
  \mas \phi\drm{out}\dgg (\omega_2)
\\  & \hspace{12mm}
\times 
 \int\frac{d\omega_4}{2\pi} 
\contraction{}{\mas}{\hspace{2mm}}{}
 \mas \phi\drm{in}\dgg(\omega+\omega_4) \
%%%%%
 (- \contraction{}{\mas}{\hspace{4mm}}{}
 \phi\drm{in}) \mas\dgg (\omega_4).
\end{align}\normalsize
\end{subequations}
The first and second integrals correspond to
the left and right loops in Figure \ref{fig-correlation-1st}(b),
respectively.
Each integral is calculated as 
\begin{subequations}
\small\begin{align}
&
 \int\frac{d\omega_2}{2\pi} \
\contraction{}{\phi}{\hspace{6mm}}{}
 \phi\drm{out} \mas\dgg(\omega+\omega_2) \
\contraction{}{\mas}{\hspace{2mm}}{}
 \mas \phi\drm{out}\dgg (\omega_2)
\\ 
= \ &
 \int \frac{d\omega_2}{2\pi} \
\contraction{}{\phi}{\hspace{6mm}}{}
 \phi\drm{out} \mas\dgg(\omega+\omega_2) \
\Bigl[
  2\pi|\ves|^2\delta(\omega_2) 
\Bigr]
 = 
 |\ves|^2 \ipg\dd{\phi\drm{out} | \mas}\urm{\ f}(\omega),
%%%%%%%%%%%
\\ \nn \\ &
 \int\frac{d\omega_4}{2\pi} \
\contraction{}{\phi}{\hspace{3mm}}{}
 \mas \phi\drm{in}\dgg(\omega+\omega_4) \
 (- \contraction{}{\phi}{\hspace{4mm}}{}
 \phi\drm{in}) \mas\dgg(\omega_4) \
\\
= \ &
 \int \frac{d\omega_2}{2\pi} \
\contraction{}{\phi}{\hspace{3mm}}{}
 \mas \phi\drm{in}\dgg(\omega+\omega_4) \
\Bigl[
 - 2\pi|\ves|^2 \delta(\omega_4) 
\Bigr]
= 
 -|\ves|^2 \ipg\dd{\mas|\phi\drm{in}}\urm{\ f}(\omega).
\end{align}\normalsize
\end{subequations}
As a result,
the first order approximation (\ref{cor1}) is given by
\small\begin{align}
\hspace{15mm}
 \ipg\dd{\phi\drm{out} | \mas}\urm{\ f} \ 
 \bigl( -\isel\drm{cor1} \bigr) \ 
 \ipg\dd{\mas|\phi\drm{in} }\urm{\ f},
\end{align}\normalsize
where 
\small\begin{align}
 \isel\drm{cor1} \equiv \im 2\lambda |\ves|^4.
\end{align}\normalsize

\newpage

\subsection{Second order approximation}

\begin{wrapfigure}[5]{r}[57mm]{73mm} 
\vspace{0mm}
\centering
\includegraphics[keepaspectratio,width=69mm]{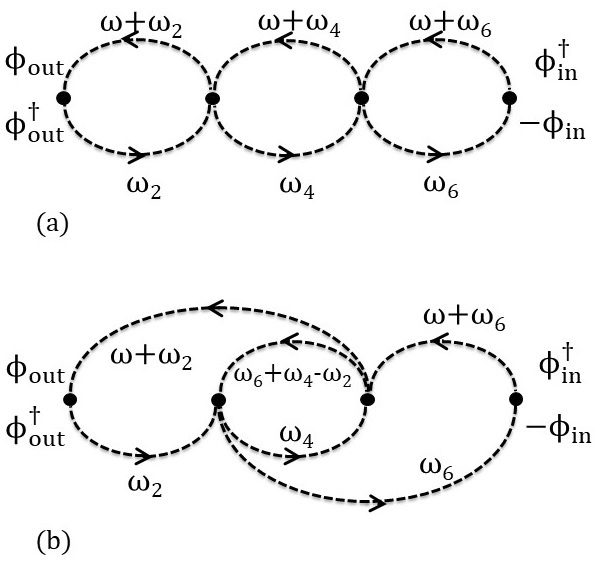}
\caption{
\small
Second order correction.
\normalsize
}
\label{fig-correlation2}
\end{wrapfigure}

The second order correction of the four-point transfer function 
is depicted as Figure \ref{fig-correlation2}.
These two diagrams are expressed as
\begin{subequations}
\label{circ2-2}
\small\begin{align}
\mbox{(a)}
=
 \frac{1}{2!}2^4 2! \left(\frac{\im\lambda}{2}\right)^2
 &\int \frac{d\omega_2}{2\pi} \
\contraction{}{\phi}{\hspace{6mm}}{}
 \phi\drm{out} \mas\dgg (\omega+\omega_2) \
%%%%%%%%%%%%%%%%%%
\contraction{}{\mas}{\hspace{1mm}}{}
 \mas \phi\drm{out}\dgg(\omega_2)
\label{cir31}\\  
%%%%%%%%%%%%%%%%%%
\times 
 & \int \frac{d\omega_4}{2\pi} \
\contraction{}{\mas}{\hspace{2mm}}{}
 \mas \mas\dgg(\omega+\omega_4) \
%%%%%%%%%%%%%%%%%%
\contraction{}{\mas}{\hspace{2mm}}{}
 \mas \mas\dgg(\omega_4)
\label{cir3}\\  
%%%%%%%%%%%%%%%%%%
\times 
 & \int \frac{d\omega_6}{2\pi} \
\contraction{}{\mas}{\hspace{1mm}}{}
 \mas \phi\drm{in}\dgg(\omega+\omega_6) \
%%%%%%%%%%%%%%%%%%
 (-\contraction{}{\phi}{\hspace{6mm}}{}
 \phi\drm{in}) \mas\dgg(\omega_6),
\label{cir32} \\ 
%%%%%%%%%%%%%%
%%%%%%%%%%%%%%
\nn\\ 
 \mbox{(b)}
= 
 \frac{1}{2!}2^4 2! \left(\frac{\im\lambda}{2}\right)^2
 & \int \frac{d\omega_2}{2\pi} \
\contraction{}{\phi}{\hspace{6mm}}{}
 \phi\drm{out} \mas\dgg (\omega+\omega_2) \
%%%%%%%%%%%%%%
\contraction{}{\mas}{\hspace{1mm}}{}
 \mas \phi\drm{out}\dgg(\omega_2)
%%%%%%%%%%%%%%
\\  
\times 
 & \int \frac{d\omega_4}{2\pi} \
\contraction{}{\mas}{\hspace{2mm}}{}
 \mas \mas\dgg (\omega_4+\omega_6-\omega_2) \
%%%%%%%%%%%%%%
\contraction{}{\mas}{\hspace{2mm}}{}
 \mas \mas\dgg (\omega_4)
%%%%%%%%%%%%%%
\\ 
\times 
 & \int \frac{d\omega_6}{2\pi} \
\contraction{}{\mas}{\hspace{1mm}}{}
 \mas \phi\drm{in}\dgg (\omega+\omega_6) \
%%%%%%%%%%%%%%
 (-\contraction{}{\phi}{\hspace{6mm}}{}
 \phi\drm{in}) \mas\dgg (\omega_6).
\end{align}\normalsize
\end{subequations}
The integrals (\ref{cir31},\ \ref{cir32}) 
have been given in the first order approximation (\ref{cor1}),
and (\ref{cir3}) is the same form as (\ref{int12}).
The remaining integrals are calculated in the same way.
As a result, (\ref{circ2-2}) is given as
\small\begin{align}
 \ipg\dd{\phi\drm{out} | \mas}\urm{\ f} \  
 \bigl(- \isel\drm{cor2} \bigr) \  
 \ipg\dd{\mas | \phi\drm{in} }\urm{\ f}.
\end{align}\normalsize
In the second order approximation,
the self-energy is not a constant:
\small\begin{align}
 \isel\drm{cor2}(s)
&\equiv
 (\im 2\lambda)^2 |\ves|^6 
\Bigl(
 \frac{1}{s+2g^2}-\frac{1}{s-2g^2}
+ \frac{1}{g^2}
\Bigr).
\end{align}\normalsize

Let us define a total self-energy as 
\small\begin{align}
 \isel\drm{cor}
&\equiv
 \isel\drm{cor1} + \isel\drm{cor2}.
\end{align}\normalsize
Note that the zeroth order correction 
is a direct connection of the two fields:
\begin{subequations}
\small\begin{align}
 \contraction[2ex]{}{\phi}{\hspace{23mm}}{}
 \phi\drm{out}
 \contraction{}{\phi}{\hspace{10mm}}{}
 \phi\drm{out}\dgg \
 (-\phi\drm{in})\phi\drm{in}\dgg
&=
 -\underbrace{\ipg\dd{\phi\drm{in} | \phi\drm{out}}\urm{\ f}}\dd{=-1} 
  \ipg\dd{\phi\drm{out} | \phi\drm{in} }\urm{\ f}
\\ &=
 \ipg\dd{\phi\drm{out} | \phi\drm{in} }\urm{\ f}.
\end{align}\normalsize
\end{subequations}
Consequently, 
up to second order,
the four-point transfer function is given
in the same form as the two-point transfer function (\ref{pzchecktf}):
\begin{subequations}
\small\begin{align}
  \ipg\dd{\phi\drm{out} \phi\drm{out}\dgg | 
          \phi\drm{in} \phi\drm{in}\dgg}\urm{f+int}
&=
 \ipg\dd{\phi\drm{out} | \phi\drm{in} }\urm{\ f}
-
 \ipg\dd{\phi\drm{out} | \mas}\urm{\ f} \
 \bigl( \isel\drm{cor} \bigr)  \
 \ipg\dd{\mas | \phi\drm{in}}\urm{\ f}
\\ &\sim
 \left[  s+2g^2+\isel\drm{cor}(s)  \right]\inv 
 \left[  s-2g^2+\isel\drm{cor}(s)  \right].
\end{align}\normalsize
\end{subequations}

\chapter{Spin $1/2$ in SU(2) systems}
\label{chap:intspin}
\thispagestyle{fancy}

We consider interactions between 
spins \en{ 1/2 } particles and \textrm{SU(2)} systems.
The decay processes of excited spins are examined.
Spin-spin correlations are also investigated 
with a four-point transfer function
for the creation of spin entanglement.

\section{Interaction with a spin field}

\begin{wrapfigure}[0]{r}[53mm]{49mm} 
\vspace{-0mm}
\centering
\includegraphics[keepaspectratio,width=30mm]{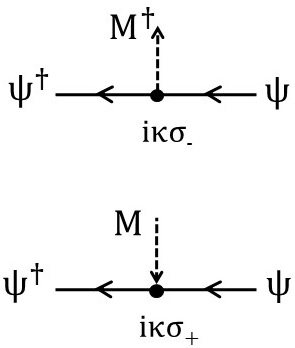}
\caption{
\small
Elementary diagrams corresponding to 
the first and second terms of $ \lag^{\mathrm{JC}}$.
\normalsize
}
\label{fig-jc-2}
\end{wrapfigure}

Assume that 
nonrelativistic spin \en{ 1/2 } particles are 
placed in an SU(2) system
and they interact with each other through a Lagrangian 
\small\begin{align}
 \lag &= \lag\urm{SU(2)} + \lag\urm{1/2} + \lag\urm{JC},
\end{align}\normalsize
where \en{ \lag\urm{SU(2)} }  and \en{ \lag\urm{1/2} } 
describe the SU(2) system  \en{ \mas } and 
the spin \en{ 1/2 } field \en{ \psi }, 
respectively.
\en{ \lag\urm{JC} } is the interaction Lagrangian 
given as
\small\begin{align}
 \lag\urm{JC}
&=
 \kappa 
 : (\psi\dgg\sigma\dd{-}\psi)\mas\dgg
+
 (\psi\dgg\sigma\dd{+}\psi)\mas :,
\label{ljc}
\end{align}\normalsize
where \en{ \kappa } is a coupling constant and 
\en{ \sigma\dd{\pm} } are raising and lowering matrices
defined in (\ref{rlop}).
The first and second terms of \en{ \lag\urm{JC} } 
are depicted in Figure \ref{fig-jc-2}.

\subsection{Preliminary remarks}

\begin{figure}[t]
\centering
\includegraphics[keepaspectratio,width=85mm]{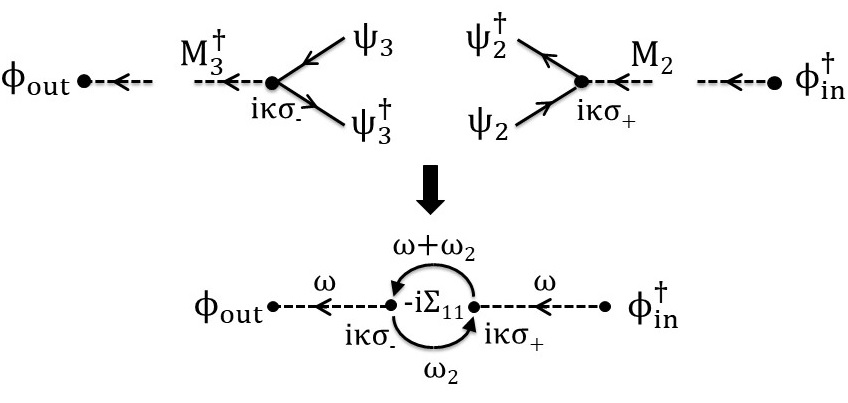}
\caption{
\small
Diagram of second order correction for the transfer function 
 $\phi_{\mathrm{out}} \gets \phi_{\mathrm{in}}$.
The input
$\phi_{\mathrm{in}}\dgg$ is connected to $\mas_2$, which results in 
$\ipg_{\mas | \phi_{\mathrm{in}}}^{\mathrm{SU(2)}}$.
The loop represents vacuum polarization
and the corresponding self-energy is $-\im \Sigma_{11}$.
$\mas_3\dgg$ is connected to $\phi_{\mathrm{out}}$,
which results in $\ipg_{\phi_{\mathrm{out}}|\mas}^{\mathrm{SU(2)}}$.
The entire process is therefore expressed as (\ref{jc11}).
\normalsize
}
\label{fig-jc11}
\end{figure}

The 
spin field is similar to the Dirac field to some degrees,
but there is a critical difference between them.
First, 
there are no antiparticles for the spin field.
Second, 
particles cannot be created and eliminated in the spin field.

For example, 
if we consider a transfer function 
\en{ \phi\drm{out}\gets\phi\drm{in} },
it seems possible to draw a diagram such as Figure \ref{fig-jc11}.
The loop is known as vacuum polarization\index{vacuum polarization}
for the creation of virtual particle-antiparticle pairs.
According to the Feynman-Stueckelberg interpretation,
the backward (time-reversal) arrow corresponds to 
the propagation of antiparticles in the Dirac field.
Unfortunately, 
this interpretation is not applicable to the spin field 
because 
there are no antiparticles in nonrelativistic fields.
In fact, the corresponding self-energy is always zero.

%\newpage

To see this,
let us calculate the transfer function of Figure \ref{fig-jc11}:
\small\begin{align}
\ipg\dd{\phi\drm{out} | \mas}\urm{SU(2)}(\omega) \
\left(
 -\im \Sigma_{11}
\right) \
\ipg\dd{\mas|\phi\drm{in}}\urm{SU(2)}(\omega).
\label{jc11}
\end{align}\normalsize
\en{ \im\Sigma_{11} } is self-energy 
corresponding to the vacuum polarization
given as
\small\begin{align}
- \im \Sigma_{11}(\omega)
&=
 \int \frac{d\omega_2}{2\pi} \,
\Tr \Bigl[
 \ipg\dd{\psi|\psi}(\omega+\omega_2)
\
 \im \kappa\sigma\dd{+}
\
 \ipg\dd{\psi|\psi}(\omega_2)
\
 \im \kappa\sigma\dd{-}
\Bigr],
\label{antispin}
\end{align}\normalsize
where \en{ \ipg\dd{\psi|\psi} } 
is the spin transfer function given in Theorem \ref{thm:spint}.
Note that we have taken the trace 
because the spin transfer functions are matrices.
Since the spin Hamiltonian can be always diagonalized 
with an adequate unitary matrix as
\small\begin{align}
 H
&=
 \uni\dgg \A{E_1}{}{}{E_2} \uni,
\end{align}\normalsize
the spin transfer function is given by
\begin{subequations}
\small\begin{align}
 \ipg\dd{\psi|\psi}(\omega+\omega_2)
&=
 \uni\dgg
 \A{\ffrac{\im}{\omega+\omega_2-E_1+\im\epsilon}}{}{}
   {\ffrac{\im}{\omega+\omega_2-E_2+\im\epsilon}}
 \uni,
\\
 \ipg\dd{\psi|\psi}(\omega_2)
&=
 \uni\dgg
 \A{\ffrac{\im}{\omega_2-E_1+\im\epsilon}}{}{}
   {\ffrac{\im}{\omega_2-E_2+\im\epsilon}}
 \uni.
\end{align}\normalsize
\end{subequations}
Both functions have poles only in the lower half plane of \en{ \omega_2 }.
According to the residue theorem,
\en{ \Sigma_{11}=0 } 
after integrating with respect to \en{ \omega_2 } in (\ref{antispin}).
Consequently,
this process is irrelevant to the nonrelativistic spin field.

The problem here is that
spin \en{ 1/2 } particles can go up and down 
between the ground and excited states,
but they cannot be created from (eliminated into) the vacuum state
in the nonrelativistic treatment.
This is the difference from the Dirac field.
In other words,
massive fermions can change their spin direction 
through the interaction with the SU(2) system \en{ \mas },
but they cannot be created or eliminated in the system.
(On the other hand, massless forward traveling bosons can be created
and eliminated in the system \en{ \mas }.)
This means that 
when we consider a scattering process involving nonrelativistic spins,
the Feynman diagram always has to starts from \en{ \psi\dgg } 
and ends at \en{ \psi }, unlike Figure \ref{fig-jc11}.
We consider examples of this in the next subsections.

\newpage

\subsection{Transfer function from spin to spin}
\label{spin-spin-tf}

\begin{wrapfigure}[0]{r}[53mm]{49mm} 
\vspace{-0mm}
\centering
\includegraphics[keepaspectratio,width=47mm]{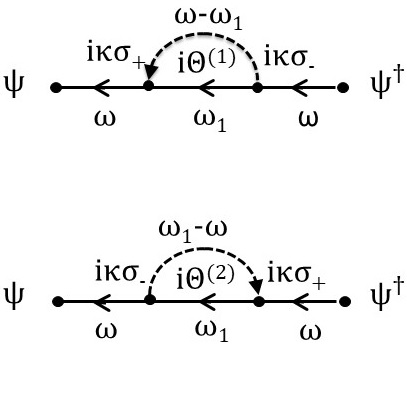}
\caption{
\small
Two possible processes 
for the second order correction. % of $\ipg_{\psi|\psi}^{\mathrm{JC}}$.
In the lower diagram,
the the forward traveling field (dashed arrow) is 
in the backward direction.
The corresponding self-energy $\im\Theta^{(2)}$ is zero.
\normalsize
}
\label{fig-jcss-1}
\end{wrapfigure}

The simplest example is a spin-spin transfer function 
\en{ (\psi\gets \psi) } under \en{ \lag\urm{JC} }
\small\begin{align}
 \ipg\urm{JC}\dd{\psi|\psi}
&= 
 \wick{\psi}{\psi\dgg}\urm{JC}.
\end{align}\normalsize
In the zeroth order term, 
the spin does not interact with \en{ \mas }:
\small\begin{align}
\left( 0\urm{th}\mbox{ order of } \ \ipg\dd{\psi|\psi}\urm{JC}\right)
&=
 \ipg\dd{\psi|\psi},
\end{align}\normalsize
where \en{ \ipg\dd{\psi|\psi} } is the free field transfer function 
given in Theorem \ref{thm:spint}.

For the second order correction,
there are two possible diagrams as in Figure \ref{fig-jcss-1}.
Both diagrams are of the form
\small\begin{align}
\hspace{5mm}
\left( 2\urm{nd}\mbox{ order of } \ \ipg\dd{\psi|\psi}\urm{JC}\right)
&=
 \ipg\dd{\psi|\psi}
\Bigl[
 -\im \Theta\uu{(a)}
\Bigr]
 \ipg\dd{\psi|\psi},
\hspace{5mm}
(a=1,2)
\end{align}\normalsize
where the upper and lower diagrams are, respectively, given as
\begin{subequations}
\label{dd}
\small\begin{align}
 \im \Theta^{(1)}(\omega)
&=
 \kappa^2 \int d\omega_1 \
 \ipg\dd{\mas|\mas}\urm{SU(2)} (\omega-\omega_1) \
 \sigma\dd{+} \ipg\dd{\psi|\psi}(\omega_1) \sigma\dd{-},
\label{dd1}\\
 \im \Theta^{(2)}(\omega)
&=
 \kappa^2 \int d\omega_1 \
 \ipg\dd{\mas|\mas}\urm{SU(2)} (\omega_1-\omega) \
 \sigma\dd{-} \ipg\dd{\psi|\psi}(\omega_1) \sigma\dd{+}.
\label{dd2}
\end{align}\normalsize
\end{subequations}
The integrand of \en{ \im \Theta^{(2)} } 
has poles only in the lower half plane:
\begin{subequations}
\small\begin{align}
 \ipg\dd{\mas|\mas}\urm{SU(2)}(\omega_1-\omega)
&=
 \frac{\im}{\omega_1-\omega+\im 2g^2},
\\
 \ipg\dd{\psi|\psi}(\omega_1)
&=
 \frac{\im}{\omega_1-H+\im\epsilon}.
\end{align}\normalsize
\end{subequations}
Again, 
according to the residue theorem, 
\en{ \im \Theta^{(2)}=0 }.

Let us calculate \en{ \im \Theta^{(1)} }.
To perform the integral, 
we need to specify the spin Hamiltonian.
Here we assume a diagonal matrix
\small\begin{align}
\label{hdiag}
 H&=\A{E_1}{}{}{E_2}.
\end{align}\normalsize
In this case, 
\small\begin{align}
 \im\Theta^{(1)}
&=
 \A{\ffrac{\im\kappa^2}{\omega+\im 2g^2-E_2+\im \epsilon}}{}
   {}{0}.
\end{align}\normalsize

\begin{figure}[t]
\centering
\includegraphics[keepaspectratio,width=130mm]{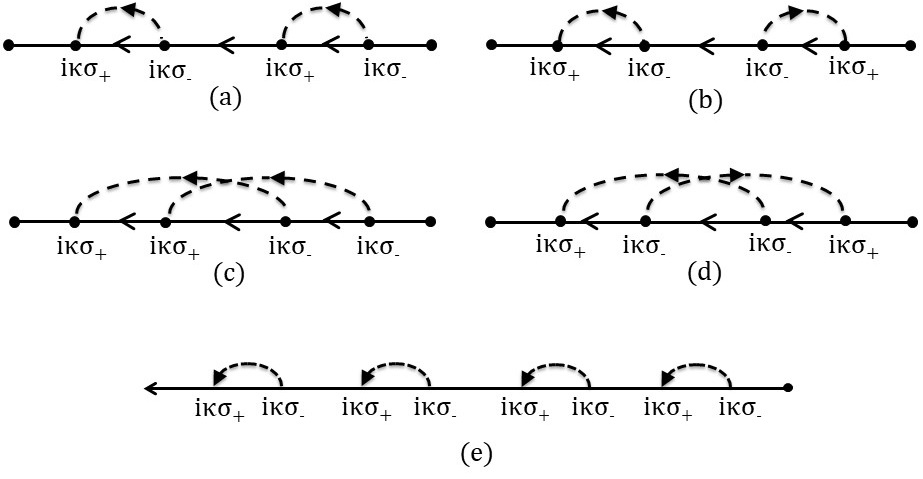}
\caption{
\small
(a $\sim$ d) Diagrams for 
the fourth order correction of $\ipg_{\psi|\psi}^{\mathrm{JC}}$.
To get nonzero self-energy,
the forward traveling fields (dashed arrows) should point
in the forward direction (right to left) for the same reason as the second order correction
(Figure \ref{fig-jcss-1}).
In addition, 
if the spin Hamiltonian is diagonal,
$\im\kappa\sigma_{+}$ and $\im\kappa\sigma_{-}$
 should be alternatively located on the spin propagation (solid line)
because $(\sigma_+)^2=(\sigma_-)^2=0$.
As a result, (a) is the only possible diagram for the fourth order correction.
Likewise,
higher order corrections can be obtained from the diagram (e).
\normalsize
}
\label{fig-jcss2}
\end{figure}

The fourth order correction is shown 
in Figure \ref{fig-jcss2} (a$\sim$d).
First, 
note that all dashed arrows (forward traveling field) 
should point from right to left 
in the forward direction of time. 
Otherwise, 
the integral results in zero, 
as in the case of \en{ \im\Theta^{(2)} }.
Secondly, 
\en{ \im\kappa\sigma\dd{+} } and \en{ \im\kappa\sigma\dd{-} } 
should be located one after the other 
on spin's propagation line (straight solid line).
Otherwise, 
the correction term involves
\en{ (\im \kappa\sigma\dd{-}) 
     \ipg\dd{\psi|\psi}(\omega) 
    (\im\kappa\sigma\dd{-})=0 }
because 
\en{ \ipg\dd{\psi|\psi} } is a diagonal matrix.
Consequently, 
only Figure \ref{fig-jcss2}(a) is 
relevant for the fourth order correction.
Likewise, 
higher order corrections can be given 
from the diagram of Figure \ref{fig-jcss2}(e).
The transfer function is therefore given by 
the following infinite series:
\begin{subequations}
\label{su2-s-dyson}
\small\begin{align}
 \ipg\dd{\psi|\psi}\urm{JC}
&=
 \ipg\dd{\psi|\psi}
-
 \bigl[\ipg\dd{\psi|\psi} \ 
       \im \Theta^{(1)}\bigr] \ 
       \ipg\dd{\psi|\psi}
+
 \bigl[\ipg\dd{\psi|\psi} \ 
       \im \Theta^{(1)}\bigr]^2 \ 
       \ipg\dd{\psi|\psi}
+\cdots
%%%%%%%%%%%%%%%%%%%
\\ &=
\Bigl[
 1
+
 \ipg\dd{\psi|\psi} \ \im \Theta^{(1)} 
\Bigr]\inv 
 \ipg\dd{\psi|\psi}
\\ &=
\A{\ffrac{s+2g^2+\im E_2}{(s+2g^2+\im E_2)(s+\im E_1)+\kappa^2}}{}
   {}{\ffrac{1}{s+\im E_2}}.
\end{align}\normalsize
\end{subequations}

This transfer function describes two possible processes.
Since the SU(2) system is initially in the vacuum state,
if the spin is in the ground state \en{ \ket{-} },
no interactions occur and the spin remains in \en{ \ket{-} }.
This results in the free evolution of the spin, 
which is the (2,2)-element.

The (1,1)-element is the  transition of the spin
from the excited state \en{ \ket{+} } to \en{ \ket{+} }.
Figure \ref{fig-jcss2} (e) shows 
that the spin simply repeats the emission and absorption of 
bosons in the SU(2) system.
Every time the spin emits bosons,
the dashed arrows appear in the diagram.
The arrows come back to the solid line
when the bosons are absorbed.
Then the spin returns to \en{ \ket{+} }.

%\newpage

The transfer function 
has no off-diagonal elements.
This does not mean that 
no transitions occur between \en{ \ket{+} } and \en{ \ket{-} }.
To describe a transition,
we need to include the SU(2) system \en{ \mas } 
in the initial or final state explicitly.
For example, 
in the case of a transition from \en{ \ket{+} } to \en{ \ket{-} },
the spin emits bosons in the system \en{ \mas }.
This is regarded as a decay process \en{ \psi+\mas \gets \psi }.
Let us consider this case next.

\newpage

\subsection{Decay process $\psi+\mas \gets \psi$}
\label{sec:spindecay}

\begin{wrapfigure}[0]{r}[53mm]{49mm} 
\vspace{-0mm}
\centering
\includegraphics[keepaspectratio,width=40mm]{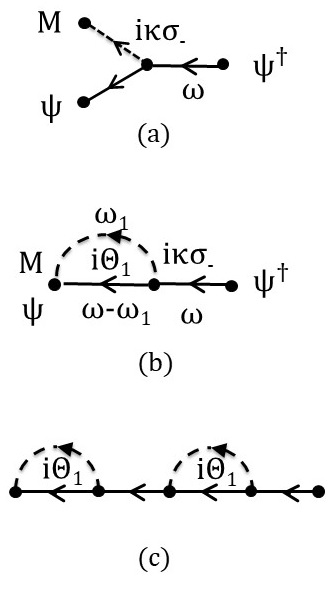}
\caption{
\small
(a) First order correction of the decay process. % (\ref{decayp}).
(b) The diagram is modified to end at the same point in time 
because we consider $t_4=t_3$ in (\ref{tf1}).
(c) Third order correction.
\normalsize
}
\label{fig-decay-1}
\end{wrapfigure}

Suppose that the spin particles are initially in the excited state
\en{ \ket{+} }
and the SU(2) system is in the vacuum state \en{ \ket{0} }.
If the spin particles decay to the ground state \en{ \ket{-} },
bosons (of frequency \en{ \Omega }) are created in the
SU(2) system. 
The probability amplitude of this process
is given by the overlap of the following two states
(see Example \ref{ex:spin}):
\begin{subequations}
\label{ifstate}
\small\begin{align}
 \mbox{\normalsize initial state: \small}&
 \ket{+,0} 
\hspace{1mm}
\sim
\hspace{2mm}
 \Bigl( \psi\dgg \otimes I \ket{0}\Bigr) \chi\dd{+},
%%%%%%%%%%%%%%%
\\
 \mbox{\normalsize final state: \small}&
 \ket{-,\Omega}
\sim
\Bigl( \psi\dgg \otimes \mas\dgg \ket{0}\Bigr) \chi\dd{-}.
\end{align}\normalsize
\end{subequations}
This is described by the off-diagonal element of 
a three-point transfer function\index{transfer function (three-point)}
\small\begin{align}
 \ipg\dd{\psi \mas|\psi}\urm{SU(2)+JC}(t_4,t_3;t_1)
&=
 \wick{\psi_4 \mas_3}{\psi_1\dgg}\urm{SU(2)+JC},
\label{tf1}
\end{align}\normalsize
which corresponds to a decay process
\en{ \psi + \mas \leftarrow \psi  }.

The lowest order correction of the decay process is 
depicted in Figure \ref{fig-decay-1} (a).
We are specifically interested in the case of \en{ t_4=t_3 },
for which the diagram ends at the same point %in time
as in Figure \ref{fig-decay-1} (b).
This is written as
\small\begin{align}
\underbrace{
 \int \frac{d\omega_1}{2\pi} 
\
 \ipg\dd{\mas|\mas}\urm{SU(2)}(\omega_1)
\
 \ipg\dd{\psi|\psi}(\omega-\omega_1)
\
 \im \kappa \sigma\dd{-}}\dd{\hspace{20mm}\equiv\im\Theta\dd{1}}
\
 \ipg\dd{\psi|\psi}(\omega).
\label{theta1}
\end{align}\normalsize
Higher order corrections can be obtained in the same way.
For example, 
the third order correction is in
Figure \ref{fig-decay-1} (c),
which results in
\en{  \im \Theta_1
\
 \ipg\dd{\psi|\psi}
\
 \im \kappa \sigma\dd{+} 
\
 \im\Theta_1
\
 \ipg\dd{\psi|\psi}. }
The transfer function is therefore written as
\begin{subequations}
\small\begin{align}
 \ipg\dd{\psi \mas|\psi}\urm{SU(2)+JC}
&=
 \Bigl(
 1 + 
 \im\Theta_1
\
 \ipg\dd{\psi|\psi}
\
 \im \kappa \sigma\dd{+} + \cdots
\Bigr)
 \im\Theta_1
\
 \ipg\dd{\psi|\psi}
%%%%%%%%%%%%%%%%%%%%%%%%%
\\ &=
\left(
 1 -
 \im\Theta_1
\
 \ipg\dd{\psi|\psi}
\
 \im \kappa \sigma\dd{+} 
\right)\inv
 \im\Theta_1
\
 \ipg\dd{\psi|\psi}.
\end{align}\normalsize
\end{subequations}
If the spin Hamiltonian is diagonal as in (\ref{hdiag}),
the self-energy is given as
\small\begin{align}
 \im\Theta_1(s)
&=
 \A{\ffrac{1}{s+2g^2+\im E_1}}{}
   {}{\ffrac{1}{s+2g^2+\im E_2}}\im \kappa \sigma\dd{-}.
\end{align}\normalsize
Then we have
\small\begin{align}
 \ipg\dd{\psi \mas| \psi}\urm{SU(2)+JC}
&=
 \A{0}{0}{\ffrac{\im\kappa}{ (s+\im E_1)(s+2g^2+\im E_2)+\kappa^2}}{0}.
\end{align}\normalsize

From (\ref{ifstate}),
the probability amplitude from \en{ \ket{+} } to \en{ \ket{-} } 
is given by
\small\begin{align}
 \chi\dd{-}\dgg
\
 \ipg\dd{\psi \mas|\psi}\urm{SU(2)+JC}
\
 \chi\dd{+} 
&=
 \frac{ \im\kappa}{(s+\im E_1)(s+2g^2+\im E_2)+\kappa^2}.
\end{align}\normalsize
If \en{ E_1=E_2=g=0 }, this is written as (via the inverse Laplace transform)
\small\begin{align}
\hspace{15mm}
 \frac{\im\kappa}{s^2 +\kappa^2} 
\ \to \
 \im \sin\kappa t,
\end{align}\normalsize
which simply describes energy exchange 
between the spin \en{ \psi } and the system \en{ \mas }.

\newpage

\subsection{Decay process $\psi + \mas^n \gets \psi$}

\begin{wrapfigure}[0]{r}[53mm]{49mm} 
\vspace{-14mm}
\centering
\includegraphics[keepaspectratio,width=42mm]{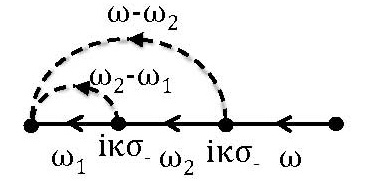}
\caption{
\small
Decay process for $n=2$.
\normalsize
}
\label{fig-ndecay}
\end{wrapfigure}

The decay process can be generalized as 
\small\begin{align}
 \psi 
+ 
 \underbrace{\mas+\mas+\cdots +\mas}\dd{\hspace{9mm}n} 
\leftarrow 
 \psi.
\label{decayn}
\end{align}\normalsize
For \en{ n=2 },
the lowest order correction is depicted in Figure \ref{fig-ndecay}
and given as
\small\begin{align}
&
%\hspace{-7mm}
\underbrace{
 \int \frac{d\omega_2}{2\pi} 
\
 \ipg\dd{\mas|\mas}\urm{SU(2)}(\omega-\omega_2)
\
\underbrace{
 \int \frac{d\omega_1}{2\pi}
\
 \ipg\dd{\mas|\mas}\urm{SU(2)}(\omega_2-\omega_1)
\
 \ipg\dd{\psi|\psi}(\omega_1) 
\
 \im\kappa\sigma\dd{-} }\dd{\hspace{10mm}\equiv\im\Theta_1(\omega_2)}
\
 \ipg\dd{\psi|\psi}(\omega_2) \im\kappa\sigma\dd{-} 
}\dd{ \hspace{30mm}
\equiv \im\Theta_2(\omega)
}
\ \
 \ipg\dd{\psi|\psi}(\omega).
\end{align}\normalsize
Likewise, 
the \en{ n }-decay process (\ref{decayn}) is given by
\small\begin{align}
 \im \Theta_n  
\
 \ipg\dd{\psi|\psi},
\end{align}\normalsize
where \en{ \im\Theta_n } is defined recursively as
\begin{subequations}
\small\begin{align}
 \im\Theta_n(\omega)&=
 \int \frac{d\omega'}{2\pi}
\
 \ipg\dd{\mas|\mas}\urm{SU(2)} (\omega-\omega')
\
 \im \Theta_{n-1}(\omega')
\
 \ipg\dd{\psi|\psi}(\omega') 
\
 \im\kappa\sigma\dd{-},
\\ 
 \im\Theta_0 &= 1. 
\end{align}\normalsize
\end{subequations}
Note that \en{ \im\Theta_2=0 } if the spin Hamiltonian is diagonal.
The off-diagonal elements of the Hamiltonian are essential 
to create more than one particle in the SU(2) system.

%\newpage

\section{Spin-spin scattering}

In this section,
we consider a correlation between spins.
Assume that 
distinguishable spin particles 
\en{ \psi_A } and \en{ \psi_B } 
are placed in an SU(2) system \en{ \mas },
and they individually
interact with \en{ \mas } through a Lagrangian
\begin{subequations}
\small\begin{align}
 \lag\urm{SC}
= \
 \kappa_A &:(\psi_A\dgg \sigma\dd{-}\psi_A)\mas\dgg
          +(\psi_A\dgg \sigma\dd{+}\psi_A)\mas:
\\ +
 \kappa_B & :(\psi_B\dgg \sigma\dd{-}\psi_B)\mas\dgg
          +(\psi_B\dgg \sigma\dd{+}\psi_B)\mas:.
\end{align}\normalsize
\end{subequations}

Let us consider a transfer function \en{ \psi\dd{A} \gets \psi\dd{B} }
to see a relationship between the two spins.
Note that they do not interact with each other directly.
The zeroth order term is written as
\small\begin{align}
\left(0\urm{th}\mbox{ order of } \ 
 \ipg\dd{\psi\dd{A} | \psi\dd{B}}\urm{SU(2)+SC}\right)
=
 \wick{\psi_{A}}{\psi_{B}\dgg}
=0.
\end{align}\normalsize
Higher order corrections also involve the same term
because we cannot eliminate \en{ \psi_B } and create \en{ \psi_A }.
Consequently,
the transfer functions between the two spins are zero:
\small\begin{align}
 \ipg\dd{\psi\dd{A} | \psi\dd{B}}\urm{SU(2)+SC}
=
 \ipg\dd{\psi\dd{B} | \psi\dd{A}}\urm{SU(2)+SC}
&=0.
\end{align}\normalsize
However, this does not necessarily mean no correlation between them.
The correlation is actually evaluated 
by four-point transfer functions.\index{transfer function (four-point)}

\newpage

\subsection{Four-point transfer function}

\begin{wrapfigure}[0]{r}[53mm]{49mm} 
\vspace{-5mm}
\centering
\includegraphics[keepaspectratio,width=30mm]{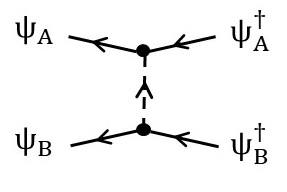}
\caption{
\small
Spin-spin scattering.
\normalsize
}
\label{fig-scatg}
\end{wrapfigure}

A four-point transfer function\index{transfer function (four-point)}
is defined as
\small\begin{align}
 \ipg\dd{\psi\dd{A 4} \psi\dd{B 3} | 
         \psi\dd{B 2} \psi\dd{A 1}}\urm{SU(2)+SC}
&=
 \wick{\psi_A(x_4)\psi_B(x_3)}
     {\psi\ddgg_B(x_2)\psi\ddgg_A(x_1)}\urm{SU(2)+SC}.
\end{align}\normalsize
The lowest order approximation of this transfer function 
is depicted as Figure \ref{fig-scatg}
in which the two spins interact with each other
via the SU(2) system \en{ \mas }.
Here we consider a case where 
\en{ t_4=t_3 } and \en{ t_2=t_1 }, 
to which
the Feynman diagrams are depicted in Figure \ref{fig-scattering-1}.

\begin{wrapfigure}[0]{r}[53mm]{49mm} 
\vspace{15mm}
\centering
\includegraphics[keepaspectratio,width=48mm]{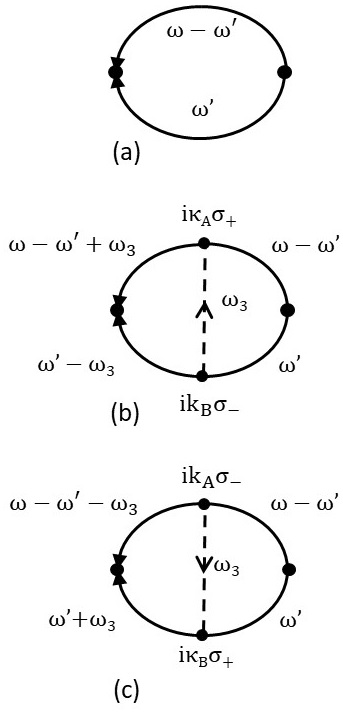}
\caption{
\small
(a) Zeroth order approximation of the four-point transfer function.
 The two arrowed lines start from $t_2=t_1$ and end at $t_4=t_3$.
 The momentum is defined to be conserved at each vertex.
(b,c) Two diagrams for the second order correction. 
\normalsize
}
\label{fig-scattering-1}
\end{wrapfigure}

The zeroth order approximation is 
a direct propagation of the two spins, 
as in Figure \ref{fig-scattering-1}(a).
This is expressed in the frequency domain as
\small\begin{align}
 \int \frac{d\omega'}{2\pi} \
 \ipg\dd{\psi\dd{A} | \psi\dd{A}}(\omega-\omega')
\
 \ipg\dd{\psi\dd{B} | \psi\dd{B}}(\omega').
\label{ss0}
\end{align}\normalsize
The second order corrections is given by 
the sum of Figure \ref{fig-scattering-1}(b) and (c):
\begin{subequations}
\label{ss2}
\small\begin{align}
\int\frac{d\omega_3}{2\pi}\frac{d\omega'}{2\pi}
\Bigl\{ 
\ipg\dd{\psi\dd{A} | \psi\dd{A}}(\omega-\omega'+\omega_3)
\
& \bigl[\im \kappa_A \sigma\dd{+}\bigr]
\
 \ipg\dd{\psi\dd{A} | \psi\dd{A}}(\omega-\omega')
\\  \times \,
 \ipg\dd{\psi\dd{B} | \psi\dd{B}}(\omega'-\omega_3)
\
&\bigl[ \im \kappa_B \sigma\dd{-}\bigr]
\
 \ipg\dd{\psi\dd{B} | \psi\dd{B}}(\omega')
\
 \ipg\dd{\mas|\mas}\urm{SU(2)}(\omega_3)
\\ 
+ \ 
 \ipg\dd{\psi\dd{A} | \psi\dd{A}}(\omega-\omega'-\omega_3)
\
&\bigl[ \im \kappa_A \sigma\dd{-} \bigr]
\
 \ipg\dd{\psi\dd{A} | \psi\dd{A}}(\omega-\omega')
\\ \times \,
 \ipg\dd{\psi\dd{B} | \psi\dd{B}}(\omega'+\omega_3)
\
&\bigl[ \im \kappa_B \sigma\dd{+}\bigr]
\
 \ipg\dd{\psi\dd{B} | \psi\dd{B}}(\omega')
\
 \ipg\dd{\mas|\mas}\urm{SU(2)}(\omega_3)
\Bigr\}.
\end{align}\normalsize
\end{subequations}

%\newpage

\subsection{Four-point transfer function - element-wise calculation}

We need to specify the spin Hamiltonians
to calculate the four-point transfer function explicitly.
Here we assume that they are diagonal:
\small\begin{align}
 H_A=H_B=\A{E_1}{}{}{E_2}.
\end{align}\normalsize
The zeroth order term (\ref{ss0}) is given as
\small\begin{align}
% \left[
%\begin{array}{cccc}
% \ffrac{\im}{\omega-2E_1+\im\epsilon}& & & \\
% & \ffrac{\im}{\omega-(E_1+E_2)+\im\epsilon} & & \\
% & & \ffrac{\im}{\omega-(E_1+E_2)+\im\epsilon} & \\
% & & & \ffrac{\im}{\omega-2E_2+\im\epsilon}
%\end{array}
%\right].
 \left[
\begin{array}{cccc}
 \ffrac{1}{ s +\im 2E_1}& & & \\
 & \ffrac{1}{s + \im (E_1+E_2) } & & \\
 & & \ffrac{1}{s + \im (E_1+E_2) } & \\
 & & & \ffrac{1}{s+2E_2}
\end{array}
\right].
\end{align}\normalsize
The second order correction (\ref{ss2}) is 
given as
\small\begin{align}
 \left[
\begin{array}{cccc}
 0& 0 & 0 & 0\\
 0 & 0& \hspace{-9mm} \left(\ffrac{1}{s+\im(E_1+E_2)}\right)^2
      \ffrac{-\kappa_A\kappa_B}{s+2g^2+2\im E_2} & 0 \\
 0 & \left(\ffrac{1}{s+\im(E_1+E_2)}\right)^2
      \ffrac{-\kappa_A\kappa_B}{s+2g^2+2\im E_1} & 0& 0 \\
 0& 0 & 0 & 0
\end{array}
\right].
\end{align}\normalsize

The four-point transfer function
describes the transition between spin states.
For example, for a transition process
\small\begin{align}
 \ket{+ -}
= \left[
\begin{array}{c}
 0\\
 1\\
 0\\
 0\\
\end{array}
\right] 
\quad
\longrightarrow
\quad
 \frac{1}{\sqrt{2}}\left(
 \ket{+ -}+\ket{- +}
\right)
= \frac{1}{\sqrt{2}}
 \left[
\begin{array}{c}
 0\\
 1\\
 1\\
 0\\
\end{array}
\right],
\end{align}\normalsize
the probability amplitude is given by
\small\begin{align}
&
 \frac{1}{\sqrt{2}}\frac{1}{s+\im(E_1+E_2)}
\left[
 1-
 \frac{1}{s+\im(E_1+E_2)}
 \frac{\kappa_A\kappa_B}{s+2g^2+2\im E_2}
\right].
\label{entanglepa}
\end{align}\normalsize
If the SU(2) system is not dissipative \en{ g=0 } 
and 
the lower energy level is set to be zero \en{ E_2=0 },
it follows from the final value theorem that 
this probability amplitude converges to
\small\begin{align}
 \frac{\kappa_A\kappa_B}{\sqrt{2} E_1^2}.
\end{align}\normalsize

\chapter{Gravitational wave detection}
\label{chap:gravity}
\thispagestyle{fancy}

A model of 
gravitational wave detection\index{gravitational wave detection}
is examined with \textit{S}-matrices.
A gravitational wave detector consists of optical cavities
in which one of mirrors is fluctuated by 
gravitational waves and thermal noise.
The change of the mirror position is reflected in 
the optical field traveling in the cavities.
The signal triggered by gravitational waves is detected through 
homodyne measurement on the output of the cavities.
We first develop a model of the mirror under thermal noise
using the technique of the SU(2) system.
It turns out to be a well-known form in systems theory, called 
a second-order system.
Then we consider a perturbation series for the interaction between 
the gravitational force and the cavity field.
This is also modeled using the SU(2) system.
To investigate how the gravitational force influences the cavity field,
we consider a transfer function.
However, 
we cannot define a transfer function 
from the gravitational force to the output of the cavities
because the gravitational force is not a quantum field.
Instead,
we introduce a sensitivity function.

\section{A harmonic oscillator and noise}

Before introducing gravitational wave detection,
we first model a mirror fluctuated by thermal noise.
Consider a harmonic oscillator \en{ p }
described by
\small\begin{align}
\hspace{20mm}
 (\partial\dd{t}^2+\Omega^2)p
&=
 0, \qquad \Omega\not=0.
\end{align}\normalsize
In the frequency domain, 
this can be expressed as
\begin{subequations}
\label{ho}
\small\begin{align}
 \im\Delta(s)
&=
 \frac{-1}{s^2+\Omega^2}
\\ &=
 \Bigl(
 \frac{1}{s+\im \Omega}-\frac{1}{s-\im \Omega}
 \Bigr)
 \frac{1}{2\im \Omega},
\label{ho-2}
\end{align}\normalsize
\end{subequations}

Let us introduce noise into \en{ p }.
We expect that noise causes the radiation of energy 
and hence damping of oscillations.
In a classical model of dumping, 
this is usually described by 
introducing frictional force into equations of motion.

Here we model damping in a different way.
First note that 
each term of (\ref{ho-2}) is the same as 
a single-mode closed-loop field \en{ \mas } in (\ref{singlen}):
\small\begin{align}
 \ipg\dd{\mas|\mas}=\frac{1}{s+\im \Omega},
\label{nise1}
\end{align}\normalsize
Recall that 
the stability (damping) of the SU(2) system \en{ \mas }
is created from 
the interaction with a free field \en{ \phi } 
through the SU(2) gate,
as in Figure \ref{fig-cavity-2}.
It is described by the Dyson equation (\ref{pdyson2}):
\begin{subequations}
\label{nise2}
\small\begin{align}
 \ipg\dd{\mas|\mas}\urm{SU(2)} 
&= 
 \left(1+2h^2 \ipg\dd{\mas|\mas}\right)\inv \ipg\dd{\mas|\mas} 
\\ &=
 \frac{1}{s+2h^2+\im\Omega}, 
\end{align}\normalsize
\end{subequations}
where \en{ h } is the coupling constant of the SU(2) gate.
This system is stable.
The change from (\ref{nise1}) to (\ref{nise2}) 
is caused by the interaction with \en{ \phi }.
In other words,
\en{ \phi } can be regarded as noise to the system \en{ \mas }.

\begin{wrapfigure}[0]{r}[53mm]{49mm} 
\vspace{-50mm}
\centering
\includegraphics[keepaspectratio,width=30mm]{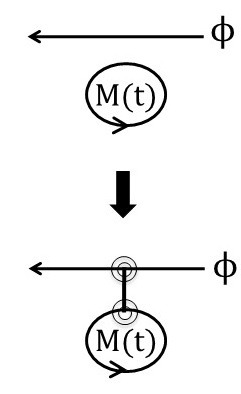}
\caption{
\small
(Upper) Independent closed-loop field $\mas$ and free field $\phi$.
(Lower) SU(2) system created from $\mas$ and $\phi$.
\normalsize
}
\label{fig-cavity-2}
\end{wrapfigure}

We apply this idea to the harmonic oscillator \en{ p }, 
i.e.,
we regard \en{  \im\Delta(s) } %(\ref{ho})
as two closed loops 
and introduce noise into each of them \textit{independently}:
\begin{subequations}
\small\begin{align}
 \im\Delta(s) 
&=
 \Bigl(
 \frac{1}{s+\im \Omega}-\frac{1}{s-\im \Omega}
 \Bigr)
 \frac{1}{2\im \Omega},
\label{hon-1} \\ \Rightarrow \quad 
 \im\Delta\urm{DO}(s) &=
 \Bigl(
 \frac{1}{s+2h^2+\im \Omega}-\frac{1}{s+2h^2-\im \Omega}
 \Bigr)
 \frac{1}{2\im \Omega}
\\ &=
 \frac{-1}{(s+2h^2)^2+\Omega^2}
\\ &\equiv
 \frac{-1}{s^2+2\zeta \omega_n s+\omega_n^2}.
\label{hon} 
\end{align}\normalsize
\end{subequations}
This is a well-known form of 
the \textit{second-order system}
in systems theory.\index{second-order system (systems theory)}
\en{ \zeta } and \en{ \omega_n } are, 
respectively, 
called 
a \textit{system damping ratio}\index{system damping ratio} 
and \textit{system natural frequency},\index{system natural frequency}
and defined as
\begin{subequations}
\small\begin{align}
 \omega_n&\equiv \sqrt{4h^4+\Omega^2},
\\
 \zeta&\equiv \frac{2h^2}{\sqrt{4h^4+\Omega^2}} < 1.
\end{align}\normalsize
\end{subequations}
The transfer function \en{ \im \Delta\urm{DO} } has two poles at
\small\begin{align}
 \pole(\im\Delta\urm{DO})
=
 -\zeta\omega_n\pm\im\omega_n\sqrt{1-\zeta^2}.
\end{align}\normalsize
Since \en{ 0<\zeta<1 }, 
these poles are in the left half of the complex (\en{ s }) plane
and hence \en{ \im\Delta\urm{DO} } is stable.
As a result, 
\en{ p } can be regarded as a damped harmonic oscillator.

\begin{remark}
It is critical to introduce two noise sources
and 
let them interact with the first and second terms of (\ref{hon-1}) 
independently.
If we introduce a single noise source,
the Dyson equation reads
\small\begin{align}
 \bigl(1+2h^2 \im\Delta \bigr)\inv \im\Delta
=
 \frac{1}{s^2+\Omega^2-2h^2},
\end{align}\normalsize
which is irrelevant to damping of oscillations.
In fact, 
if \en{ \Omega^2-2h^2 \ge 0 },
the system is still oscillating and no damping occurs.
If \en{ \Omega^2-2h^2 < 0 },
the system has both stable and unstable poles,
which means that the state of the system is amplified by the noise.
This is obviously incorrect as a model of damping.
\end{remark}

\newpage

\section{Lagrangian of gravitational wave detection}

Optical cavities are indispensable in detection of gravitational waves.
%described as tidal force.
One of optical mirrors is designed to oscillate at a certain frequency
so that the mirror can resonate with gravitational waves.
Meanwhile, 
the mirror is fluctuated by disturbance such as thermal noise.
Its position is therefore described by the damped harmonic oscillator
shown in the preceding section.
Once the resonance is triggered by gravitational force,
the signal appears in the cavity.

The SU(2) system \en{ \mas } is a good approximate model of
the optical cavity.
The Lagrangian of the detection process is given as
\small\begin{align}
 \lag
&=
 \lag\urm{SU(2)}\dd{\mas} + \lag\dd{p}\urm{DO}
+ 
 \kappa p \mas\dgg \mas + \lambda p u,
\label{gravl}
\end{align}\normalsize
The first and second terms describe the cavity field \en{ \mas }
and damped harmonic oscillator \en{ p }, 
respectively.
The third term is the coupling of \en{ p } to \en{ \mas },
whereas 
the fourth term is 
the coupling of \en{ p } to the gravitational force \en{ u }.

The explicit form of \en{ \lag\urm{DO} } is not necessary.
All we need is 
the corresponding transfer function \en{ \im \Delta\urm{DO} }.
For the cavity \en{ \lag\urm{SU(2)}\dd{\mas} },
its transfer function 
\en{ \ipg\dd{\mas|\mas}\urm{SU(2)} } is known as well.
Accordingly, 
we express the Lagrangian \en{ \lag } as
\small\begin{align}
 \lag
=
  \lag\urm{f} + \lag\urm{GD},
\end{align}\normalsize
where 
\begin{subequations}
\label{gde}
\small\begin{align}
 \lag\urm{f} &\equiv \lag\urm{SU(2)}\dd{\mas} + \lag\dd{p}\urm{DO},
\\
 \lag\urm{GD} &\equiv \kappa p \mas\dgg \mas + \lambda p u.
\end{align}\normalsize
\end{subequations}
\en{ \lag\urm{GD} } is regarded as a perturbation.
In the Feynman diagram,
\en{ \lag\urm{GD} } is expressed by the following elementary diagrams:

\vspace{0mm}

\begin{figure}[H]
\centering
\includegraphics[keepaspectratio,width=64mm]{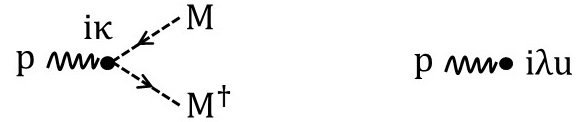}
\end{figure}

\vspace{0mm}

\noindent
where the dashed arrow and wiggly solid line
correspond to 
\en{ \ipg\dd{\mas|\mas}\urm{SU(2)} }
and 
\en{ \im \Delta\urm{DO} }, 
respectively.
It is worth noting that 
the gravitational force \en{ u } 
is not a quantum field.
In the right diagram above, 
\en{ u } is regarded as a \en{ c }-number attached to the vertex.
In other words,
\en{ \mas } and \en{ p } are treated as quantum fields, 
whereas
\en{ u } is a classical parameter.

We are interested in seeing how \en{ \mas } is influenced by \en{ u }
because gravitational waves are detected through the optical signal.
This is usually done by 
calculating a transfer function \en{ \mas \gets u }.
In a conventional approach,
we regard all quantities as \en{ c }-numbers,
write down the Euler-Lagrange (nonlinear) equation
from \en{ \lag=\lag\urm{f} + \lag\urm{GD} },
and then, linearize it to derive a transfer function \en{ \mas \gets u } 
in a classical way.

Here
we take a different approach.
Given \en{ \lag=\lag\urm{f} + \lag\urm{GD} },
we regard \en{ \lag\urm{GD} } as a perturbation
and expand the corresponding \textit{S}-matrix.
However, 
a transfer function \en{ \mas \gets u } is not well defined 
in a field theoretical sense
because \en{ u } is not a quantum field.
To investigate how \en{ \mas } is influenced by \en{ u },
we introduce a sensitivity function.

\section{Sensitivity function}

Consider a transfer function \en{ \mas\gets\mas }
under \en{ \lag=\lag\urm{f} + \lag\urm{GD} }.
This transfer function depends on \en{ u }.
According to the discussion in Section \ref{sec:dyson2},
it can be expressed as
\small\begin{align}
 \ipg\dd{\mas|\mas}\urm{f+GD}(s)
&=
 \ipg\dd{\mas|\mas}\urm{SU(2)} \
\Bigl[
 \im f_0(s)
+
 u \im f_1(s)
+
 u^2 \im f_2(s) +\cdots
\Bigr] \
 \ipg\dd{\mas|\mas}\urm{SU(2)}.
\label{gdseries}
\end{align}\normalsize
The gravitational force \en{ u } is so weak 
that we can ignore \en{ u^2 } and higher order.
Then \en{ \im f_1 } represents \en{ \mas }'s dependency on \en{ u }.
This is regarded as a sensitivity function.
We evaluate \en{ \im f_1 }
by calculating the second and forth order corrections 
of the \textit{S}-matrix.

\subsection{Second order correction}

For the second order correction of 
\en{ \ipg\dd{\mas|\mas}\urm{f+GD} },
there are two possible diagrams:

\vspace{0mm}
\begin{figure}[H]
\centering
\includegraphics[keepaspectratio,width=104mm]{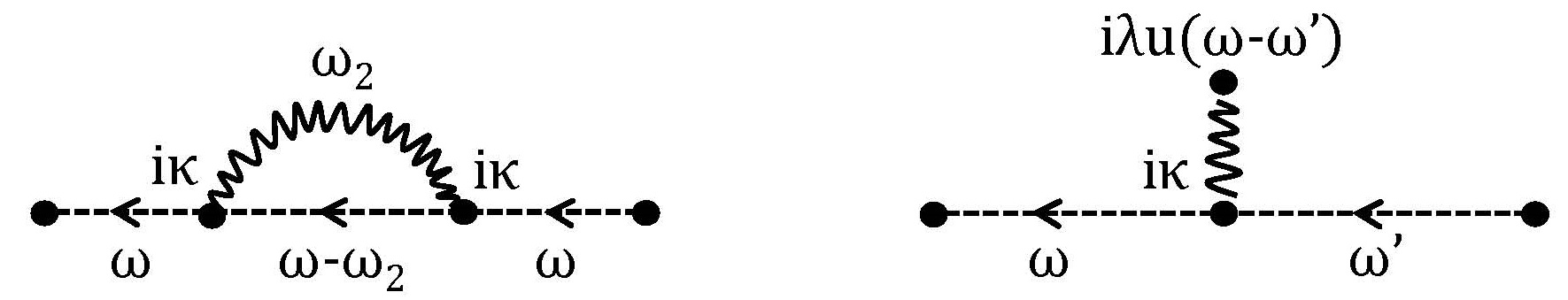}
\end{figure}
\vspace{0mm}

\noindent
It is not necessary to consider the left diagram 
because \en{ u } is not involved.
The right one is written as
\small\begin{align}
&
\int dt_4 dt_1 \ex\uu{\im \omega t\dd{4} - \im\omega' t\dd{1}}
\int dt_3dt_2 \
 \contraction{}{\phi\hspace{1mm}}{\hspace{15mm}}{}
 \mas_4 : \im \kappa
 \contraction[2ex]{}{\phi}{\hspace{25mm}}{}
 p_3 \mas_3\dgg 
 \contraction{}{\phi\hspace{1mm}}{\hspace{23mm}}{}
 \mas_3 : : \im \lambda u_2 p_2 
 : \mas_1\dgg
%%%%%%%%%%%%%%%%%%%%%%%%
\nn\\ =&
\int dt_4 dt_1 \ex\uu{\im \omega t\dd{4} - \im\omega' t\dd{1}}
\int dt_3dt_2 \
 \ipg\dd{\mas|\mas}\urm{SU(2)}(t_4;t_3)
\ \Bigl[
 \im \kappa  \im\Delta\urm{DO}(t_3;t_2) 
\
 \im \lambda u(t_2)
 \Bigr] \
 \ipg\dd{\mas|\mas}\urm{SU(2)}(t_3;t_1)
%%%%%%%%%%%%%%%%%%%%%%%%
\nn\\ =& \
 \ipg\dd{\mas|\mas}\urm{SU(2)}(\omega) 
\
\Bigl[ 
 \im \kappa \im \Delta\urm{DO}(\omega-\omega')
\
 \im \lambda u(\omega-\omega')
\Bigr] \
 \ipg\dd{\mas|\mas}\urm{SU(2)}(\omega').
\end{align}\normalsize
%\end{subequations}

If the gravitational force is slowly varying 
compared to the cavity field,
and 
passing through the mirror for a very short period of time,
then the process can be approximated to adiabatic scattering.
In this case, 
there is no energy loss
so that the incoming and outgoing momenta are the same
\en{ \omega\sim\omega' } approximately.
The second order correction is then rewritten as
\small\begin{align}
 \ipg\dd{\mas|\mas}\urm{SU(2)}(\omega) 
\
\Bigl[ -\isel \, u\Bigr]
\
 \ipg\dd{\mas|\mas}\urm{SU(2)}(\omega),
\label{gwd-se1}
\end{align}\normalsize
where 
\small\begin{align}
- \isel
\equiv
- \kappa \lambda  \im \Delta\urm{DO}(0)
=
 \frac{\kappa \lambda}{\omega_n^2}
\end{align}\normalsize
is self-energy. 
Compared to (\ref{gdseries}),
\small\begin{align}
 \im f_1
&=
 \frac{\kappa \lambda}{\omega_n^2}
\end{align}\normalsize
is a sensitivity function in the lowest order approximation.
Note that 
the influence of the thermal noise is static.

\newpage

\subsection{Fourth order correction}

\begin{wrapfigure}[0]{r}[53mm]{49mm} 
\vspace{10mm}
\centering
\includegraphics[keepaspectratio,width=48mm]{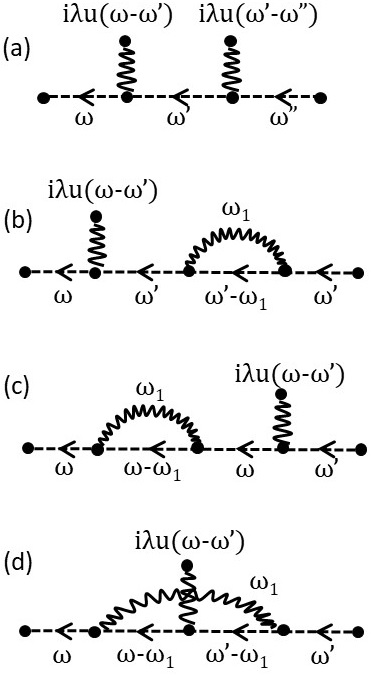}
\caption{
\small
Fourth order corrections.
\normalsize
}
\label{fig-gravity2-1}
\end{wrapfigure}

There are four possible diagrams for the fourth order correction,
as in Figure \ref{fig-gravity2-1}.
Note that Figure \ref{fig-gravity2-1}(a) can be ignored 
because of \en{ u^2 }.
The fourth order correction is given by the sum of the 
remaining three diagrams:
\begin{subequations}
\small\begin{align}
 \ipg\dd{\mas|\mas}\urm{SU(2)}(\omega)&
%\
 \Bigl[-\im \Xi_1(\omega)\Bigr]
%\
\Bigl[
 \im \kappa \im \Delta\urm{DO}(\omega-\omega')
%\
 \im \lambda u(\omega-\omega')
\Bigr]
%\
 \ipg\dd{\mas|\mas}\urm{SU(2)}(\omega')
%%%%%%%%%%%%%%
\\  
+ \ipg\dd{\mas|\mas}\urm{SU(2)}(\omega)&
%\
\Bigl[
 \im \kappa \im \Delta\urm{DO}(\omega-\omega')
%\
 \im \lambda u(\omega-\omega')
\Bigr]
%\
\Bigl[- \im \Xi_1(\omega')\Bigr]
%\
 \ipg\dd{\mas|\mas}\urm{SU(2)}(\omega')
%%%%%%%%%%%%%%
\\  
+  \ipg\dd{\mas|\mas}\urm{SU(2)}(\omega)&
%\
\Bigl[-\im \Xi_2(\omega,\omega')\Bigr] 
%\
\Bigl[
 \im \kappa \im\Delta\urm{DO}(\omega-\omega')\ \im \lambda u(\omega-\omega')
\Bigr]
%\
 \ipg\dd{\mas|\mas}\urm{SU(2)}(\omega'),
\end{align}\normalsize
\end{subequations}
where
\begin{subequations}
\label{resint}
\small\begin{align}
- \im \Xi_1(\omega) 
&=
 (\im\kappa)^2 
\
\ipg\dd{\mas|\mas}\urm{SU(2)}(\omega)
 \int \frac{d\omega_1}{2\pi} \
 \ipg\dd{\mas|\mas}\urm{SU(2)}(\omega-\omega_1)
\
 \im \Delta\urm{DO}(\omega_1),
\\
- \im \Xi_2(\omega,\omega') 
&=
 (\im\kappa)^2 \int \frac{d\omega_1}{2\pi} \
 \ipg\dd{\mas|\mas}\urm{SU(2)}(\omega-\omega_1)
\
 \ipg\dd{\mas|\mas}\urm{SU(2)}(\omega'-\omega_1)
\
 \im \Delta\urm{DO}(\omega_1).
\end{align}\normalsize
\end{subequations}

If the scattering process is adiabatic 
as assumed in the second order correction,
the fourth order correction is expressed as
\small\begin{align}
 \ipg\dd{\mas|\mas}\urm{SU(2)}(s)
\
\Bigl[-\im \Xi (s) \Bigr]
\
\Bigl[
-\isel \, u
\Bigr]
\
 \ipg\dd{\mas|\mas}\urm{SU(2)}(s),
\label{gwd-se1-2}
\end{align}\normalsize
where 
\begin{subequations}
\small\begin{align}
- \im \Xi
&\equiv
- 2 \im \Xi_1 - \im \Xi_2.
\\&=
 (\im\kappa)^2
 \left[2\ipg\dd{\mas|\mas}\urm{SU(2)}(s)-\frac{d}{ds}\right]
 \im \Delta\urm{DO}(s+2g^2)
\\ &=
 \frac{-\kappa^2 \bigl[ 2\zeta\omega_n(s+2g^2)+\omega_n^2 \bigr]}
      {\bigl[ (s+2g^2)^2+2\zeta\omega_n(s+2g^2)+\omega_n^2 \bigr]^2(s+2g^2)}.
\end{align}\normalsize
\end{subequations}
In this case, 
the self-energy \en{ \im\Xi } depends on the thermal noise dynamically.

\subsection{Sensitivity function}

Up to fourth order,
the transfer function is given as
\begin{subequations}
\label{mmf1}
\small\begin{align}
 \ipg\dd{\mas|\mas}\urm{f+GD}(s)
&=
 \ipg\dd{\mas|\mas}\urm{SU(2)} \
\Bigl[ 1-\im\Xi(s) \Bigr] 
\Bigl[ -\isel \, u\Bigr] \
 \ipg\dd{\mas|\mas}\urm{SU(2)} 
%%%%%%%%%%%%%%%%%%%%%%%%%%%%%%%%%%%%%%%%%%%
\\ &\sim
 \ipg\dd{\mas|\mas}\urm{SU(2)} \
 u \, \frac{ -\isel }{1+\im\Xi(s)} \
 \ipg\dd{\mas|\mas}\urm{SU(2)}.
\end{align}\normalsize
\end{subequations}
Compared to (\ref{gdseries}),
the sensitivity function is given as
\small\begin{align}
 \im f_1
&=
\frac{ -\isel }{1+\im\Xi(s)}.
\label{f1}
\end{align}\normalsize

\chapter{Nonlinear gates and systems}
\label{chap:nfb}
\thispagestyle{fancy}

D-feedback is a process in which
the parameter of a displacement gate is defined 
as being proportional to one of the output quadratures.
This has been examined 
classically in Section \ref{sec:fffb} 
and
with \textit{S}-matrices in Sections \ref{sec:circuitss}.
In this chapter,
we apply the same feedback to the XX and QND gates.
The input-output relations of the resulting gates 
are investigated through multi-point transfer functions.

\section{XX+SU(2) gate and feedback coupling}

Consider the XX+SU(2) gate introduced in Section \ref{sec:xxsu2}.
Assume that the parameter (reactance matrix) of the XX gate
is defined as
\small\begin{align}
 g_0 
&= 
 \frac{k}{\sqrt{2}} \xi\drm{2,out}.
\label{xxsu2p}
\end{align}\normalsize
The Lagrangian is then expressed as
\small\begin{align}
 \lag
&=
 \lag\urm{SU(2)} 
+
 \lag\urm{int},
\end{align}\normalsize
where \en{ \lag\urm{SU(2)} } is from the SU(2) gate and 
\small\begin{align}
  \lag\urm{int}
&= \
\frac{k}{\sqrt{2}}
 : 
 \xi\drm{2,out} 
\Bigl[
\mm{\xi}_1 \mm{\xi}_2
+
 g\dd{l} 
\bigl(
 \mm{\xi}_1^2 - \mm{\xi}_2^2
\bigr) \Bigr]:.
\end{align}\normalsize

In this section,
we examine a decay process
\small\begin{align}
 \phi\drm{1,out} + \phi\drm{1,out} \gets \phi\drm{1,in}
\end{align}\normalsize
with a three-point transfer function
\small\begin{align}
 \ipg\dd{\phi\drm{1,out}^2 | \phi\drm{1,in} }\urm{SU(2)+int}
&=
 \wick{\phi\drm{1,out}^2}{\phi\drm{1,in}\dgg}\urm{SU(2)+int}.
\label{tptf1}
\end{align}\normalsize
The transfer functions of the SU(2) gate have been obtained 
in Theorem \ref{prop:bsfull}.
The three-point transfer function is calculated by
expanding an \textit{S}-matrix corresponding to
\en{ \lag\urm{int} }.

\newpage

We first note that 
\en{ \lag\urm{int} } is a polynomial function of 
the quadrature \en{ \xi }.
The second and higher order corrections involve 
the following contraction:
\small\begin{align}
\contraction{}{\phi\hspace{2mm}}{\hspace{3mm}}{}
 \cdots \ \xi\dd{i}
 \contraction{}{\phi\hspace{2mm}}{\hspace{4mm}}{}
 \xi\dd{j} : : \xi\dd{k}
\contraction{}{\phi\hspace{1mm}}{\hspace{3mm}}{}
 \xi\dd{l} \ \cdots.
\label{middlecont}
\end{align}\normalsize
Since the SU(2) gate is unitary, 
\en{ \xi } and \en{ \eta } are always independent
and hence 
all transfer functions between \en{ \xi } and \en{ \eta } 
are zero (Theorem \ref{prop:bsfull}):
\small\begin{align}
 \ipg\dd{\xi\dd{j} | \eta\dd{k}}\urm{SU(2)}
=
 \im \wick{\xi\dd{j}}{\xi\dd{k}}\urm{SU(2)}
=
 0.
\end{align}\normalsize
As a result, 
(\ref{middlecont}) is zero and 
only the first order term remains
in the expansion of the \textit{S}-matrix.

\begin{wrapfigure}[0]{r}[53mm]{49mm} 
\vspace{0mm}
\centering
\includegraphics[keepaspectratio,width=48mm]{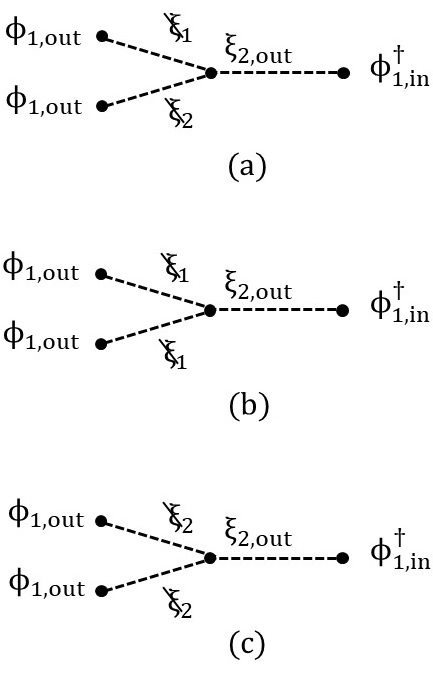}
\caption{
\small
Three diagrams for the first order correction
\normalsize
}
\label{fig-xxfb-1}
\end{wrapfigure}

For the first order correction,
there are three possible diagrams as in Figure \ref{fig-xxfb-1}.
The diagram (a) is expressed as
\begin{subequations}
\small\begin{align}
 (a)
&=
 2! \times \im \frac{k}{\sqrt{2}} 
 \wick{\phi\drm{1,out}}{\mm{\xi}_1} \
 \wick{\phi\drm{1,out}}{\mm{\xi}_2} \
 \wick{\xi\drm{2,out}}{\phi\drm{1,in}\dgg},
\end{align}\normalsize
\end{subequations}
where \en{ 2! } comes from \en{ (\phi\drm{1,out})^2 }.
It follows from Theorem \ref{prop:bsfull} that
\begin{subequations}
\small\begin{align}
 \wick{\phi\drm{1,out}}{\mm{\xi}_1}
&=
 \frac{1}{\sqrt{2}} \frac{1}{1+g\dd{l}^2},
\\
 \wick{\phi\drm{1,out}}{\mm{\xi}_2} 
&=
 \frac{1}{\sqrt{2}} \frac{g\dd{l}}{1+g\dd{l}^2},
\\
 \wick{\xi\drm{2,out}}{\phi\drm{1,in}\dgg} 
&=
 \frac{1}{\sqrt{2}} \frac{-2g\dd{l}}{1+g\dd{l}^2}.
\end{align}\normalsize
\end{subequations}
As a result, we have
\begin{subequations}
\small\begin{align}
 (a)
&=
 \im k \frac{-g\dd{l}^2}{(1+g\dd{l}^2)^3},
\\
 (b)
&=
 \im k \frac{-g\dd{l}^2}{(1+g\dd{l}^2)^3},
\\
 (c)
&=
 \im k \frac{g\dd{l}^4}{(1+g\dd{l}^2)^3}.
\end{align}\normalsize
\end{subequations}
The three-point transfer function is given by 
the sum of these diagrams:
\small\begin{align}
 \ipg\dd{\phi\drm{1,out}^2 | \phi\drm{1,in}}\urm{SU(2)+int}
&=
 \im k
 \frac{g\dd{l}^2(g\dd{l}^2-2)}{(1+g\dd{l}^2)^3}.
\end{align}\normalsize

Likewise, 
the three-point transfer function for a decay process
\small\begin{align}
 \phi\drm{1,out} + \phi\drm{1,out} \gets \phi\drm{2,in} 
\end{align}\normalsize
is given as
\small\begin{align}
 \ipg\dd{\phi\drm{1,out}^2 | \phi\drm{2,in}}\urm{SU(2)+int}
&=
 \frac{\im k}{2}
 \frac{g\dd{l}(1-g\dd{l}^2)^2}{(1+g\dd{l}^2)^3}.
\end{align}\normalsize

\newpage

\section{QND system and feedback coupling}

\begin{wrapfigure}[0]{r}[53mm]{49mm} 
\vspace{27mm}
\centering
\includegraphics[keepaspectratio,width=30mm]{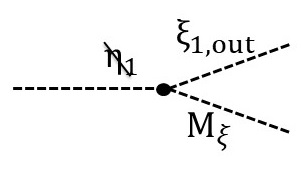}
\caption{
\small
Diagram for $\lag^{\mathrm{int}}$
\normalsize
}
\label{fig-lag-qnd-int}
\end{wrapfigure}

Consider the QND system introduced in Section \ref{sec:sumsys}.
As in the preceding section,
the coupling constant \en{ g } is defined as
\small\begin{align}
 g
&=
 g_0 + k \xi\drm{1,out},
\end{align}\normalsize
where \en{ g_0 } is constant.
The resulting Lagrangian is written as
\small\begin{align}
 \lag
&=
 \lag_{\mas}\urm{QND} + \lag\urm{int},
\end{align}\normalsize
where \en{ \lag_{\mas}\urm{QND} } is 
the Lagrangian of the QND system and 
\small\begin{align}
 \lag\urm{int}
&=
 k \xi\drm{1,out}  \ \mm{\eta}_1 \mas_{\xi}.
\end{align}\normalsize
This is depicted in Figure \ref{fig-lag-qnd-int}.
Before examining this interaction,
let us revisit the QND system.
As in (\ref{sumsysio}), it is expressed as a state equation
\begin{subequations}
\small\begin{align}
 \AV{\dot{\mas}\dd{\xi}}{\dot{\mas}\dd{\eta}}
&=
\hspace{37mm} \A{0}{0}{0}{g}
 \AV{\xi_1}{\eta_1}\drm{in},
\label{qndstate-1}\\
 \AV{\xi_1}{\eta_1}\drm{out}
&=
 \A{-g}{0}{0}{0}\AV{\mas\dd{\xi}}{\mas\dd{\eta}}
+ \hspace{17mm}
 \AV{\xi_1}{\eta_1}\drm{in}.
\label{qndstate-2}
\end{align}\normalsize
\end{subequations}
Note that \en{ \mas\dd{\xi} } is completely decoupled 
from the input \en{ \phi\drm{1,in}=(\xi\drm{1,in} + \im \eta\drm{1,in})/\sqrt{2} }
in (\ref{qndstate-1}),
whereas 
it is coupled to the output \en{ \xi\drm{1,out} } in (\ref{qndstate-2}).
This state equation is also expressed as the following contractions:

\begin{lemma}
\label{s0s1}
In the quadrature basis,
contractions under \en{ \lag^{\mathrm{QND}} } are given as
\begin{subequations}
\small\begin{align}
 \wick{\xi\dd{1,\mathrm{in}}}{\eta\dd{1,\mathrm{in}}}^{\mathrm{QND}}&=\im,
\hspace{18mm}
 \wick{\eta\dd{1,\mathrm{in}}}{\xi\dd{1,\mathrm{in}}}^{\mathrm{QND}}=-\im,
\\
 \wick{\xi\dd{1,\mathrm{out}}}{\eta\dd{1,\mathrm{in}}}^{\mathrm{QND}}&=\im,
\hspace{17mm}
 \wick{\eta\dd{1,\mathrm{out}}}{\xi\dd{1,\mathrm{in}}}^{\mathrm{QND}}=-\im,
\\
 \wick{\xi\dd{1,\mathrm{in}}}{\eta\dd{1,\mathrm{out}}}^{\mathrm{QND}}&=-\im,
\hspace{15mm}
 \wick{\eta\dd{1,\mathrm{in}}}{\xi\dd{1,\mathrm{out}}}^{\mathrm{QND}}=\im,
\\
 \wick{\xi\dd{1,\mathrm{out}}}{\eta\dd{1,\mathrm{out}}}^{\mathrm{QND}}&=\im,
\hspace{16mm}
 \wick{\eta\dd{1,\mathrm{out}}}{\xi\dd{1,\mathrm{out}}}^{\mathrm{QND}}=-\im,
\\ \nn\\
%%%%%%%%%%%%%%%%%%%%%%%%%%%%%%%%%%%%%%%%%%%%%%
 & \hspace{-2mm}
\wick{\mas_{\eta}}{\xi\dd{1,\mathrm{in}}}^{\mathrm{QND}}=-\im\ffrac{g_0}{s},
\\
 & \wick{\mas_{\xi}}{\mas_{\eta}}^{\mathrm{QND}}=\ffrac{\im}{s},
\\
 & \wick{\mas_{\eta}}{\mas_{\xi}}^{\mathrm{QND}}=-\ffrac{\im}{s},
\\
 & \hspace{-3mm}
\wick{\xi\dd{1,\mathrm{out}}}{\mas_{\eta}}^{\mathrm{QND}}=-\im\ffrac{g_0}{s}.
\end{align}\normalsize
\end{subequations}
All other contractions are zero.
In particular, 
\en{ \mas\dd{\xi} } is decoupled from the input \en{ \phi\dd{\mathrm{1,in}} }
\small\begin{align}
 \ipg\dd{\mas\dd{\xi} | \phi\dd{\mathrm{1,in}}}^{\mathrm{QND}}
=
 \wick{\mas_{\xi}}{\phi\dd{\mathrm{1,in}}\dgg}^{\mathrm{QND}} 
= 
 0.
\label{s0s1-1}
\end{align}\normalsize
\end{lemma}

\newpage

The input process is not affected by the feedback.
Hence
\en{ \mas\dd{\xi} } is still decoupled from \en{ \phi\drm{1,in} } 
under \en{ \lag\urm{int} }.
To see this, note that 
an operator \en{ \maths{O} }
such that \en{ \wick{\mas_{\xi}}{\maths{O}}\urm{QND}\not= 0 }
is only \en{ \maths{O}=\mas_{\eta} }
in Lemma \ref{s0s1}.
However, 
\en{ \mas_{\eta} } is not in \en{ \lag\urm{int} }.
As a result, 
the following \en{ (l+m+n) }-point transfer function is zero
and 
\en{ \mas\dd{\xi} } remains decoupled from \en{ \phi\drm{1,in} }:
\small\begin{align}
 \ipg\dd{\mas\dd{\xi}\uu{l} | 
         \phi\drm{1,in}\uu{m} \phi\drm{1,in}\uu{\dag n}}\urm{QND+int} 
=
 \wick{\mas_{\xi}^{l}}
      {(\phi\drm{1,in}\dgg)^{m} (-\phi\drm{1,in})^{n}}\urm{QND+int}
= 
 0.
\qquad 
 \forall l+m+n\in\mathbb{N}
\end{align}\normalsize

\begin{wrapfigure}[0]{r}[53mm]{49mm} 
\vspace{27mm}
\centering
\includegraphics[keepaspectratio,width=48mm]{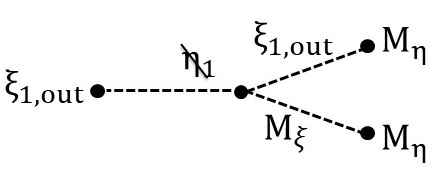}
\caption{
\small
Nonzero term of the transfer function (\ref{fu2-tf}).
\normalsize
}
\label{fig-fusion1}
\end{wrapfigure}

On the other hand,
the output process is influenced by the feedback.
To see this,
we examine a fusion process
\small\begin{align}
 \xi\drm{1,out} \gets \mas_{\xi}+\mas_{\xi}
\label{fu2}
\end{align}\normalsize
with a three-point transfer function
\small\begin{align}
 \ipg\dd{\xi_1 | \mas\dd{\xi} \mas\dd{\xi}}\urm{QND+int}
 (\{t\drm{out}\},\{t, t'\})
&=
 \wick{\xi_1(t\drm{out})\, }
      { \mas\dd{\xi}\ddgg(t)\mas\dd{\xi}\ddgg(t')}\urm{QND+int}
\nn\\ &=
 (-\im)^2 
 \wick{\xi_1(t\drm{out})\ }
      {\mas\dd{\eta}(t)\mas\dd{\eta}(t')}\urm{QND+int}.
\label{fu2-tf}
\end{align}\normalsize
For this process,
only the first order correction depicted in Figure \ref{fig-fusion1}
is nonzero.
Each edge of this diagram is given as
\begin{subequations}
\small\begin{align}
 F_1&\equiv\wick{\xi\drm{1,out}}{\mas\dd{\eta}}\urm{QND} 
= 
 -\im \frac{g_0}{s},
\\
 F_2&\equiv\wick{\mas\dd{\xi}}{\mas\dd{\eta}}\urm{QND} 
= 
 \frac{\im}{s},
\\
 F_3&\equiv\wick{\xi\drm{1,out}}{\mm{\eta}_1}\urm{QND} 
= 
 2\im.
\end{align}\normalsize
\end{subequations}
The three-point transfer function is given as
\small\begin{align}
& \ipg\dd{\xi_1 |\mas\dd{\xi} \mas\dd{\xi}}\urm{QND+int}
       (\{t\drm{out}\},\{t,t'\})
\nn\\ 
&= 2! (-\im)^2  \frac{\im k}{2}
   \int dt_2 \ F_3(t\drm{out},t_2) F_2(t_2,t) F_1(t_2,t')
\nn \\ &=
 2! (-\im)^2  \frac{\im k}{2}
 \int \frac{d\omega_2}{2\pi}\frac{d\omega_1}{2\pi}
 \ex\uu{-\im\omega\dd{2}(t\drm{out} -t)}
 \ex\uu{-\im\omega\dd{1}(t\drm{out} -t')}
 F_2(\omega_2)
 F_1(\omega_1)
 F_3(\omega_2+\omega_1)
%%%%%%%%%%%%%%%%%%%%
\nn \\ &=
 2! g_0 k
 \int \frac{d\omega_2}{2\pi} \ex\uu{-\im\omega\dd{2}(t\drm{out} -t)}
 \frac{\im}{\omega_2+\im\epsilon}
 \int \frac{d\omega_1}{2\pi} \ex\uu{-\im\omega\dd{1}(t\drm{out} -t')}
 \frac{\im}{\omega_1+\im\epsilon}
%%%%%%%%%%%%%%%%%%%%
\nn \\ &=
 2! g_0 k \
 \step(t\drm{out} -t) \ 
 \step(t\drm{out} -t'),
\label{fusionstep}
\end{align}\normalsize
where \en{ 2! } comes from the permutation of 
\en{ \mas_{\eta}(t) } and \en{ \mas_{\eta}(t') }.
We have also used the integral representation of the step function
\small\begin{align}
 \step(t)
&=
 \int \frac{d\omega}{2\pi} 
 \ex\uu{-\im\omega t}\frac{\im}{\omega+\im\epsilon}.
\end{align}\normalsize
In (\ref{fusionstep}), 
the step functions indicate that
the output \en{ \xi\dd{1}(t\drm{out}) } is correlated with 
the system \en{ \mas\dd{\xi}(t)\mas\dd{\xi}(t') }
only when \en{ t\drm{out} > t } and \en{ t\drm{out} > t' }.
This is expected from 
the unidirectionality of the forward traveling field.

%%%%%%%%%%%%
\newpage

We can consider the following process in the same way:
\small\begin{align}
 \xi\drm{1,out} \gets \mas_{\xi} + \mas_{\xi} + \mas_{\xi}.
\end{align}\normalsize
The corresponding probability amplitude is given by
a \en{ (1+3)}-point transfer function
\small\begin{align}
 \ipg\dd{\xi_1 |\mas\dd{\xi} \mas\dd{\xi} \mas\dd{\xi}}\urm{QND+int}
 (\{t\drm{out}\},\{t,t',t''\})
&=
 (-\im)^3 
 \wick{\xi_1(t\drm{out})}
      {\mas\dd{\eta}(t) \mas\dd{\eta}(t') \mas\dd{\eta}(t'')}\urm{QND+int}.
\end{align}\normalsize
In this case, only the following  second order correction is nonzero:

\vspace{1mm}
\begin{figure}[H]
\centering
\includegraphics[keepaspectratio,width=63mm]{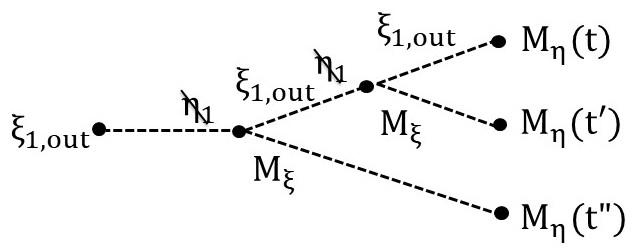}
\end{figure}
\vspace{-1mm}

\noindent
The \en{ (1+3)}-point transfer function is therefore given as
\small\begin{align}
& \ipg\dd{\xi_1 |\mas\dd{\xi} \mas\dd{\xi} \mas\dd{\xi}}\urm{QND+int}
 (\{t\drm{out}\},\{t,t',t''\})
\nn\\
&= 
 3!(-\im)^3 \left(\frac{\im k}{2}\right)^2
 \int dt_3 dt_2\ 
 F_3(t\drm{out},t_3) F_2(t_3,t'') F_3(t_3,t_2) F_2(t_2,t') F_1(t_2,t)
%%%%%%%%%%%%%%%%%%%%
\nn\\ &=
 -3! g_0 k^2 \
 \step(t\drm{out} -t) \
 \step(t\drm{out} -t') \
 \step(t\drm{out} -t'') .
\end{align}\normalsize
Likewise,
for a process
\small\begin{align}
 \xi\drm{1,out}
\gets 
 \underbrace{\mas_{\xi}+\cdots +\mas_{\xi}}\dd{\hspace{6mm}n},
\end{align}\normalsize
a \en{ (1+n)}-point transfer function is given by
(without step functions)
\small\begin{align}
 \ipg\dd{\xi\drm{1,out} |\mas\dd{\xi} \cdots \mas\dd{\xi}}\urm{QND+int}
=
 (-1)^n n! g_0 k^{n-1}.
\label{dffn}
\end{align}\normalsize

Let us consider a fusion process
\small\begin{align}
 \xi\drm{1,out}
\gets
 \exp\left[\mu \mas_{\xi}\right],
\end{align}\normalsize
where \en{ \mu } is a constant such that \en{ |\mu k|<1 }.
Note that 
\small\begin{align}
 \ipg\dd{\xi\drm{1,out} |\mas\dd{\xi}}\urm{QND+int}
=
 -g_0.
\end{align}\normalsize
It follows from (\ref{dffn}) that 
\small\begin{align}
 \ipg\dd{\xi\drm{1,out} |\exp(\mu \mas\dd{\xi})}\urm{QND+int}
&=
 \frac{-\mu g_0}{1+\mu k}.
\end{align}\normalsize

\chapter{Afterword: Gauge field as a system}
\label{epilogue}
\thispagestyle{fancy}

In Chapter \ref{chap:1},
we mentioned that 
the gauge field\index{gauge field}
was regarded as a system
(the wiggly line of Figure \ref{fig-pointlike}).
However,
we have not clearly explained this %in our examples
because 
quantum gates are usually designed 
to operate instanteneously
in quantum computing,
and accordingly,
the corresponding gauge fields
are effectively static.
In this case,
quantum gates are described 
by pointlike interactions\index{pointlike interaction}.
In other words,
the wiggly line shrinks as in Figure \ref{fig-pointlike}.

\begin{wrapfigure}[0]{r}[53mm]{49mm} 
\vspace{-40mm}
\centering
\includegraphics[keepaspectratio,width=23mm]{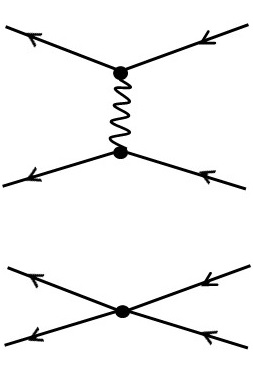}
\caption{
\small
(Upper)
Interaction mediated by a gauge field.
(Lower)
It can be approximated 
by a pointlike interaction
when 
the dynamics of the gauge field is negligible.
\normalsize
}
\label{fig-pointlike}
\end{wrapfigure}

This simplified model is useful as an approximation
unless
quantum gates operate on high-energy particles
or the length scale of the interactions is large
(effective field theory).\index{effective field theory}
In fact,
we have reconstructed existing theories of quantum computing and control 
using this model.
However,
it is not satisfactory in some sense 
because
considering the wiggly line explicitly 
is the essence of the gauge theory.
Here we briefly discuss this issue.

\begin{wrapfigure}[0]{r}[53mm]{49mm} 
\vspace{7mm}
\centering
\includegraphics[keepaspectratio,width=43mm]{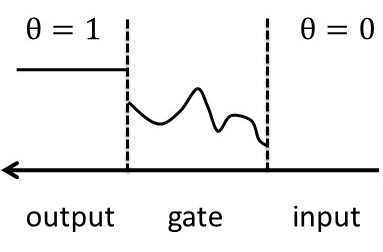}
\caption{
\small
The gauge function is $\win=0$ 
in the domain of the input
and 
$\win=1$
in the domain of the output.
In the domain of the gate,
we can choose other conditions.
\normalsize
}
\label{fig-lorenz}
\end{wrapfigure}

To define the propagator of the wiggly line,
we need to quantize the gauge field.
This is usually done 
by choosing a specific gauge (gauge fixing).\index{gauge fixing}
For example,
in the case of U(1) symmetry,
a standard technique of second quantization
can be applied to the gauge field by choosing 
the Feynman gauge\index{Feynman gauge} 
(or simply the Lorenz gauge).\index{Lorenz gauge}
However, 
this is not straightforward in our case
because
the gauge function serves as a weight function.\index{weight function}

As shown in Section \ref{subsec:passive},
a weight function is \en{ \win=0 } for the input
and \en{ \win=1 } for the output (Figure \ref{fig-lorenz}).
In a domain where the gate operates
(the middle area of Figure \ref{fig-lorenz}),
the input and the output are mixed 
and 
there are degrees of freedom for the choice of the gauge.
The weight function is not necessarily continuous or smooth.
If we choose the Lorenz gauge in this domain,
the quantization of the gauge field is well defined.

\begin{wrapfigure}[0]{r}[53mm]{49mm} 
\vspace{19mm}
\centering
\includegraphics[keepaspectratio,width=30mm]{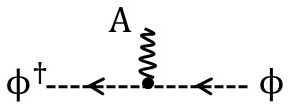}
\caption{
\small
Feynman diagram corresponding to (\ref{u1g-int}).
\normalsize
}
\label{fig-u1g}
\end{wrapfigure}

Let us perform a quick calculation 
for a forward traveling field \en{ \phi }.
Consider a Lagrangian 
\small\begin{align}
 \lag
&=
 \lag\uu{\phi}
+
 \lag\urm{int}
+
 \lag\urm{\gel}.
\end{align}\normalsize
For a U(1) symmetry,
the interaction Lagrangian \en{ \lag\urm{int}}
is written as
\small\begin{align}
 \lag\urm{int}
&=
 g \phi\dgg \gel \phi,
\label{u1g-int}
\end{align}\normalsize
where \en{\gel} is a gauge field and 
\en{ g\in\mathbb{R} } is a coupling constant.
The corresponding diagram is depicted in 
Figure \ref{fig-u1g}.

\newpage

\begin{wrapfigure}[0]{r}[53mm]{49mm} 
\vspace{-0mm}
\centering
\includegraphics[keepaspectratio,width=40mm]{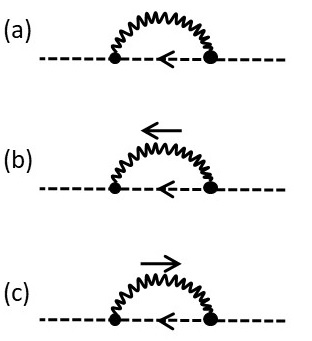}
\caption{
\small
(a) Lowest order correction.
The transfer function of the gauge field
is decomposed into (b) forward and (c) backward traveling parts.
\normalsize
}
\label{fig-u1g-1}
\end{wrapfigure}

We are interested in a transfer function \en{ \phi\gets \phi }
under \en{ \lag }.
The lowest order term 
is depicted in Figure \ref{fig-u1g-1}(a).
The dotted line is 
the forward traveling field \en{ \phi } in free space.
On the other hand,
the wiggly line is the gauge field \en{ \gel }
defined only in a finite interval (of a distance \en{l>0})
and hence described by discrete modes.
It is also represented as in Figure \ref{fig-u1g-1}(b) and (c)
where the wiggly line is decomposed to 
forward and backward components,
respectively.

These diagrams are similar to Section \ref{spin-spin-tf}
where
we have considered interactions between spin 1/2 and SU(2) systems.
If \en{ l \ll 1 },
we can actually calculate the corresponding self-energy in the same way.
According to Section \ref{sec:dfde},
free field transfer functions for \en{ \phi } and \en{ \gel } are,
respectively, 
written as
\begin{subequations}
\small\begin{align}
\hspace{25mm}
 \ipg\dd{\phi|\phi}(s)
&=
 \frac{1}{s + \im k\dd{\phi} + \epsilon},
\qquad
(k\dd{\phi}:\mbox{\normalsize constant \small})
\\
 \ipg\dd{\gel|\gel}(s)
&=
 \frac{1}{sl + \epsilon} 
- 
 \frac{1}{sl - \epsilon}.
\end{align}\normalsize
\end{subequations}
Then it is easy to see
that Figure \ref{fig-u1g-1}(b) and (c), 
respectively, 
correspond to 
\en{ \im \Theta^{(1)} } and \en{ \im \Theta^{(2)} } in (\ref{dd}).
As a result,
the self-energy of this process is given by
\small\begin{align}
 \im \Theta(s)
&=
 \frac{2g^2/l}{s + \im k\dd{\phi} + \epsilon}.
\end{align}\normalsize
The transfer function \en{ \phi\gets \phi } under \en{ \lag }
(inside the gate)
is expressed as  (\ref{su2-s-dyson}):
\small\begin{align}
 \ipg\dd{\phi|\phi}\uu{\lag}
&=
\Bigl[
 1
+
 \ipg\dd{\phi|\phi} \ \im \Theta
\Bigr]\inv 
 \ipg\dd{\phi|\phi}.
\end{align}\normalsize
If the coupling is very weak (\en{ g\ll 1}),
this can be simplified as
\small\begin{align}
 \ipg\dd{\phi|\phi}\uu{\lag}
&\sim
\frac{1}
     {s+\im k\dd{\phi} + 2g^2/l}.
\label{u1-dynamic}
\end{align}\normalsize

\begin{wrapfigure}[0]{r}[53mm]{49mm} 
\vspace{-43mm}
\centering
\includegraphics[keepaspectratio,width=37mm]{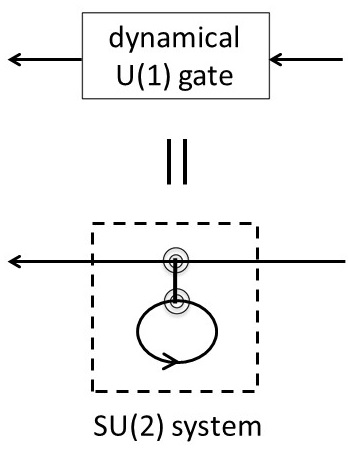}
\caption{
\small
A dynamical U(1) gate (upper)
and
an SU(2) system (lower).
These two are equivalent because
the dyncamics of the gauge field creates 
the same effect as the closed loop.
\normalsize
}
\label{fig-u1-su2}
\end{wrapfigure}

This is the same form as the SU(2) system
(Theorem \ref{thm:tfbs})
\small\begin{align}
\ipg\dd{\mas|\mas}\urm{SU(2)}
&=
\frac{1}
     {s + 2g^2}
=\dtf{-2g^2}{1}{1}{0},
\end{align}\normalsize
with a detuning parameter \en{k\dd{\phi}}. %(Section \ref{detuning}).
In Chapter \ref{chap:feedback},
the SU(2) system has been defined 
by forming a closed loop 
across a static SU(2) gate
(Figure \ref{fig-u1-su2}).
Hence (\ref{u1-dynamic})
indicates that 
the dynamics of the gauge field
creates 
the same effect as the closed loop,
and consequently,
the dynamical U(1) gate ends up being equivalent to 
the SU(2) system.

%%%%%%%%%%%%%%%%%%%%%%%%%%%%%%%%%%%%%%%%%%%%%%%%%%%%%
%%%%%%%%%%%%%%%%%%%%%%%%%%%%%%%%%%%%%%%%%%%%%%%%%%%%%
%%%%%%%%%%%%%%%%%%%%%%%%%%%%%%%%%%%%%%%%%%%%%%%%%%%%%

\chapter*{}

\begin{center}
\begin{Large}
\begin{bfseries}
\textsf{Appendix}
\end{bfseries}
\end{Large}
\end{center}

\appendix

%\addcontentsline{toc}{chapter}{Appendix}
\addcontentsline{toc}{part}{Appendix}

\setcounter{equation}{0} % eq number A.1 
\renewcommand{\theequation}{\Alph{chapter}.\arabic{equation}}
\setcounter{figure}{0} % fig number 
\renewcommand{\thefigure}{\Alph{chapter}.\arabic{figure}}

\chapter{Chain-scattering representation}
\label{csr}
\thispagestyle{fancy}

A chain-scattering representation is commonly used in circuit theory
because it has the advantage of being able to 
express 
the feedback connection of 
multiple-input and multiple-output (MIMO) systems 
as a simple cascade connection.
In this appendix,
we introduce basic properties of this representation
both in the time domain and in the frequency domain.

\section{In the frequency domain}

\begin{wrapfigure}[0]{r}[53mm]{49mm}
\vspace{-5mm}
\centering
\includegraphics[keepaspectratio,width=45mm]{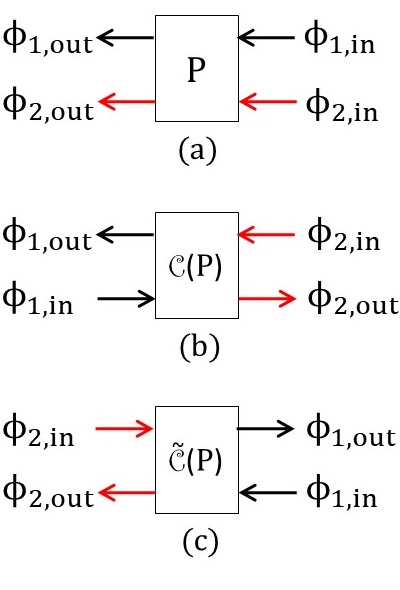}
\caption{
\small
(a) Multi-input and multi-output system.
(b) Chain-scattering representation.
(c) Dual chain-scattering representation.
\normalsize
}
\label{fig-chain1-1}
\end{wrapfigure}

Consider a system \en{ \tf } 
as shown in Figure \ref{fig-chain1-1}(a)
\small\begin{align}
 \AVl{\phi\drm{1,out}}{\phi\drm{2,out}}
=
 \tf
 \AVl{\phi\drm{1,in}}{\phi\drm{2,in}}
\equiv
 \A{\tf_{11}}{\tf_{12}}{\tf_{21}}{\tf_{22}}
 \AVl{\phi\drm{1,in}}{\phi\drm{2,in}}.
\label{sec4.1-1}
\end{align}\normalsize
Let us rearrange the inputs and the outputs
as in Figure \ref{fig-chain1-1}(b)
\small\begin{align}
 \AVl{\phi\drm{1,out}}{\phi\drm{1,in}}
= 
 \chain{\tf}
 \AVl{\phi\drm{2,in}}{\phi\drm{2,out}},
\label{sec4.1-3}
\end{align}\normalsize
where \en{ \chain{\tf} } is called a 
\textit{chain-scattering representation}\index{chain-scattering representation}.
To find the expression of \en{ \chain{\tf} },
assume that \en{ \tf_{21} } is invertible.
Then 
\small\begin{align}
 \phi\drm{1,in}& = \tf_{21}\inv (\phi\drm{2,out} - \tf_{22}\phi\drm{2,in}),
\label{csp-1}
\end{align}\normalsize
from which we have
\small\begin{align}
 \phi\drm{1,out}
&=
 (\tf_{12} - \tf_{11} \tf_{21}\inv \tf_{22})\phi\drm{2,in} 
+ 
 \tf_{11} \tf_{21}\inv \phi\drm{2,out}.
\label{csp-2}
\end{align}\normalsize
It follows from (\ref{csp-1}, \ref{csp-2}) that
\begin{subequations}
\label{sec4.1-2}
\small\begin{align}
\chain{\tf} 
&\equiv
\A{\tf_{12} - \tf_{11} \tf_{21}\inv  \tf_{22}}
  {\tf_{11} \tf_{21}\inv }
  {- \tf_{21}\inv  \tf_{22}}
  {\tf_{21}\inv \rule[0mm]{0mm}{5mm}}
\\ &=
 \A{\tf_{12}}{\tf_{11}}{0}{I}
 \A{I}{0}{\tf_{22}}{\tf_{21}}\inv 
\\ &=
 \A{I}{-\tf_{11}}{0}{\tf_{21}}\inv 
 \A{\tf_{12}}{0}{-\tf_{22}}{I},
\end{align}\normalsize
\end{subequations}
where we have used
\begin{subequations}
\small\begin{align}
 \A{I}{0}{\tf_{22}}{\tf_{21}}\inv 
&=
 \A{I}{0}
   {-\tf_{21}\inv  \tf_{22}}{\tf_{12}\inv },
\\
 \A{I}{-\tf_{11}}{0}{\tf_{21}}\inv 
&=
 \A{I}{\tf_{11}\tf_{21}\inv }
   {0}{\tf_{21}\inv }.
\end{align}\normalsize
\end{subequations}
Conversely,, 
given \en{ \chain{\tf} },
the system \en{ \tf } is expressed as
\begin{subequations}
\label{chaininv}
\small\begin{align}
 \tf&=
 \A{\chain{\tf}_{11}}{\chain{\tf}_{12}}{0}{I}
 \A{\chain{\tf}_{21}}{\chain{\tf}_{22}}{I}{0}\inv
\\ &=
 \A{I}{-\chain{\tf}_{12}}{0}{\chain{\tf}_{22}}\inv
 \A{0}{ \chain{\tf}_{11}}{I}{-\chain{\tf}_{21}}.
\end{align}\normalsize
\end{subequations}

\begin{wrapfigure}[0]{r}[53mm]{49mm} 
\vspace{40mm}
\centering
\includegraphics[keepaspectratio,width=42mm]{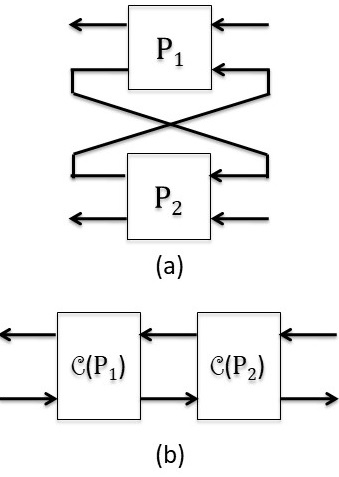}
\caption{
\small
(a) Feedback connection. 
(b) Corresponding chain-scattering representation.
\normalsize
}
\label{fig-chain2-1}
\end{wrapfigure}

Likewise, assuming that \en{ \tf_{12} } is invertible,
we get an alternative form 
shown in Figure \ref{fig-chain1-1}(c):
\small\begin{align}
 \AVl{\phi\drm{2,in}}{\phi\drm{2,out}}
&=
 \dchain{\tf}
 \AVl{\phi\drm{1,out}}{\phi\drm{1,in}},
\label{dualc}
\end{align}\normalsize
where \en{ \dchain{\tf} } is called 
a \textit{dual chain-scattering representation},\index{dual chain-scattering representation}
expressed as
\begin{subequations}
\label{dualcs}
\small\begin{align}
\dchain{\tf}
&\equiv
\A{\tf_{12}\inv }
  {-\tf_{12}\inv \tf_{11}}
  {\tf_{22}\tf_{12}\inv }
  {\tf_{21}-\tf_{22}\tf_{12}\inv \tf_{11}}
\\ &=
 \A{I}{0}{\tf_{22}}{\tf_{21}}
 \A{\tf_{12}}{\tf_{11}}{0}{I}\inv 
\\ &=
 \A{\tf_{12}}{0}{-\tf_{22}}{I}\inv 
 \A{I}{-\tf_{11}}{0}{\tf_{21}}.
\end{align}\normalsize
\end{subequations}
Conversely,
given \en{ \dchain{\tf} },
the system \en{ P } is expressed as
\begin{subequations}
\label{dchaininv}
\small\begin{align}
 P
&=
 \A{I}{0}{\dchain{\tf}_{21}}{\dchain{\tf}_{22}}
 \A{0}{I}{\dchain{\tf}_{11}}{\dchain{\tf}_{12}}\inv 
\\ &=
 \A{ \dchain{\tf}_{11}}{0}{-\dchain{\tf}_{21}}{I}\inv 
 \A{-\dchain{\tf}_{12}}{I}{ \dchain{\tf}_{22}}{0}.
\end{align}\normalsize
\end{subequations}

\begin{remark}
If neither \en{ \tf_{12} } nor \en{ \tf_{21} } is invertible,
the input(output) is extended with virtual inputs(outputs)
so that \en{ \tf_{12}(\tf_{21}) } becomes invertible.
Then a (dual) chain-scattering representation
is defined.
This procedure is called \textit{augmentation}.
\end{remark}

We have considered a feedback connection of two systems 
as shown in Figure \ref{fig-chain2-1}(a)
in Section \ref{sec:fc}.
In the chain-scattering representation,
it can be expressed 
as a simple cascade connection as in Figure \ref{fig-chain2-1}(b).
This property is quite useful.

\newpage

\section{In the time domain}

The state space realization of the chain-scattering representation
 can be readily obtained from a direct calculation.
Assume that the system \en{ \tf } has a state space realization 
\small\begin{align}
 \tf
=
    \left[  \begin{array}{cc}
          \tf_{11} & \tf_{12} \\
          \tf_{21} & \tf_{22} 
        \end{array} \right]
=
    \left[\begin{array}{c|cc}
          A   & B_1    & B_2 \\ \hline
          C_1 & D_{11} & D_{12} \\
          C_2 & D_{21} & D_{22}
    \end{array} \right].
\label{4.17}
\end{align}\normalsize
Then we have
\begin{subequations}
\small\begin{align}
  \chain{\tf}
=
  \left[  
\renewcommand{\arraystretch}{1.2}
  \begin{array}{c|cc}
  A-B_1 D_{21}\inv C_2 & B_2 -B_1 D_{21}\inv D_{22} & B_1 D_{21}\inv  \\
  \hline
  C_1 - D_{11}D_{21}\inv C_2 & D_{12}-D_{11}D_{21}\inv D_{22} &
  D_{11}D_{21}\inv  \\
 -D_{21}\inv C_2 & -D_{21}\inv D_{22} & D_{21}\inv 
 \end{array}     
\renewcommand{\arraystretch}{1}
 \right],
\\
  \dchain{\tf}
=
  \left[  
\renewcommand{\arraystretch}{1.2}
 \begin{array}{c|cc}
  A-B_2 D_{12}\inv C_1 & B_2 D_{12}\inv  & B_1 -B_2 D_{12}\inv D_{11} \\
  \hline
  -D_{12}\inv C_1 & D_{12}\inv  & -D_{12}\inv D_{11} \\
  C_2 - D_{22}D_{12}\inv C_1 & D_{22}D_{12}\inv  & D_{21}-D_{22}D_{12}\inv D_{11}
 \end{array}     
\renewcommand{\arraystretch}{1}
\right].
\end{align}\normalsize
\end{subequations}

%\newpage

\section{Feedback in the chain-scattering representation}

\begin{wrapfigure}[0]{r}[53mm]{49mm} 
\vspace{-0mm}
\centering
\includegraphics[keepaspectratio,width=45mm]{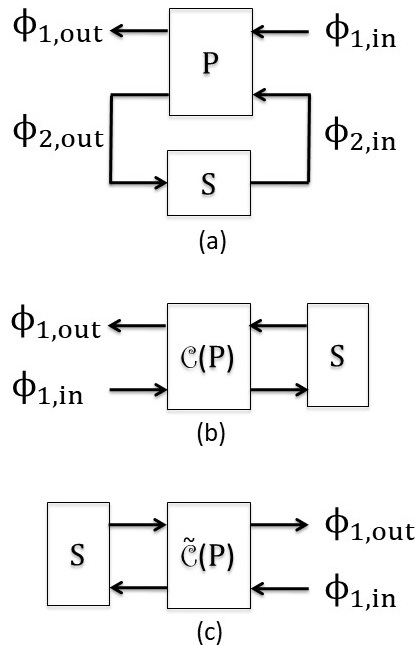}
\caption{
\small
(a) Feedback connection,
(b) Homographic transform,
(c) Dual homographic transform.
\normalsize
}
\label{fig-chain3-1}
\end{wrapfigure}

In this section, 
we consider a feedback process
in which \en{ \phi\drm{2,out} } is fed back to \en{ \phi\drm{2,in} } through another system \en{ S },
as shown in Figure \ref{fig-chain3-1}(a):
\small\begin{align}
 \phi\drm{2,in}&=S\phi\drm{2,out}.
\label{fbc}
\end{align}\normalsize
In the chain-scattering representation,
this corresponds to a homographic transformation.

\subsection{Homographic transformation}

Applying the feedback process (\ref{fbc}) to \en{ \tf },
we get the feedback transfer function as 
\begin{subequations}
\small\begin{align}
 \phi\drm{1,out}
&= 
 [\tf_{11} + \tf_{12} S (1-\tf_{22}S)\inv \tf_{21}]\phi\drm{1,in}
\\ &=
 [\tf_{11} + \tf_{12}(1-S \tf_{22})\inv S \tf_{21}]\phi\drm{1,in}.
\label{fb22}
\end{align}\normalsize
\end{subequations}
In the chain-scattering representation,
(\ref{fbc}) is written as
\begin{subequations}
\small\begin{align}
 \AV{\phi\drm{1,out}}{\phi\drm{1,in}}
&=
 \A{\chain{\tf}_{11}}{\chain{\tf}_{12}}
   {\chain{\tf}_{21}}{\chain{\tf}_{22}} \AV{\phi\drm{2,in}}{\phi\drm{2,out}}
 \\ 
&=
 \A{\chain{\tf}_{11}}{\chain{\tf}_{12}}
   {\chain{\tf}_{21}}{\chain{\tf}_{22}} 
 \AV{S}{1}\phi\drm{2,out}
\equiv
 \AV{Q}{R} \phi\drm{2,out}.
\end{align}\normalsize
\end{subequations}
Eliminating \en{ \phi\drm{2,out} } yields
\small\begin{align}
 \phi\drm{1,out}
&= 
 \homo(\chain{\tf};S) \phi\drm{1,in},
\end{align}\normalsize
where 
\small\begin{align}
 \homo(\chain{\tf};S) 
\equiv
 QR\inv
\label{homog}
\end{align}\normalsize
is called a \textit{homographic transformation}.\index{homographic transformation}

\newpage

Likewise, 
the feedback transfer function is obtained 
in the dual chain-scattering representation as
\small\begin{align}
 \phi\drm{1,out}
&= 
 \dhomo(\dchain{\tf};S) \phi\drm{1,in},
\end{align}\normalsize
where
a \textit{dual homographic transformation}\index{dual homographic transformation}
\en{ \dhomo(\dchain{\tf};S) }
is defined as
\small\begin{align}
 \dhomo(\dchain{\tf};S) 
&= -Q\inv R
\label{homo}
\end{align}\normalsize
with
\small\begin{align}
 \AH{Q}{R}= \AH{I}{-S} \dchain{\tf}.
\end{align}\normalsize
The two transformations are related to each other as
\begin{subequations}
\small\begin{align}
  \phi\drm{1,out} 
&= 
 \homo(\chain{\tf};S) \phi\drm{1,in}
\\ &=
 \dhomo(\dchain{\tf};S) \phi\drm{1,in}.
  \label{4.76}
\end{align}\normalsize
\end{subequations}

\begin{wrapfigure}[0]{r}[53mm]{49mm} 
\vspace{-6mm}
\centering
\includegraphics[keepaspectratio,width=48mm]{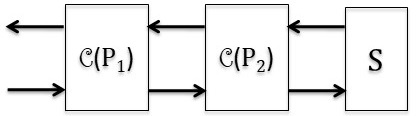}
\caption{
\small
Cascade of feedback connections in 
the chain-scattering representation.
\normalsize
}
\label{fig-chain4}
\end{wrapfigure}

Again, 
the advantage of the homographic transformation is that 
the feedback connection can be visually simplified.
as shown in Figure \ref{fig-chain4}.
In this case, 
the concatenation of two feedback connections
is written as a simple cascade:
\small\begin{align}
 \homo\bigl( \chain{\tf_1}; \homo(\chain{\tf_2};S)\bigr)
&=
 \homo\bigl(\chain{\tf_1} \chain{\tf_2};S\bigr).
\label{lem:4.13.3} 
\end{align}\normalsize

\subsection{State space form of homographic transformation}

The state space realization of the homographic transformation is also 
obtained from a direct calculation.
Assume that 
\en{ \chain{\tf} } and \en{ S } are, respectively, given by
\begin{subequations}
\small\begin{align}
   \chain{\tf}
&=
    \left[  
\renewcommand{\arraystretch}{1.2}
 \begin{array}{c|cc}
            A & B_1 & B_2 \\ \hline
            C_1 & D_{11} & D_{12} \\
            C_2 & D_{21} & D_{22}
        \end{array}
\renewcommand{\arraystretch}{1}
     \right],
\qquad
 S 
= 
 \dtf{A\dd{s}}{B\dd{s}}{C\dd{s}}{D\dd{s}}.
\end{align}\normalsize
\end{subequations}
If \en{ D_2 \equiv D_{21}D\dd{s} +D_{22} } is invertible,
we get
\small\begin{align}
\label{4.83}
    \homo(\chain{\tf};S) =
     \left[  
\renewcommand{\arraystretch}{1.2}
 \begin{array}{c|c}
 \A{A}{B_1 C\dd{s}}
   {0}{A\dd{s}}
- 
 \AV{\hat{B}}
    {B\dd{s}}
 D_{2}\inv 
 \AH{C_2}{D_{21}C\dd{s}}
&        
 \AV{\hat{B}}
    {B\dd{s}}
 D_{2}\inv  
\\ \hline
 C_1-D\dd{c} C_2 \hspace{13mm} (D_{11}-D\dd{c} D_{21})C\dd{s}
&
 D_1 D_{2}\inv 
\end{array}  
\renewcommand{\arraystretch}{1}
 \right],
\end{align}\normalsize
where
\small\begin{align}
 \BV{\hat{B}}{D_1}{D_2}
&\equiv
 \left[
 \begin{array}{cc}
  B_1& B_2 \\
  D_{11}& D_{12}\\
  D_{21}& D_{22}
 \end{array}
\right]
 \AV{D\dd{s}}{I}.
\end{align}\normalsize

\chapter{Pole-zero symmetry for interconnections}
\label{sec:connection}
\thispagestyle{fancy}

In this appendix,
we show that 
when two systems possess the pole-zero symmetry,
the feedforward and feedback connections of them
also possess the same symmetry.
We first consider the cascade connection.

\begin{theorem}
 \label{thm2}
For \en{ \tf\dd{\alpha} \ (\alpha=1,2) } such that
\en{ \tf\dd{\alpha}\Pi \tf\dd{\alpha}\simm = \Pi },
a cascade system \en{ \tf=\tf_1 \tf_2 }
satisfies \en{ \tf\Pi \tf\simm = \Pi }.
\end{theorem}
\begin{proof}
By definition, we have
\small\begin{align}
 \tf \Pi \tf\simm 
&=
 \tf_1 \tf_2 \Pi \tf_2\simm \tf_1\simm
=
 \tf_1 \Pi \tf_1\simm
= \Pi.
\end{align}\normalsize
This establishes the assertion.
\end{proof}

It is not straightforward to show it for the feedback connection.
However, 
it can be simplified 
in the (dual) chain-scattering 
representation\index{chain-scattering representation}\index{chain-scattering representation}
because 
the feedback connection can be expressed as a cascade connection.
%as seen in Appendix \ref{csr}.

\begin{lemma}
\label{lem:cs}
In the dual chain-scattering representation, 
\en{ \tf \Pi \tf\simm=\Pi } is expressed as
\small\begin{align}
 \dchain{\tf} J \Pi \dchain{\tf}\simm =&J \Pi,
\end{align}\normalsize
where
\small\begin{align}
 J&=
\A{I}{}{}{-I}.
\end{align}\normalsize
\end{lemma}
\begin{proof}
Recall that the inverse relation of the dual chain-scattering
 representation has been given in (\ref{dchaininv}):
\small\begin{align}
 \tf &= N\inv M,
\end{align}\normalsize
where
\small\begin{align}
 N &= 
 \A{\dchain{\tf}_{11}}{0}{-\dchain{\tf}_{21}}{1},
\qquad
 M =  
 \A{-\dchain{\tf}_{12}}{1}{\dchain{\tf}_{22}}{0}.
\end{align}\normalsize
It follows from this expression that 
\small\begin{align}
 \tf \Pi \tf\simm=\Pi
\hspace{5mm} \Rightarrow \hspace{5mm}
 M \Pi M\simm 
=  
 N \Pi N\simm,
\end{align}\normalsize
which is rewritten as
\small\begin{align}
& 
\renewcommand{\arraystretch}{1.3}
\A{\dchain{\tf}_{11}\Pi \dchain{\tf}_{11}\simm
   -\dchain{\tf}_{12}\Pi \dchain{\tf}_{12}\simm }
   {\dchain{\tf}_{12}\Pi \dchain{\tf}_{22}\simm
   -\dchain{\tf}_{11}\Pi \dchain{\tf}_{21}\simm }
   {\dchain{\tf}_{22}\Pi \dchain{\tf}_{12}\simm
   -\dchain{\tf}_{21}\Pi \dchain{\tf}_{11}\simm }
   {\dchain{\tf}_{21}\Pi \dchain{\tf}_{21}\simm
   -\dchain{\tf}_{22}\Pi \dchain{\tf}_{12}\simm }
\renewcommand{\arraystretch}{1}
=
 \A{\Pi}{}{}{-\Pi}.
\end{align}\normalsize
Note that the off-diagonal elements are zero.
Changing their sign, 
we get
\small\begin{align}
\A{\dchain{\tf}_{11}}
  {\dchain{\tf}_{12}}
  {\dchain{\tf}_{21}}
  {\dchain{\tf}_{22}}
\A{\Pi}{}{}{-\Pi}
\A{\dchain{\tf}_{11}}
  {\dchain{\tf}_{12}}
  {\dchain{\tf}_{21}}
  {\dchain{\tf}_{22}}\simm
&=
 \A{\Pi}{}{}{-\Pi},
\end{align}\normalsize
which establishes the assertion.
\end{proof}

\begin{wrapfigure}[0]{r}[53mm]{49mm}
\vspace{-10mm}
\centering
\includegraphics[keepaspectratio,width=30mm]{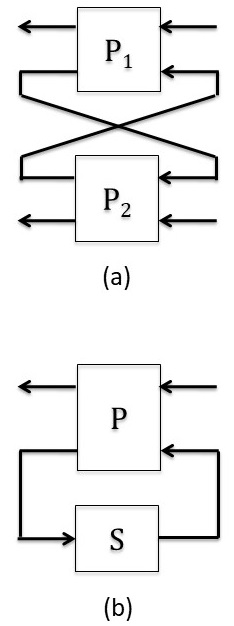}
\caption{
\small
Feedback connections.
\normalsize
}
\label{fig-chain5-1}
\end{wrapfigure}

Now we show the pole-zero symmetry for the feedback connections.

\begin{theorem}
Assume that \en{ \tf\dd{\alpha} \ (\alpha=1,2) } possess
the pole-zero symmetry.
Then the feedback connection of Figure \ref{fig-chain5-1}(a)
also possesses the same symmetry.
\end{theorem}

\begin{proof}
Figure \ref{fig-chain5-1}(a) is expressed as
\en{ \dchain{\tf_1}\dchain{\tf_2} }
in the dual chain-scattering representation.
Each system satisfies
\small\begin{align}
 \dchain{\tf\dd{\alpha}} J\Pi \dchain{\tf\dd{\alpha}}\simm 
&= 
 J\Pi,
\end{align}\normalsize
from which it follows that 
\small\begin{align}
 [\dchain{\tf_1} \dchain{\tf_2}] J \Pi  
 [\dchain{\tf_1} \dchain{\tf_2}]\simm
&= 
 \dchain{\tf_1} \dchain{\tf_2}   J \Pi  
 \dchain{\tf_2}\simm \dchain{\tf_1}\simm
=
 J \Pi.
\end{align}\normalsize
This establishes the assertion.
\end{proof}

\begin{theorem}
\label{thm:fbc}
Assume that \en{ \tf } and \en{ S } possess the pole-zero symmetry.
Then the feedback connection of Figure \ref{fig-chain5-1}(b)
also possesses the same symmetry.
\end{theorem}
\begin{proof}
In the dual chain-scattering representation,
Figure \ref{fig-chain5-1}(b) is expressed as
\small\begin{align}
 \dhomo(\dchain{\tf};S)
&=
 -Q\inv R,
\end{align}\normalsize
where
\small\begin{align}
 \AH{Q}{R}
&=
 \AH{I}{-S}\dchain{\tf}.
\end{align}\normalsize
Since \en{ P } and \en{ S } satisfy
\begin{subequations}
\small\begin{align}
 \dchain{\tf} J\Pi \dchain{\tf}\simm &= J\Pi,
\qquad
 S\Pi S\simm = \Pi,
\end{align}\normalsize
\end{subequations}
we have 
\small\begin{align}
 0=&\AH{I}{-S} J\Pi \AV{I}{-S\simm}
\nonumber \\ =&
 \AH{I}{-S} \dchain{\tf} J\Pi
 \dchain{\tf}\simm \AV{I}{-S\simm}
\nonumber \\ =&
 \AH{Q}{R} J \Pi \AV{Q\simm }{R\simm}
\nonumber \\ =&
 Q\Pi Q\simm - R\Pi R\simm.
\end{align}\normalsize
It follows from (\ref{homo}) that 
\small\begin{align}
 \dhomo(\dchain{\tf};S) \ \Pi \ 
 \dhomo(\dchain{\tf};S)\simm
= \Pi,
\end{align}\normalsize
which completes the assertion.
\end{proof}

\chapter{Quantum stochastic differential equation}
\label{app:qsde}
\thispagestyle{fancy}

For the SU(2) system,
the input-output relation has been considered 
with the Euler-Lagrange equation in Section \ref{sec:sysfor} 
and with \textit{S}-matrices in Section \ref{sec:su2system}.
The same input-output relation can be obtained in another way: 
a quantum stochastic differential equation.
This method is often used in quantum optics
because coherent states are 
approximately described as quantum Wiener processes.
The stochastic approach is less powerful for some problems,
especially when feedback is involved,
but 
it has the advantage of being able to describe a system under measurement
using conditional expectations such as a Kalman filter.
Here we briefly review quantum It\^o calculus and 
show the input-output relation of the SU(2) system.

\section{Quantum Wiener process and It\^o calculus}

Consider a forward traveling field \en{ \phi(t,z) }.
As shown in Section \ref{sec:fpfbasics}, it satisfies
\small\begin{align}
 [\phi(t,z),\phi\dgg(t,z')]
&=
 \delta(z-z').
\end{align}\normalsize
For massless particles, the commutation relation is always rewritten as
\small\begin{align}
 [\phi(t,z),\phi\dgg(t',z)]
&=
 \delta(t-t'),
\label{stocomm}
\end{align}\normalsize
hence we omit the spatial parameter \en{ z } in what follows.
For example, a vacuum state is expressed as
\small\begin{align}
 \phi(t,z)\ket{0}
=
 \phi(t)\ket{0}
&=
 0.
\end{align}\normalsize

Let us define an operator \en{ \Phi(t) } as
\small\begin{align}
 \Phi(t)
&\equiv
 \int\dd{0}\uu{t}d\tau \, \phi(\tau).
\end{align}\normalsize
It follows from (\ref{stocomm}) that 
\small\begin{align}
 [d\Phi(t),d\Phi\dgg(t)]
&=
 \int\dd{t}\uu{t+dt} \hspace{-5mm} d\tau 
 \int\dd{t}\uu{t+dt} \hspace{-5mm} d\sigma
 \ \delta(\tau-\sigma)
=
 dt.
\label{phicom}
\end{align}\normalsize
Usually the properties of fields
are examined through a vacuum expectation.
For example, the Feynman propagator is defined by 
the vacuum expectation of two field operators 
at different points in spacetime.
Keeping this in mind,
we express (\ref{phicom}) as
\begin{subequations}
\label{itope}
\small\begin{align}
 d\Phi(t) d\Phi\dgg(t)
&=
 dt,
\label{itope-1} \\
 d\Phi\dgg(t) d\Phi(t)
&=
 0,
\label{itope-2}
\end{align}\normalsize
\end{subequations}
where \en{ \bra{0}\cdot\ket{0} } is implicitly taken.
This implies that \en{ d\Phi } is of order \en{ \maths{O}(d\Phi)\sim \sqrt{dt} },
which 
corresponds to It\^o's rule for the Wiener process.
In quantum optics,
\en{ \Phi } is considered as a good approximation of vacuum noise
in optical lasers.

\begin{remark}
Note that 
time symmetry is broken in (\ref{itope}).
If there are antiparticles in the field \en{ \phi },
(\ref{itope-2}) is not zero.
Defining \en{ \Phi } as in (\ref{itope}),
we get rid of the negative energy component of the field
and extract the forward traveling component.
Conversely,
if we focus on the negative energy component,
(\ref{itope-2}) is rewritten as
\small\begin{align}
 d\Phi\dgg(t)d\Phi(t)=-dt.
\end{align}\normalsize
For antiparticles, 
the creation operator is in \en{ \phi }.
Redefining \en{ \phi\to \widetilde{\phi}\dgg }.
we have 
\small\begin{align}
 d\widetilde{\Phi}(t)d\widetilde{\Phi}\dgg(t)=-dt.
\end{align}\normalsize
Compared to (\ref{itope}),
the direction of time is reversed. 
\end{remark}

\section{Quantum stochastic differential equation}
\label{sec:qsde}

Let us consider a mode operator \en{ \mas } satisfying 
\small\begin{align}
 [\mas,\mas\dgg]
&=
 1.
\end{align}\normalsize
We regard this as a system
and let it interact with \en{ \Phi } through an interaction Hamiltonian 
that is the same form as the SU(2) system (\ref{su2gaugelag}):
\small\begin{align}
 \ham
&=
 \im 2(g \mas\Phi\dgg - g^* \mas\dgg \Phi).
\end{align}\normalsize
In this case, 
\en{ \Phi } is considered to be an input to the system \en{ \mas }.
The infinitesimal time evolution of the system 
is described by a unitary operator
\begin{subequations}
\label{infunit}
\small\begin{align}
 \uni(dt)
&=
 \exp\left[ -\im \ham dt \right]
\\ &\sim
 1
+
 \left[ 2g \mas d\Phi\dgg - 2g^* \mas\dgg d\Phi \right]
-
 2|g|^2 \mas\dgg \mas dt,
\end{align}\normalsize
\end{subequations}
where we have used (\ref{itope}).
Denoted by \en{ f } is a function of \en{ \mas } and \en{ \mas\dgg }.
Its differential is written as
\begin{subequations}
\small\begin{align}
 df
&=
 \uni\dgg(dt) \, f \, \uni(dt) -f
\\ &=
-
 \left[ 2g \mas d\Phi\dgg - 2g^* \mas\dgg d\Phi , f \right]
\\ & \hspace{4mm} 
+
 2|g|^2
\Bigl(
 2\mas\dgg f \mas - \mas\dgg \mas f - f \mas\dgg \mas
\Bigr)dt.
\end{align}\normalsize
\end{subequations}
For \en{ f=\mas }, this is written as
\small\begin{align}
 d\mas
&=
 -2|g|^2 \mas dt - 2g^* d\Phi.
\label{apsu2in}
\end{align}\normalsize

It is also possible to consider 
the infinitesimal time evolution of \en{ d\Phi }.
The resulting operator \en{ d\Phi' } is considered to be an output:
\begin{subequations}
\label{apsu2out}
\small\begin{align}
 d\Phi'
&=
 \uni\dgg(dt) \, d\Phi \, \uni(dt) 
\\ &=
 d\Phi + 2g\mas dt.
\end{align}\normalsize
\end{subequations}

Consequently,
the input-output relation of the system 
is given by (\ref{apsu2in},\ \ref{apsu2out}):
\small\begin{align}
\kakkon{\dot{M}=-2|g|^2 M  - 2g^* \phi,}
       {\phi'= 2gM + \phi.} 
\end{align}\normalsize
In the frequency domain, this is expressed as
\small\begin{align}
 \phi'
&=
 \dtf{-2|g|^2 M}{- 2g^*}
     {2g}{1}
 \phi
=
 \frac{s-2|g|^2}{s+2|g|^2}\,  \phi,
\label{qsdecav}
\end{align}\normalsize
which is the same as the transfer function of 
the SU(2) system (\ref{su2d}) and (\ref{41ms-2}).

\chapter{Magnetic monopole, Berry phase and gauge transformations}
\label{app:berry}
\thispagestyle{fancy}

Here 
we review the singularity of a vector potential
created by a magnetic monopole.
In Section \ref{subsec:uniLgauge},
we have seen that the quantum gates are activated only when
there is a gap in the gauge field.
The same type of singularity can be seen for a magnetic monopole
through the Berry phase.

\section{Magnetic monopole}
\label{sec:monopole}

Suppose that a magnetic monopole is placed at the origin \en{ \bm{r}=0 }
in \en{ \mathbb{R}^3 }.
By analogy with electric charge,
a magnetic field is assumed to be 
\small\begin{align}
 \bm{B}(\bm{r})
&=
 g\frac{\bm{r}}{r^3}.
\label{monob}
\end{align}\normalsize
A total magnetic flux is given as
\small\begin{align}
 \int \bm{B} \cdot d\bm{S}
=
 4\pi g.
\label{monoflux}
\end{align}\normalsize

\begin{wrapfigure}[0]{r}[53mm]{49mm}
\vspace{-0mm}
\centering
\includegraphics[keepaspectratio,width=45mm]{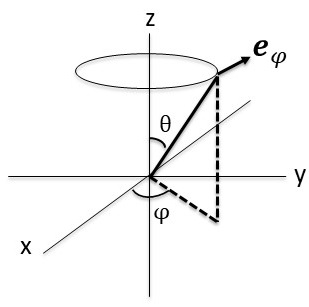}
\caption{
\small
Spherical coordinate system.
\normalsize
}
\label{fig-spherical}
\end{wrapfigure}

We are interested in a vector potential \en{ \bm{A} } 
corresponding to this magnetic field \en{ \bm{B} = \nabla\times \bm{\gel} }.
However, such a potential \en{ \bm{\gel} } is inevitably singular
because if it is well-defined in the entire space,
\en{ \nabla\cdot\bm{B}=\nabla\cdot(\nabla\times\bm{\gel})=0 },
which contradicts (\ref{monob}).
In fact, there are two different solutions for \en{ \bm{\gel} }.
In a spherical coordinate system (Figure \ref{fig-spherical}),
they are written as
\begin{subequations}
\label{mono-2}
\small\begin{align}
 \bm{A}^N(r,\theta,\varphi)
&=
 +\frac{g}{r}\frac{1-\cos \theta}{\sin \theta} \bm{e}\dd{\varphi},
\label{monon}\\
 \bm{A}^S(r,\theta,\varphi)
&=
 -\frac{g}{r}\frac{1+\cos \theta}{\sin \theta} \bm{e}\dd{\varphi}.
\label{monos}
\end{align}\normalsize
\end{subequations}
Note that \en{ \bm{A}^N \ (\bm{A}^S) } is not well-defined on 
the negative (positive) \en{ z }-axis.
To avoid these singularities, 
we regard \en{ \bm{A}^N } and \en{ \bm{A}^S } as solutions
on the north and south hemispheres, respectively.

Note that 
there is a gap in the vector potentials (\ref{mono-2}).
For example, on the equator 
\begin{subequations}
\label{monogap}
\small\begin{align}
 \bm{A}^N(\theta=\frac{\pi}{2}\Bigr)
&=
 +\frac{g}{r} \bm{e}\dd{\varphi},
\\
 \bm{A}^S\Bigl(\theta=\frac{\pi}{2}\Bigr)
&=
 -\frac{g}{r} \bm{e}\dd{\varphi},
\end{align}\normalsize
\end{subequations}
which means that when we cross the equator,
the vector potential `jumps'.
We have seen a similar situation when we 
described quantum gates through the gauge theory
in Section \ref{subsec:uniLgauge},
where 
the gauge field jumped when a quantum gate was turned on.
In fact, 
the two solutions are related to each other 
through a gauge transformation
\small\begin{align}
 \bm{A}^N
=
 \bm{A}^S + \nabla (2g\varphi),
\label{monogauge}
\end{align}\normalsize
which corresponds to a transformation for the gauge field of quantum gates
\small\begin{align}
 \gel'
=
 \gel + \partial\dd{+} \win.
\end{align}\normalsize

The gap \en{ 2g \nabla \varphi } is critical for the magnetic monopole.
To see this,
let us calculate the magnetic flux from the vector potential.
It follows from Stokes' theorem that 
\begin{subequations}
\label{totfluxmono}
\small\begin{align}
 \int_{N+S} d\bm{S} \cdot \bm{B}
&=
 \oint\dd{\textrm{Equator}} d\bm{r} \cdot (\bm{A}^N - \bm{A}^S)
\\ &=
 \oint\dd{\textrm{Equator}} 
 d\bm{r} \cdot ( 2g \nabla \varphi)
\\ &=
 4\pi g,
\end{align}\normalsize
\end{subequations}
which is the same as (\ref{monoflux}).
If there is not a gap in the vector potential,
this is identically zero.

No magnetic monopoles exist in the real world,
but the discussion above is still useful 
when we consider the Berry phase.
In the next sections,
We briefly review 
the Berry phase and discuss a relationship with gauge transformations.

\section{Berry phase}

Consider a set of parameters \en{ \bm{R}=(R_1,R_2,\cdots) }
and a Hamiltonian of the form \en{ \ham=\ham(\bm{R}) }.
Let \en{ E_n(\bm{R}) } be the non-degenerate \en{ n }th eigenvalue:
\small\begin{align}
 \ham(\bm{R})  \ket{n,\bm{R}}
&=
 E_n(\bm{R}) \ket{n, \bm{R}}.
\end{align}\normalsize

Assume that \en{ \bm{R} } is a function of time, \en{ \bm{R}=\bm{R}\dd{t} },
and the initial state is \en{ \ket{n,  \bm{R}\dd{t=0}} }.
If \en{ \bm{R}\dd{t} } changes adiabatically, 
the corresponding energy fluctuation is so small \en{ \Delta E \ll 1 }
that \en{ \Delta E < |E_n-E_{n\pm 1}| }.
In this case,
no particles jump between different energy levels
and 
\en{ \ket{n,\bm{R}\dd{t=0}} } remains the \en{ n }th eigenvector 
of \en{ \ham(\bm{R}\dd{t}) }.
Then the effect of time evolution appears only in the phase factor of the eigenstate:
\small\begin{align}
 \ket{n,\bm{R}\dd{t=0}}
&\to
 \ex\uu{\im\alpha(t)}\ket{n,\bm{R}\dd{t}},
\end{align}\normalsize
where \en{ \alpha(0)=0 }.
Substituting this into the Schr\"{o}dinger equation
\small\begin{align}
 \im \partial\dd{t}\ket{\psi(t)}
&=
 \ham(\bm{R}\dd{t})\ket{\psi(t)},
\label{besh}
\end{align}\normalsize
we have
\small\begin{align}
 \left(
 -\frac{d\alpha}{dt} + \im\partial\dd{t}
 \right) \ket{n,\bm{R}\dd{t}}
&=
 E_n(\bm{R}\dd{t}) \  \ket{n,\bm{R}\dd{t}}
\end{align}\normalsize
Premultiplying by \en{ \bra{n,\bm{R}\dd{t}} } yields
\small\begin{align}
 \frac{d\alpha}{dt} 
&=
- E_n(\bm{R}\dd{t})
+
 \bra{n,\bm{R}\dd{t}} \im \partial\dd{t} \ket{n,\bm{R}\dd{t}}
\end{align}\normalsize
Integrating both sides, we get
\small\begin{align}
 \alpha(t)
&=
 -\int\dd{0}\uu{t} d\sigma \, E_n(\bm{R}\dd{\sigma})
+
 \gamma(t),
\end{align}\normalsize
where 
\small\begin{align}
 \gamma(t)
\equiv
 \int\dd{0}\uu{t} d\sigma 
 \bra{n,\bm{R}\dd{\sigma}} \im \partial\dd{\sigma} \ket{n,\bm{R}\dd{\sigma}}.
\label{berryphase}
\end{align}\normalsize
This is called the \textit{Berry phase}.\index{Berry phase}
Introducing a path \en{ C } from \en{ \bm{R}\dd{t=0} } to \en{ \bm{R}\dd{t} } 
in the \en{ \bm{R} } space,
we can rewrite it as 
\small\begin{align}
 \gamma(t)
=
  \int\dd{C} d\bm{R} \cdot \bm{\gel}\dd{n}(\bm{R}),
\label{berryphase-1}
\end{align}\normalsize
where we have introduced a vector field,
called the \textit{Berry connection},\index{Berry connection} as
\small\begin{align}
 \bm{\gel}\dd{n}(\bm{R})
\equiv
 \bra{n,\bm{R}} \im \partial\dd{\bm{R}} \ket{n,\bm{R}}
\label{def:berrya}
\end{align}\normalsize
Note that the Berry phase (\ref{berryphase-1})
is similar to the magnetic flux (\ref{totfluxmono}).

\section{Berry phase and U(1) symmetry}

The Schr\"{o}dinger equation is always invariant 
under global phase transformations:
\small\begin{align}
 \ket{n,\bm{R}}
\to
 \ket{n,\bm{R}}'
=
 \ex\uu{\im\win\dd{n}}\ket{n,\bm{R}}.
\end{align}\normalsize
If the Hamiltonian \en{ \ham(\bm{R}) } does not involve any derivative 
with respect to \en{ \bm{R} },
it is also invariant under local phase transformations:
\small\begin{align}
 \ket{n,\bm{R}}
\to
 \ket{n,\bm{R}}'
=
 \ex\uu{\im\win\dd{n}(\bm{R})}\ket{n,\bm{R}}.
\label{berryg}
\end{align}\normalsize
In this case, 
substituting (\ref{berryg}) into (\ref{def:berrya}) yields
\small\begin{align}
 \bm{\gel}\dd{n}
\to
 \bm{\gel}\dd{n}'
=
 \bm{\gel}\dd{n}
-
 \partial\dd{\bm{R}}\win\dd{n},
\label{berrya}
\end{align}\normalsize
which can be regarded as a U(1) gauge transformation
in the \en{ \bm{R} } space.
If the \en{ \bm{R} } space is three dimensional,
we define the \textit{Berry curvature}\index{Berry curvature}
in the same way as a magnetic field
\small\begin{align}
\hspace{20mm}
 \bm{B}\dd{n}
\equiv
 \nabla\dd{\bm{R}}\times \bm{\gel}\dd{n},
\qquad
 (\nabla\dd{\bm{R}} \equiv \partial\dd{\bm{R}}.)
\end{align}\normalsize
Obviously, \en{ \bm{B}\dd{n} } is invariant under the local gauge transformation,
i.e.,
\en{ \bm{B}\dd{n}'=\bm{B}\dd{n} }.

Now let us see how the Berry phase transforms.
Substituting (\ref{berrya}) into (\ref{berryphase-1}) yields
\begin{subequations}
\small\begin{align}
 \gamma(t)
\to 
 \gamma'(t)
&=
 \int\dd{C} d\bm{R} \cdot \bm{\gel}\dd{n}'
\\&=
 \int\dd{C} d\bm{R} \cdot (\bm{\gel}\dd{n} - \nabla\dd{\bm{R}}\win\dd{n})
\\ &=
 \gamma(t) -\bigl[ \win(\bm{R}\dd{t}) - \win(\bm{R}\dd{t=0}) \bigr],
\end{align}\normalsize
\end{subequations}
which is not invariant because of the second term.
However, if we choose a closed path \en{ C } so that 
\en{ \bm{R}\dd{t=T}=\bm{R}\dd{t=0} },
it becomes invariant.
As mentioned in the preceding section,
the Berry phase corresponds to the magnetic flux,
hence this is equivalent to the invariance of the magnetic flux
under gauge transformations.
In fact, for a closed path \en{ C },
it follows from Stokes' theorem that
\begin{subequations}
\small\begin{align}
 \gamma(t)
\to 
 \gamma'(T)
=
 \oint\dd{C} d\bm{R} \cdot \bm{\gel}\dd{n}'
&=
 \int\dd{S} d\bm{S} \cdot \bm{B}\dd{n}'
\\ &=
 \int\dd{S} d\bm{S} \cdot \bm{B}\dd{n}
= \gamma(T).
\end{align}\normalsize
\end{subequations}

\section{A relation to the magnetic monopole}

Let us consider an example in a three-dimensional \en{ \bm{R}=(X,Y,Z) } space.
Suppose a Hamiltonian of the form
\small\begin{align}
 \ham(\bm{R})
=
 \bm{\sigma}\cdot \bm{R}
=
 \A{Z}{X-\im Y}{X+\im Y}{-Z},
\end{align}\normalsize
where \en{ \bm{\sigma} } are the Pauli matrices.
The eigenvalues of this Hamiltonian are given by
\small\begin{align}
 \hspace{20mm}
 E\dd{\pm} = \pm R,  \qquad \bigl(R = \sqrt{X^2+Y^2+Z^2}.\bigr)
\end{align}\normalsize
The corresponding eigenvectors are expressed as
\begin{subequations}
\label{set1}
\small\begin{align}
 \ket{\phi\dd{+}^N}
&=
 \frac{1}{\sqrt{2R(R+Z)}} \AV{Z+R}{X+\im Y},
\\
 \ket{\phi\dd{-}^N}
&=
 \frac{1}{\sqrt{2R(R+Z)}} \AV{-X+\im Y}{Z+R},
\end{align}\normalsize
\end{subequations}
which are singular at \en{ Z=-R }.
(We accept a singularity at the origin \en{ R=0 }.)
Another set of eigenvectors is given as
\begin{subequations}
\label{set2}
\small\begin{align}
 \ket{\phi\dd{+}^S}
&=
 \frac{1}{\sqrt{2R(R-Z)}} \AV{X-\im Y}{R-Z},
\\
 \ket{\phi\dd{-}^S}
&=
 \frac{1}{\sqrt{2R(R-Z)}} \AV{-R+Z}{X+\im Y},
\end{align}\normalsize
\end{subequations}
which are singular at \en{ Z=R }.
To avoid these singularities,
we regard \en{ \ket{\phi\dd{\pm}^N} } and \en{ \ket{\phi\dd{\pm}^S} } 
as eigenvectors on the north and south hemispheres, respectively.

This situation is the same as the magnetic monopole 
in Section \ref{sec:monopole}.
In fact, 
(\ref{set1}) and (\ref{set2}) are related to each other 
through a U(1) gauge transformation
\small\begin{align}
\hspace{20mm}
 \ket{\phi\dd{\pm}^S}
=
 \ex\uu{\mp \im \varphi}\ket{\phi\dd{\pm}^N}, 
\qquad
 \left(
 \ex\uu{\mp \im \varphi} = \frac{X \mp \im Y}{\sqrt{X^2+Y^2}}.
\right)
\end{align}\normalsize
It follows from the definition of the Berry connection (\ref{def:berrya})
that
\small\begin{align}
 \bm{\gel}\dd{\pm}^N
=
 \bm{\gel}\dd{\pm}^S \mp \nabla\dd{\bm{R}} \varphi,
\end{align}\normalsize
which is the same form as 
the gauge transformation for the magnetic monopole (\ref{monogauge}).
Accordingly, a total magnetic flux is given in the same way as 
(\ref{totfluxmono}):
\small\begin{align}
% \int_{N+S} d\bm{S}\dd{\bm{R}} \cdot \bm{B}\dd{\pm}
%&=
 \oint\dd{\textrm{Equator}} d\bm{R} \cdot (\bm{A}\dd{\pm}^N - \bm{A}\dd{\pm}^S)
&=
 \oint\dd{\textrm{Equator}} d\bm{R} \cdot (\mp \nabla\dd{\bm{R}} \varphi) 
=
 \mp 2\pi,
\end{align}\normalsize
which means that 
there exists negative (positive) magnetic charge at the origin
in the \en{ \bm{R} } space.

%%%%%%%%%%%%%%%%%%%%%%%%%%%%%%%%%%%%%%%%%%%%%%%%%%%%
%%%%%%%%%%%%%%%%%%%%%%%%%%%%%%%%%%%%%%%%%%%%%%%%%%%%

\printindex


\begin{thebibliography}{99}
\bibitem[Chapter \ref{chap:1}]{1}~

%\noindent
%\textbf{Chapter \ref{chap:1}}
%\vspace{2mm}

\noindent
The system theoretical perspective to the Dyson equation 
was inspired by a personal communication with 
Professor John Doyle at Caltech,
which subsequently led me to the gauge theoretical approach to 
quantum computing and control.
\end{thebibliography}
\end{document}